\documentclass{lmcs}
\usepackage{amssymb,amsmath}
\usepackage{cmll}
\usepackage{txfonts}
\usepackage{graphicx}
\usepackage{stmaryrd}
\usepackage{todonotes}
\usepackage{mathpartir}
\usepackage{hyperref}
\usepackage{mdframed}
\usepackage[all]{xypic}
\usepackage[barr]{xy}

\newtheorem{theorem}{Theorem}
\newtheorem{lemma}[theorem]{Lemma}
\newtheorem{corollary}[theorem]{Corollary}
\newtheorem{definition}[theorem]{Definition}

\newtheorem{remark}[theorem]{Remark}

\let\mto\to
\let\to\relax
\newcommand{\to}{\rightarrow}

\newcommand{\interp}[1]{\llbracket #1 \rrbracket}

\date{}

\let\t\relax
\let\r\relax
\let\c\relax
\let\j\relax
\let\wn\relax
\let\H\relax

\newcommand{\cat}[1]{\mathcal{#1}}
\newcommand{\func}[1]{\mathsf{#1}}
\newcommand{\iso}[0]{\mathsf{iso}}
\newcommand{\H}[0]{\func{H}}
\newcommand{\J}[0]{\func{J}}
\newcommand{\catop}[1]{\cat{#1}^{\mathsf{op}}}
\newcommand{\Hom}[3]{\mathsf{Hom}_{\cat{#1}}(#2,#3)}
\newcommand{\limp}[0]{\multimap}
\newcommand{\colimp}[0]{\multimapdotinv}

\newcommand{\id}[0]{\mathsf{id}}

\newcommand{\m}[1]{\mathsf{m}_{#1}}
\newcommand{\n}[1]{\mathsf{n}_{#1}}

\newcommand{\h}[1]{\mathsf{h}_{#1}}

\newcommand{\t}[0]{\mathsf{t}}
\newcommand{\r}[1]{\mathsf{r}_{#1}}
\newcommand{\s}[1]{\mathsf{s}_{#1}}
\newcommand{\w}[1]{\mathsf{w}_{#1}}
\newcommand{\c}[1]{\mathsf{c}_{#1}}
\newcommand{\j}[1]{\mathsf{j}_{#1}}
\newcommand{\jinv}[1]{\mathsf{j}^{-1}_{#1}}
\newcommand{\wn}[0]{\mathop{?}}
\newcommand{\codiag}[1]{\bigtriangledown_{#1}}

\newenvironment{diagram}{
  \begin{center}
    \begin{math}
      \bfig
}{
      \efig
    \end{math}
  \end{center}
}

\def\BPmessage{Proof Tree (bussproofs) style macros. Version 0.6c.}
\def\EnableBpAbbreviations{%
	\let\AX\Axiom
	\let\AXC\AxiomC
	\let\UI\UnaryInf
	\let\UIC\UnaryInfC
	\let\BI\BinaryInf
	\let\BIC\BinaryInfC
	\let\TI\TrinaryInf
	\let\TIC\TrinaryInfC
	\let\LL\LeftLabel
	\let\RL\RightLabel
	\let\DP\DisplayProof
}

\def\ScoreOverhang{4pt}			
\def\ScoreOverhangLeft{\ScoreOverhang}
\def\ScoreOverhangRight{\ScoreOverhang}

\def\extraVskip{2pt}			
\def\ruleScoreFiller{\hrule}		

\def\defaultScoreFiller{\ruleScoreFiller}  
\def\defaultBuildScore{\buildSingleScore}  

\def\defaultHypSeparation{\hskip.2in}   

\def\labelSpacing{3pt}		

\def\proofSkipAmount{\vskip.8ex plus.8ex minus.4ex}

\def\theHypSeparation{\defaultHypSeparation}
\def\alwaysScoreFiller{\defaultScoreFiller}	
\def\alwaysBuildScore{\defaultBuildScore}
\def\theScoreFiller{\alwaysScoreFiller}	
\def\buildScore{\alwaysBuildScore}   
\def\hypKernAmt{0pt}	

\def\defaultLeftLabel{}
\def\defaultRightLabel{}

\def\myTrue{Y}
\def\bottomAlignFlag{N}
\def\centerAlignFlag{N}

\expandafter\ifx\csname newenvironment\endcsname\relax%
\def\makeatletter{\catcode`\@=11\relax}
\def\makeatother{\catcode`\@=12\relax}
\makeatletter
\def\newcount{\alloc@0\count\countdef\insc@unt}
\def\newdimen{\alloc@1\dimen\dimendef\insc@unt}
\def\newskip{\alloc@2\skip\skipdef\insc@unt}
\def\newbox{\alloc@4\box\chardef\insc@unt}
\makeatother
\else
{\begin{center}\proofSkipAmount \leavevmode}%
{\DisplayProof \proofSkipAmount \end{center} }
\fi
 
\def\thecur#1{\csname#1\number\theLevel\endcsname}

\newcount\theLevel    
\global\theLevel=0		
\newcount\myMaxLevel
\global\myMaxLevel=0
\newbox\myBoxA		
\newbox\myBoxB
\newbox\myBoxC
\newbox\myBoxD
\newbox\myBoxLL		
\newbox\myBoxRL
\newdimen\thisAboveSkip		
\newdimen\thisBelowSkip		
\newdimen\newScoreStart		
\newdimen\newScoreEnd
\newdimen\newCenter
\newdimen\displace
\newdimen\leftLowerAmt
\newdimen\rightLowerAmt
\newdimen\scoreHeight
\newdimen\scoreDepth

\setbox\myBoxLL=\hbox{\defaultLeftLabel}%
\setbox\myBoxRL=\hbox{\defaultRightLabel}%

\def\allocatemore{%
	\ifnum\theLevel>\myMaxLevel%
		\expandafter\newbox\curBox%
		\expandafter\newdimen\curScoreStart%
		\expandafter\newdimen\curCenter%
		\expandafter\newdimen\curScoreEnd%
		\global\advance\myMaxLevel by1%
	\fi%
}

\def\prepAxiom{%
	\advance\theLevel by1%
	\edef\curBox{\thecur{myBox}}%
	\edef\curScoreStart{\thecur{myScoreStart}}%
	\edef\curCenter{\thecur{myCenter}}%
	\edef\curScoreEnd{\thecur{myScoreEnd}}%
	\allocatemore%
}

\def\Axiom$#1\fCenter#2${%
	\prepAxiom%
	\setbox\myBoxA=\hbox{$\mathord{#1}\fCenter\mathord{\relax}$}%
	\setbox\myBoxB=\hbox{$#2$}%
	\global\setbox\curBox=%
	     \hbox{\hskip\ScoreOverhangLeft\relax%
		\unhcopy\myBoxA\unhcopy\myBoxB\hskip\ScoreOverhangRight\relax}%
	\global\curScoreStart=0pt \relax
	\global\curScoreEnd=\wd\curBox \relax
	\global\curCenter=\wd\myBoxA \relax
	\global\advance \curCenter by \ScoreOverhangLeft%
	\ignorespaces
}

\def\AxiomC#1{		
	\prepAxiom%
	\setbox\myBoxA=\hbox{#1}%
	\global\setbox\curBox =%
		\hbox{\hskip\ScoreOverhangLeft\relax%
                        \unhcopy\myBoxA\hskip\ScoreOverhangRight\relax}%
        \global\curScoreStart=0pt \relax
        \global\curScoreEnd=\wd\curBox \relax
        \global\curCenter=.5\wd\curBox \relax
        \global\advance \curCenter by \ScoreOverhangLeft%
	\ignorespaces
}

\def\prepUnary{%
	\ifnum \theLevel<1 
		\errmessage{Hypotheses missing!}
	\fi%
	\edef\curBox{\thecur{myBox}}%
	\edef\curScoreStart{\thecur{myScoreStart}}%
	\edef\curCenter{\thecur{myCenter}}%
	\edef\curScoreEnd{\thecur{myScoreEnd}}%
}

\def\UnaryInf$#1\fCenter#2${%
	\prepUnary%
	\buildConclusion{#1}{#2}%
	\joinUnary%
	\resetInferenceDefaults%
	\ignorespaces%
}

\def\UnaryInfC#1{
	\prepUnary%
	\buildConclusionC{#1}%
	\joinUnary%
	\resetInferenceDefaults%
	\ignorespaces%
}

\def\prepBinary{%
	\ifnum\theLevel<2 
		\errmessage{Hypotheses missing!}
	\fi%
	\edef\rcurBox{\thecur{myBox}}
	\edef\rcurScoreStart{\thecur{myScoreStart}}%
	\edef\rcurCenter{\thecur{myCenter}}%
	\edef\rcurScoreEnd{\thecur{myScoreEnd}}%
	\advance\theLevel by-1
	\edef\lcurBox{\thecur{myBox}}
	\edef\lcurScoreStart{\thecur{myScoreStart}}%
	\edef\lcurCenter{\thecur{myCenter}}%
	\edef\lcurScoreEnd{\thecur{myScoreEnd}}%
}

\def\BinaryInf$#1\fCenter#2${%
	\prepBinary%
	\buildConclusion{#1}{#2}%
	\joinBinary%
	\resetInferenceDefaults%
	\ignorespaces%
}

\def\BinaryInfC#1{%
	\prepBinary%
	\buildConclusionC{#1}%
	\joinBinary%
	\resetInferenceDefaults%
	\ignorespaces%
}

\def\prepTrinary{%
	\ifnum\theLevel<3 
		\errmessage{Hypotheses missing!}
	\fi%
	\edef\rcurBox{\thecur{myBox}}
	\edef\rcurScoreStart{\thecur{myScoreStart}}%
	\edef\rcurCenter{\thecur{myCenter}}%
	\edef\rcurScoreEnd{\thecur{myScoreEnd}}%
	\advance\theLevel by-1
	\edef\ccurBox{\thecur{myBox}}
	\edef\ccurScoreStart{\thecur{myScoreStart}}%
	\edef\ccurCenter{\thecur{myCenter}}%
	\edef\ccurScoreEnd{\thecur{myScoreEnd}}%
	\advance\theLevel by-1
	\edef\lcurBox{\thecur{myBox}}
	\edef\lcurScoreStart{\thecur{myScoreStart}}%
	\edef\lcurCenter{\thecur{myCenter}}%
	\edef\lcurScoreEnd{\thecur{myScoreEnd}}%
}

\def\TrinaryInf$#1\fCenter#2${%
	\prepTrinary%
	\buildConclusion{#1}{#2}%
	\joinTrinary%
	\resetInferenceDefaults%
	\ignorespaces%
}

\def\TrinaryInfC#1{%
	\prepTrinary%
	\buildConclusionC{#1}%
	\joinTrinary%
	\resetInferenceDefaults%
	\ignorespaces%
}

\def\buildConclusion#1#2{
        \setbox\myBoxA=\hbox{$\mathord{#1}\fCenter\mathord{\relax}$}%
        \setbox\myBoxB=\hbox{$#2$}%
	\setbox\myBoxC =%
	      \hbox{\hskip\ScoreOverhangLeft\relax%
		\unhcopy\myBoxA\unhcopy\myBoxB\hskip\ScoreOverhangRight\relax}%
	\newScoreStart=0pt \relax%
	\newCenter=\wd\myBoxA \relax%
	\advance \newCenter by \ScoreOverhangLeft%
	\newScoreEnd=\wd\myBoxC%
}

\def\buildConclusionC#1{
	\setbox\myBoxA=\hbox{#1}%
	\setbox\myBoxC =%
		\hbox{\hbox{\hskip\ScoreOverhangLeft\relax%
                        \unhcopy\myBoxA\hskip\ScoreOverhangRight\relax}}%
	\newScoreStart=0pt \relax%
	\newCenter=.5\wd\myBoxC \relax%
	\newScoreEnd=\wd\myBoxC%
        \advance \newCenter by \ScoreOverhangLeft%
}

\def\joinUnary{
	\global\advance\curCenter by -\hypKernAmt%
	\ifnum\curCenter<\newCenter%
		\displace=\newCenter%
		\advance \displace by -\curCenter%
		\kernUpperBox%
	\else%
		\displace=\curCenter%
		\advance \displace by -\newCenter%
		\kernLowerBox%
	\fi%
        \ifnum \newScoreStart < \curScoreStart %
		\global \curScoreStart = \newScoreStart \fi%
	\ifnum \curScoreEnd < \newScoreEnd %
		\global \curScoreEnd = \newScoreEnd \fi%
	\ifnum \curScoreStart<\wd\myBoxLL%
		\global\displace = \wd\myBoxLL%
		\global\advance\displace by -\curScoreStart%
		\kernUpperBox%
		\kernLowerBox%
	\fi%
	\buildScore%
	\buildScoreLabels%
	\global \setbox \curBox =%
		\vbox{\box\curBox%
			\vskip\thisAboveSkip \relax%
			\nointerlineskip\box\myBoxD%
			\vskip\thisBelowSkip \relax%
			\nointerlineskip\box\myBoxC}%
	\global \curScoreStart=\newScoreStart%
	\global \curScoreEnd=\newScoreEnd%
	\global \curCenter=\newCenter%
}
\def\kernUpperBox{%
		\global\setbox\curBox =%
			\hbox{\hskip\displace\box\curBox}%
		\global\advance \curScoreStart by \displace%
		\global\advance \curScoreEnd by \displace%
		\global\advance\curCenter by \displace%
}
\def\kernLowerBox{%
		\global\setbox\myBoxC =%
			\hbox{\hskip\displace\unhbox\myBoxC}%
		\global\advance \newScoreStart by \displace%
		\global\advance \newScoreEnd by \displace%
		\global\advance\newCenter by \displace%
}

\def\joinBinary{
	\setbox\myBoxA=\hbox{\theHypSeparation}
	\lcurScoreEnd=\rcurScoreEnd%
	\advance\lcurScoreEnd by\wd\lcurBox%
	\advance\lcurScoreEnd by\wd\myBoxA%
	\displace=\lcurScoreEnd%
	\advance\displace by -\lcurScoreStart%
	\lcurCenter=.5\displace%
	\advance\lcurCenter by\lcurScoreStart%
	\setbox\lcurBox=%
		\hbox{\box\lcurBox\unhcopy\myBoxA\box\rcurBox}%
	\displace=\newCenter%
	\advance\displace by -.5\newScoreStart%
	\advance\displace by -.5\newScoreEnd%
	\advance\lcurCenter by \displace%
	\edef\curBox{\lcurBox}%
	\edef\curScoreStart{\lcurScoreStart}%
	\edef\curScoreEnd{\lcurScoreEnd}%
	\edef\curCenter{\lcurCenter}%
	\joinUnary%
}

\def\joinTrinary{
	\setbox\myBoxA=\hbox{\theHypSeparation}
	\lcurScoreEnd=\rcurScoreEnd%
	\advance\lcurScoreEnd by\wd\lcurBox%
	\advance\lcurScoreEnd by\wd\ccurBox%
	\advance\lcurScoreEnd by2\wd\myBoxA%
	\displace=\lcurScoreEnd%
	\advance\displace by -\lcurScoreStart%
	\lcurCenter=.5\displace%
	\advance\lcurCenter by\lcurScoreStart%
	\setbox\lcurBox=%
		\hbox{\box\lcurBox\unhcopy\myBoxA\box\ccurBox%
				  \unhcopy\myBoxA\box\rcurBox}%
	\displace=\newCenter%
	\advance\displace by -.5\newScoreStart%
	\advance\displace by -.5\newScoreEnd%
	\advance\lcurCenter by \displace%
	\edef\curBox{\lcurBox}%
	\edef\curScoreStart{\lcurScoreStart}%
	\edef\curScoreEnd{\lcurScoreEnd}%
	\edef\curCenter{\lcurCenter}%
	\joinUnary%
}

\def\DisplayProof{%
	\ifnum \theLevel=1 \relax \else
		\errmessage{Proof tree badly specified.}%
	\fi%
	\edef\curBox{\thecur{myBox}}%
	\ifx\bottomAlignFlag\myTrue%
		\displace=0pt%
	\else%
		\displace=.5\ht\curBox%
		\ifx\centerAlignFlag\myTrue\relax
		\else%
		      	\advance\displace by -3pt%
		\fi%
	\fi%
	\leavevmode%
	\lower\displace\hbox{\copy\curBox}%
	\global\theLevel=0%
	\global\def\alwaysBuildScore{\defaultBuildScore}
	\global\def\alwaysScoreFiller{\defaultScoreFiller}
	\def\bottomAlignFlag{N}
	\def\centerAlignFlag{N}
	\resetInferenceDefaults%
	\ignorespaces
}

\def\buildSingleScore{
	\displace=\curScoreEnd%
	\advance \displace by -\curScoreStart%
	\global\setbox \myBoxD =%
		\hbox to \displace{\expandafter\xleaders\theScoreFiller\hfill}%
}

\def\buildDoubleScore{
	\buildSingleScore%
	\global\setbox\myBoxD=%
		\hbox{\hbox to0pt{\copy\myBoxD\hss}\raise2pt\copy\myBoxD}%
}

\def\buildNoScore{
	\global\setbox\myBoxD=\hbox{\vbox{\vskip1pt}}%
}

\def\doubleLine{%
	\gdef\buildScore{\buildDoubleScore}
	\ignorespaces
}

\def\noLine{%
	\gdef\buildScore{\buildNoScore}
	\ignorespaces
}

\def\LeftLabel#1{%
	\global\setbox\myBoxLL=\hbox{{#1}\hskip\labelSpacing}%
	\ignorespaces
}
\def\RightLabel#1{%
	\global\setbox\myBoxRL=\hbox{\hskip\labelSpacing #1}%
	\ignorespaces
}

\def\buildScoreLabels{%
	\scoreHeight = \ht\myBoxD%
	\scoreDepth = \dp\myBoxD%
	\leftLowerAmt=\ht\myBoxLL%
	\advance \leftLowerAmt by -\dp\myBoxLL%
	\advance \leftLowerAmt by -\scoreHeight%
	\advance \leftLowerAmt by \scoreDepth%
	\leftLowerAmt=.5\leftLowerAmt%
	\rightLowerAmt=\ht\myBoxRL%
	\advance \rightLowerAmt by -\dp\myBoxRL%
	\advance \rightLowerAmt by -\scoreHeight%
	\advance \rightLowerAmt by \scoreDepth%
	\rightLowerAmt=.5\rightLowerAmt%
	\displace = \curScoreStart%
	\advance\displace by -\wd\myBoxLL%
	\global\setbox\myBoxD =%
		\hbox{\hskip\displace%
			\lower\leftLowerAmt\copy\myBoxLL%
			\box\myBoxD%
			\lower\rightLowerAmt\copy\myBoxRL}%
	\global\thisAboveSkip = \ht\myBoxLL%
	\global\advance \thisAboveSkip by -\leftLowerAmt%
	\global\advance \thisAboveSkip by -\scoreHeight%
	\ifnum \thisAboveSkip<0 %
		\global\thisAboveSkip=0pt%
	\fi%
	\displace = \ht\myBoxRL%
	\advance \displace by -\rightLowerAmt%
	\advance \displace by -\scoreHeight%
	\ifnum \displace<0 %
		\displace=0pt%
	\fi%
	\ifnum \displace>\thisAboveSkip %
		\global\thisAboveSkip=\displace%
	\fi%
	\global\thisBelowSkip = \dp\myBoxLL%
	\global\advance\thisBelowSkip by \leftLowerAmt%
	\global\advance\thisBelowSkip by -\scoreDepth%
	\ifnum\thisBelowSkip<0 %
		\global\thisBelowSkip = 0pt%
	\fi%
	\displace = \dp\myBoxLL%
	\advance\displace by \rightLowerAmt%
	\advance\displace by -\scoreDepth%
	\ifnum\displace<0 %
		\displace = 0pt%
	\fi%
	\ifnum\displace>\thisBelowSkip%
		\global\thisBelowSkip = \displace%
	\fi
	\global\thisAboveSkip = -\thisAboveSkip%
	\global\thisBelowSkip = -\thisBelowSkip%
	\global\advance\thisAboveSkip by\extraVskip
	\global\advance\thisBelowSkip by\extraVskip
}

\def\resetInferenceDefaults{%
	\global\def\theHypSeparation{\defaultHypSeparation}%
	\global\setbox\myBoxLL=\hbox{\defaultLeftLabel}%
	\global\setbox\myBoxRL=\hbox{\defaultRightLabel}%
	\global\def\buildScore{\alwaysBuildScore}%
	\global\def\theScoreFiller{\alwaysScoreFiller}%
	\gdef\hypKernAmt{0pt}
}
\def\limp{\mathrel{-\!\circ}}
\def\lsub{\mathrel{\bullet\!-}}

\newcommand{\DualLNLLogicdrule}[4][]{{\displaystyle\frac{\begin{array}{l}#2\end{array}}{#3}\quad\DualLNLLogicdrulename{#4}}}

\newcommand{\DualLNLLogicpremise}[1]{ #1 \\}
\newenvironment{DualLNLLogicdefnblock}[3][]{ \framebox{\mbox{#2}} \quad #3 \\[0pt]}{}

\newcommand{\DualLNLLogicnt}[1]{\mathit{#1}}
\newcommand{\DualLNLLogicmv}[1]{\mathit{#1}}

\newcommand{\DualLNLLogicsym}[1]{#1}

\newcommand{\DualLNLLogicdrulename}[1]{\textsc{#1}}

\newcommand{\DualLNLLogicdruleCXXidName}[0]{\DualLNLLogicdrulename{C\_id}}
\newcommand{\DualLNLLogicdruleCXXid}[1]{\DualLNLLogicdrule[#1]{%
}{
 \DualLNLLogicnt{S}  \vdash_{\mathsf{C} }  \DualLNLLogicnt{S} }{%
{\DualLNLLogicdruleCXXidName}{}%
}}

\newcommand{\DualLNLLogicdruleCXXwkName}[0]{\DualLNLLogicdrulename{C\_wk}}
\newcommand{\DualLNLLogicdruleCXXwk}[1]{\DualLNLLogicdrule[#1]{%
\DualLNLLogicpremise{ \DualLNLLogicnt{S}  \vdash_{\mathsf{C} }  \Psi }%
}{
 \DualLNLLogicnt{S}  \vdash_{\mathsf{C} }  \DualLNLLogicnt{T}  \DualLNLLogicsym{,}  \Psi }{%
{\DualLNLLogicdruleCXXwkName}{}%
}}

\newcommand{\DualLNLLogicdruleCXXcrName}[0]{\DualLNLLogicdrulename{C\_cr}}
\newcommand{\DualLNLLogicdruleCXXcr}[1]{\DualLNLLogicdrule[#1]{%
\DualLNLLogicpremise{ \DualLNLLogicnt{S}  \vdash_{\mathsf{C} }  \DualLNLLogicnt{T}  \DualLNLLogicsym{,}  \DualLNLLogicnt{T}  \DualLNLLogicsym{,}  \Psi }%
}{
 \DualLNLLogicnt{S}  \vdash_{\mathsf{C} }  \DualLNLLogicnt{T}  \DualLNLLogicsym{,}  \Psi }{%
{\DualLNLLogicdruleCXXcrName}{}%
}}

\newcommand{\DualLNLLogicdruleCXXexName}[0]{\DualLNLLogicdrulename{C\_ex}}
\newcommand{\DualLNLLogicdruleCXXex}[1]{\DualLNLLogicdrule[#1]{%
\DualLNLLogicpremise{ \DualLNLLogicnt{R}  \vdash_{\mathsf{C} }  \Psi_{{\mathrm{1}}}  \DualLNLLogicsym{,}  \DualLNLLogicnt{S}  \DualLNLLogicsym{,}  \DualLNLLogicnt{T}  \DualLNLLogicsym{,}  \Psi_{{\mathrm{2}}} }%
}{
 \DualLNLLogicnt{R}  \vdash_{\mathsf{C} }  \Psi_{{\mathrm{1}}}  \DualLNLLogicsym{,}  \DualLNLLogicnt{T}  \DualLNLLogicsym{,}  \DualLNLLogicnt{S}  \DualLNLLogicsym{,}  \Psi_{{\mathrm{2}}} }{%
{\DualLNLLogicdruleCXXexName}{}%
}}

\newcommand{\DualLNLLogicdruleCXXfLName}[0]{\DualLNLLogicdrulename{C\_fL}}
\newcommand{\DualLNLLogicdruleCXXfL}[1]{\DualLNLLogicdrule[#1]{%
}{
 \DualLNLLogicsym{0}  \vdash_{\mathsf{C} }  \Psi }{%
{\DualLNLLogicdruleCXXfLName}{}%
}}

\newcommand{\DualLNLLogicdruleCXXdLName}[0]{\DualLNLLogicdrulename{C\_dL}}
\newcommand{\DualLNLLogicdruleCXXdL}[1]{\DualLNLLogicdrule[#1]{%
\DualLNLLogicpremise{  \DualLNLLogicnt{T_{{\mathrm{1}}}}  \vdash_{\mathsf{C} }  \Psi_{{\mathrm{1}}}   \quad   \DualLNLLogicnt{T_{{\mathrm{2}}}}  \vdash_{\mathsf{C} }  \Psi_{{\mathrm{2}}}  }%
}{
  \DualLNLLogicnt{T_{{\mathrm{1}}}}  +  \DualLNLLogicnt{T_{{\mathrm{2}}}}   \vdash_{\mathsf{C} }  \Psi_{{\mathrm{1}}}  \DualLNLLogicsym{,}  \Psi_{{\mathrm{2}}} }{%
{\DualLNLLogicdruleCXXdLName}{}%
}}

\newcommand{\DualLNLLogicdruleCXXdROneName}[0]{\DualLNLLogicdrulename{C\_dR1}}
\newcommand{\DualLNLLogicdruleCXXdROne}[1]{\DualLNLLogicdrule[#1]{%
\DualLNLLogicpremise{ \DualLNLLogicnt{R}  \vdash_{\mathsf{C} }  \Psi  \DualLNLLogicsym{,}  \DualLNLLogicnt{T_{{\mathrm{1}}}} }%
}{
 \DualLNLLogicnt{R}  \vdash_{\mathsf{C} }  \Psi  \DualLNLLogicsym{,}   \DualLNLLogicnt{T_{{\mathrm{1}}}}  +  \DualLNLLogicnt{T_{{\mathrm{2}}}}  }{%
{\DualLNLLogicdruleCXXdROneName}{}%
}}

\newcommand{\DualLNLLogicdruleCXXdRTwoName}[0]{\DualLNLLogicdrulename{C\_dR2}}
\newcommand{\DualLNLLogicdruleCXXdRTwo}[1]{\DualLNLLogicdrule[#1]{%
\DualLNLLogicpremise{ \DualLNLLogicnt{R}  \vdash_{\mathsf{C} }  \Psi  \DualLNLLogicsym{,}  \DualLNLLogicnt{T_{{\mathrm{2}}}} }%
}{
 \DualLNLLogicnt{R}  \vdash_{\mathsf{C} }  \Psi  \DualLNLLogicsym{,}   \DualLNLLogicnt{T_{{\mathrm{1}}}}  +  \DualLNLLogicnt{T_{{\mathrm{2}}}}  }{%
{\DualLNLLogicdruleCXXdRTwoName}{}%
}}

\newcommand{\DualLNLLogicdruleCXXsLName}[0]{\DualLNLLogicdrulename{C\_sL}}
\newcommand{\DualLNLLogicdruleCXXsL}[1]{\DualLNLLogicdrule[#1]{%
\DualLNLLogicpremise{ \DualLNLLogicnt{T_{{\mathrm{1}}}}  \vdash_{\mathsf{C} }  \DualLNLLogicnt{T_{{\mathrm{2}}}}  \DualLNLLogicsym{,}  \Psi }%
}{
  \DualLNLLogicnt{T_{{\mathrm{1}}}}  -  \DualLNLLogicnt{T_{{\mathrm{2}}}}   \vdash_{\mathsf{C} }  \Psi }{%
{\DualLNLLogicdruleCXXsLName}{}%
}}

\newcommand{\DualLNLLogicdruleCXXsRName}[0]{\DualLNLLogicdrulename{C\_sR}}
\newcommand{\DualLNLLogicdruleCXXsR}[1]{\DualLNLLogicdrule[#1]{%
\DualLNLLogicpremise{  \DualLNLLogicnt{S}  \vdash_{\mathsf{C} }  \Psi_{{\mathrm{1}}}  \DualLNLLogicsym{,}  \DualLNLLogicnt{T_{{\mathrm{1}}}}   \quad   \DualLNLLogicnt{T_{{\mathrm{2}}}}  \vdash_{\mathsf{C} }  \Psi_{{\mathrm{2}}}  }%
}{
 \DualLNLLogicnt{S}  \vdash_{\mathsf{C} }  \Psi_{{\mathrm{1}}}  \DualLNLLogicsym{,}  \Psi_{{\mathrm{2}}}  \DualLNLLogicsym{,}   \DualLNLLogicnt{T_{{\mathrm{1}}}}  -  \DualLNLLogicnt{T_{{\mathrm{2}}}}  }{%
{\DualLNLLogicdruleCXXsRName}{}%
}}

\newcommand{\DualLNLLogicdruleCXXcutName}[0]{\DualLNLLogicdrulename{C\_cut}}
\newcommand{\DualLNLLogicdruleCXXcut}[1]{\DualLNLLogicdrule[#1]{%
\DualLNLLogicpremise{  \DualLNLLogicnt{S}  \vdash_{\mathsf{C} }  \Psi_{{\mathrm{1}}}  \DualLNLLogicsym{,}  \DualLNLLogicnt{T}   \quad   \DualLNLLogicnt{T}  \vdash_{\mathsf{C} }  \Psi_{{\mathrm{2}}}  }%
}{
 \DualLNLLogicnt{S}  \vdash_{\mathsf{C} }  \Psi_{{\mathrm{1}}}  \DualLNLLogicsym{,}  \Psi_{{\mathrm{2}}} }{%
{\DualLNLLogicdruleCXXcutName}{}%
}}

\newcommand{\DualLNLLogicdruleCXXmcutName}[0]{\DualLNLLogicdrulename{C\_mcut}}
\newcommand{\DualLNLLogicdruleCXXmcut}[1]{\DualLNLLogicdrule[#1]{%
\DualLNLLogicpremise{  \DualLNLLogicnt{S}  \vdash_{\mathsf{C} }  \Psi  \DualLNLLogicsym{,}   \DualLNLLogicnt{S} ^{\, \DualLNLLogicmv{n} }    \quad   \DualLNLLogicnt{S}  \vdash_{\mathsf{C} }  \Psi'  }%
}{
 \DualLNLLogicnt{S}  \vdash_{\mathsf{C} }  \Psi  \DualLNLLogicsym{,}  \Psi' }{%
{\DualLNLLogicdruleCXXmcutName}{}%
}}

\newcommand{\DualLNLLogicdruleCXXhLName}[0]{\DualLNLLogicdrulename{C\_hL}}
\newcommand{\DualLNLLogicdruleCXXhL}[1]{\DualLNLLogicdrule[#1]{%
\DualLNLLogicpremise{ \DualLNLLogicnt{A}  \vdash_{\mathsf{L} }   \cdot   ;  \Psi }%
}{
  \mathsf{H}\, \DualLNLLogicnt{A}   \vdash_{\mathsf{C} }  \Psi }{%
{\DualLNLLogicdruleCXXhLName}{}%
}}

\newcommand{\DualLNLLogicdruleLXXidName}[0]{\DualLNLLogicdrulename{L\_id}}
\newcommand{\DualLNLLogicdruleLXXid}[1]{\DualLNLLogicdrule[#1]{%
}{
 \DualLNLLogicnt{A}  \vdash_{\mathsf{L} }  \DualLNLLogicnt{A}  ;   \cdot  }{%
{\DualLNLLogicdruleLXXidName}{}%
}}

\newcommand{\DualLNLLogicdruleLXXwkName}[0]{\DualLNLLogicdrulename{L\_wk}}
\newcommand{\DualLNLLogicdruleLXXwk}[1]{\DualLNLLogicdrule[#1]{%
\DualLNLLogicpremise{ \DualLNLLogicnt{A}  \vdash_{\mathsf{L} }  \Delta  ;  \Psi }%
}{
 \DualLNLLogicnt{A}  \vdash_{\mathsf{L} }  \Delta  ;  \DualLNLLogicnt{T}  \DualLNLLogicsym{,}  \Psi }{%
{\DualLNLLogicdruleLXXwkName}{}%
}}

\newcommand{\DualLNLLogicdruleLXXctrName}[0]{\DualLNLLogicdrulename{L\_ctr}}
\newcommand{\DualLNLLogicdruleLXXctr}[1]{\DualLNLLogicdrule[#1]{%
\DualLNLLogicpremise{ \DualLNLLogicnt{A}  \vdash_{\mathsf{L} }  \Delta  ;  \DualLNLLogicnt{T}  \DualLNLLogicsym{,}  \DualLNLLogicnt{T}  \DualLNLLogicsym{,}  \Psi }%
}{
 \DualLNLLogicnt{A}  \vdash_{\mathsf{L} }  \Delta  ;  \DualLNLLogicnt{T}  \DualLNLLogicsym{,}  \Psi }{%
{\DualLNLLogicdruleLXXctrName}{}%
}}

\newcommand{\DualLNLLogicdruleLXXexName}[0]{\DualLNLLogicdrulename{L\_ex}}
\newcommand{\DualLNLLogicdruleLXXex}[1]{\DualLNLLogicdrule[#1]{%
\DualLNLLogicpremise{ \DualLNLLogicnt{A}  \vdash_{\mathsf{L} }  \Delta_{{\mathrm{1}}}  \DualLNLLogicsym{,}  \DualLNLLogicnt{A}  \DualLNLLogicsym{,}  \DualLNLLogicnt{B}  \DualLNLLogicsym{,}  \Delta_{{\mathrm{2}}}  ;  \Psi }%
}{
 \DualLNLLogicnt{A}  \vdash_{\mathsf{L} }  \Delta_{{\mathrm{1}}}  \DualLNLLogicsym{,}  \DualLNLLogicnt{B}  \DualLNLLogicsym{,}  \DualLNLLogicnt{A}  \DualLNLLogicsym{,}  \Delta_{{\mathrm{2}}}  ;  \Psi }{%
{\DualLNLLogicdruleLXXexName}{}%
}}

\newcommand{\DualLNLLogicdruleLXXCexName}[0]{\DualLNLLogicdrulename{L\_Cex}}
\newcommand{\DualLNLLogicdruleLXXCex}[1]{\DualLNLLogicdrule[#1]{%
\DualLNLLogicpremise{ \DualLNLLogicnt{A}  \vdash_{\mathsf{L} }  \Delta  ;  \Psi_{{\mathrm{1}}}  \DualLNLLogicsym{,}  \DualLNLLogicnt{S}  \DualLNLLogicsym{,}  \DualLNLLogicnt{T}  \DualLNLLogicsym{,}  \Psi_{{\mathrm{2}}} }%
}{
 \DualLNLLogicnt{A}  \vdash_{\mathsf{L} }  \Delta  ;  \Psi_{{\mathrm{1}}}  \DualLNLLogicsym{,}  \DualLNLLogicnt{T}  \DualLNLLogicsym{,}  \DualLNLLogicnt{S}  \DualLNLLogicsym{,}  \Psi_{{\mathrm{2}}} }{%
{\DualLNLLogicdruleLXXCexName}{}%
}}

\newcommand{\DualLNLLogicdruleLXXcutName}[0]{\DualLNLLogicdrulename{L\_cut}}
\newcommand{\DualLNLLogicdruleLXXcut}[1]{\DualLNLLogicdrule[#1]{%
\DualLNLLogicpremise{  \DualLNLLogicnt{A}  \vdash_{\mathsf{L} }  \Delta_{{\mathrm{1}}}  \DualLNLLogicsym{,}  \DualLNLLogicnt{B}  ;  \Psi_{{\mathrm{1}}}   \quad   \DualLNLLogicnt{B}  \vdash_{\mathsf{L} }  \Delta_{{\mathrm{2}}}  ;  \Psi_{{\mathrm{2}}}  }%
}{
 \DualLNLLogicnt{A}  \vdash_{\mathsf{L} }  \Delta_{{\mathrm{1}}}  \DualLNLLogicsym{,}  \Delta_{{\mathrm{2}}}  ;  \Psi_{{\mathrm{1}}}  \DualLNLLogicsym{,}  \Psi_{{\mathrm{2}}} }{%
{\DualLNLLogicdruleLXXcutName}{}%
}}

\newcommand{\DualLNLLogicdruleLXXCcutName}[0]{\DualLNLLogicdrulename{L\_Ccut}}
\newcommand{\DualLNLLogicdruleLXXCcut}[1]{\DualLNLLogicdrule[#1]{%
\DualLNLLogicpremise{  \DualLNLLogicnt{A}  \vdash_{\mathsf{L} }  \Delta  ;  \Psi_{{\mathrm{1}}}  \DualLNLLogicsym{,}  \DualLNLLogicnt{T}   \quad   \DualLNLLogicnt{T}  \vdash_{\mathsf{C} }  \Psi_{{\mathrm{2}}}  }%
}{
 \DualLNLLogicnt{A}  \vdash_{\mathsf{L} }  \Delta  ;  \Psi_{{\mathrm{1}}}  \DualLNLLogicsym{,}  \Psi_{{\mathrm{2}}} }{%
{\DualLNLLogicdruleLXXCcutName}{}%
}}

\newcommand{\DualLNLLogicdruleLXXflLName}[0]{\DualLNLLogicdrulename{L\_flL}}
\newcommand{\DualLNLLogicdruleLXXflL}[1]{\DualLNLLogicdrule[#1]{%
}{
  \perp   \vdash_{\mathsf{L} }   \cdot   ;   \cdot  }{%
{\DualLNLLogicdruleLXXflLName}{}%
}}

\newcommand{\DualLNLLogicdruleLXXflRName}[0]{\DualLNLLogicdrulename{L\_flR}}
\newcommand{\DualLNLLogicdruleLXXflR}[1]{\DualLNLLogicdrule[#1]{%
\DualLNLLogicpremise{ \DualLNLLogicnt{A}  \vdash_{\mathsf{L} }  \Delta  ;  \Psi }%
}{
 \DualLNLLogicnt{A}  \vdash_{\mathsf{L} }   \perp   \DualLNLLogicsym{,}  \Delta  ;  \Psi }{%
{\DualLNLLogicdruleLXXflRName}{}%
}}

\newcommand{\DualLNLLogicdruleLXXdROneName}[0]{\DualLNLLogicdrulename{L\_dR1}}
\newcommand{\DualLNLLogicdruleLXXdROne}[1]{\DualLNLLogicdrule[#1]{%
\DualLNLLogicpremise{ \DualLNLLogicnt{A}  \vdash_{\mathsf{L} }  \Delta  ;  \Psi  \DualLNLLogicsym{,}  \DualLNLLogicnt{T_{{\mathrm{1}}}} }%
}{
 \DualLNLLogicnt{A}  \vdash_{\mathsf{L} }  \Delta  ;  \Psi  \DualLNLLogicsym{,}   \DualLNLLogicnt{T_{{\mathrm{1}}}}  +  \DualLNLLogicnt{T_{{\mathrm{2}}}}  }{%
{\DualLNLLogicdruleLXXdROneName}{}%
}}

\newcommand{\DualLNLLogicdruleLXXdRTwoName}[0]{\DualLNLLogicdrulename{L\_dR2}}
\newcommand{\DualLNLLogicdruleLXXdRTwo}[1]{\DualLNLLogicdrule[#1]{%
\DualLNLLogicpremise{ \DualLNLLogicnt{A}  \vdash_{\mathsf{L} }  \Delta  ;  \Psi  \DualLNLLogicsym{,}  \DualLNLLogicnt{T_{{\mathrm{2}}}} }%
}{
 \DualLNLLogicnt{A}  \vdash_{\mathsf{L} }  \Delta  ;  \Psi  \DualLNLLogicsym{,}   \DualLNLLogicnt{T_{{\mathrm{1}}}}  +  \DualLNLLogicnt{T_{{\mathrm{2}}}}  }{%
{\DualLNLLogicdruleLXXdRTwoName}{}%
}}

\newcommand{\DualLNLLogicdruleLXXpLName}[0]{\DualLNLLogicdrulename{L\_pL}}
\newcommand{\DualLNLLogicdruleLXXpL}[1]{\DualLNLLogicdrule[#1]{%
\DualLNLLogicpremise{  \DualLNLLogicnt{B_{{\mathrm{1}}}}  \vdash_{\mathsf{L} }  \Delta_{{\mathrm{1}}}  ;  \Psi_{{\mathrm{1}}}   \quad   \DualLNLLogicnt{B_{{\mathrm{2}}}}  \vdash_{\mathsf{L} }  \Delta_{{\mathrm{2}}}  ;  \Psi_{{\mathrm{2}}}  }%
}{
  \DualLNLLogicnt{B_{{\mathrm{1}}}}  \oplus  \DualLNLLogicnt{B_{{\mathrm{2}}}}   \vdash_{\mathsf{L} }  \Delta_{{\mathrm{1}}}  \DualLNLLogicsym{,}  \Delta_{{\mathrm{2}}}  ;  \Psi_{{\mathrm{1}}}  \DualLNLLogicsym{,}  \Psi_{{\mathrm{2}}} }{%
{\DualLNLLogicdruleLXXpLName}{}%
}}

\newcommand{\DualLNLLogicdruleLXXpRName}[0]{\DualLNLLogicdrulename{L\_pR}}
\newcommand{\DualLNLLogicdruleLXXpR}[1]{\DualLNLLogicdrule[#1]{%
\DualLNLLogicpremise{ \DualLNLLogicnt{A}  \vdash_{\mathsf{L} }  \Delta  \DualLNLLogicsym{,}  \DualLNLLogicnt{B}  \DualLNLLogicsym{,}  \DualLNLLogicnt{C}  ;  \Psi }%
}{
 \DualLNLLogicnt{A}  \vdash_{\mathsf{L} }  \Delta  \DualLNLLogicsym{,}   \DualLNLLogicnt{B}  \oplus  \DualLNLLogicnt{C}   ;  \Psi }{%
{\DualLNLLogicdruleLXXpRName}{}%
}}

\newcommand{\DualLNLLogicdruleLXXsLName}[0]{\DualLNLLogicdrulename{L\_sL}}
\newcommand{\DualLNLLogicdruleLXXsL}[1]{\DualLNLLogicdrule[#1]{%
\DualLNLLogicpremise{ \DualLNLLogicnt{B_{{\mathrm{1}}}}  \vdash_{\mathsf{L} }  \DualLNLLogicnt{B_{{\mathrm{2}}}}  \DualLNLLogicsym{,}  \Delta  ;  \Psi }%
}{
  \DualLNLLogicnt{B_{{\mathrm{1}}}}  \colimp  \DualLNLLogicnt{B_{{\mathrm{2}}}}   \vdash_{\mathsf{L} }  \Delta  ;  \Psi }{%
{\DualLNLLogicdruleLXXsLName}{}%
}}

\newcommand{\DualLNLLogicdruleLXXsRName}[0]{\DualLNLLogicdrulename{L\_sR}}
\newcommand{\DualLNLLogicdruleLXXsR}[1]{\DualLNLLogicdrule[#1]{%
\DualLNLLogicpremise{  \DualLNLLogicnt{A}  \vdash_{\mathsf{L} }  \DualLNLLogicnt{B_{{\mathrm{1}}}}  \DualLNLLogicsym{,}  \Delta_{{\mathrm{1}}}  ;  \Psi_{{\mathrm{1}}}   \quad   \DualLNLLogicnt{B_{{\mathrm{2}}}}  \vdash_{\mathsf{L} }  \Delta_{{\mathrm{2}}}  ;  \Psi_{{\mathrm{2}}}  }%
}{
 \DualLNLLogicnt{A}  \vdash_{\mathsf{L} }   \DualLNLLogicnt{B}  \colimp  \DualLNLLogicnt{C}   \DualLNLLogicsym{,}  \Delta_{{\mathrm{1}}}  \DualLNLLogicsym{,}  \Delta_{{\mathrm{2}}}  ;  \Psi_{{\mathrm{1}}}  \DualLNLLogicsym{,}  \Psi_{{\mathrm{2}}} }{%
{\DualLNLLogicdruleLXXsRName}{}%
}}

\newcommand{\DualLNLLogicdruleLXXCsRName}[0]{\DualLNLLogicdrulename{L\_CsR}}
\newcommand{\DualLNLLogicdruleLXXCsR}[1]{\DualLNLLogicdrule[#1]{%
\DualLNLLogicpremise{  \DualLNLLogicnt{A}  \vdash_{\mathsf{L} }  \Delta  ;  \Psi_{{\mathrm{1}}}  \DualLNLLogicsym{,}  \DualLNLLogicnt{T_{{\mathrm{1}}}}   \quad   \DualLNLLogicnt{T_{{\mathrm{2}}}}  \vdash_{\mathsf{C} }  \Psi_{{\mathrm{2}}}  }%
}{
 \DualLNLLogicnt{A}  \vdash_{\mathsf{L} }  \Delta  ;  \Psi_{{\mathrm{1}}}  \DualLNLLogicsym{,}  \Psi_{{\mathrm{2}}}  \DualLNLLogicsym{,}   \DualLNLLogicnt{T_{{\mathrm{1}}}}  -  \DualLNLLogicnt{T_{{\mathrm{2}}}}  }{%
{\DualLNLLogicdruleLXXCsRName}{}%
}}

\newcommand{\DualLNLLogicdruleLXXjLName}[0]{\DualLNLLogicdrulename{L\_jL}}
\newcommand{\DualLNLLogicdruleLXXjL}[1]{\DualLNLLogicdrule[#1]{%
\DualLNLLogicpremise{ \DualLNLLogicnt{T}  \vdash_{\mathsf{C} }  \Psi }%
}{
  \mathsf{J}\, \DualLNLLogicnt{T}   \vdash_{\mathsf{L} }   \cdot   ;  \Psi }{%
{\DualLNLLogicdruleLXXjLName}{}%
}}

\newcommand{\DualLNLLogicdruleLXXjRName}[0]{\DualLNLLogicdrulename{L\_jR}}
\newcommand{\DualLNLLogicdruleLXXjR}[1]{\DualLNLLogicdrule[#1]{%
\DualLNLLogicpremise{ \DualLNLLogicnt{A}  \vdash_{\mathsf{L} }  \Delta  ;  \DualLNLLogicnt{T}  \DualLNLLogicsym{,}  \Psi }%
}{
 \DualLNLLogicnt{A}  \vdash_{\mathsf{L} }  \Delta  \DualLNLLogicsym{,}   \mathsf{J}\, \DualLNLLogicnt{T}   ;  \Psi }{%
{\DualLNLLogicdruleLXXjRName}{}%
}}

\newcommand{\DualLNLLogicdruleLXXhRName}[0]{\DualLNLLogicdrulename{L\_hR}}
\newcommand{\DualLNLLogicdruleLXXhR}[1]{\DualLNLLogicdrule[#1]{%
\DualLNLLogicpremise{ \DualLNLLogicnt{A}  \vdash_{\mathsf{L} }  \Delta  \DualLNLLogicsym{,}  \DualLNLLogicnt{B}  ;  \Psi }%
}{
 \DualLNLLogicnt{A}  \vdash_{\mathsf{L} }  \Delta  ;   \mathsf{H}\, \DualLNLLogicnt{B}   \DualLNLLogicsym{,}  \Psi }{%
{\DualLNLLogicdruleLXXhRName}{}%
}}

\newcommand{\DualLNLLogicdruleLXXCmcutName}[0]{\DualLNLLogicdrulename{L\_Cmcut}}
\newcommand{\DualLNLLogicdruleLXXCmcut}[1]{\DualLNLLogicdrule[#1]{%
\DualLNLLogicpremise{  \DualLNLLogicnt{A}  \vdash_{\mathsf{L} }  \Delta  ;  \Psi  \DualLNLLogicsym{,}   \DualLNLLogicnt{S} ^{\, \DualLNLLogicmv{n} }    \quad   \DualLNLLogicnt{S}  \vdash_{\mathsf{C} }  \Psi'  }%
}{
 \DualLNLLogicnt{A}  \vdash_{\mathsf{L} }  \Delta  ;  \Psi  \DualLNLLogicsym{,}  \Psi' }{%
{\DualLNLLogicdruleLXXCmcutName}{}%
}}

\newcommand{\DualLNLLogicdruleNCXXidName}[0]{\DualLNLLogicdrulename{NC\_id}}
\newcommand{\DualLNLLogicdruleNCXXid}[1]{\DualLNLLogicdrule[#1]{%
}{
 \DualLNLLogicnt{S}  \vdash_{\mathsf{C} }  \DualLNLLogicnt{S} }{%
{\DualLNLLogicdruleNCXXidName}{}%
}}

\newcommand{\DualLNLLogicdruleNCXXzEName}[0]{\DualLNLLogicdrulename{NC\_zE}}
\newcommand{\DualLNLLogicdruleNCXXzE}[1]{\DualLNLLogicdrule[#1]{%
\DualLNLLogicpremise{  \DualLNLLogicnt{S}  \vdash_{\mathsf{C} }  \DualLNLLogicsym{0}  \DualLNLLogicsym{,}  \Psi   \quad   \DualLNLLogicnt{S_{{\mathrm{1}}}}  \vdash_{\mathsf{C} }  \Psi_{{\mathrm{1}}}   \DualLNLLogicsym{,} \, ... \, \DualLNLLogicsym{,}   \DualLNLLogicnt{S_{\DualLNLLogicmv{n}}}  \vdash_{\mathsf{C} }  \Psi_{\DualLNLLogicmv{n}}  }%
}{
 \DualLNLLogicnt{S}  \vdash_{\mathsf{C} }  \Psi  \DualLNLLogicsym{,}  \Psi_{{\mathrm{1}}}  \DualLNLLogicsym{,} \, ... \, \DualLNLLogicsym{,}  \Psi_{\DualLNLLogicmv{n}} }{%
{\DualLNLLogicdruleNCXXzEName}{}%
}}

\newcommand{\DualLNLLogicdruleNCXXdIOneName}[0]{\DualLNLLogicdrulename{NC\_dI1}}
\newcommand{\DualLNLLogicdruleNCXXdIOne}[1]{\DualLNLLogicdrule[#1]{%
\DualLNLLogicpremise{ \DualLNLLogicnt{S}  \vdash_{\mathsf{C} }  \Psi  \DualLNLLogicsym{,}  \DualLNLLogicnt{T_{{\mathrm{1}}}} }%
}{
 \DualLNLLogicnt{S}  \vdash_{\mathsf{C} }  \Psi  \DualLNLLogicsym{,}   \DualLNLLogicnt{T_{{\mathrm{1}}}}  +  \DualLNLLogicnt{T_{{\mathrm{2}}}}  }{%
{\DualLNLLogicdruleNCXXdIOneName}{}%
}}

\newcommand{\DualLNLLogicdruleNCXXdITwoName}[0]{\DualLNLLogicdrulename{NC\_dI2}}
\newcommand{\DualLNLLogicdruleNCXXdITwo}[1]{\DualLNLLogicdrule[#1]{%
\DualLNLLogicpremise{ \DualLNLLogicnt{S}  \vdash_{\mathsf{C} }  \Psi  \DualLNLLogicsym{,}  \DualLNLLogicnt{T_{{\mathrm{2}}}} }%
}{
 \DualLNLLogicnt{S}  \vdash_{\mathsf{C} }  \Psi  \DualLNLLogicsym{,}   \DualLNLLogicnt{T_{{\mathrm{1}}}}  +  \DualLNLLogicnt{T_{{\mathrm{2}}}}  }{%
{\DualLNLLogicdruleNCXXdITwoName}{}%
}}

\newcommand{\DualLNLLogicdruleNCXXdEName}[0]{\DualLNLLogicdrulename{NC\_dE}}
\newcommand{\DualLNLLogicdruleNCXXdE}[1]{\DualLNLLogicdrule[#1]{%
\DualLNLLogicpremise{    \DualLNLLogicnt{S}  \vdash_{\mathsf{C} }  \Psi_{{\mathrm{1}}}  \DualLNLLogicsym{,}   \DualLNLLogicnt{T_{{\mathrm{1}}}}  +  \DualLNLLogicnt{T_{{\mathrm{2}}}}    \quad   \DualLNLLogicnt{T_{{\mathrm{1}}}}  \vdash_{\mathsf{C} }  \Psi_{{\mathrm{2}}}     \quad   \DualLNLLogicnt{T_{{\mathrm{2}}}}  \vdash_{\mathsf{C} }  \Psi_{{\mathrm{2}}}  }%
}{
 \DualLNLLogicnt{S}  \vdash_{\mathsf{C} }  \Psi_{{\mathrm{1}}}  \DualLNLLogicsym{,}  \Psi_{{\mathrm{2}}} }{%
{\DualLNLLogicdruleNCXXdEName}{}%
}}

\newcommand{\DualLNLLogicdruleNCXXsubIName}[0]{\DualLNLLogicdrulename{NC\_subI}}
\newcommand{\DualLNLLogicdruleNCXXsubI}[1]{\DualLNLLogicdrule[#1]{%
\DualLNLLogicpremise{  \DualLNLLogicnt{S}  \vdash_{\mathsf{C} }  \Psi_{{\mathrm{1}}}  \DualLNLLogicsym{,}  \DualLNLLogicnt{T_{{\mathrm{1}}}}   \quad   \DualLNLLogicnt{T_{{\mathrm{2}}}}  \vdash_{\mathsf{C} }  \Psi_{{\mathrm{2}}}  }%
}{
 \DualLNLLogicnt{S}  \vdash_{\mathsf{C} }  \Psi_{{\mathrm{1}}}  \DualLNLLogicsym{,}  \Psi_{{\mathrm{2}}}  \DualLNLLogicsym{,}   \DualLNLLogicnt{T_{{\mathrm{1}}}}  -  \DualLNLLogicnt{T_{{\mathrm{2}}}}  }{%
{\DualLNLLogicdruleNCXXsubIName}{}%
}}

\newcommand{\DualLNLLogicdruleNCXXsubEName}[0]{\DualLNLLogicdrulename{NC\_subE}}
\newcommand{\DualLNLLogicdruleNCXXsubE}[1]{\DualLNLLogicdrule[#1]{%
\DualLNLLogicpremise{  \DualLNLLogicnt{S}  \vdash_{\mathsf{C} }  \Psi_{{\mathrm{1}}}  \DualLNLLogicsym{,}   \DualLNLLogicnt{T_{{\mathrm{1}}}}  -  \DualLNLLogicnt{T_{{\mathrm{2}}}}    \quad   \DualLNLLogicnt{T_{{\mathrm{1}}}}  \vdash_{\mathsf{C} }  \DualLNLLogicnt{T_{{\mathrm{2}}}}  \DualLNLLogicsym{,}  \Psi_{{\mathrm{2}}}  }%
}{
 \DualLNLLogicnt{S}  \vdash_{\mathsf{C} }  \Psi_{{\mathrm{1}}}  \DualLNLLogicsym{,}  \Psi_{{\mathrm{2}}} }{%
{\DualLNLLogicdruleNCXXsubEName}{}%
}}

\newcommand{\DualLNLLogicdruleNCXXHEName}[0]{\DualLNLLogicdrulename{NC\_HE}}
\newcommand{\DualLNLLogicdruleNCXXHE}[1]{\DualLNLLogicdrule[#1]{%
\DualLNLLogicpremise{  \DualLNLLogicnt{S}  \vdash_{\mathsf{C} }  \Psi_{{\mathrm{1}}}  \DualLNLLogicsym{,}   \mathsf{H}\, \DualLNLLogicnt{A}    \quad   \DualLNLLogicnt{A}  \vdash_{\mathsf{L} }   \cdot  ; \Psi_{{\mathrm{2}}}  }%
}{
 \DualLNLLogicnt{S}  \vdash_{\mathsf{C} }  \Psi_{{\mathrm{1}}}  \DualLNLLogicsym{,}  \Psi_{{\mathrm{2}}} }{%
{\DualLNLLogicdruleNCXXHEName}{}%
}}

\newcommand{\DualLNLLogicdruleNCXXweakName}[0]{\DualLNLLogicdrulename{NC\_weak}}
\newcommand{\DualLNLLogicdruleNCXXweak}[1]{\DualLNLLogicdrule[#1]{%
\DualLNLLogicpremise{ \DualLNLLogicnt{S}  \vdash_{\mathsf{C} }  \Psi }%
}{
 \DualLNLLogicnt{S}  \vdash_{\mathsf{C} }  \DualLNLLogicnt{T}  \DualLNLLogicsym{,}  \Psi }{%
{\DualLNLLogicdruleNCXXweakName}{}%
}}

\newcommand{\DualLNLLogicdruleNCXXcontrName}[0]{\DualLNLLogicdrulename{NC\_contr}}
\newcommand{\DualLNLLogicdruleNCXXcontr}[1]{\DualLNLLogicdrule[#1]{%
\DualLNLLogicpremise{ \DualLNLLogicnt{S}  \vdash_{\mathsf{C} }  \DualLNLLogicnt{T}  \DualLNLLogicsym{,}  \DualLNLLogicnt{T}  \DualLNLLogicsym{,}  \Psi }%
}{
 \DualLNLLogicnt{S}  \vdash_{\mathsf{C} }  \DualLNLLogicnt{T}  \DualLNLLogicsym{,}  \Psi }{%
{\DualLNLLogicdruleNCXXcontrName}{}%
}}

\newcommand{\DualLNLLogicdruleNCXXcutName}[0]{\DualLNLLogicdrulename{NC\_cut}}
\newcommand{\DualLNLLogicdruleNCXXcut}[1]{\DualLNLLogicdrule[#1]{%
\DualLNLLogicpremise{  \DualLNLLogicnt{S}  \vdash_{\mathsf{C} }  \Psi_{{\mathrm{1}}}  \DualLNLLogicsym{,}  \DualLNLLogicnt{T}   \quad   \DualLNLLogicnt{T}  \vdash_{\mathsf{C} }  \Psi_{{\mathrm{2}}}  }%
}{
 \DualLNLLogicnt{S}  \vdash_{\mathsf{C} }  \Psi_{{\mathrm{1}}}  \DualLNLLogicsym{,}  \Psi_{{\mathrm{2}}} }{%
{\DualLNLLogicdruleNCXXcutName}{}%
}}

\newcommand{\DualLNLLogicdruleNLXXidName}[0]{\DualLNLLogicdrulename{NL\_id}}
\newcommand{\DualLNLLogicdruleNLXXid}[1]{\DualLNLLogicdrule[#1]{%
}{
 \DualLNLLogicnt{A}  \vdash_{\mathsf{L} }  \DualLNLLogicnt{A} ;  \cdot  }{%
{\DualLNLLogicdruleNLXXidName}{}%
}}

\newcommand{\DualLNLLogicdruleNLXXpIName}[0]{\DualLNLLogicdrulename{NL\_pI}}
\newcommand{\DualLNLLogicdruleNLXXpI}[1]{\DualLNLLogicdrule[#1]{%
\DualLNLLogicpremise{ \DualLNLLogicnt{A}  \vdash_{\mathsf{L} }  \Delta ; \Psi }%
}{
 \DualLNLLogicnt{A}  \vdash_{\mathsf{L} }  \Delta  \DualLNLLogicsym{,}   \perp  ; \Psi }{%
{\DualLNLLogicdruleNLXXpIName}{}%
}}

\newcommand{\DualLNLLogicdruleNLXXpEName}[0]{\DualLNLLogicdrulename{NL\_pE}}
\newcommand{\DualLNLLogicdruleNLXXpE}[1]{\DualLNLLogicdrule[#1]{%
\DualLNLLogicpremise{ \DualLNLLogicnt{A}  \vdash_{\mathsf{L} }   \perp   \DualLNLLogicsym{,}  \Delta ;  \cdot  }%
}{
 \DualLNLLogicnt{A}  \vdash_{\mathsf{L} }  \Delta ;  \cdot  }{%
{\DualLNLLogicdruleNLXXpEName}{}%
}}

\newcommand{\DualLNLLogicdruleNLXXparIName}[0]{\DualLNLLogicdrulename{NL\_parI}}
\newcommand{\DualLNLLogicdruleNLXXparI}[1]{\DualLNLLogicdrule[#1]{%
\DualLNLLogicpremise{ \DualLNLLogicnt{A}  \vdash_{\mathsf{L} }  \Delta  \DualLNLLogicsym{,}  \DualLNLLogicnt{B_{{\mathrm{1}}}}  \DualLNLLogicsym{,}  \DualLNLLogicnt{B_{{\mathrm{2}}}} ; \Psi }%
}{
 \DualLNLLogicnt{A}  \vdash_{\mathsf{L} }  \Delta  \DualLNLLogicsym{,}   \DualLNLLogicnt{B_{{\mathrm{1}}}}  \oplus  \DualLNLLogicnt{B_{{\mathrm{2}}}}  ; \Psi }{%
{\DualLNLLogicdruleNLXXparIName}{}%
}}

\newcommand{\DualLNLLogicdruleNLXXparEName}[0]{\DualLNLLogicdrulename{NL\_parE}}
\newcommand{\DualLNLLogicdruleNLXXparE}[1]{\DualLNLLogicdrule[#1]{%
\DualLNLLogicpremise{    \DualLNLLogicnt{A}  \vdash_{\mathsf{L} }  \Delta  \DualLNLLogicsym{,}   \DualLNLLogicnt{B_{{\mathrm{1}}}}  \oplus  \DualLNLLogicnt{B_{{\mathrm{2}}}}  ; \Psi   \quad   \DualLNLLogicnt{B_{{\mathrm{1}}}}  \vdash_{\mathsf{L} }  \Delta_{{\mathrm{1}}} ; \Psi_{{\mathrm{1}}}     \quad   \DualLNLLogicnt{B_{{\mathrm{2}}}}  \vdash_{\mathsf{L} }  \Delta_{{\mathrm{2}}} ; \Psi_{{\mathrm{2}}}  }%
}{
 \DualLNLLogicnt{A}  \vdash_{\mathsf{L} }  \Delta  \DualLNLLogicsym{,}  \Delta_{{\mathrm{1}}}  \DualLNLLogicsym{,}  \Delta_{{\mathrm{2}}} ; \Psi  \DualLNLLogicsym{,}  \Psi_{{\mathrm{1}}}  \DualLNLLogicsym{,}  \Psi_{{\mathrm{2}}} }{%
{\DualLNLLogicdruleNLXXparEName}{}%
}}

\newcommand{\DualLNLLogicdruleNLXXsubIName}[0]{\DualLNLLogicdrulename{NL\_subI}}
\newcommand{\DualLNLLogicdruleNLXXsubI}[1]{\DualLNLLogicdrule[#1]{%
\DualLNLLogicpremise{  \DualLNLLogicnt{A}  \vdash_{\mathsf{L} }  \Delta_{{\mathrm{1}}}  \DualLNLLogicsym{,}  \DualLNLLogicnt{B_{{\mathrm{1}}}} ; \Psi_{{\mathrm{1}}}   \quad   \DualLNLLogicnt{B_{{\mathrm{2}}}}  \vdash_{\mathsf{L} }  \Delta_{{\mathrm{2}}} ; \Psi_{{\mathrm{2}}}  }%
}{
 \DualLNLLogicnt{A}  \vdash_{\mathsf{L} }   \DualLNLLogicnt{B_{{\mathrm{1}}}}  \colimp  \DualLNLLogicnt{B_{{\mathrm{2}}}}   \DualLNLLogicsym{,}  \Delta_{{\mathrm{1}}}  \DualLNLLogicsym{,}  \Delta_{{\mathrm{2}}} ; \Psi_{{\mathrm{1}}}  \DualLNLLogicsym{,}  \Psi_{{\mathrm{2}}} }{%
{\DualLNLLogicdruleNLXXsubIName}{}%
}}

\newcommand{\DualLNLLogicdruleNLXXsubEName}[0]{\DualLNLLogicdrulename{NL\_subE}}
\newcommand{\DualLNLLogicdruleNLXXsubE}[1]{\DualLNLLogicdrule[#1]{%
\DualLNLLogicpremise{  \DualLNLLogicnt{A}  \vdash_{\mathsf{L} }  \Delta_{{\mathrm{1}}}  \DualLNLLogicsym{,}   \DualLNLLogicnt{B_{{\mathrm{1}}}}  \colimp  \DualLNLLogicnt{B_{{\mathrm{2}}}}  ; \Psi_{{\mathrm{1}}}   \quad   \DualLNLLogicnt{B_{{\mathrm{1}}}}  \vdash_{\mathsf{L} }  \DualLNLLogicnt{B_{{\mathrm{1}}}}  \DualLNLLogicsym{,}  \Delta_{{\mathrm{2}}} ; \Psi_{{\mathrm{2}}}  }%
}{
 \DualLNLLogicnt{A}  \vdash_{\mathsf{L} }  \Delta_{{\mathrm{1}}}  \DualLNLLogicsym{,}  \Delta_{{\mathrm{2}}} ; \Psi_{{\mathrm{1}}}  \DualLNLLogicsym{,}  \Psi_{{\mathrm{2}}} }{%
{\DualLNLLogicdruleNLXXsubEName}{}%
}}

\newcommand{\DualLNLLogicdruleNLXXJIName}[0]{\DualLNLLogicdrulename{NL\_JI}}
\newcommand{\DualLNLLogicdruleNLXXJI}[1]{\DualLNLLogicdrule[#1]{%
\DualLNLLogicpremise{ \DualLNLLogicnt{A}  \vdash_{\mathsf{L} }  \Delta ; \DualLNLLogicnt{T}  \DualLNLLogicsym{,}  \Psi }%
}{
 \DualLNLLogicnt{A}  \vdash_{\mathsf{L} }  \Delta  \DualLNLLogicsym{,}   \mathsf{J}\, \DualLNLLogicnt{T}  ; \Psi }{%
{\DualLNLLogicdruleNLXXJIName}{}%
}}

\newcommand{\DualLNLLogicdruleNLXXJEName}[0]{\DualLNLLogicdrulename{NL\_JE}}
\newcommand{\DualLNLLogicdruleNLXXJE}[1]{\DualLNLLogicdrule[#1]{%
\DualLNLLogicpremise{  \DualLNLLogicnt{A}  \vdash_{\mathsf{L} }  \Delta  \DualLNLLogicsym{,}   \mathsf{J}\, \DualLNLLogicnt{T}  ; \Psi_{{\mathrm{1}}}   \quad   \DualLNLLogicnt{T}  \vdash_{\mathsf{C} }  \Psi_{{\mathrm{2}}}  }%
}{
 \DualLNLLogicnt{A}  \vdash_{\mathsf{L} }  \Delta ; \Psi_{{\mathrm{1}}}  \DualLNLLogicsym{,}  \Psi_{{\mathrm{2}}} }{%
{\DualLNLLogicdruleNLXXJEName}{}%
}}

\newcommand{\DualLNLLogicdruleNLXXHIName}[0]{\DualLNLLogicdrulename{NL\_HI}}
\newcommand{\DualLNLLogicdruleNLXXHI}[1]{\DualLNLLogicdrule[#1]{%
\DualLNLLogicpremise{ \DualLNLLogicnt{A}  \vdash_{\mathsf{L} }  \Delta  \DualLNLLogicsym{,}  \DualLNLLogicnt{B} ; \Psi }%
}{
 \DualLNLLogicnt{A}  \vdash_{\mathsf{L} }  \Delta ;  \mathsf{H}\, \DualLNLLogicnt{B}   \DualLNLLogicsym{,}  \Psi }{%
{\DualLNLLogicdruleNLXXHIName}{}%
}}

\newcommand{\DualLNLLogicdruleNLXXHEName}[0]{\DualLNLLogicdrulename{NL\_HE}}
\newcommand{\DualLNLLogicdruleNLXXHE}[1]{\DualLNLLogicdrule[#1]{%
\DualLNLLogicpremise{  \DualLNLLogicnt{A}  \vdash_{\mathsf{L} }  \Delta ; \Psi_{{\mathrm{1}}}  \DualLNLLogicsym{,}   \mathsf{H}\, \DualLNLLogicnt{A}    \quad   \DualLNLLogicnt{A}  \vdash_{\mathsf{L} }   \cdot  ; \Psi_{{\mathrm{2}}}  }%
}{
 \DualLNLLogicnt{A}  \vdash_{\mathsf{L} }  \Delta ; \Psi_{{\mathrm{1}}}  \DualLNLLogicsym{,}  \Psi_{{\mathrm{2}}} }{%
{\DualLNLLogicdruleNLXXHEName}{}%
}}

\newcommand{\DualLNLLogicdruleNLXXweakName}[0]{\DualLNLLogicdrulename{NL\_weak}}
\newcommand{\DualLNLLogicdruleNLXXweak}[1]{\DualLNLLogicdrule[#1]{%
\DualLNLLogicpremise{ \DualLNLLogicnt{A}  \vdash_{\mathsf{L} }  \Delta ; \Psi }%
}{
 \DualLNLLogicnt{A}  \vdash_{\mathsf{L} }  \Delta ; \DualLNLLogicnt{T}  \DualLNLLogicsym{,}  \Psi }{%
{\DualLNLLogicdruleNLXXweakName}{}%
}}

\newcommand{\DualLNLLogicdruleNLXXcontrName}[0]{\DualLNLLogicdrulename{NL\_contr}}
\newcommand{\DualLNLLogicdruleNLXXcontr}[1]{\DualLNLLogicdrule[#1]{%
\DualLNLLogicpremise{ \DualLNLLogicnt{A}  \vdash_{\mathsf{L} }  \Delta ; \DualLNLLogicnt{T}  \DualLNLLogicsym{,}  \DualLNLLogicnt{T}  \DualLNLLogicsym{,}  \Psi }%
}{
 \DualLNLLogicnt{A}  \vdash_{\mathsf{L} }  \Delta ; \DualLNLLogicnt{T}  \DualLNLLogicsym{,}  \Psi }{%
{\DualLNLLogicdruleNLXXcontrName}{}%
}}

\newcommand{\DualLNLLogicdruleNLXXCcutName}[0]{\DualLNLLogicdrulename{NL\_Ccut}}
\newcommand{\DualLNLLogicdruleNLXXCcut}[1]{\DualLNLLogicdrule[#1]{%
\DualLNLLogicpremise{  \DualLNLLogicnt{A}  \vdash_{\mathsf{L} }  \Delta ; \Psi_{{\mathrm{1}}}  \DualLNLLogicsym{,}  \DualLNLLogicnt{T}   \quad   \DualLNLLogicnt{T}  \vdash_{\mathsf{C} }  \Psi_{{\mathrm{2}}}  }%
}{
 \DualLNLLogicnt{A}  \vdash_{\mathsf{L} }  \Delta ; \Psi_{{\mathrm{1}}}  \DualLNLLogicsym{,}  \Psi_{{\mathrm{2}}} }{%
{\DualLNLLogicdruleNLXXCcutName}{}%
}}

\newcommand{\DualLNLLogicdruleNLXXcutName}[0]{\DualLNLLogicdrulename{NL\_cut}}
\newcommand{\DualLNLLogicdruleNLXXcut}[1]{\DualLNLLogicdrule[#1]{%
\DualLNLLogicpremise{  \DualLNLLogicnt{A}  \vdash_{\mathsf{L} }  \Delta_{{\mathrm{1}}}  \DualLNLLogicsym{,}  \DualLNLLogicnt{B} ; \Psi_{{\mathrm{1}}}   \quad   \DualLNLLogicnt{B}  \vdash_{\mathsf{L} }  \Delta_{{\mathrm{2}}} ; \Psi_{{\mathrm{2}}}  }%
}{
 \DualLNLLogicnt{A}  \vdash_{\mathsf{L} }  \Delta_{{\mathrm{1}}}  \DualLNLLogicsym{,}  \Delta_{{\mathrm{2}}} ; \Psi_{{\mathrm{1}}}  \DualLNLLogicsym{,}  \Psi_{{\mathrm{2}}} }{%
{\DualLNLLogicdruleNLXXcutName}{}%
}}

\newcommand{\DualLNLLogicdruleTCXXidName}[0]{\DualLNLLogicdrulename{TC\_id}}
\newcommand{\DualLNLLogicdruleTCXXid}[1]{\DualLNLLogicdrule[#1]{%
}{
 \DualLNLLogicmv{x}  :  \DualLNLLogicnt{S}  \vdash_{\mathsf{C} }  \DualLNLLogicmv{x}  \DualLNLLogicsym{:}  \DualLNLLogicnt{S} }{%
{\DualLNLLogicdruleTCXXidName}{}%
}}

\newcommand{\DualLNLLogicdruleTCXXzIName}[0]{\DualLNLLogicdrulename{TC\_zI}}
\newcommand{\DualLNLLogicdruleTCXXzI}[1]{\DualLNLLogicdrule[#1]{%
\DualLNLLogicpremise{  \DualLNLLogicmv{x}  :  \DualLNLLogicnt{S}  \vdash_{\mathsf{C} }  \DualLNLLogicnt{t}  \DualLNLLogicsym{:}  \DualLNLLogicsym{0}  \DualLNLLogicsym{,}  \Psi   \quad   \DualLNLLogicmv{x_{{\mathrm{1}}}}  :  \DualLNLLogicnt{S_{{\mathrm{1}}}}  \vdash_{\mathsf{C} }  \Psi_{{\mathrm{1}}}   \DualLNLLogicsym{,} \, ... \, \DualLNLLogicsym{,}   \DualLNLLogicmv{x_{\DualLNLLogicmv{n}}}  :  \DualLNLLogicnt{S_{\DualLNLLogicmv{n}}}  \vdash_{\mathsf{C} }  \Psi_{\DualLNLLogicmv{n}}  }%
}{
 \DualLNLLogicmv{x}  :  \DualLNLLogicnt{S}  \vdash_{\mathsf{C} }  \Psi  \DualLNLLogicsym{,}  \DualLNLLogicsym{[}   \mathsf{false}\, \DualLNLLogicnt{t}   \DualLNLLogicsym{/}  \DualLNLLogicmv{x_{{\mathrm{1}}}}  \DualLNLLogicsym{]}  \Psi_{{\mathrm{1}}}  \DualLNLLogicsym{,} \, ... \, \DualLNLLogicsym{,}  \DualLNLLogicsym{[}   \mathsf{false}\, \DualLNLLogicnt{t}   \DualLNLLogicsym{/}  \DualLNLLogicmv{x_{\DualLNLLogicmv{n}}}  \DualLNLLogicsym{]}  \Psi_{\DualLNLLogicmv{n}} }{%
{\DualLNLLogicdruleTCXXzIName}{}%
}}

\newcommand{\DualLNLLogicdruleTCXXdIOneName}[0]{\DualLNLLogicdrulename{TC\_dI1}}
\newcommand{\DualLNLLogicdruleTCXXdIOne}[1]{\DualLNLLogicdrule[#1]{%
\DualLNLLogicpremise{ \DualLNLLogicmv{x}  :  \DualLNLLogicnt{S}  \vdash_{\mathsf{C} }  \Psi  \DualLNLLogicsym{,}  \DualLNLLogicnt{t}  \DualLNLLogicsym{:}  \DualLNLLogicnt{T_{{\mathrm{1}}}} }%
}{
 \DualLNLLogicmv{x}  :  \DualLNLLogicnt{S}  \vdash_{\mathsf{C} }  \Psi  \DualLNLLogicsym{,}   \mathsf{inl}\, \DualLNLLogicnt{t}   \DualLNLLogicsym{:}   \DualLNLLogicnt{T_{{\mathrm{1}}}}  +  \DualLNLLogicnt{T_{{\mathrm{2}}}}  }{%
{\DualLNLLogicdruleTCXXdIOneName}{}%
}}

\newcommand{\DualLNLLogicdruleTCXXdITwoName}[0]{\DualLNLLogicdrulename{TC\_dI2}}
\newcommand{\DualLNLLogicdruleTCXXdITwo}[1]{\DualLNLLogicdrule[#1]{%
\DualLNLLogicpremise{ \DualLNLLogicmv{x}  :  \DualLNLLogicnt{S}  \vdash_{\mathsf{C} }  \Psi  \DualLNLLogicsym{,}  \DualLNLLogicnt{t}  \DualLNLLogicsym{:}  \DualLNLLogicnt{T_{{\mathrm{2}}}} }%
}{
 \DualLNLLogicmv{x}  :  \DualLNLLogicnt{S}  \vdash_{\mathsf{C} }  \Psi  \DualLNLLogicsym{,}   \mathsf{inr}\, \DualLNLLogicnt{t}   \DualLNLLogicsym{:}   \DualLNLLogicnt{T_{{\mathrm{1}}}}  +  \DualLNLLogicnt{T_{{\mathrm{2}}}}  }{%
{\DualLNLLogicdruleTCXXdITwoName}{}%
}}

\newcommand{\DualLNLLogicdruleTCXXdEName}[0]{\DualLNLLogicdrulename{TC\_dE}}
\newcommand{\DualLNLLogicdruleTCXXdE}[1]{\DualLNLLogicdrule[#1]{%
\DualLNLLogicpremise{      \DualLNLLogicmv{x}  :  \DualLNLLogicnt{S}  \vdash_{\mathsf{C} }  \Psi_{{\mathrm{1}}}  \DualLNLLogicsym{,}  \DualLNLLogicnt{t}  \DualLNLLogicsym{:}   \DualLNLLogicnt{T_{{\mathrm{1}}}}  +  \DualLNLLogicnt{T_{{\mathrm{2}}}}    \quad   \DualLNLLogicmv{y}  :  \DualLNLLogicnt{T_{{\mathrm{1}}}}  \vdash_{\mathsf{C} }  \Psi_{{\mathrm{2}}}     \quad   \DualLNLLogicmv{z}  :  \DualLNLLogicnt{T_{{\mathrm{2}}}}  \vdash_{\mathsf{C} }  \Psi_{{\mathrm{3}}}     \quad  \DualLNLLogicsym{\mbox{$\mid$}}  \Psi_{{\mathrm{2}}}  \DualLNLLogicsym{\mbox{$\mid$}}  \DualLNLLogicsym{=}  \DualLNLLogicsym{\mbox{$\mid$}}  \Psi_{{\mathrm{3}}}  \DualLNLLogicsym{\mbox{$\mid$}} }%
}{
 \DualLNLLogicmv{x}  :  \DualLNLLogicnt{S}  \vdash_{\mathsf{C} }  \Psi_{{\mathrm{1}}}  \DualLNLLogicsym{,}   \mathsf{case}\, \DualLNLLogicnt{t} \,\mathsf{of}\, \DualLNLLogicmv{y} . \Psi_{{\mathrm{2}}} ,  \DualLNLLogicmv{z} . \Psi_{{\mathrm{3}}}  }{%
{\DualLNLLogicdruleTCXXdEName}{}%
}}

\newcommand{\DualLNLLogicdruleTCXXsubIName}[0]{\DualLNLLogicdrulename{TC\_subI}}
\newcommand{\DualLNLLogicdruleTCXXsubI}[1]{\DualLNLLogicdrule[#1]{%
\DualLNLLogicpremise{  \DualLNLLogicmv{x}  :  \DualLNLLogicnt{S}  \vdash_{\mathsf{C} }  \Psi_{{\mathrm{1}}}  \DualLNLLogicsym{,}  \DualLNLLogicnt{t}  \DualLNLLogicsym{:}  \DualLNLLogicnt{T_{{\mathrm{1}}}}   \quad   \DualLNLLogicmv{y}  :  \DualLNLLogicnt{T_{{\mathrm{2}}}}  \vdash_{\mathsf{C} }  \Psi_{{\mathrm{2}}}  }%
}{
 \DualLNLLogicmv{x}  :  \DualLNLLogicnt{S}  \vdash_{\mathsf{C} }  \Psi_{{\mathrm{1}}}  \DualLNLLogicsym{,}   \mathsf{mkc}( \DualLNLLogicnt{t} , \DualLNLLogicmv{y} )   \DualLNLLogicsym{:}   \DualLNLLogicnt{T_{{\mathrm{1}}}}  -  \DualLNLLogicnt{T_{{\mathrm{2}}}}   \DualLNLLogicsym{,}  \DualLNLLogicsym{[}  \DualLNLLogicmv{y}  \DualLNLLogicsym{(}  \DualLNLLogicnt{t}  \DualLNLLogicsym{)}  \DualLNLLogicsym{/}  \DualLNLLogicmv{y}  \DualLNLLogicsym{]}  \Psi_{{\mathrm{2}}} }{%
{\DualLNLLogicdruleTCXXsubIName}{}%
}}

\newcommand{\DualLNLLogicdruleTCXXsubEName}[0]{\DualLNLLogicdrulename{TC\_subE}}
\newcommand{\DualLNLLogicdruleTCXXsubE}[1]{\DualLNLLogicdrule[#1]{%
\DualLNLLogicpremise{  \DualLNLLogicmv{x}  :  \DualLNLLogicnt{S}  \vdash_{\mathsf{C} }  \Psi_{{\mathrm{1}}}  \DualLNLLogicsym{,}  \DualLNLLogicnt{s}  \DualLNLLogicsym{:}   \DualLNLLogicnt{T_{{\mathrm{1}}}}  -  \DualLNLLogicnt{T_{{\mathrm{2}}}}    \quad   \DualLNLLogicmv{y}  :  \DualLNLLogicnt{T_{{\mathrm{1}}}}  \vdash_{\mathsf{C} }  \DualLNLLogicnt{t}  \DualLNLLogicsym{:}  \DualLNLLogicnt{T_{{\mathrm{2}}}}  \DualLNLLogicsym{,}  \Psi_{{\mathrm{2}}}  }%
}{
 \DualLNLLogicmv{x}  :  \DualLNLLogicnt{S}  \vdash_{\mathsf{C} }  \Psi_{{\mathrm{1}}}  \DualLNLLogicsym{,}    \mathsf{postp}\,( \DualLNLLogicmv{y}  \mapsto  \DualLNLLogicnt{t} ,  \DualLNLLogicnt{s} )    \DualLNLLogicsym{,}  \DualLNLLogicsym{[}  \DualLNLLogicmv{y}  \DualLNLLogicsym{(}  \DualLNLLogicnt{s}  \DualLNLLogicsym{)}  \DualLNLLogicsym{/}  \DualLNLLogicmv{y}  \DualLNLLogicsym{]}  \Psi_{{\mathrm{2}}} }{%
{\DualLNLLogicdruleTCXXsubEName}{}%
}}

\newcommand{\DualLNLLogicdruleTCXXHEName}[0]{\DualLNLLogicdrulename{TC\_HE}}
\newcommand{\DualLNLLogicdruleTCXXHE}[1]{\DualLNLLogicdrule[#1]{%
\DualLNLLogicpremise{   \DualLNLLogicmv{x}  :  \DualLNLLogicnt{S}  \vdash_{\mathsf{C} }  \Psi_{{\mathrm{1}}}  \DualLNLLogicsym{,}  \DualLNLLogicnt{t}  \DualLNLLogicsym{:}   \mathsf{H}\, \DualLNLLogicnt{A}    \quad   \DualLNLLogicmv{y}  :  \DualLNLLogicnt{A}  \vdash_{\mathsf{L} }   \cdot  ; \Psi_{{\mathrm{2}}}   }%
}{
 \DualLNLLogicmv{x}  :  \DualLNLLogicnt{S}  \vdash_{\mathsf{C} }  \Psi_{{\mathrm{1}}}  \DualLNLLogicsym{,}   \mathsf{let}\,\mathsf{H}\, \DualLNLLogicmv{y}  =  \DualLNLLogicnt{t} \,\mathsf{in}\, \Psi_{{\mathrm{2}}}  }{%
{\DualLNLLogicdruleTCXXHEName}{}%
}}

\newcommand{\DualLNLLogicdruleTCXXweakName}[0]{\DualLNLLogicdrulename{TC\_weak}}
\newcommand{\DualLNLLogicdruleTCXXweak}[1]{\DualLNLLogicdrule[#1]{%
\DualLNLLogicpremise{ \DualLNLLogicmv{x}  :  \DualLNLLogicnt{S}  \vdash_{\mathsf{C} }  \Psi }%
}{
 \DualLNLLogicmv{x}  :  \DualLNLLogicnt{S}  \vdash_{\mathsf{C} }  \Psi  \DualLNLLogicsym{,}   \varepsilon   \DualLNLLogicsym{:}  \DualLNLLogicnt{T} }{%
{\DualLNLLogicdruleTCXXweakName}{}%
}}

\newcommand{\DualLNLLogicdruleTCXXcontrName}[0]{\DualLNLLogicdrulename{TC\_contr}}
\newcommand{\DualLNLLogicdruleTCXXcontr}[1]{\DualLNLLogicdrule[#1]{%
\DualLNLLogicpremise{ \DualLNLLogicmv{x}  :  \DualLNLLogicnt{S}  \vdash_{\mathsf{C} }  \DualLNLLogicnt{t_{{\mathrm{1}}}}  \DualLNLLogicsym{:}  \DualLNLLogicnt{T}  \DualLNLLogicsym{,}  \DualLNLLogicnt{t_{{\mathrm{2}}}}  \DualLNLLogicsym{:}  \DualLNLLogicnt{T}  \DualLNLLogicsym{,}  \Psi }%
}{
 \DualLNLLogicmv{x}  :  \DualLNLLogicnt{S}  \vdash_{\mathsf{C} }  \DualLNLLogicsym{(}   \DualLNLLogicnt{t_{{\mathrm{1}}}}  \cdot  \DualLNLLogicnt{t_{{\mathrm{2}}}}   \DualLNLLogicsym{)}  \DualLNLLogicsym{:}  \DualLNLLogicnt{T}  \DualLNLLogicsym{,}  \Psi }{%
{\DualLNLLogicdruleTCXXcontrName}{}%
}}

\newcommand{\DualLNLLogicdruleTCXXcutName}[0]{\DualLNLLogicdrulename{TC\_cut}}
\newcommand{\DualLNLLogicdruleTCXXcut}[1]{\DualLNLLogicdrule[#1]{%
\DualLNLLogicpremise{  \DualLNLLogicmv{x}  :  \DualLNLLogicnt{S}  \vdash_{\mathsf{C} }  \Psi_{{\mathrm{1}}}  \DualLNLLogicsym{,}  \DualLNLLogicnt{t}  \DualLNLLogicsym{:}  \DualLNLLogicnt{T}   \quad   \DualLNLLogicmv{y}  :  \DualLNLLogicnt{T}  \vdash_{\mathsf{C} }  \Psi_{{\mathrm{2}}}  }%
}{
 \DualLNLLogicmv{x}  :  \DualLNLLogicnt{S}  \vdash_{\mathsf{C} }  \Psi_{{\mathrm{1}}}  \DualLNLLogicsym{,}  \DualLNLLogicsym{[}  \DualLNLLogicnt{t}  \DualLNLLogicsym{/}  \DualLNLLogicmv{y}  \DualLNLLogicsym{]}  \Psi_{{\mathrm{2}}} }{%
{\DualLNLLogicdruleTCXXcutName}{}%
}}

\newcommand{\DualLNLLogicdruleTLXXidName}[0]{\DualLNLLogicdrulename{TL\_id}}
\newcommand{\DualLNLLogicdruleTLXXid}[1]{\DualLNLLogicdrule[#1]{%
}{
 \DualLNLLogicmv{x}  :  \DualLNLLogicnt{A}  \vdash_{\mathsf{L} }  \DualLNLLogicmv{x}  \DualLNLLogicsym{:}  \DualLNLLogicnt{A} ;  \cdot  }{%
{\DualLNLLogicdruleTLXXidName}{}%
}}

\newcommand{\DualLNLLogicdruleTLXXpIName}[0]{\DualLNLLogicdrulename{TL\_pI}}
\newcommand{\DualLNLLogicdruleTLXXpI}[1]{\DualLNLLogicdrule[#1]{%
\DualLNLLogicpremise{  \DualLNLLogicmv{x}  :  \DualLNLLogicnt{A}  \vdash_{\mathsf{L} }  \Delta ; \Psi   \quad   \DualLNLLogicnt{e} : \DualLNLLogicnt{B}  \in  \Delta  }%
}{
 \DualLNLLogicmv{x}  :  \DualLNLLogicnt{A}  \vdash_{\mathsf{L} }  \Delta  \DualLNLLogicsym{,}   \mathsf{connect}_\perp\,\mathsf{to}\, \DualLNLLogicnt{e}   \DualLNLLogicsym{:}   \perp  ; \Psi }{%
{\DualLNLLogicdruleTLXXpIName}{}%
}}

\newcommand{\DualLNLLogicdruleTLXXpEName}[0]{\DualLNLLogicdrulename{TL\_pE}}
\newcommand{\DualLNLLogicdruleTLXXpE}[1]{\DualLNLLogicdrule[#1]{%
\DualLNLLogicpremise{ \DualLNLLogicmv{x}  :  \DualLNLLogicnt{A}  \vdash_{\mathsf{L} }  \DualLNLLogicnt{e}  \DualLNLLogicsym{:}   \perp   \DualLNLLogicsym{,}  \Delta ;  \cdot  }%
}{
 \DualLNLLogicmv{x}  :  \DualLNLLogicnt{A}  \vdash_{\mathsf{L} }   \mathsf{postp}_\perp\, \DualLNLLogicnt{e}   \DualLNLLogicsym{,}  \Delta ;  \cdot  }{%
{\DualLNLLogicdruleTLXXpEName}{}%
}}

\newcommand{\DualLNLLogicdruleTLXXparIName}[0]{\DualLNLLogicdrulename{TL\_parI}}
\newcommand{\DualLNLLogicdruleTLXXparI}[1]{\DualLNLLogicdrule[#1]{%
\DualLNLLogicpremise{ \DualLNLLogicmv{x}  :  \DualLNLLogicnt{A}  \vdash_{\mathsf{L} }  \Delta  \DualLNLLogicsym{,}  \DualLNLLogicnt{e_{{\mathrm{1}}}}  \DualLNLLogicsym{:}  \DualLNLLogicnt{B_{{\mathrm{1}}}}  \DualLNLLogicsym{,}  \DualLNLLogicnt{e_{{\mathrm{2}}}}  \DualLNLLogicsym{:}  \DualLNLLogicnt{B_{{\mathrm{2}}}} ; \Psi }%
}{
 \DualLNLLogicmv{x}  :  \DualLNLLogicnt{A}  \vdash_{\mathsf{L} }  \Delta  \DualLNLLogicsym{,}   \DualLNLLogicnt{e_{{\mathrm{1}}}}  \oplus  \DualLNLLogicnt{e_{{\mathrm{2}}}}   \DualLNLLogicsym{:}   \DualLNLLogicnt{B_{{\mathrm{1}}}}  \oplus  \DualLNLLogicnt{B_{{\mathrm{2}}}}  ; \Psi }{%
{\DualLNLLogicdruleTLXXparIName}{}%
}}

\newcommand{\DualLNLLogicdruleTLXXparEName}[0]{\DualLNLLogicdrulename{TL\_parE}}
\newcommand{\DualLNLLogicdruleTLXXparE}[1]{\DualLNLLogicdrule[#1]{%
\DualLNLLogicpremise{    \DualLNLLogicmv{x}  :  \DualLNLLogicnt{A}  \vdash_{\mathsf{L} }  \Delta  \DualLNLLogicsym{,}  \DualLNLLogicnt{e}  \DualLNLLogicsym{:}   \DualLNLLogicnt{B_{{\mathrm{1}}}}  \oplus  \DualLNLLogicnt{B_{{\mathrm{2}}}}  ; \Psi   \quad   \DualLNLLogicmv{y}  :  \DualLNLLogicnt{B_{{\mathrm{1}}}}  \vdash_{\mathsf{L} }  \Delta_{{\mathrm{1}}} ; \Psi_{{\mathrm{1}}}     \quad   \DualLNLLogicmv{z}  :  \DualLNLLogicnt{B_{{\mathrm{2}}}}  \vdash_{\mathsf{L} }  \Delta_{{\mathrm{2}}} ; \Psi_{{\mathrm{2}}}  }%
}{
 \DualLNLLogicmv{x}  :  \DualLNLLogicnt{A}  \vdash_{\mathsf{L} }  \Delta  \DualLNLLogicsym{,}  \DualLNLLogicsym{[}   \mathsf{casel}\, \DualLNLLogicsym{(}  \DualLNLLogicnt{e}  \DualLNLLogicsym{)}   \DualLNLLogicsym{/}  \DualLNLLogicmv{y}  \DualLNLLogicsym{]}  \Delta_{{\mathrm{1}}}  \DualLNLLogicsym{,}  \DualLNLLogicsym{[}   \mathsf{caser}\, \DualLNLLogicsym{(}  \DualLNLLogicnt{e}  \DualLNLLogicsym{)}   \DualLNLLogicsym{/}  \DualLNLLogicmv{z}  \DualLNLLogicsym{]}  \Delta_{{\mathrm{2}}} ; \Psi  \DualLNLLogicsym{,}  \DualLNLLogicsym{[}   \mathsf{casel}\, \DualLNLLogicsym{(}  \DualLNLLogicnt{e}  \DualLNLLogicsym{)}   \DualLNLLogicsym{/}  \DualLNLLogicmv{y}  \DualLNLLogicsym{]}  \Psi_{{\mathrm{1}}}  \DualLNLLogicsym{,}  \DualLNLLogicsym{[}   \mathsf{caser}\, \DualLNLLogicsym{(}  \DualLNLLogicnt{e}  \DualLNLLogicsym{)}   \DualLNLLogicsym{/}  \DualLNLLogicmv{z}  \DualLNLLogicsym{]}  \Psi_{{\mathrm{2}}} }{%
{\DualLNLLogicdruleTLXXparEName}{}%
}}

\newcommand{\DualLNLLogicdruleTLXXsubIName}[0]{\DualLNLLogicdrulename{TL\_subI}}
\newcommand{\DualLNLLogicdruleTLXXsubI}[1]{\DualLNLLogicdrule[#1]{%
\DualLNLLogicpremise{  \DualLNLLogicmv{x}  :  \DualLNLLogicnt{A}  \vdash_{\mathsf{L} }  \Delta_{{\mathrm{1}}}  \DualLNLLogicsym{,}  \DualLNLLogicnt{e}  \DualLNLLogicsym{:}  \DualLNLLogicnt{B_{{\mathrm{1}}}} ; \Psi_{{\mathrm{1}}}   \quad   \DualLNLLogicmv{y}  :  \DualLNLLogicnt{B_{{\mathrm{2}}}}  \vdash_{\mathsf{L} }  \Delta_{{\mathrm{2}}} ; \Psi_{{\mathrm{2}}}  }%
}{
 \DualLNLLogicmv{x}  :  \DualLNLLogicnt{A}  \vdash_{\mathsf{L} }   \mathsf{mkc}( \DualLNLLogicnt{e} , \DualLNLLogicmv{y} )   \DualLNLLogicsym{:}   \DualLNLLogicnt{B_{{\mathrm{1}}}}  \colimp  \DualLNLLogicnt{B_{{\mathrm{2}}}}   \DualLNLLogicsym{,}  \Delta_{{\mathrm{1}}}  \DualLNLLogicsym{,}  \DualLNLLogicsym{[}  \DualLNLLogicmv{y}  \DualLNLLogicsym{(}  \DualLNLLogicnt{e}  \DualLNLLogicsym{)}  \DualLNLLogicsym{/}  \DualLNLLogicmv{y}  \DualLNLLogicsym{]}  \Delta_{{\mathrm{2}}} ; \Psi_{{\mathrm{1}}}  \DualLNLLogicsym{,}  \DualLNLLogicsym{[}  \DualLNLLogicmv{y}  \DualLNLLogicsym{(}  \DualLNLLogicnt{e}  \DualLNLLogicsym{)}  \DualLNLLogicsym{/}  \DualLNLLogicmv{y}  \DualLNLLogicsym{]}  \Psi_{{\mathrm{2}}} }{%
{\DualLNLLogicdruleTLXXsubIName}{}%
}}

\newcommand{\DualLNLLogicdruleTLXXsubEName}[0]{\DualLNLLogicdrulename{TL\_subE}}
\newcommand{\DualLNLLogicdruleTLXXsubE}[1]{\DualLNLLogicdrule[#1]{%
\DualLNLLogicpremise{  \DualLNLLogicmv{x}  :  \DualLNLLogicnt{A}  \vdash_{\mathsf{L} }  \Delta_{{\mathrm{1}}}  \DualLNLLogicsym{,}  \DualLNLLogicnt{e_{{\mathrm{1}}}}  \DualLNLLogicsym{:}   \DualLNLLogicnt{B_{{\mathrm{1}}}}  \colimp  \DualLNLLogicnt{B_{{\mathrm{2}}}}  ; \Psi_{{\mathrm{1}}}   \quad   \DualLNLLogicmv{y}  :  \DualLNLLogicnt{B_{{\mathrm{1}}}}  \vdash_{\mathsf{L} }  \DualLNLLogicnt{e_{{\mathrm{2}}}}  \DualLNLLogicsym{:}  \DualLNLLogicnt{B_{{\mathrm{1}}}}  \DualLNLLogicsym{,}  \Delta_{{\mathrm{2}}} ; \Psi_{{\mathrm{2}}}  }%
}{
 \DualLNLLogicmv{x}  :  \DualLNLLogicnt{A}  \vdash_{\mathsf{L} }  \Delta_{{\mathrm{1}}}  \DualLNLLogicsym{,}   \mathsf{postp}\,( \DualLNLLogicmv{y}  \mapsto  \DualLNLLogicnt{e_{{\mathrm{2}}}} ,  \DualLNLLogicnt{e_{{\mathrm{1}}}} )   \DualLNLLogicsym{,}  \Delta_{{\mathrm{2}}} ; \Psi_{{\mathrm{1}}}  \DualLNLLogicsym{,}  \Psi_{{\mathrm{2}}} }{%
{\DualLNLLogicdruleTLXXsubEName}{}%
}}

\newcommand{\DualLNLLogicdruleTLXXJIName}[0]{\DualLNLLogicdrulename{TL\_JI}}
\newcommand{\DualLNLLogicdruleTLXXJI}[1]{\DualLNLLogicdrule[#1]{%
\DualLNLLogicpremise{ \DualLNLLogicmv{x}  :  \DualLNLLogicnt{A}  \vdash_{\mathsf{L} }  \Delta ; \DualLNLLogicnt{t}  \DualLNLLogicsym{:}  \DualLNLLogicnt{T}  \DualLNLLogicsym{,}  \Psi }%
}{
 \DualLNLLogicmv{x}  :  \DualLNLLogicnt{A}  \vdash_{\mathsf{L} }  \Delta  \DualLNLLogicsym{,}   \mathsf{J}\, \DualLNLLogicnt{t}   \DualLNLLogicsym{:}   \mathsf{J}\, \DualLNLLogicnt{T}  ; \Psi }{%
{\DualLNLLogicdruleTLXXJIName}{}%
}}

\newcommand{\DualLNLLogicdruleTLXXJEName}[0]{\DualLNLLogicdrulename{TL\_JE}}
\newcommand{\DualLNLLogicdruleTLXXJE}[1]{\DualLNLLogicdrule[#1]{%
\DualLNLLogicpremise{  \DualLNLLogicmv{x}  :  \DualLNLLogicnt{A}  \vdash_{\mathsf{L} }  \Delta  \DualLNLLogicsym{,}  \DualLNLLogicnt{e}  \DualLNLLogicsym{:}   \mathsf{J}\, \DualLNLLogicnt{T}  ; \Psi_{{\mathrm{1}}}   \quad   \DualLNLLogicmv{y}  :  \DualLNLLogicnt{T}  \vdash_{\mathsf{C} }  \Psi_{{\mathrm{2}}}  }%
}{
 \DualLNLLogicmv{x}  :  \DualLNLLogicnt{A}  \vdash_{\mathsf{L} }  \Delta ; \Psi_{{\mathrm{1}}}  \DualLNLLogicsym{,}   \mathsf{let}\,\mathsf{J}\, \DualLNLLogicmv{y}  =  \DualLNLLogicnt{e} \,\mathsf{in}\, \Psi_{{\mathrm{2}}}  }{%
{\DualLNLLogicdruleTLXXJEName}{}%
}}

\newcommand{\DualLNLLogicdruleTLXXHIName}[0]{\DualLNLLogicdrulename{TL\_HI}}
\newcommand{\DualLNLLogicdruleTLXXHI}[1]{\DualLNLLogicdrule[#1]{%
\DualLNLLogicpremise{ \DualLNLLogicmv{x}  :  \DualLNLLogicnt{A}  \vdash_{\mathsf{L} }  \Delta  \DualLNLLogicsym{,}  \DualLNLLogicnt{e}  \DualLNLLogicsym{:}  \DualLNLLogicnt{B} ; \Psi }%
}{
 \DualLNLLogicmv{x}  :  \DualLNLLogicnt{A}  \vdash_{\mathsf{L} }  \Delta ;  \mathsf{H}\, \DualLNLLogicnt{e}   \DualLNLLogicsym{:}   \mathsf{H}\, \DualLNLLogicnt{B}   \DualLNLLogicsym{,}  \Psi }{%
{\DualLNLLogicdruleTLXXHIName}{}%
}}

\newcommand{\DualLNLLogicdruleTLXXHEName}[0]{\DualLNLLogicdrulename{TL\_HE}}
\newcommand{\DualLNLLogicdruleTLXXHE}[1]{\DualLNLLogicdrule[#1]{%
\DualLNLLogicpremise{  \DualLNLLogicmv{x}  :  \DualLNLLogicnt{A}  \vdash_{\mathsf{L} }  \Delta ; \Psi_{{\mathrm{1}}}  \DualLNLLogicsym{,}  \DualLNLLogicnt{t}  \DualLNLLogicsym{:}   \mathsf{H}\, \DualLNLLogicnt{A}    \quad   \DualLNLLogicmv{y}  :  \DualLNLLogicnt{A}  \vdash_{\mathsf{L} }   \cdot  ; \Psi_{{\mathrm{2}}}  }%
}{
 \DualLNLLogicmv{x}  :  \DualLNLLogicnt{A}  \vdash_{\mathsf{L} }  \Delta ; \Psi_{{\mathrm{1}}}  \DualLNLLogicsym{,}   \mathsf{let}\,\mathsf{H}\, \DualLNLLogicmv{y}  =  \DualLNLLogicnt{t} \,\mathsf{in}\, \Psi_{{\mathrm{2}}}  }{%
{\DualLNLLogicdruleTLXXHEName}{}%
}}

\newcommand{\DualLNLLogicdruleTLXXweakName}[0]{\DualLNLLogicdrulename{TL\_weak}}
\newcommand{\DualLNLLogicdruleTLXXweak}[1]{\DualLNLLogicdrule[#1]{%
\DualLNLLogicpremise{ \DualLNLLogicmv{x}  :  \DualLNLLogicnt{A}  \vdash_{\mathsf{L} }  \Delta ; \Psi }%
}{
 \DualLNLLogicmv{x}  :  \DualLNLLogicnt{A}  \vdash_{\mathsf{L} }  \Delta ; \Psi  \DualLNLLogicsym{,}   \varepsilon   \DualLNLLogicsym{:}  \DualLNLLogicnt{T} }{%
{\DualLNLLogicdruleTLXXweakName}{}%
}}

\newcommand{\DualLNLLogicdruleTLXXcontrName}[0]{\DualLNLLogicdrulename{TL\_contr}}
\newcommand{\DualLNLLogicdruleTLXXcontr}[1]{\DualLNLLogicdrule[#1]{%
\DualLNLLogicpremise{ \DualLNLLogicmv{x}  :  \DualLNLLogicnt{A}  \vdash_{\mathsf{L} }  \Delta ; \DualLNLLogicnt{t_{{\mathrm{1}}}}  \DualLNLLogicsym{:}  \DualLNLLogicnt{T}  \DualLNLLogicsym{,}  \DualLNLLogicnt{t_{{\mathrm{2}}}}  \DualLNLLogicsym{:}  \DualLNLLogicnt{T}  \DualLNLLogicsym{,}  \Psi }%
}{
 \DualLNLLogicmv{x}  :  \DualLNLLogicnt{A}  \vdash_{\mathsf{L} }  \Delta ;  \DualLNLLogicnt{t_{{\mathrm{1}}}}  \cdot  \DualLNLLogicnt{t_{{\mathrm{2}}}}   \DualLNLLogicsym{:}  \DualLNLLogicnt{T}  \DualLNLLogicsym{,}  \Psi }{%
{\DualLNLLogicdruleTLXXcontrName}{}%
}}

\newcommand{\DualLNLLogicdruleTLXXCcutName}[0]{\DualLNLLogicdrulename{TL\_Ccut}}
\newcommand{\DualLNLLogicdruleTLXXCcut}[1]{\DualLNLLogicdrule[#1]{%
\DualLNLLogicpremise{  \DualLNLLogicmv{x}  :  \DualLNLLogicnt{A}  \vdash_{\mathsf{L} }  \Delta ; \Psi_{{\mathrm{1}}}  \DualLNLLogicsym{,}  \DualLNLLogicnt{t}  \DualLNLLogicsym{:}  \DualLNLLogicnt{T}   \quad   \DualLNLLogicmv{y}  :  \DualLNLLogicnt{T}  \vdash_{\mathsf{C} }  \Psi_{{\mathrm{2}}}  }%
}{
 \DualLNLLogicmv{x}  :  \DualLNLLogicnt{A}  \vdash_{\mathsf{L} }  \Delta ; \Psi_{{\mathrm{1}}}  \DualLNLLogicsym{,}  \DualLNLLogicsym{[}  \DualLNLLogicnt{t}  \DualLNLLogicsym{/}  \DualLNLLogicmv{y}  \DualLNLLogicsym{]}  \Psi_{{\mathrm{2}}} }{%
{\DualLNLLogicdruleTLXXCcutName}{}%
}}

\newcommand{\DualLNLLogicdruleTLXXcutName}[0]{\DualLNLLogicdrulename{TL\_cut}}
\newcommand{\DualLNLLogicdruleTLXXcut}[1]{\DualLNLLogicdrule[#1]{%
\DualLNLLogicpremise{  \DualLNLLogicmv{x}  :  \DualLNLLogicnt{A}  \vdash_{\mathsf{L} }  \Delta_{{\mathrm{1}}}  \DualLNLLogicsym{,}  \DualLNLLogicnt{e}  \DualLNLLogicsym{:}  \DualLNLLogicnt{B} ; \Psi_{{\mathrm{1}}}   \quad   \DualLNLLogicmv{y}  :  \DualLNLLogicnt{B}  \vdash_{\mathsf{L} }  \Delta_{{\mathrm{2}}} ; \Psi_{{\mathrm{2}}}  }%
}{
 \DualLNLLogicmv{x}  :  \DualLNLLogicnt{A}  \vdash_{\mathsf{L} }  \Delta_{{\mathrm{1}}}  \DualLNLLogicsym{,}  \DualLNLLogicsym{[}  \DualLNLLogicnt{e}  \DualLNLLogicsym{/}  \DualLNLLogicmv{y}  \DualLNLLogicsym{]}  \Delta_{{\mathrm{2}}} ; \Psi_{{\mathrm{1}}}  \DualLNLLogicsym{,}  \DualLNLLogicsym{[}  \DualLNLLogicnt{e}  \DualLNLLogicsym{/}  \DualLNLLogicmv{y}  \DualLNLLogicsym{]}  \Psi_{{\mathrm{2}}} }{%
{\DualLNLLogicdruleTLXXcutName}{}%
}}

\renewcommand{\DualLNLLogicdrule}[4][]{{\displaystyle\frac{\begin{array}{l}#2\end{array}}{#3}\,\DualLNLLogicdrulename{#4}}}
\renewcommand{\DualLNLLogicdrulename}[1]{#1}

\renewcommand{\DualLNLLogicdruleCXXidName}{\text{C\_}\text{id}}
\renewcommand{\DualLNLLogicdruleCXXwkName}{\text{C\_}\text{weak}}
\renewcommand{\DualLNLLogicdruleCXXcrName}{\text{C\_}\text{contr}}
\renewcommand{\DualLNLLogicdruleCXXexName}{\text{C\_}\text{ex}}
\renewcommand{\DualLNLLogicdruleCXXfLName}{\text{C\_}0}
\renewcommand{\DualLNLLogicdruleCXXdLName}{\text{C\_}+_L}
\renewcommand{\DualLNLLogicdruleCXXdROneName}{\text{C\_}+_{R_1}}
\renewcommand{\DualLNLLogicdruleCXXdRTwoName}{\text{C\_}+_{R_2}}
\renewcommand{\DualLNLLogicdruleCXXsLName}{\text{C\_}-_L}
\renewcommand{\DualLNLLogicdruleCXXsRName}{\text{C\_}-_R}
\renewcommand{\DualLNLLogicdruleCXXcutName}{\text{C\_}\text{cut}}
\renewcommand{\DualLNLLogicdruleCXXhLName}{\mathsf{H}_L}
\renewcommand{\DualLNLLogicdruleCXXmcutName}{\text{C\_cut}_n}

\renewcommand{\DualLNLLogicdruleLXXidName}{\text{LL\_id}}
\renewcommand{\DualLNLLogicdruleLXXwkName}{\text{LC\_weak}}
\renewcommand{\DualLNLLogicdruleLXXctrName}{\text{LC\_contr}}
\renewcommand{\DualLNLLogicdruleLXXexName}{\text{LL\_ex}}
\renewcommand{\DualLNLLogicdruleLXXCexName}{\text{LC\_ex}}
\renewcommand{\DualLNLLogicdruleLXXcutName}{\text{LL\_cut}}
\renewcommand{\DualLNLLogicdruleLXXCcutName}{\text{LC\_cut}}
\renewcommand{\DualLNLLogicdruleLXXflLName}{\text{LL\_}\hspace{-3px}\perp_L}
\renewcommand{\DualLNLLogicdruleLXXflRName}{\text{LL\_}\hspace{-3px}\perp_R}
\renewcommand{\DualLNLLogicdruleLXXdROneName}{\text{LC\_}+_{R_1}}
\renewcommand{\DualLNLLogicdruleLXXdRTwoName}{\text{LC\_}+_{R_2}}
\renewcommand{\DualLNLLogicdruleLXXpLName}{\text{LL\_}\oplus_L}
\renewcommand{\DualLNLLogicdruleLXXpRName}{\text{LL\_}\oplus_R}
\renewcommand{\DualLNLLogicdruleLXXsLName}{\text{LL\_}\hspace{-3px}\colimp_L}
\renewcommand{\DualLNLLogicdruleLXXsRName}{\text{LL\_}\hspace{-3px}\colimp_R}
\renewcommand{\DualLNLLogicdruleLXXCsRName}{\text{LL\_}-_R}
\renewcommand{\DualLNLLogicdruleLXXjLName}{\func{J}_L}
\renewcommand{\DualLNLLogicdruleLXXjRName}{\func{J}_R}
\renewcommand{\DualLNLLogicdruleLXXhRName}{\func{H}_R}
\renewcommand{\DualLNLLogicdruleLXXCmcutName}{\text{LC\_cut}_n}

\renewcommand{\DualLNLLogicdruleNCXXidName}{\text{NC\_}\text{id}}
\renewcommand{\DualLNLLogicdruleNCXXzEName}{\text{NC\_}0_E}
\renewcommand{\DualLNLLogicdruleNCXXdIOneName}{\text{NC\_}+_{I_1}}
\renewcommand{\DualLNLLogicdruleNCXXdITwoName}{\text{NC\_}+_{I_2}}
\renewcommand{\DualLNLLogicdruleNCXXdEName}{\text{NC\_}+_E}
\renewcommand{\DualLNLLogicdruleNCXXsubIName}{\text{NC\_}-_I}
\renewcommand{\DualLNLLogicdruleNCXXsubEName}{\text{NC\_}-_E}
\renewcommand{\DualLNLLogicdruleNCXXHEName}{\text{NC\_}\func{H}_E}
\renewcommand{\DualLNLLogicdruleNCXXweakName}{\text{NC\_}\text{weak}}
\renewcommand{\DualLNLLogicdruleNCXXcontrName}{\text{NC\_}\text{contr}}
\renewcommand{\DualLNLLogicdruleNCXXcutName}{\text{NC\_}\text{cut}}

\renewcommand{\DualLNLLogicdruleNLXXidName}{\text{NLL\_}\text{id}}
\renewcommand{\DualLNLLogicdruleNLXXpIName}{\text{NLL\_}\hspace{-3px}\perp_I}
\renewcommand{\DualLNLLogicdruleNLXXpEName}{\text{NLL\_}\hspace{-3px}\perp_E}
\renewcommand{\DualLNLLogicdruleNLXXparIName}{\text{NLL\_}\oplus_I}
\renewcommand{\DualLNLLogicdruleNLXXparEName}{\text{NLL\_}\oplus_E}
\renewcommand{\DualLNLLogicdruleNLXXsubIName}{\text{NLL\_}\hspace{-3px}\colimp_I}
\renewcommand{\DualLNLLogicdruleNLXXsubEName}{\text{NLL\_}\hspace{-3px}\colimp_E}
\renewcommand{\DualLNLLogicdruleNLXXJIName}{\text{NLL\_}\func{J}_I}
\renewcommand{\DualLNLLogicdruleNLXXJEName}{\text{NLL\_}\func{J}_E}
\renewcommand{\DualLNLLogicdruleNLXXHIName}{\text{NLL\_}\func{H}_I}
\renewcommand{\DualLNLLogicdruleNLXXHEName}{\text{NLL\_}\func{H}_E}
\renewcommand{\DualLNLLogicdruleNLXXweakName}{\text{NLC\_}\text{weak}}
\renewcommand{\DualLNLLogicdruleNLXXcontrName}{\text{NLC\_}\text{contr}}
\renewcommand{\DualLNLLogicdruleNLXXCcutName}{\text{NLC\_}\text{cut}}
\renewcommand{\DualLNLLogicdruleNLXXcutName}{\text{NLL\_}\text{cut}}

\renewcommand{\DualLNLLogicdruleTCXXidName}{\text{TC\_}\text{id}}
\renewcommand{\DualLNLLogicdruleTCXXweakName}{\text{TC\_}\text{weak}}
\renewcommand{\DualLNLLogicdruleTLXXweakName}{\text{TL\_}\text{weak}}
\renewcommand{\DualLNLLogicdruleTCXXzIName}{\text{TC\_}0_E}
\renewcommand{\DualLNLLogicdruleTCXXdIOneName}{\text{TC\_}+_{I_1}}
\renewcommand{\DualLNLLogicdruleTCXXdITwoName}{\text{TC\_}+_{I_2}}
\renewcommand{\DualLNLLogicdruleTCXXdEName}{\text{TC\_}+_E}
\renewcommand{\DualLNLLogicdruleTCXXsubIName}{\text{TC\_}-_I}
\renewcommand{\DualLNLLogicdruleTCXXsubEName}{\text{TC\_}-_E}
\renewcommand{\DualLNLLogicdruleTCXXHEName}{\text{TC\_}\func{H}_E}
\renewcommand{\DualLNLLogicdruleTCXXcontrName}{\text{TC\_}\text{contr}}
\renewcommand{\DualLNLLogicdruleTCXXcutName}{\text{TC\_}\text{cut}}

\renewcommand{\DualLNLLogicdruleTLXXidName}{\text{TLL\_}\text{id}}
\renewcommand{\DualLNLLogicdruleTLXXpIName}{\text{TLL\_}\hspace{-3px}\perp_I}
\renewcommand{\DualLNLLogicdruleTLXXpEName}{\text{TLL\_}\hspace{-3px}\perp_E}
\renewcommand{\DualLNLLogicdruleTLXXparIName}{\text{TLL\_}\oplus_I}
\renewcommand{\DualLNLLogicdruleTLXXparEName}{\text{TLL\_}\oplus_E}
\renewcommand{\DualLNLLogicdruleTLXXsubIName}{\text{TLL\_}\hspace{-3px}\colimp_I}
\renewcommand{\DualLNLLogicdruleTLXXsubEName}{\text{TLL\_}\hspace{-3px}\colimp_E}
\renewcommand{\DualLNLLogicdruleTLXXJIName}{\text{TLL\_}\func{J}_I}
\renewcommand{\DualLNLLogicdruleTLXXJEName}{\text{TLL\_}\func{J}_E}
\renewcommand{\DualLNLLogicdruleTLXXHIName}{\text{TLL\_}\func{H}_I}
\renewcommand{\DualLNLLogicdruleTLXXHEName}{\text{TLL\_}\func{H}_E}
\renewcommand{\DualLNLLogicdruleTLXXcontrName}{\text{TLC\_}\text{contr}}
\renewcommand{\DualLNLLogicdruleTLXXCcutName}{\text{NLC\_}\text{cut}}
\renewcommand{\DualLNLLogicdruleTLXXcutName}{\text{TLL\_}\text{cut}}

\newcommand{\DLNLP}{\text{DLNL}^+}

\urldef{\mailsa}\path|{heades}@augusta.edu|

\begin{document}

\title{A Cointuitionistic Adjoint Logic}
\author{Harley Eades III}
\email{heades@augusta.edu}
\address{Computer Science, Augusta University, Augusta, GA}

\author{Gianluigi Bellin}
\email{gianluigi.bellin@univr.it}
\address{Dipartimento di Informatica, Universit\`{a} di Verona, Strada Le Grazie, 37134 Verona, Italy}

\maketitle 

\begin{abstract}
Bi-intuitionistic logic (BINT) is a conservative extension of
intuitionistic logic to include the duals of each logical
connective. One leading question with respect to BINT is, what does
BINT look like across the three arcs -- logic, typed
$\lambda$-calculi, and category theory -- of the Curry-Howard-Lambek
correspondence?  Categorically, BINT can be seen as a mixing of two
worlds: the first being intuitionistic logic (IL), which is modeled by
a cartesian closed category, and the second being the dual to
intuitionistic logic called cointuitionistic logic (coIL), which is
modeled by a cocartesian coclosed category.  Crolard
\cite{Crolard:2001} showed that combining these two categories into
the same category results in it degenerating to a poset.  However,
this degeneration does not occur when both logics are linear.  We
propose that IL and coIL need to be separated, and then mixed in a
controlled way using the modalities from linear logic.  This
separation can be ultimately achieved by an adjoint formalization of
bi-intuitionistic logic.  This formalization consists of three worlds
instead of two: the first is intuitionistic logic, the second is
linear bi-intuitionistic (Bi-ILL), and the third is cointuitionistic
logic.  They are then related via two adjunctions.  The adjunction
between IL and ILL is known as a Linear/Non-linear model (LNL model)
of ILL, and is due to Benton \cite{Benton:1994}.  However, the dual to
LNL models which would amount to the adjunction between coILL and coIL
has yet to appear in the literature.  In this paper we fill this gap
by studying the dual to LNL models which we call dual LNL models.  We
conduct a similar analysis to that of Benton for dual LNL models by
showing that dual LNL models correspond to dual linear categories, the
dual to Bierman's \cite{Bierman:1994} linear categories proposed by
Bellin~\cite{Bellin:2012}.  Following this we give the definition of
bi-LNL models by combining our dual LNL models with Benton's LNL
models to obtain a categorical model of bi-intuitionistic logic, but
we leave its analysis and corresponding logic to a future paper.
Finally, we give a corresponding sequent calculus, natural deduction,
and term assignment for dual LNL models.
\end{abstract}

\section{Introduction}
\label{sec:introduction}
Bi-intuitionistic logic (BINT) is a conservative extension of
intuitionistic logic to include the duals of each logical connective.
That is, BINT contains the usual intuitionistic logical connectives
such as true, conjunction, and implication, but also their duals
false, disjunction, and coimplication.  One leading question with
respect to BINT is, what does BINT look like across the three arcs --
logic, typed $\lambda$-calculi, and category theory -- of the
Curry-Howard-Lambek correspondence?  A non-trivial (does not
degenerate to a poset) categorical model of BINT is currently an open
problem.  This paper directly contributes to the solution of this open
problem by giving a new categorical model based on adjunctions for
cointuitionistic logic, and then proposing a new categorical model for
BINT.

BINT can be seen as a mixing of two worlds: the first being
intuitionistic logic (IL), which is modeled categorically by a
cartesian closed category (CCC), and the second being the dual to
intuitionistic logic called cointuitionistic logic (coIL), which is
modeled by a cocartesian coclosed category (coCCC).  Crolard
\cite{Crolard:2001} showed that combining these two categories into
the same category results in it degenerating to a poset, i.e.
there is at most one morphism between any two objects; we review this
result in
Section~\ref{subsec:cartesian_closed_and_cocartesian_coclosed_categories}.
However, this degeneration does not occur when both logics are linear.

Notice that atoms are not dualized, at least in the main stream
tradition of BINT started by C. Rauszer \cite{Rauszer:1974,
  Rauszer:1980}. For this reason T. Crolard \cite{Crolard:2001}
p. 160, describes the relation between IL and coIL within BINT as
``pseudo duality''. A duality on atoms could be added and this has
been attempted with linguistic motivations \cite{Bellin:2014} (see the
section on Related Work).  This avoids the collapse but yields a
different framework.  Here we are concerned mainly with the main
stream tradition.

We propose that IL and coIL need to be separated, and then mixed in a
controlled way using the modalities from linear logic.  This
separation can be ultimately achieved by an adjoint formalization of
bi-intuitionistic logic.  This formalization consists of three worlds
instead of two: the first is intuitionistic logic, the second is
linear bi-intuitionistic (Bi-ILL), and the third is cointuitionistic
logic.  They are then related via two adjunctions as depicted by the
following diagram:
\begin{center}
  \begin{tikzpicture} 
    \node (img) {\includegraphics[scale=0.4]{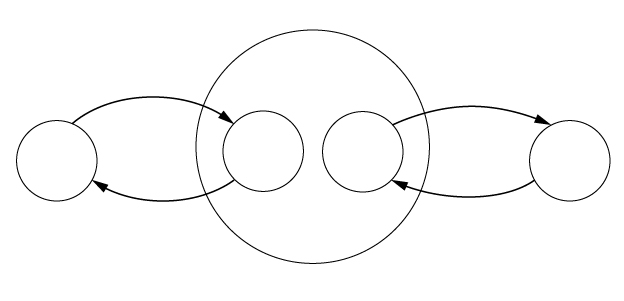}};
    \node (dv) at (-2.3, 0.0) {\huge $\dashv$};
    \node (IL) at (-3.58, -0.1) {IL};
    \node (vd) at (2.3, 0.0) {\huge $\dashv$};
    \node (coIL) at (3.65, -0.1) {coIL};

    \node (ILL) at (-0.67, 0.0) {ILL};
    \node (coILL) at (0.71, 0.0) {coILL};
    \node (BiILL) at (0, 1.0) {Bi-ILL};
  \end{tikzpicture}    
\end{center}
The adjunction between IL and ILL is known as a Linear/Non-linear model
(LNL model) of ILL, and is due to Benton \cite{Benton:1994}.  However,
the dual to LNL models which would amount to the adjunction between coILL
and coIL has yet to appear in the literature.

Suppose $(\cat{I}, 1, \times, \to)$ is a cartesian closed category,
and $(\cat{L}, \top, \otimes, \limp)$ is a symmetric monoidal closed
category.  Then relate these two categories with a symmetric monoidal
adjunction $\cat{I} : \func{F} \dashv \func{G} : \cat{L}$
(Definition~\ref{def:SMCADJ}), where $\func{F}$ and $\func{G}$ are
symmetric monoidal functors.  The later point implies that there are
natural transformations $\m{X,Y} : \func{F}X \otimes \func{F}Y \mto
\func{F}(X \times Y)$ and $\n{A,B} : \func{G}A \times \func{G}B \mto
\func{G}(A \otimes B)$, and maps $\m\top : \top \mto \func{F}1$ and
$\n1 : 1 \mto \func{G}\top$ subject to several coherence conditions;
see Definition~\ref{def:SMCFUN}.  Furthermore, the functor $\func{F}$
is strong which means that $\m{X,Y}$ and $\m{\top}$ are isomorphisms.
This setup turns out to be one of the most beautiful models of
intuitionistic linear logic called a LNL model due to Benton
\cite{Benton:1994}.  In fact, the linear modality of-course can be
defined by $!A = \func{F}(\func{G}(A))$ which defines a symmetric
monoidal comonad using the adjunction; see Section~2.2 of
\cite{Benton:1994}.  This model is much simpler than other known
models, and resulted in a logic called LNL logic which supports mixing
intuitionistic logic with linear logic.  The main contribution of this
paper is the definition and study of the dual to Benton's LNL models
as models of cointuitionistic logic.

Taking the dual of the previous model results in what we call dual LNL
models. They consist of a cocartesian coclosed category, $(\cat{C}, 0,
+, -)$ where $- : \cat{C} \times \cat{C} \mto \cat{C}$ is left adjoint
to the coproduct, a symmetric monoidal coclosed category
(Definition~\ref{def:SMCCC}), $(\cat{L}', \perp, \oplus, \colimp)$,
where $\colimp : \cat{L}' \times \cat{L}' \mto \cat{L}'$ is left
adjoint to cotensor (usually called \emph{par}), and a symmetric
comonoidal adjunction (Definition~\ref{def:coSMCADJ}) $\cat{L'} :
\func{H} \dashv \func{J} : \cat{C}$, where $\func{H}$ and $\func{J}$
are symmetric comonoidal functors. Dual to the above, this implies
that there are natural transformations $\m{X,Y} : \func{J}(X + Y) \mto
\func{J}X \oplus \func{J}Y$ and $\n{A,B} : \func{H}(A \oplus B) \mto
\func{H}A + \func{H}B$, and maps $\m0 : \func{J}0 \mto \perp$ and
$\n\perp : \func{H}\perp \mto 0$ subject to several coherence conditions;
see Definition~\ref{def:coSMCFUN}.  In fact, one can define Girard's
exponential why-not by $\wn A = \func{JH}A$, and hence, is the monad
induced by the adjunction.

Bellin \cite{Bellin:2012} was the first to propose the dual to
Bierman's \cite{Bierman:1994} linear categories which he names dual
linear categories as a model of cointuitionistic linear logic.  We
conduct a similar analysis to that of Benton for dual LNL models by
showing that dual LNL models are dual linear categories
(Section~\ref{subsec:dual_lnl_model_implies_dual_category}), and that
from a dual linear category we may obtain a dual LNL model
(Section~\ref{subsec:dual_category_implies_dual_lnl_model}).
Following this we give the definition of bi-LNL models by combining
our dual LNL models with Benton's LNL models to obtain a categorical
model of bi-intuitionistic logic
(Section~\ref{subsec:a_mixed_bi-linear_non-linear_model}), but we
leave its analysis and corresponding logic to a future paper.

Benton~\cite{Benton:1994} showed that, syntactically, LNL models have
a corresponding logic by first defining intuitionistic logic, whose
sequent is denoted, $\Theta \vdash_{\cat{C}} X$, and then
intuitionistic linear logic, $\Theta;\Gamma \vdash_{\cat{L}} A$, but
the key insight was that $\Theta$ contains non-linear assumptions
while $\Gamma$ contains linear assumptions, but one should view their
separation as merely cosmetic; all assumptions can consistently be
mixed within a single context.  The two logics are then connected by
syntactic versions of the functors $\func{F}$ and $\func{G}$ which
allow formulas to move between both fragments.

Following Benton's lead the design of dual LNL logic is similar.  We
have a non-linear cointuitionistic fragment, $T \vdash_{\cat{C}}
\Psi$, and a linear cointuitionistic fragment, $A \vdash_{\cat{C}}
\Delta;\Psi$, where $\Delta$ contains linear conclusions and $\Psi$
contains non-linear conclusions, but again the separation of contexts
is only cosmetic. The non-linear fragment has the following structural rules:
\begin{mathpar}
  \DualLNLLogicdruleCXXwk{}  
  \and
  \DualLNLLogicdruleCXXcr{} 
\end{mathpar}
Then we connect these two fragments together using the following rules
for the functors $\func{H}$ and $\func{J}$:
\begin{mathpar}
  \DualLNLLogicdruleCXXhL{}
  \and
  \DualLNLLogicdruleLXXhR{}
  \and
  \DualLNLLogicdruleLXXjL{}
  \and
  \DualLNLLogicdruleLXXjR{}
\end{mathpar}
These allow for linear and non-linear formulas to move from one
fragment to the other.  We will give a sequent calculus and natural
deduction formalization (Section~\ref{sec:sequent_calculus} and
Section~\ref{sec:sequent-style_natural_deduction}) as well as a term
assignment (Section~\ref{sec:term_assignment}).  The latter is
particularly interesting, because of the fact that cointuitionistic
logic has multiple conclusions, but only a single hypothesis.



\section{The Adjoint Model}
\label{sec:adjoint_model}
In this section we define dual LNL models (Definition~\ref{def:dual
  LNL-model}) and then relate them to Bellin's dual linear categories
(Definition~\ref{def:dual-linear-cat}), but first we introduce the
basic categorical machinery needed for the later sections and
summarize Crolard's result showing that the combination of cartesian
closed categories with cocartesian coclosed categories is degenerate.
Following these we conclude this section by introducing a categorical
model for full BINT called a mixed bilinear/non-linear model that
combines LNL models with dual LNL models
(Definition~\ref{subsec:a_mixed_bi-linear_non-linear_model}).

\subsection{Symmetric (co)Monoidal Categories}
\label{subsec:symmetric_monoidal_categories}
We now introduce the necessary definitions related to symmetric
monoidal categories that our model will depend on.  Most of these
definitions are equivalent to the ones given by Benton
\cite{Benton:1994}, but we give a lesser known definition of symmetric
comonoidal functors due to Bellin \cite{Bellin:2012}.  In this
section we also introduce distributive categories, the notion of
coclosure, and finally, the definition of bilinear categories.  The
reader may wish to simply skim this section, but refer back to it when
they encounter a definition or result they do not know.

\begin{definition}
  \label{def:monoidal-category}
  A \textbf{symmetric monoidal category (SMC)} is a category, $\cat{M}$,
  with the following data:
  \begin{itemize}
  \item An object $\top$ of $\cat{M}$,
  \item A bi-functor $\otimes : \cat{M} \times \cat{M} \mto \cat{M}$,
  \item The following natural isomorphisms:
    \[
    \begin{array}{lll}
      \lambda_A : \top \otimes A \mto A\\
      \rho_A : A \otimes \top \mto A\\      
      \alpha_{A,B,C} : (A \otimes B) \otimes C \mto A \otimes (B \otimes C)\\
    \end{array}
    \]
  \item A symmetry natural transformation:
    \[
    \beta_{A,B} : A \otimes B \mto B \otimes A
    \]
  \item Subject to the following coherence diagrams:
    \begin{mathpar}
      \bfig
      \vSquares|ammmmma|/->`->```->``<-/[
        ((A \otimes B) \otimes C) \otimes D`
        (A \otimes (B \otimes C)) \otimes D`
        (A \otimes B) \otimes (C \otimes D)``
        A \otimes (B \otimes (C \otimes D))`
        A \otimes ((B \otimes C) \otimes D);
        \alpha_{A,B,C} \otimes \id_D`
        \alpha_{A \otimes B,C,D}```
        \alpha_{A,B,C \otimes D}``
        \id_A \otimes \alpha_{B,C,D}]      
      
      \morphism(1206,1000)|m|<0,-1000>[
        (A \otimes (B \otimes C)) \otimes D`
        A \otimes ((B \otimes C) \otimes D);
        \alpha_{A,B \otimes C,D}]
      \efig
      \and
      \bfig
      \hSquares|aammmaa|/->`->`->``->`->`->/[
        (A \otimes B) \otimes C`
        A \otimes (B \otimes C)`
        (B \otimes C) \otimes A`
        (B \otimes A) \otimes C`
        B \otimes (A \otimes C)`
        B \otimes (C \otimes A);
        \alpha_{A,B,C}`
        \beta_{A,B \otimes C}`
        \beta_{A,B} \otimes \id_C``
        \alpha_{B,C,A}`
        \alpha_{B,A,C}`
        \id_B \otimes \beta_{A,C}]
      \efig      
    \end{mathpar}
    \begin{mathpar}
      \bfig
      \Vtriangle[
        (A \otimes \top) \otimes B`
        A \otimes (\top \otimes B)`
        A \otimes B;
        \alpha_{A,\top,B}`
        \rho_{A}`
        \lambda_{B}]
      \efig
      \and
      \bfig
      \btriangle[
        A \otimes B`
        B \otimes A`
        A \otimes B;
        \beta_{A,B}`
        \id_{A \otimes B}`
        \beta_{B,A}]
      \efig
      \and
      \bfig
      \Vtriangle[
        \top \otimes A`
        A \otimes \top`
        A;
        \beta_{\top,A}`
        \lambda_A`
        \rho_A]
      \efig
    \end{mathpar}    
  \end{itemize}
\end{definition}

Categorical modeling implication requires that the model be closed;
which can be seen as an internalization of the notion of a morphism.
\begin{definition}
  \label{def:SMCC}
  A \textbf{symmetric monoidal closed category (SMCC)} is a symmetric
  monoidal category, $(\cat{M},\top,\otimes)$, such that, for any object
  $B$ of $\cat{M}$, the functor $- \otimes B : \cat{M} \mto \cat{M}$
  has a specified right adjoint.  Hence, for any objects $A$ and $C$
  of $\cat{M}$ there is an object $B \limp C$ of $\cat{M}$ and a
  natural bijection:
  \[
  \Hom{\cat{M}}{A \otimes B}{C} \cong \Hom{\cat{M}}{A}{B \limp C}
  \]
  We call the functor $\limp : \cat{M} \times \cat{M} \mto \cat{M}$
  the internal hom of $\cat{M}$.
\end{definition}

Symmetric monoidal closed categories can be seen as a model of
intuitionistic linear logic with a tensor product and implication
\cite{Bierman:1994}.  What happens when we take the dual?  First, we
have the following result:
\begin{lemma}[Dual of Symmetric Monoidal Categories]
  \label{lemma:dual_of_symmetric_monoidal_categories}
  If $(\cat{M},\top,\otimes)$ is a symmetric monoidal category, then
  $\catop{M}$ is also a symmetric monoidal category.
\end{lemma}
The previous result follows from the fact that the structures making
up symmetric monoidal categories are isomorphisms, and so naturally
taking their opposite will yield another symmetric monoidal category.
To emphasize when we are thinking about a symmetric monoidal category
in the opposite we use the notation $(\cat{M},\perp,\oplus)$ which gives
the suggestion of $\oplus$ corresponding to a disjunctive tensor
product which we call the \textit{cotensor} of $\cat{M}$. The next
definition describes when a symmetric monoidal category is coclosed.
\begin{definition}
  \label{def:SMCCC}
  A \textbf{symmetric monoidal coclosed category (SMCCC)} is a symmetric
  monoidal category, $(\cat{M},\perp,\oplus)$, such that, for any object
  $B$ of $\cat{M}$, the functor $- \oplus B : \cat{M} \mto \cat{M}$
  has a specified left adjoint.  Hence, for any objects $A$ and $C$
  of $\cat{M}$ there is an object $C \colimp B$ of $\cat{M}$ and a
  natural bijection:
  \[
  \Hom{\cat{M}}{C}{A \oplus B} \cong \Hom{\cat{M}}{C \colimp B}{A}
  \]
  We call the functor $\colimp : \cat{M} \times \cat{M} \mto \cat{M}$
  the internal cohom of $\cat{M}$.
\end{definition}

We combine a symmetric monoidal closed category with a symmetric
monoidal coclosed category in a single category.  First, we define the
notion of a distributive category due to Cockett and Seely \cite{Cockett:1997}.
\begin{definition}
  \label{def:dist-cat}
  We call a symmetric monoidal category, $(\cat{M}, \top, \otimes,
  \perp, \oplus)$ equipped with the structure of a cotensor $(\cat{M},
  \perp, \oplus)$, a \textbf{distributive category} if there are
  natural transformations:
  \[
  \begin{array}{lll}
    \delta^L_{A,B,C} : A \otimes (B \oplus C) \mto (A \otimes B) \oplus C\\
    \delta^R_{A,B,C} : (B \oplus C) \otimes A \mto B \oplus (C \otimes A)
  \end{array}
  \]
  subject to several coherence diagrams.  Due to the large number of
  coherence diagrams we do not list them here, but they all can be
  found in Cockett and Seely's paper \cite{Cockett:1997}.
\end{definition}
\noindent
Requiring that the tensor and cotensor products have the corresponding
right and left adjoints results in the following definition.
\begin{definition}
  \label{def:bilinear-cat}
  A \textbf{bilinear category} is a distributive category $(\cat{M},
  \top, \otimes, \perp, \oplus)$ such that $(\cat{M}, \top, \otimes)$
  is closed, and $(\cat{M}, \perp, \oplus)$ is coclosed.  We will
  denote bi-linear categories by $(\cat{M}, \top, \otimes, \limp, \perp,
  \oplus, \colimp)$.
\end{definition}
Originally, Lambek defined bilinear categories to be similar to the
previous definition, but the tensor and cotensor were non-commutative
\cite{Cockett:1997a}, however, the bilinear categories given here
are. We retain the name in homage to his original work.  As we will
see below bilinear categories form the core of a categorical model for
bi-intuitionism.

A symmetric monoidal category is a category with additional structure
subject to several coherence diagrams.  Thus, an ordinary functor is
not enough to capture this structure, and hence, the introduction of
symmetric monoidal functors.
\begin{definition}
  \label{def:SMCFUN}
  Suppose we are given two symmetric monoidal
  categories\\ $(\cat{M}_1,\top_1,\otimes_1,\alpha_1,\lambda_1,\rho_1,\beta_1)$
  and
  $(\cat{M}_2,\top_2,\otimes_2,\alpha_2,\lambda_2,\rho_2,\beta_2)$.
  Then a \textbf{symmetric monoidal functor} is a functor $F :
  \cat{M}_1 \mto \cat{M}_2$, a map $m_{\top_1} : \top_2 \mto F\top_1$
  and a natural transformation $m_{A,B} : FA \otimes_2 FB \mto F(A
  \otimes_1 B)$ subject to the following coherence conditions:
  \begin{mathpar}
    \bfig
    \vSquares|ammmmma|/->`->`->``->`->`->/[
      (FA \otimes_2 FB) \otimes_2 FC`
      FA \otimes_2 (FB \otimes_2 FC)`
      F(A \otimes_1 B) \otimes_2 FC`
      FA \otimes_2 F(B \otimes_1 C)`
      F((A \otimes_1 B) \otimes_1 C)`
      F(A \otimes_1 (B \otimes_1 C));
      {\alpha_2}_{FA,FB,FC}`
      m_{A,B} \otimes \id_{FC}`
      \id_{FA} \otimes m_{B,C}``
      m_{A \otimes_1 B,C}`
      m_{A,B \otimes_1 C}`
      F{\alpha_1}_{A,B,C}]
    \efig
    \end{mathpar}
\begin{mathpar}
    \bfig
    \square|amma|/->`->`<-`->/<1000,500>[
      \top_2 \otimes_2 FA`
      FA`
      F\top_1 \otimes_2 FA`
      F(\top_1 \otimes_1 A);
      {\lambda_2}_{FA}`
      m_{\top_1} \otimes \id_{FA}`
      F{\lambda_1}_{A}`
      m_{\top_1,A}]
    \efig
    \and
    \bfig
    \square|amma|/->`->`<-`->/<1000,500>[
      FA \otimes_2 \top_2`
      FA`
      FA \otimes_2 F\top_1`
      F(A \otimes_1 \top_1);
      {\rho_2}_{FA}`
      \id_{FA} \otimes m_{\top_1}`
      F{\rho_1}_{A}`
      m_{A,\top_1}]
    \efig
     \end{mathpar}
     
      \begin{mathpar}
    \bfig
    \square|amma|/->`->`->`->/<1000,500>[
      FA \otimes_2 FB`
      FB \otimes_2 FA`
      F(A \otimes_1 B)`
      F(B \otimes_1 A);
      {\beta_2}_{FA,FB}`
      m_{A,B}`
      m_{B,A}`
      F{\beta_1}_{A,B}]
    \efig
  \end{mathpar}
\end{definition}
\noindent
The following is dual to the previous definition.
\begin{definition}
  \label{def:coSMCFUN}
  Suppose we are given two symmetric monoidal
  categories\\ $(\cat{M}_1,\perp_1,\oplus_1,\alpha_1,\lambda_1,\rho_1,\beta_1)$
  and
  $(\cat{M}_2,\perp_2,\oplus_2,\alpha_2,\lambda_2,\rho_2,\beta_2)$.
  Then a \textbf{symmetric comonoidal functor} is a functor $F :
  \cat{M}_1 \mto \cat{M}_2$, a map $m_{\perp_1} : F\perp_1 \mto
  \perp_2$ and a natural transformation $m_{A,B} : F(A \oplus_1 B)
  \mto FA \oplus_2 FB$ subject to the following coherence conditions:
  \begin{mathpar}
    \bfig
    \vSquares|ammmmma|/->`->`->``->`->`->/[
      F((A \oplus_1 B) \oplus_1 C)`
      F(A \oplus_1 B) \oplus_2 FC`
      F(A \oplus_1 (B \oplus_1 C))`
      (FA \oplus_2 FB) \oplus_2 FC`
      FA \oplus_2 F(B \oplus_1 C))`
      FA \oplus_2 (FB \oplus_2 FC);
      m_{A \oplus_1 B,C}`
      F\alpha_{A,B,C}`
      m_{A,B} \oplus_2 \id_{FC}``
      m_{A,B \oplus_1 C}`
      \alpha_{FA,FB,FC}`
      \id_{FA} \oplus_2 m_{B,C}]    
    \efig
  \end{mathpar}
  \begin{mathpar}
    \bfig
    \square|amma|/->`->`->`->/<1000,500>[
      F(\perp_1 \oplus_1 A)`
      F\perp_1 \oplus_2 FA`
      FA`
      \perp_2 \oplus_2 FA;
      m_{\perp_1,A}`
      F{\lambda_1}_{A}`
      m_{\perp_1} \oplus \id_{FA}`
      {\lambda^{-1}_2}_{FA}]
    \efig
    \and
    \bfig
    \square|amma|/->`->`->`->/<1000,500>[
      F(A \oplus_1 \perp_1)`
      FA \oplus_2 F\perp_1`
      FA`
      FA \oplus_2 \perp_2;
      m_{A,\perp_1}`
      F{\rho_1}_{A}`
      \id_{FA} \oplus m_{\perp_1}`
      {\rho^{-1}_2}_{FA}]
    \efig
  \end{mathpar}
      
  \begin{diagram}
    \square|amma|/->`->`->`->/<1000,500>[
      F(A \oplus_1 B)`
      FA \oplus_2 FB`
      F(B \oplus_1 A)`
      FB \oplus_2 FA;
      m_{A,B}`
      F{\beta_1}_{A,B}`
      {\beta_2}_{FA,FB}`
      m_{B,A}]
  \end{diagram}
\end{definition}

Naturally, since functors are enhanced to handle the additional
structure found in a symmetric monoidal category we must also extend
natural transformations, and adjunctions.
\begin{definition}
  \label{def:SMCNAT}
  Suppose $(\cat{M}_1,\top_1,\otimes_1)$ and $(\cat{M}_2,\top_2,\otimes_2)$
  are SMCs, and $(F,m)$ and $(G,n)$ are a symmetric monoidal functors
  between $\cat{M}_1$ and $\cat{M}_2$.  Then a \textbf{symmetric
    monoidal natural transformation} is a natural transformation,
  $f : F \mto G$, subject to the following coherence diagrams:
  \begin{mathpar}
    \bfig
    \square<1000,500>[
      FA \otimes_2 FB`
      F(A \otimes_1 B)`
      GA \otimes_2 GB`
      G(A \otimes_1 B);
      m_{A,B}`
      f_A \otimes_2 f_B`
      f_{A \otimes_1 B}`
      n_{A,B}]
    \efig
    \and
    \bfig
    \Vtriangle/->`<-`<-/[
      F\top_1`
      G\top_1`
      \top_2;
      f_{\top_1}`
      m_{\top_1}`
      n_{\top_1}]
    \efig
  \end{mathpar}  
\end{definition}
\begin{definition}
  \label{def:coSMCNAT}
  Suppose $(\cat{M}_1,\perp_1,\oplus_1)$ and $(\cat{M}_2,\perp_2,\oplus_2)$
  are SMCs, and $(F,m)$ and $(G,n)$ are a symmetric comonoidal functors
  between $\cat{M}_1$ and $\cat{M}_2$.  Then a \textbf{symmetric
    comonoidal natural transformation} is a natural transformation,
  $f : F \mto G$, subject to the following coherence diagrams:
  \begin{mathpar}
    \bfig
    \square<1000,500>[
      F(A \oplus_1 B)`
      FA \oplus_2 FB`
      G(A \oplus_1 B)`
      GA \oplus_2 GB;
      m_{A,B}`
      f_{A \oplus_1 B}`
      f_A \oplus_2 f_B`
      n_{A,B}]
    \efig
    \and
    \bfig
    \Vtriangle/<-`<-`<-/[
      \perp_2`
      G\perp_1`
      F\perp_1;
      n_{\perp_1}`
      m_{\perp_1}`
      f_{\perp_1}]
    \efig
  \end{mathpar}  
\end{definition}  
\begin{definition}
  \label{def:SMCADJ}
  Suppose $(\cat{M}_1,\top_1,\otimes_1)$ and $(\cat{M}_2,\top_2,\otimes_2)$
  are SMCs, and $(F,m)$ is a symmetric monoidal functor between
  $\cat{M}_1$ and $\cat{M}_2$ and $(G,n)$ is a symmetric monoidal
  functor between $\cat{M}_2$ and $\cat{M}_1$.  Then a
  \textbf{symmetric monoidal adjunction} is an ordinary adjunction
  $\cat{M}_1 : F \dashv G : \cat{M}_2$ such that the unit,
  $\eta_A : A \to GFA$, and the counit, $\varepsilon_A : FGA \to A$, are
  symmetric monoidal natural transformations.  Thus, the following
  diagrams must commute:
  \begin{mathpar}
    \bfig
    \square|amma|/->`->`->`<-/<1000,500>[
      FGA \otimes_2 FGB`
      F(GA \otimes_1 GB)`
      A \otimes_2 B`
      FGA \otimes_2 FGB;
      m_{GA,GB}`
      \varepsilon_A \otimes_1 \varepsilon_B`
      Fn_{A,B}`
      \varepsilon_{A \otimes_1 B}]
    \efig
    \and
    \bfig
    \square|amma|/->`<-`->`=/<1000,500>[
      F\top_1`
      FG\top_2`
      \top_2`
      \top_2;
      Fn_{\top_2}`
      m_{\top_1}`
      \varepsilon_{\top_1}`]    
    \efig
    \and
    \bfig
    \square|amma|/<-`->`->`->/<1000,500>[
      GFA \otimes_1 GFB`
      A \otimes_1 B`
      G(FA \otimes_2 FB)`
      GF(A \otimes_1 B);
      \eta_A \otimes_1 \eta_B`
      n_{FA,FB}`
      \eta_{A \otimes_1 B}`
      m_{A,B}]
    \efig
    \and
    \bfig
    \square|amma|/->`<-`<-`=/<1000,500>[
      G\top_2`
      GF\top_1`
      \top_1`
      \top_1;
      Gm_{\top_1}`
      n_{\top_2}`
      \eta_{\top_1}`]      
    \efig
  \end{mathpar} 
\end{definition}
\begin{definition}
  \label{def:coSMCADJ}
  Suppose $(\cat{M}_1,\perp_1,\oplus_1)$ and $(\cat{M}_2,\perp_2,\oplus_2)$
  are SMCs, and $(F,m)$ is a symmetric comonoidal functor between
  $\cat{M}_1$ and $\cat{M}_2$ and $(G,n)$ is a symmetric comonoidal
  functor between $\cat{M}_2$ and $\cat{M}_1$.  Then a
  \textbf{symmetric comonoidal adjunction} is an ordinary adjunction
  $\cat{M}_1 : F \dashv G : \cat{M}_2$ such that the unit,
  $\eta_A : A \to GFA$, and the counit, $\varepsilon_A : FGA \to A$, are
  symmetric comonoidal natural transformations.  Thus, the following
  diagrams must commute:
  \begin{mathpar}
    \bfig
    \square|amma|/->`->`->`<-/<1000,500>[
      A \oplus_1 B`
      GF(A \oplus_1 B)`
      GFA \oplus_1 GFB`
      G(FA \oplus_2 FB);
      \eta_{A \oplus_1 B}`
      \eta_A \oplus_1 \eta_B`
      Gm_{A,B}`
      m_{FA,FB}]
    \efig
    \and
    \bfig
    \square|amma|/->`<-`->`=/<1000,500>[
      GF \perp_1`
      G \perp_2`
      \perp_1`
      \perp_1;
      Gm_{\perp_1}`
      \eta_{\perp_1}`
      n_{\perp_2}`]
    \efig
    \and
    \bfig
    \square|amma|/->`->`->`<-/<1000,500>[
      FG(A \oplus_2 B)`
      F(GA \oplus_1 GB)`
      A \oplus_2 B`
      FGA \oplus_2 FGB;
      Fn_{A,B}`
      \varepsilon_{A \oplus_2 B}`
      m_{GA,GB}`
      \varepsilon_A \oplus_2 \varepsilon_B]
    \efig
    \and
    \bfig
    \square|amma|/->`=`<-`->/<1000,500>[
      FG\perp_2`
      \perp_2`
      FG\perp_2`
      F\perp_1;
      \varepsilon_{\perp_2}``
      m_{\perp_1}`
      Fn_{\perp_2}]
    \efig
  \end{mathpar}  
\end{definition}
We will be defining, and making use of the why-not exponentials from
linear logic, but these correspond to a symmetric comonoidal monad.
In addition, whenever we have a symmetric comonoidal adjunction, we
immediately obtain a symmetric comonoidal comonad on the left, and a
symmetric comonoidal monad on the right.
\begin{definition}
  \label{def:symm-comonoidal-monad}
  A \textbf{symmetric comonoidal monad} on a symmetric monoidal
  category $\cat{C}$ is a triple $(T,\eta, \mu)$, where
  $(T,\n{})$ is a symmetric comonoidal endofunctor on $\cat{C}$,
  $\eta_A : A \mto TA$ and $\mu_A : T^2A \to TA$ are
  symmetric comonoidal natural transformations, which make the following
  diagrams commute:
  \begin{mathpar}
    \bfig
    \square|ammb|<600,600>[
      T^3 A`
      T^2A`
      T^2A`
      TA;
      \mu_{TA}`
      T\mu_A`
      \mu_A`
      \mu_A]
    \efig
    \and
    \bfig
    \Atrianglepair/=`<-`=`->`<-/<600,600>[
      TA`
      TA`
      T^2 A`
      TA;`
      \mu_A``
      \eta_{TA}`
      T\eta_A]
    \efig
  \end{mathpar}
  The assumption that $\eta$ and $\mu$ are symmetric
  comonoidal natural transformations amount to the following diagrams
  commuting:
  \begin{mathpar}
    \bfig
    \ptriangle|amm|/->`->`<-/<1000,600>[
      A \oplus B`
      TA \oplus TB`
      T(A \oplus B);
      \eta_A \oplus \eta_B`
      \eta_A`
      \n{A,B}]    
    \efig
    \and
    \bfig
    \Vtriangle/->`=`->/<600,600>[
      \perp`
      T\perp`
      \perp;
      \eta_\perp``
      \n{\perp}]
    \efig
  \end{mathpar}
  \begin{mathpar}
    \bfig
    \square|ammm|/->`->``/<1050,600>[
      T^2(A \oplus B)`
      T(TA \oplus TB)`
      T(A \oplus B)`;
      T\n{A,B}`
      \mu_{A \oplus B}``]

    \square(850,0)|ammm|/->``->`/<1050,600>[
      T(TA \oplus TB)`
      T^2 A \oplus T^2 B``
      TA \oplus TB;
      \n{TA,TB}``
      \mu_A \oplus \mu_B`]
    \morphism(-200,0)<2100,0>[T(A \oplus B)`TA \oplus TB;\n{A,B}]
    \efig
    \and
    \bfig
    \square|ammb|/->`->`->`->/<600,600>[
      T^2\perp`
      T\perp`
      T\perp`
      \perp;
      T\n{\perp}`
      \mu_\perp`
      \n{\perp}`
      \n{\perp}]
    \efig
  \end{mathpar}
\end{definition}
\noindent
Finally, the dual concept of a symmetric comonoidal comonad.
\begin{definition}
  \label{def:symm-comonoidal-comonad}
  A \textbf{symmetric comonoidal comonad} on a symmetric monoidal
  category $\cat{C}$ is a triple $(T,\varepsilon, \delta)$, where
  $(T,\m{})$ is a symmetric comonoidal endofunctor on $\cat{C}$,
  $\varepsilon_A : TA \mto A$ and $\delta_A : TA \to T^2 A$ are
  symmetric comonoidal natural transformations, which make the
  following diagrams commute:
  \begin{mathpar}
    \bfig
    \square|amma|<600,600>[
      TA`
      T^2A`
      T^2A`
      T^3A;
      \delta_A`
      \delta_A`
      T\delta_A`
      \delta_{TA}]
    \efig
    \and
    \bfig
    \Atrianglepair/=`->`=`<-`->/<600,600>[
      TA`
      TA`
      T^2 A`
      TA;`
      \delta_A``
      \varepsilon_{TA}`
      T\varepsilon_A]
    \efig
  \end{mathpar}
  The assumption that $\varepsilon$ and $\delta$ are symmetric
  monoidal natural transformations amount to the following diagrams
  commuting:
  \begin{mathpar}
    \bfig
    \qtriangle|mmb|<1000,500>[
      T(A \oplus B)`
      TA \oplus TB`
      A \oplus B;
      \m{A,B}`
      \varepsilon_{A \oplus B}`
      \varepsilon_A \oplus \varepsilon_B]
    \efig
    \and
    \bfig
    \Vtriangle|amm|/->`=`<-/[
      T\perp`
      \perp`
      T\perp;
      \varepsilon_{\perp}``
      \m{\perp}]
    \efig    
  \end{mathpar}
  \begin{mathpar}
    \bfig
    \square|amab|/`->``->/<1050,600>[
      T(A \oplus B)``
      T^2(A \oplus B)`
      T(TA \oplus TB);`
      \delta_{A \oplus B}``
      T\m{A,B}]
    \square(1050,0)|mmmb|/``->`->/<1050,600>[`
      TA \oplus TB`
      T(TA \oplus TB)`
      T^2A \oplus T^2B;``
      \delta_A \oplus \delta_B`
      \m{TA,TB}]
    \morphism(0,600)<2100,0>[T(A \oplus B)`TA \oplus TB;\m{A,B}]
    \efig
    \and
    \bfig
    \square|amma|/->`->`<-`->/<600,600>[
      T\perp`
      \perp`
      T^2 \perp`
      T\perp;
      \m{\perp}`
      \delta_\perp`
      \m{\perp}`
      T\m{\perp}]
    \efig
  \end{mathpar}
\end{definition}

\subsection{Cartesian Closed and Cocartesian Coclosed Categories}
\label{subsec:cartesian_closed_and_cocartesian_coclosed_categories}
The notion of a cartesian closed category is well-known, but for
completeness we define them here.  However, their dual is lesser
known, especially in computer science, and so we given their full
definition.  We also review some know results concerning cocartesian
coclosed categories and categories that are both cartesian closed and
cocartesian coclosed.
\begin{definition}
  \label{def:CC}
  A \textbf{cartesian category} is a category, $(\cat{C}, 1, \times)$,
  with an object, $1$, and a bi-functor, $\times : \cat{C} \times
  \cat{C} \mto \cat{C}$, such that for any object $A$ there is exactly
  one morphism $\diamond : A \to 1$, and for any morphisms $f : C \mto
  A$ and $g : C \mto B$ there is a morphism $\langle f , g \rangle : C
  \to A \times B$ subject to the following diagram:
  \[
  \bfig
  \Atrianglepair/->`->`->`<-`->/[C`A`A\times B`B;
    f`
    \langle f , g \rangle`
    g`
    \pi_1`
    \pi_2]
  \efig
  \]
\end{definition}
A cartesian category models conjunction by the product functor,
$\times : \cat{C} \times \cat{C} \mto \cat{C}$ , and the unit of
conjunction by the terminal object.  As we mention above modeling
implication requires closure, and since it is well-known that any
cartesian category is also a symmetric monoidal category the
definition of closure for a cartesian category is the same as the
definition of closure for a symmetric monoidal category
(Definition~\ref{def:SMCC}).  We denote the internal hom for cartesian
closed categories by $A \to B$.

The dual of a cartesian category is a cocartesian category.  They are
a model of intuitionistic logic with disjunction and its unit.
\begin{definition}
  \label{def:coCC}
  A \textbf{cocartesian category} is a category, $(\cat{C}, 0, +)$,
  with an object, $0$, and a bi-functor,
  $+ : \cat{C} \times \cat{C} \mto \cat{C}$, such that for any object $A$ there is exactly
  one morphism $\Box : 0 \to A$, and for any morphisms $f : A \mto C$ and $g : B \mto C$
  there is a morphism $\lbrack f , g \rbrack : A + B \mto C$
  subject to the following diagram:
  \[
  \bfig
  \Atrianglepair/<-`<-`<-`->`<-/[C`A`A+B`B;
    f`
    \lbrack f, g \rbrack`
    g`
    \iota_1`
    \iota_2]
  \efig
  \]  
\end{definition}
Coclosure, just like closure for cartesian categories, is defined in
the same way that coclosure is defined for symmetric monoidal
categories, because cocartesian categories are also symmetric
monoidal categories.  Thus, a cocartesian category is coclosed if
there is a specified left-adjoint, which we denote $S - T$, to the
coproduct.

There are many examples of cocartesian coclosed categories.
Basically, any interesting cartesian category has an interesting dual,
and hence, induces an interesting cocartesian coclosed category.
The opposite of the category of sets and functions between them is
isomorphic to the category of complete atomic boolean algebras, and
both of which, are examples of cocartesian coclosed categories.  As
we mentioned above bi-linear categories \cite{Cockett:1997a} are
models of bi-linear logic where the left adjoint to the cotensor
models coimplication.  Similarly, cocartesian coclosed categories
model cointuitionistic logic with disjunction and intuitionistic
coimplication

We might now ask if a category can be both cartesian closed and
cocartesian coclosed just as bi-linear categories, but this turns out
to be where the matter meets antimatter in such away that the category
degenerates to a preorder.  That is, every homspace contains at most
one morphism.  We recall this proof here, which is due to Crolard
\cite{Crolard:2001}. We need a couple basic facts about cartesian
closed categories with initial objects.
\begin{lemma}
  \label{lemma:iso-prod-initial}
  In any cartesian category $\cat{C}$, if $0$ is an initial object in
  $\cat{C}$ and $\Hom{C}{A}{0}$ is non-empty, then $A \cong A \times 0$.
\end{lemma}
\begin{proof}
  This follows easily from the universial mapping property for products.
\end{proof}

\begin{lemma}
  \label{lemma:products-of-initial-gives-initial}
  In any cartesian closed category $C$, if $0$ is an initial
  object in $\cat{C}$, then so is $0 \times A$ for any object $A$
  of $\cat{C}$.
\end{lemma}
\begin{proof}
    We know that the universal morphism for the initial object is
    unique, and hence, the homspace $\Hom{C}{0}{A \Rightarrow B}$ for
    any object $B$ of $\cat{C}$ contains exactly one morphism.  Then
    using the right adjoint to the product functor we know that
    $\Hom{C}{0}{A \Rightarrow B} \cong \Hom{C}{0 \times A}{B}$, and
    hence, there is only one arrow between $0 \times A$ and $B$.
\end{proof}
\noindent
The following lemma is due to Joyal \cite{Lambek:1988}, and is key to
the next theorem.
\begin{lemma}[Joyal]
  \label{lemma:joyals}
  In any cartesian closed category $\cat{C}$, if $0$ is an initial
  object in $\cat{C}$ and $\Hom{C}{A}{0}$ is non-empty, then $A$ is an
  initial object in $\cat{C}$.
\end{lemma}
\begin{proof}
  Suppose $\cat{C}$ is a cartesian closed category, such that, $0$ is
  an initial object in $\cat{C}$, and $A$ is an arbitrary object in
  $\cat{C}$.  Furthermore, suppose $\Hom{C}{A}{0}$ is non-empty.  By
  the first basic lemma above we know that $A \cong A \times 0$, and
  by the second $A \times 0$ is initial, thus $A$ is initial.
\end{proof}
Finally, the following theorem shows that any category that is both
cartesian closed and cocartesian coclosed is a preorder.
\begin{theorem}[(co)Cartesian (co)Closed Categories are Preorders (Crolard\cite{Crolard:2001})]
  \label{thm:dengerate-to-preorder}
  If $\cat{C}$ is both cartesian closed and cocartesian coclosed, then
  for any two objects $A$ and $B$ of $\cat{C}$, $\Hom{C}{A}{B}$ has at
  most one element.
\end{theorem}
\begin{proof}
  Suppose $\cat{C}$ is both cartesian closed and cocartesian coclosed,
  and $A$ and $B$ are objects of $\cat{C}$.  Then by using the basic
  fact that the initial object is the unit to the coproduct, and the
  coproducts left adjoint we know the following:
  \[\Hom{C}{A}{B} \cong \Hom{C}{A}{0 + B} \cong \Hom{C}{B - A}{0}\]
  Therefore, by Joyal's theorem above $\Hom{C}{A}{B}$ has at most one
  element.
\end{proof}
\noindent
Notice that the previous result hinges on the fact that there are
initial and terminal objects, and thus, this result does not hold for
bi-linear categories, because the units to the tensor and cotensor are
not initial nor terminal.

The repercussions of this result are that if we do not want to work
with preorders, but do want to work with all of the structure, then we
must separate the two worlds.  Thus, this result can be seen as the
motivation for the current work.  We enforce the separation using
linear logic, but through the power of linear logic this separation is
not over a large distance.

\subsection{A Mixed Linear/Non-Linear Model for Co-Intuitionistic Logic}
\label{subsec:a_mixed_linear/non-linear_model_for_co-intuitionistic_logic}

Benton \cite{Benton:1994} showed that from a LNL model it is possible
to construct a linear category, and vice versa.  Bellin
\cite{Bellin:2012} showed that the dual to linear categories are
sufficient to model co-intuitionistic linear logic. We show that from
the dual to a LNL model we can construct the dual to a linear
category, and vice versa, thus, carrying out the same program for
co-intuitionistic linear logic as Benton did for intuitionistic linear
logic.

Combining a symmetric monoidal coclosed category with a cocartesian
coclosed category via a symmetric comonoidal adjunction defines a
dual LNL model.
\begin{definition}
  \label{def:dual LNL-model}
  A
  \textbf{mixed linear/non-linear model for co-intuitionistic logic (dual LNL model)},
  $\cat{L} : \func{H} \dashv \func{J} : \cat{C}$, consists of the following:
  \begin{itemize}
  \item[i.] a symmetric monoidal coclosed category $(\cat{L},\perp,\oplus,\colimp)$,
  \item[ii.] a cocartesian coclosed category $(\cat{C},0,+,-)$, and
  \item[iv.] a symmetric comonoidal adjunction $\cat{L} : \func{H}
    \dashv \func{J} : \cat{C}$, where $\eta_A : A \mto
    \func{JH}A$ and $\varepsilon_R : \func{HJ}R \mto R$
    are the unit and counit of the adjunction respectively.
  \end{itemize}
\end{definition}
It is well-known that an adjunction $\cat{L} : \func{H} \dashv
\func{J} : \cat{C}$ induces a monad $\func{H};\func{J} : \cat{L} \mto
\cat{L}$, but when the adjunction is symmetric comonoidal we obtain a
symmetric comonoidal monad, in fact, $\func{H};\func{J}$ defines the
linear exponential why-not denoted $\wn A = \func{JH}A$.
By the definition of dual LNL models we know that both $\func{H}$ and
$\func{J}$ are symmetric comonoidal functors, and hence, are equipped
with natural transformations $\h{A,B} : \func{H}(A \oplus B) \mto
\func{H}A + \func{H}B$ and $\j{R,S} : \func{J}(R + S) \mto \func{J}R
\oplus \func{J}S$, and maps $\h{\perp} : \func{H}\perp \mto 0$ and
$\j{0} : \func{J}0 \mto \perp$.  We will make heavy use of these
maps throughout the sequel.

Compare this definition with that of Bellin's dual linear category from
\cite{Bellin:2012}, and we can easily see that the definition of dual
LNL models -- much like LNL models -- is more succinct.
\begin{definition}
  \label{def:dual-linear-cat}
  A \textbf{dual linear category}, $\cat{L}$, consists of the
  following data:
  \begin{enumerate}[label=\roman*.]
  \item A symmetric monoidal coclosed category $(\cat{L}, \oplus, \perp, \colimp)$ with
  \item a symmetric co-monoidal monad $(\wn,\eta, \mu)$ on $\cat{L}$
    such that
  \begin{enumerate}[label=\alph*.]
  \item each free $\wn$-algebra carries naturally the structure of a
    commutative $\oplus$-monoid.  This implies that there are
    distinguished symmetric monoidal natural transformations $\w{A} :
    \perp \mto \wn A$ and $\c{A} : \wn A \oplus \wn A \mto \wn A$
    which form a commutative monoid and are $\wn$-algebra morphisms.

  \item whenever $f : (\wn A, \mu_A) \mto (\wn B, \mu_B)$ is a
    morphism of free $\wn$-algebras, then it is also a monoid
    morphism.
  \end{enumerate}
  \end{enumerate}
\end{definition}

\subsubsection{A Useful Isomorphism}
\label{subsec:a_useful_isomorphism}
One useful property of Benton's LNL model is that the maps associated
with the symmetric monoidal left adjoint in the model are
isomorphisms.  Since dual LNL models are dual we obtain similar
isomorphisms with respect to the right adjoint.
\begin{lemma}[Symmetric Comonoidal Isomorphisms]
  \label{lemma:symmetric_comonoidal_isomorphisms}
  Given any dual LNL model $\cat{L} : \func{H} \dashv \func{J} : \cat{C}$, then there are the following isomorphisms:
  \[
  \begin{array}{lll}
    \func{J}(R + S) \cong \func{J}R \oplus \func{J}S & \text{ and } & \func{J}0 \cong \perp\\
  \end{array}
  \]
  Furthermore, the former is natural in $R$ and $S$.  
\end{lemma}
\begin{proof}
  Suppose $\cat{L} : \func{H} \dashv \func{J} : \cat{C}$ is a dual LNL
  model.  Then we can define the following family of maps:
  \[
  \begin{array}{lll}
    \jinv{R,S} := \func{J}R \oplus \func{J}S \mto^{\eta} \func{JH}(\func{J}R \oplus \func{J}S) \mto^{\func{J}\h{A,B}} \func{J}(\func{HJ}R + \func{HJ}S) \mto^{\J(\varepsilon_R + \varepsilon_S)} \func{J}(R + S)\\
  \\
  \jinv{0} := \perp \mto^{\eta} \func{JH}\perp \mto^{\func{J}\h{\perp}} \func{J}0
  \end{array}
  \]
  It is easy to see that $\jinv{R,S}$ is natural, because it is
  defined in terms of a composition of natural transformations.  All
  that is left to be shown is that $\jinv{R,s}$ and $\jinv{0}$ are
  mutual inverses with $\j{R,S}$ and $\j{0}$; for the details see
  Appendix~\ref{subsec:proof_of_lemma:symmetric_comonoidal_isomorphisms}.
\end{proof}
\noindent
Just as Benton we also do not have similar isomorphisms with respect
to the functor $\H$.  One fact that we can point out, that Benton did
not make explicit -- because he did not use the notion of symmetric
comonoidal functor -- is that $\jinv{}$ makes $\J$ also a symmetric
monoidal functor.

\begin{corollary}
  \label{corollary:J-SMMF}
  Given any dual LNL model $\cat{L} : \func{H} \dashv \func{J} :
  \cat{C}$, the functor $(\J, \jinv{})$ is symmetric monoidal.
\end{corollary}
\begin{proof}
  This holds by straightforwardly reducing the diagrams defining a
  symmetric monoidal functor, Definition~\ref{def:SMCFUN}, to the
  diagrams defining a symmetric comonoidal functor,
  Definition~\ref{def:coSMCFUN}, using the fact that $\jinv{}$ is an
  isomorphism.
\end{proof}

\subsubsection{Dual LNL Model Implies Dual Linear Category}
\label{subsec:dual_lnl_model_implies_dual_category}

The next result shows that any dual LNL model induces a symmetric
comonoidal monad.
\begin{lemma}[Symmetric Comonoidal Monad]
  \label{lemma:symmetric_comonoidal_monad}
  Given a dual LNL model $\cat{L} : \func{H} \dashv \func{J} : \cat{C}$,
  the functor, $\wn = H;J$, defines a symmetric comonoidal monad.
\end{lemma}
\begin{proof}
  Suppose $(\func{H},\h{})$ and $(\func{J},\j{})$ are two symmetric
  comonoidal functors, such that, $\cat{L} : \func{H} \dashv \func{J}
  : \cat{C}$ is a dual LNL model.  We can easily show that $\wn A = \J\H
  A$ is symmetric monoidal by defining the following maps:
  \[
  \begin{array}{lll}
    \r{\perp} := \wn \perp \mto/=/ \func{JH}\perp \mto^{\func{J}\h{\perp}} \func{J}0 \mto^{\j{\perp}} \perp\\
    \r{A,B} := \wn (A \oplus B) \mto/=/ \func{JH}(A \oplus B) \mto^{\func{J}\h{A,B}} \func{J}(\func{H}A + \func{H}B) \mto^{\j{\func{H}A,\func{H}B}} \func{JH}A \oplus \func{JH}B \mto/=/ \wn A \oplus \wn B\\
  \end{array}
  \]
  The fact that these maps satisfy the appropriate symmetric
  comonoidal functor diagrams from Definition~\ref{def:coSMCFUN} is
  obvious, because symmetric comonoidal functors are closed under
  composition.  

  We have a dual LNL model, and hence, we have the symmetric comonoidal
  natural transformations $\eta_A : A \mto \J\H A$ and $\varepsilon_R
  : \H\J R \mto R$ which correspond to the unit and counit of the
  adjunction respectively.  Define $\mu_A := \J\varepsilon_{\H A} :
  \J\H\J\H A \mto \J\H A$.  This implies that we have maps $\eta_A : A
  \mto \wn A$ and $\mu_A : \wn\wn A \mto \wn A$, and thus, we can show
  that $(\wn, \eta, \mu)$ is a symmetric comonoidal monad.  All
  the diagrams defining a symmetric comonoidal monad hold by the
  structure given by the adjunction.  For the complete proof see
  Appendix~\ref{subsec:proof_of_lemma:symmetric_comonoidal_monad}.
\end{proof}

The monad from the previous result must be equipped with the
additional structure to model the right weakening and contraction
structural rules.

\begin{lemma}[Right Weakening and Contraction]
  \label{lemma:right_weakening_and_contraction}
  Given a dual LNL model $\cat{L} : \func{H} \dashv \func{J} : \cat{C}$,
  then for any $\wn A$ there are distinguished symmetric comonoidal
  natural transformations $\w{A} : \perp \mto \wn A$ and $\c{A} : \wn
  A \oplus \wn A \mto \wn A$ that form a commutative monoid, and are
  $\wn\text{-algebra}$ morphisms with respect to the canonical
  definitions of the algebras $\wn A$, $\perp$, $\wn A \oplus \wn A$.
\end{lemma}
\begin{proof}
  Suppose $(\func{H},\h{})$ and $(\func{J},\j{})$ are two symmetric
  comonoidal functors, such that, $\cat{L} : \func{H} \dashv \func{J}
  : \cat{C}$ is a dual LNL model.  Again, we know $\wn A = H;J : \cat{L}
  \mto \cat{L}$ is a symmetric comonoidal monad by
  Lemma~\ref{lemma:symmetric_comonoidal_monad}.  
  
  We define the following morphisms:
  \[
  \begin{array}{lll}
    \w{A} := \perp \mto^{\jinv{\perp}} \func{J} 0 \mto^{\func{J}\diamond_{\func{H} A}} \func{JH}A \mto/=/ \wn A\\
    \c{A} := \wn A \oplus \wn A \mto/=/ \func{JH}A \oplus \func{JH}A \mto^{\jinv{\func{H}A,\func{H}A}} \func{J}(\func{H}A + \func{H}A) \mto^{\func{J}\codiag{\func{H}A}} \func{JHA} \mto/=/ \wn A
  \end{array}
  \]
  The remainder of the proof is by carefully checking all of the
  required diagrams.  Please see
  Appendix~\ref{subsec:proof_of_lemma:right_weakening_and_contraction}
  for the complete proof.
\end{proof}

\begin{lemma}[$\wn$-Monoid Morphisms]
  \label{lemma:monoid-morphism}
  Suppose $\cat{L} : \func{H} \dashv \func{J} : \cat{C}$ is a dual LNL
  model.  Then if $f : (\wn A, \mu_A) \mto (\wn B, \mu_B)$ is a
  morphism of free $\wn$-algebras, then it is a monoid morphism.
\end{lemma}
\begin{proof}
  Suppose $\cat{L} : \func{H} \dashv \func{J} : \cat{C}$ is a dual LNL
  model.  Then we know $\wn A = \J\H A$ is a symmetric comonoidal
  monad by Lemma~\ref{lemma:symmetric_comonoidal_monad}.  Bellin
  \cite{Bellin:2012} remarks that by Maietti, Maneggia de Paiva and
  Ritter's Proposition~25 \cite{Maietti2005}, it suffices to show that
  $\mu_A : \wn\wn A \mto \wn A$ is a monoid morphism.  For the details
  see the complete proof in
  Appendix~\ref{sec:proof_of_lemma:monoid-morphism}.
\end{proof}

\noindent
Finally, we may now conclude the following corollary.
\begin{corollary}
  \label{corollary:dual_lnl_model_implies_dual_category}
  Every dual LNL model is a dual linear category.
\end{corollary}

\subsubsection{Dual Linear Category implies Dual LNL Model}
\label{subsec:dual_category_implies_dual_lnl_model}
This section shows essentially the inverse to the result from the
previous section.  That is, from any dual linear category we may
construct a dual LNL model.  By exploiting the duality between LNL
models and dual LNL models this result follows straightforwardly from
Benton's result. The proof of this result must first find a symmetric
monoid coclosed category, a cocartesian coclosed category, and
finally, a symmetric comonoidal adjunction between them.  Take the
symmetric monoid coclosed category to be an arbitrary dual linear
category $\cat{L}$.  Then we may define the following categories.
\begin{itemize}
\item The Eilenberg-Moore category, $\cat{L}^{\wn}$, has as objects
  all $\wn$-algebras, $(A, h_A : \wn A \mto A)$, and as morphisms all
  $\wn$-algebra morphisms.
\item The Kleisli category, $\cat{L}_{\wn}$, is the full subcategory
  of $\cat{L}^{\wn}$ of all free $\wn$-algebras $(\wn A, \mu_A :
  \wn\wn A \mto \wn A)$.
\end{itemize}
\noindent
The previous three categories are related by a pair of adjunctions:
\begin{diagram}
  \morphism(0,0)|m|/{@{->}@/^1em/}/<800,0>[\cat{L}`\cat{L}^{\wn};F]
  \morphism(0,0)|m|/{@{<-}@/_1em/}/<800,0>[\cat{L}`\cat{L}^{\wn};U]

  \morphism(0,-500)|m|/{@{->}@/^1em/}/<800,0>[\cat{L}`\cat{L}_{\wn};F]
  \morphism(0,-500)|m|/{@{<-}@/_1em/}/<800,0>[\cat{L}`\cat{L}_{\wn};U]

  \square|amma|/`=`->`/<800,-500>[
    \cat{L}`
    \cat{L}_{\wn}`
    \cat{L}`
    \cat{L}^{\wn};``
    i`]
\end{diagram}
The functor $F(A) = (\wn A, \mu_A)$ is the free functor, and the
functor $U(A, h_A) = A$ is the forgetful functor.  Note that we, just
as Benton did, are overloading the symbols $F$ and $U$.  Lastly, the
functor $i : \cat{L}_{\wn} \mto \cat{L}^{\wn}$ is the injection of the
subcategory of free $\wn$-algebras into its parent category.  

We are now going to show that both $\cat{L}^{\wn}$ and $\cat{L}_{\wn}$
are induce two cocartesian coclosed categories.  Then we could take
either of those when constructing a dual LNL model from a dual linear
category.  First, we show $\cat{C}^{\wn}$ is cocartesian.
\begin{lemma}
  \label{lemma:EM-has-coproducts}
  If $\cat{L}$ is a dual linear category, then $\cat{L}^{\wn}$ has finite coproducts.
\end{lemma}
\begin{proof}
  We give a proof sketch of this result, because the proof is
  essentially by duality of Benton's corresponding proof for LNL
  models (see Lemma~9, \cite{Benton:1994}). Suppose $\cat{L}$ is a
  dual linear category.  Then we first need to identify the initial object
  which is defined by the $\wn$-algebra $(\perp, \r{\perp} : \wn \perp
  \mto \perp)$.  The unique map between the initial map and any other
  $\wn$-algebra $(A, h_A : \wn A \mto A)$ is defined by $\perp
  \mto^{\w{A}} \wn A \mto^{h_A} A$.  The coproduct of the
  $\wn$-algebras $(A, h_A : \wn A \mto A)$ and $(B, h_B : \wn B \mto
  B)$ is $(A \oplus B, \r{A,B};(h_A \oplus h_B))$.  Injections and the
  codiagonal map are defined as follows:
  \begin{itemize}
  \item Injections:
    \[
    \begin{array}{lll}
      \iota_1 := A \mto^{\rho_A} A \oplus \perp \mto^{\id_A \oplus \w{B}} A \oplus \wn B \mto^{\id \oplus h_b} A \oplus B\\
      \iota_2 := B \mto^{\lambda_A} \perp \oplus B \mto^{\w{A} \oplus \id_B} \wn A \oplus B \mto^{h_A \oplus \id_B} A \oplus B\\
    \end{array}
    \]
  \item Codiagonal map:
    \[
    \codiag{} := A \oplus A \mto^{\eta_A \oplus \eta_A} \wn A \oplus \wn A \mto^{\c{A}} \wn A \mto^{h_A} A
    \]
  \end{itemize}
  Showing that these respect the appropriate diagrams is
  straightforward.
\end{proof}
\noindent
Notice as a direct consequence of the previous result we know the following.
\begin{corollary}
  \label{corollary:Kleisli-has-coproducts}
  The Kleisli category, $\cat{L}_{\wn}$, has finite coproducts.
\end{corollary}

Thus, both $\cat{L}^{\wn}$ and $\cat{L}_{\wn}$ are cocartesian, but we
need a cocartesian coclosed category, and in general these are not
coclosed, and so we follow Benton's lead and show that there are
actually two subcategories of $\cat{L}^{\wn}$ that are coclosed.
\begin{definition}
  \label{def:subtractable}
  We call an object, $A$, of a category, $\cat{L}$,
  \textbf{subtractable} if for any object $B$ of $\cat{L}$, the
  internal cohom $A \colimp B$ exists.
\end{definition}
\noindent
We now have the following results:
\begin{lemma}
  \label{lemma:free-algebras-are-subtractable}
  In $\cat{L}^{\wn}$, all the free $\wn$-algebras are subtractable, and the
  internal cohom is a free $\wn$-algebra.
\end{lemma}
\begin{proof}
  The internal cohom is defined as follows:
  \[
  (\wn A, \delta_A) \colimp (B, h_B) := (\wn (A \colimp B), \delta_{A \colimp B})
  \]
  We can capitalize on the adjunctions involving $F$ and $U$ from
  above to lift the internal cohom of $\cat{L}$ into $\cat{L}^{\wn}$:
  \[
  \begin{array}{lll}
    \mathsf{Hom}_{\cat{L}^{\wn}}((\wn (A \colimp B), \delta_{A \colimp B}), (C, h_C))
    & =  & \mathsf{Hom}_{\cat{L}^{\wn}}(\func{F}(A \colimp B), (C, h_C))\\
    & \cong  & \mathsf{Hom}_{\cat{L}}(A \colimp B, \func{U}(C, h_C))\\
    & =  & \mathsf{Hom}_{\cat{L}}(A \colimp B, C)\\
    & \cong  & \mathsf{Hom}_{\cat{L}}(A, C \oplus B)\\
    & =  & \mathsf{Hom}_{\cat{L}}(A, \func{U}(C \oplus B,h_{C \oplus B}))\\
    & \cong  & \mathsf{Hom}_{\cat{L}^{\wn}}(\func{F}A, (C \oplus B,h_{C \oplus B}))\\
    & =  & \mathsf{Hom}_{\cat{L}^{\wn}}((\wn A, \delta_A), (C \oplus B,h_{C \oplus B}))\\    
  \end{array}
  \]
  The previous equation holds for any $h_{C \oplus B}$ making $C \oplus B$ a $\wn$-algebra, in particular, the co-product in
  $\cat{L}^{\wn}$ (Lemma~\ref{lemma:EM-has-coproducts}), and hence, we
  may instantiate the final line of the previous equation with the following:
  \[
  \mathsf{Hom}_{\cat{L}^{\wn}}((\wn A, \delta_A), (C, h_c) + (B,\delta_A))
  \]
  Thus, we obtain our result.
\end{proof}

\begin{lemma}
  \label{lemma:subtractable-subcats}
  We have the following cocartesian coclosed categories:
  \begin{enumerate}[label=\roman*.]
  \item The full subcategory, $\mathsf{Sub}(\cat{L}^{\wn})$, of
    $\cat{L}^{\wn}$ consisting of objects the subtractable
    $\wn$-algebras is cocartesian coclosed, and contains the Kleisli
    category.

  \item The full subcategory, $\cat{L}^*_{\wn}$, of
    $\mathsf{Sub}(\cat{L}^{\wn})$ consisting of finite coproducts of
    free $\wn$-algebras is cocartesian coclosed.
  \end{enumerate}
\end{lemma}
\noindent
Let $\cat{C}$ be either of the previous two categories.  Then we must
exhibit a adjunction between $\cat{C}$ and $\cat{L}$, but this is
easily done.
\begin{lemma}
  \label{lemma:dual-LNL-model}
  The adjunction $\cat{L} : \func{F} \vdash \func{U} : \cat{C}$, with
  the free functor, $\func{F}$, and the forgetful functor, $\func{U}$,
  is symmetric comonoidal.
\end{lemma}
\begin{proof}
  Showing that $\func{F}$ and $\func{U}$ are symmetric comonoidal
  follows similar reasoning to Benton's result, but in the opposite;
  see Lemma~13 and Lemma~14 of \cite{Benton:1994}.  Lastly, showing
  that the unit and the counit of the adjunction are comonoidal
  natural transformations is straightforward, and we leave it to the
  reader. The reasoning is similar to Benton's, but in the opposite;
  see Lemma~15 and Lemma~16 of \cite{Benton:1994}.
\end{proof}

\begin{corollary}
  \label{corollary:dual-categories-implies-LNL-models}
  Any dual linear category gives rise to a dual LNL model.
\end{corollary}



\subsection{A Mixed Bilinear/Non-Linear Model}
\label{subsec:a_mixed_bi-linear_non-linear_model}

The main goal of our research program is to give a non-trivial
categorical model of bi-intuitionistic logic.  In this section we give
a introduction of the model we have in mind, but leave the details and
the study of the logical and programmatic sides to future work.

The naive approach would be to try and define a LNL-style model of
bi-intuitionistic logic as an adjunction between a bilinear category
and a bi-cartesian bi-closed category, but this results in a few
problems.  First, should the adjunction be monoidal or comonoidal?
Furthermore, we know bi-cartesian bi-closed categories are trivial
(Theorem~\ref{thm:dengerate-to-preorder}), and hence, this model is
not very interesting nor correct.  We must separate the two
worlds using two dual adjunctions, and hence, we arrive at the
following definition.
\begin{definition}
  \label{def:biLNL-model}
  A \textbf{mixed bilinear/non-linear model} consists of the
  following:
  \begin{enumerate}[label=\roman*.]
  \item a bilinear category $(\cat{L},
    \top,\otimes,\limp,\perp,\oplus,\colimp)$,
  \item a cartesian closed category $(\cat{I},1,\times,\to)$,
  \item a cocartesian coclosed category $(\cat{C},0,+,-)$, 
  \item a LNL model $\cat{I} : F \dashv G : \cat{L}$, and
  \item a dual LNL model $\cat{L} : H \dashv J : \cat{C}$.
  \end{enumerate}
\end{definition}
Since $\cat{L}$ is a bilinear category then it is also a linear
category, and a dual linear category.  Thus, the LNL model intuitively
corresponds to an adjunction between $\cat{I}$ and the linear
subcategory of $\cat{L}$, and the dual LNL model corresponds to an
adjunction between the dual linear subcategory of $\cat{L}$ and
$\cat{C}$.  In addition, both intuitionistic logic and
cointuitionistic logic can be embedded into $\cat{L}$ via the linear
modalities of-course, $!A$, and why-not, $\wn A$, using the well-known
Girard embeddings.  This implies that we have a very controlled way of
mixing $\cat{I}$ and $\cat{C}$ within $\cat{L}$, and hence, linear
logic is the key.


\section{Dual LNL Logic}
\label{sec:dual_lnl_logic}
We now turn to developing the syntactic side of dual LNL models called
dual LNL logic (DLNL).  First, we give a sequent calculus
formalization which we will simply refer to as DLNL logic, then a
natural deduction formalization called DND logic, and finally a term
assignment to the natural deduction version.  Each of these systems
will consistently use the same syntax and naming conventions for
formulas, types, and contexts given by the following definition.
\begin{definition}
  \label{def:DLNL-syntax-formulas-ctx}
  The the syntax for formulas, types, and contexts are given as follows:
  \[
  \begin{array}{rllll}
    \text{(non-linear formulas/types)} & \DualLNLLogicnt{R},\DualLNLLogicnt{S},\DualLNLLogicnt{T} ::= \DualLNLLogicsym{0} \mid  \DualLNLLogicnt{S}  +  \DualLNLLogicnt{T}  \mid  \DualLNLLogicnt{S}  -  \DualLNLLogicnt{T}  \mid  \mathsf{H}\, \DualLNLLogicnt{A} \\
    \text{(linear formulas/types)}     & \DualLNLLogicnt{A},\DualLNLLogicnt{B},\DualLNLLogicnt{C} ::=  \perp  \mid  \DualLNLLogicnt{A}  \oplus  \DualLNLLogicnt{B}  \mid  \DualLNLLogicnt{A}  \colimp  \DualLNLLogicnt{B}  \mid  \mathsf{J}\, \DualLNLLogicnt{S} \\
    \text{(non-linear contexts)}       & \Psi,\Theta ::=  \cdot  \mid \DualLNLLogicnt{T} \mid \Psi  \DualLNLLogicsym{,}  \Theta\\
    \text{(linear contexts)}           & \Gamma,\Delta  ::=  \cdot  \mid \DualLNLLogicnt{A} \mid \Gamma  \DualLNLLogicsym{,}  \Delta\\
  \end{array}
  \]
\end{definition}
\noindent
The term assignment will index contexts by terms, but we will maintain
the same naming convention throughout.

\subsection{The Sequent Calculus for Dual LNL Logic}
\label{sec:sequent_calculus}

\begin{figure}
  \begin{mdframed}
    \begin{mathpar}
      \DualLNLLogicdruleCXXid{} \and
      \DualLNLLogicdruleCXXwk{} \and
      \DualLNLLogicdruleCXXcr{} \and
      \DualLNLLogicdruleCXXex{} \and
      \DualLNLLogicdruleCXXfL{} \and
      \DualLNLLogicdruleCXXdL{} \and
      \DualLNLLogicdruleCXXdROne{} \and
      \DualLNLLogicdruleCXXdRTwo{} \and
      \DualLNLLogicdruleCXXsL{} \and
      \DualLNLLogicdruleCXXsR{} \and
      \DualLNLLogicdruleCXXcut{} \and
      \DualLNLLogicdruleCXXhL{}     
    \end{mathpar}
  \end{mdframed}
  \caption{Non-linear fragment of the DLNL logic}
  \label{fig:non-linear-sequent}
\end{figure}

\begin{figure}
  \begin{mdframed}
    \begin{mathpar}
      \DualLNLLogicdruleLXXid{} \and
      \DualLNLLogicdruleLXXwk{} \and
      \DualLNLLogicdruleLXXctr{} \and
      \DualLNLLogicdruleLXXex{} \and
      \DualLNLLogicdruleLXXCex{} \and
      \DualLNLLogicdruleLXXcut{} \and
      \DualLNLLogicdruleLXXCcut{} \and
      \DualLNLLogicdruleLXXflL{} \and
      \DualLNLLogicdruleLXXflR{} \and
      \DualLNLLogicdruleLXXdROne{} \and
      \DualLNLLogicdruleLXXdRTwo{} \and
      \DualLNLLogicdruleLXXpL{} \and
      \DualLNLLogicdruleLXXpR{} \and
      \DualLNLLogicdruleLXXsL{} \and
      \DualLNLLogicdruleLXXsR{} \and
      \DualLNLLogicdruleLXXCsR{} \and
      \DualLNLLogicdruleLXXjL{} \and
      \DualLNLLogicdruleLXXjR{} \and
      \DualLNLLogicdruleLXXhR{}      
    \end{mathpar}
  \end{mdframed}
  \caption{Linear fragment of the DLNL logic}
  \label{fig:linear-fragment-sequent}
\end{figure}

In this section we take the dual of Benton's~\cite{Benton:1994}
sequent calculus for LNL logic to obtain the sequent calculus for dual
LNL logic.  The inference rules for the non-linear fragment can be
found in Figure~\ref{fig:non-linear-sequent} and the linear fragment
in Figure~\ref{fig:linear-fragment-sequent}. The remainder of this
section is devoted to proving cut-elimination.  However, the proof is
simply a dualization of Benton's~\cite{Benton:1994} proof of
cut-elimination for LNL logic.

Just as Benton we use $n$-ary cuts:
\[
\DualLNLLogicdruleCXXmcut{} \quad \DualLNLLogicdruleLXXCmcut{}
\]
where $S^n = S, \ldots, S$ $n$-times. We call $\DLNLP$ the system DLNL with $n$-cuts replacing 
ordinary 1-cuts. Such cuts are admissible in DLNL and cut-elimination for $\DLNLP$ implies 
cut-elimination for DLNL.

We begin with a few standard definitions. The \emph{rank} of a
formula, denoted by $|A|$ or $|S|$, is the number of the logical
symbols in the given formula.  The \emph{cut-rank} of a derivation
$\Pi$, denoted by $c(\Pi)$, is the maximum of the ranks of the cut
formulas in $\Pi$ plus one; if $\Pi$ is cut-free its cut rank is 0.
Finally, the \emph{depth} of a derivation $\Pi$, denoted by $d(\Pi)$,
is the length of the longest path in $\Pi$.  The following three
results establish cut elimination.

\begin{lemma}[Cut Reduction]
  \label{lemma:cut_reduction}
  The following defines the cut reduction procedure:
  \begin{enumerate}
  \item If $\Pi_1$ is a derivation of $ \DualLNLLogicnt{T}  \vdash_{\mathsf{C} }  \Psi  \DualLNLLogicsym{,}   \DualLNLLogicnt{S} ^{\, \DualLNLLogicmv{n} }  $ and  $\Pi_2$ is a derivation of $ \DualLNLLogicnt{S}  \vdash_{\mathsf{C} }  \Psi' $ with 
    $c(\Pi_1), c(\Pi_2) \leq |S|$, then there exists a derivation $\Pi$ of $T\vdash_{\mathsf{C}} \Psi, \Psi'$ with $c(\Pi) \leq |S|$;
  \item If $\Pi_1$ is a derivation of $T \vdash_{\mathsf{L}} \Delta; \Psi, S^n$ and  $\Pi_2$ is a derivation of $S \vdash_{\mathsf{C}}\Psi'$ with 
    $c(\Pi_1), c(\Pi_2) \leq |S|$, then there exists a derivation $\Pi$ of $T\vdash_{\mathsf{L}} \Delta; \Psi, \Psi'$ with $c(\Pi) \leq |S|$;
  \item If $\Pi_1$ is a derivation of $B\vdash_{\mathsf{L}} \Delta; \Psi, A^n$ and  $\Pi_2$ is a derivation of $A \vdash_{\mathsf{L}}\Delta', \Psi'$ with $c(\Pi_1), c(\Pi_2) \leq |S|$,
    then there exists a derivation $\Pi$ of $B\vdash_{\mathsf{C}}\Delta, \Delta', \Psi, \Psi'$ with $c(\Pi) \leq |A|$.
  \end{enumerate}
\end{lemma}
\begin{proof}
  By induction on $d(\Pi_1) + d(\Pi_2)$.  We give one case where the
  last inferences of $\Pi_1$ and $\Pi_2$ are logical inferences;
  please see
  Appendix~\ref{subsec:proof_of_cut_reduction_lemma:cut-reduction} for
  the complete proof.

\ \\
\noindent
$\lsub$ right / $\lsub$ left. We have 
\begin{center}
\begin{tabular}{c}
\AxiomC{$\pi_1$ }
\noLine
\UnaryInfC{$A \vdash_{\mathsf{L}} \Delta_1; \Psi_1, B_1$}
 \AxiomC{$\pi_2$ } 
\noLine
\UnaryInfC{$B_2 \vdash_{\mathsf{L}} \Delta_2; \Psi_2$}
\LeftLabel{$\Pi_1 =$}
\RightLabel{$\DualLNLLogicdruleLXXsRName$}
\BinaryInfC{$A \vdash_{\mathsf{L}} B_1 \lsub B_2,  \Delta_1, \Delta_2; \Psi_1, \Psi_2$}
\AxiomC{$\pi_3$}
\noLine
\UnaryInfC{$B_1 \vdash_{\mathsf{L}} B_2, \Delta ;  \Psi$}
\LeftLabel{$\Pi_2 =$}
\RightLabel{$\DualLNLLogicdruleLXXsLName$}
\UnaryInfC{$B_1 \lsub B_2 \vdash_{\mathsf{L}} \Delta ; \Psi$}
\RightLabel{$\DualLNLLogicdruleLXXcutName$}
\BinaryInfC{$A \vdash_{\mathsf{L}} \Delta_1, \Delta_2, \Delta ; \Psi_1, \Psi_2, \Psi$}
\DisplayProof\\
\\
reduces to $\Pi$ 
\\
\\
\AxiomC{$\pi_1$}
\noLine
\UnaryInfC{$A\vdash_{\mathsf{L}} \Delta_1, B_1; \Psi_1$}
 \AxiomC{$\pi_3$}
\noLine
\UnaryInfC{$B_1 \vdash_{\mathsf{L}} B_2, \Delta ;  \Psi$}
\RightLabel{$\DualLNLLogicdruleLXXcutName$}
\BinaryInfC{$A\vdash_{\mathsf{L}} \Delta_1, \Delta, B_2; \Psi_1,\Psi$}
 \AxiomC{$\pi_2$ } 
\noLine
\UnaryInfC{$B_2 \vdash_{\mathsf{L}} \Delta_2; \Psi_2$}
\RightLabel{$\DualLNLLogicdruleLXXcutName$}
\BinaryInfC{$A \vdash_{\mathsf{L}} \Delta_1, \Delta_2, \Delta; \Psi_1, \Psi_2, \Psi$}
\DisplayProof
\end{tabular}
\end{center}
The resulting derivation $\Pi$ has cut rank $c(\Pi) = max(|B_1|+1, c(\pi_1), c(\pi_2), |B_2|+1, c(\pi_3)) \leq |B_1\lsub B_2|$.  
\end{proof}

\begin{lemma}[Decrease in Cut-Rank]
  \label{lemma:decrease-cut-rank}
  Let $\Pi$ be a $\DLNLP$ proof of a sequent  $S \vdash_{\mathsf{C}} \Psi$ or $A \vdash_{\mathsf{L}} \Delta; \Psi$ with 
  $c(\Pi)>0$. Then there exists a proof $\Pi'$ of the same sequent with $c(\Pi') < c(\Pi)$. 
\end{lemma}
\begin{proof} 
  By induction on $d(\Pi)$. If the last inference is not a cut, then we apply the 
  induction hypothesis. If the last inference is a cut on a formula $A$, but $A$ is not of maximal rank
  among the cut formulas, so that $c(\Pi) > |A|+1$, then we apply the induction hypothesis. Finally,  
  if the last inference is a cut on $A$ and $c(\Pi) = |A| + 1$ we have the following situation:
  \begin{center}
    \AxiomC{$\Pi_1$}
    \noLine
    \UnaryInfC{$B \vdash_{\mathsf{L}} \Delta, A ; \Psi$}
    \AxiomC{$\Pi_2$}
    \noLine
    \UnaryInfC{$A \vdash_{\mathsf{L}} \Delta' ; \Psi'$}
    \LeftLabel{$\Pi =$}
    \RightLabel{$\DualLNLLogicdruleLXXcutName$}
    \BinaryInfC{$B \vdash_{\mathsf{L}} \Delta, \Delta'; \Psi, \Psi'$}
    \DisplayProof
  \end{center}
  Now since $c(\Pi_1), c(\Pi_2) \leq |A| + 1$ then by applying the
  induction hypothesis to the premises of the previous derivation we
  can construct derivations $\Pi'_1$ and $\Pi'_2$ with $c(\Pi'_1) \leq
  |A|$ and $c(\Pi'_2) \leq |A|$. Then by cut reduction we can
  construct a derivation $\Pi'$ proving $B \vdash_{\mathsf{L}} \Delta,
  \Delta'; \Psi, \Psi'$ with $c(\Pi') \leq |A|$ as required.
\end{proof}

\begin{theorem}[Cut Elimination]
  \label{thm:cut_elimination}
  Let $\Pi$ be a proof of a sequent $ \DualLNLLogicnt{S}  \vdash_{\mathsf{C} }  \Psi $ or
  $ \DualLNLLogicnt{A}  \vdash_{\mathsf{L} }  \Delta  ;  \Psi $ such that $c(\Pi)> 0$. There is
  an algorithm which yields a cut free proof of the same sequent.
\end{theorem}
\begin{proof}
  By induction on $c(\Pi)$ using the previous lemma. 
\end{proof}


\subsection{Sequent-style Natural Deduction}
\label{sec:sequent-style_natural_deduction}

The inference rules for the non-linear and linear fragments of the
sequent-style natural deduction formalization of DLNL (DND) can be
found in Figure~\ref{fig:non-linear-nd} and Figure~\ref{fig:linear-nd}
respectively. 
\begin{remark}
  \label{rem:additive-contexts}
  In DLNL logic contexts are treated multiplicatively and so are in DND. 
  Non-linear context could also be treated additively. In the case of 
  the minor premises of non-linear disjunction elimination (rule 
  $\mathrm{NLL}_{\_ +_ E}$  of  Figure \ref{fig:non-linear-nd} an additive 
  interpretation is required, namely, both minor premises must have the same 
  right context, to match the categorical interpretation of disjunction as 
  coproduct.  The same holds for the term assignment in the rule 
  $TC_{\_+E}$ of Figure \ref{fig:non-linear-ta}.
  Of course additive contexts can be simulated using weakening and 
  contraction. This is what we do in the case of disjunction elimination.
\end{remark}


\begin{figure}
  \begin{mdframed}
    \begin{mathpar}
      \DualLNLLogicdruleNCXXid{} \and
      \DualLNLLogicdruleNCXXweak{} \and
      \DualLNLLogicdruleNCXXcontr{} \and
      \DualLNLLogicdruleNCXXzE{} \and
      \DualLNLLogicdruleNCXXdIOne{} \and
      \DualLNLLogicdruleNCXXdITwo{} \and
      \DualLNLLogicdruleNCXXdE{} \and
      \DualLNLLogicdruleNCXXsubI{} \and
      \DualLNLLogicdruleNCXXsubE{} \and
      \DualLNLLogicdruleNCXXHE{}      
    \end{mathpar}
  \end{mdframed}
  \caption{Non-linear fragment of DND logic}
  \label{fig:non-linear-nd}
\end{figure}

\begin{figure}
  \begin{mdframed}
    \begin{mathpar}
      \DualLNLLogicdruleNLXXid{} \and
      \DualLNLLogicdruleNLXXweak{} \and
      \DualLNLLogicdruleNLXXcontr{} \and
      \DualLNLLogicdruleNLXXpI{} \and
      \DualLNLLogicdruleNLXXpE{} \and
      \DualLNLLogicdruleNLXXparI{} \and
      \DualLNLLogicdruleNLXXparE{} \and
      \DualLNLLogicdruleNLXXsubI{} \and
      \DualLNLLogicdruleNLXXsubE{} \and
      \DualLNLLogicdruleNLXXJI{} \and
      \DualLNLLogicdruleNLXXJE{} \and
      \DualLNLLogicdruleNLXXHI{} \and
      \DualLNLLogicdruleNLXXHE{}      
    \end{mathpar}
  \end{mdframed}
  \caption{Linear fragment of DND logic}
  \label{fig:linear-nd}
\end{figure}

We now recall a correspondence between DND and DLNL logic.  First, we 
need the admissible rule of cut, i.e., substitution.
\begin{lemma}[Admissible Rules in DND]
  \label{lemma:admissible_rules_in_dnd}
  The following rules are admissible in DND: 
  \begin{mathpar}
    \DualLNLLogicdruleNCXXcut{} \and
    \DualLNLLogicdruleNLXXCcut{} \and
    \DualLNLLogicdruleNLXXcut{}
  \end{mathpar}
\end{lemma}
Using these admissible rules we can construct a proof preserving
translation between DND and DLNL logic.
\begin{lemma}[Translations between DND and DLNL logic]
  \label{lemma:translations}
  There are functions $\mathcal{S}: DND \rightarrow DLNL$ 
  and $\mathcal{N}: DLNL \rightarrow DND$ from natural deduction to sequent 
  calculus derivations. 
\end{lemma}
\noindent
Notice that the right rules of the sequent calculus and the introductions of natural deduction
have the same form. Elimination rules are derivable from left rules with \emph{cut} and left 
rules are derivable using the admissible cut rule in DND. For instance, the $\DualLNLLogicdruleNCXXzEName$ rule 
\[ \DualLNLLogicdruleNCXXzE{} \]
is derivable in the sequent calculus as follows: 
\[
\AxiomC{$S\vdash_{\mathsf{C}} 0, \Psi$}
\AxiomC{$0\vdash S_1, \ldots S_n$}
\RightLabel{$\DualLNLLogicdruleCXXcutName{}$}
\BinaryInfC{$S\vdash_{\mathsf{C}} \Psi, S_1, \ldots S_n$}
\AxiomC{$ S_1 \vdash_{\mathsf{C}} \Psi_1$}
\RightLabel{$\DualLNLLogicdruleCXXcutName{}$}
\BinaryInfC{$S\vdash_{\mathsf{C}}  \Psi, \Psi_1, S_2, \ldots, S_n$}
\AxiomC{$\vdots$} 
\BinaryInfC{$S\vdash_{\mathsf{C}}  \Psi, \Psi_1, \ldots \Psi_{n-1}, S_n$}
\AxiomC{$S_n \vdash \Psi_n$} 
\RightLabel{$\DualLNLLogicdruleCXXcutName{}$}
\BinaryInfC{$S\vdash_{\mathsf{C}} \Psi, \Psi_1, \ldots \Psi_{n-1}, \Psi_n$ }
\DisplayProof
\]

\subsection{Term Assignment}
\label{sec:term_assignment}
We now turn to giving a term assignment to DND logic called TND, which
is greatly influenced by Crolard's 
 term assignment for subtractive logic in the paper 
 \emph{A formulae-as-types
  interpretation of subtractive logic} JLC 2004. Crolard based his term assignment on
Parigot's~\cite{Parigot:1992} $\lambda\mu$-calculus.  He then shows
that a type theory of coroutines can be given by subtractive types 
and it is this result we pull inspiration from. 

TND pushes beyond Crolard's work on subtractive logic. He restricts a
classical calculus to provide a constructive version of subtraction
called \emph{safe coroutines}. TND is based on the work of the second
author where he used a variant of Crolard's constructive calculus as a
term assignment to co-intuitionistic logic and to linear
co-intuitionistic logic \cite{Bellin:2012} without using the
$\lambda\mu$-calculus. In this formulation, distinct terms are
assigned to distinct formulas in the context and the reduction of a
term in context may impact other terms in the context.

The syntax of TND terms is defined by the following definition.
\begin{definition}
  \label{def:TND-terms-syntax}
  The syntax for TND terms and typing judgments are given by the following grammar:
  \[
  \begin{array}{l}
    \begin{array}{cllllll}
    \text{(non-linear terms)} & \DualLNLLogicnt{s},\DualLNLLogicnt{t} & ::= & \DualLNLLogicmv{x} \mid  \varepsilon  \mid  \DualLNLLogicnt{t_{{\mathrm{1}}}}  \cdot  \DualLNLLogicnt{t_{{\mathrm{2}}}}  \mid
          \mathsf{false}\, \DualLNLLogicnt{t}  \mid \DualLNLLogicmv{x}  \DualLNLLogicsym{(}  \DualLNLLogicnt{t}  \DualLNLLogicsym{)} \mid  \mathsf{mkc}( \DualLNLLogicnt{t} , \DualLNLLogicmv{x} )  \mid  \mathsf{inl}\, \DualLNLLogicnt{t}  \mid  \mathsf{inr}\, \DualLNLLogicnt{t}  \mid \\
         & & &  \mathsf{case}\, \DualLNLLogicnt{t} \,\mathsf{of}\, \DualLNLLogicmv{x} . \DualLNLLogicnt{t_{{\mathrm{1}}}} , \DualLNLLogicmv{y} . \DualLNLLogicnt{t_{{\mathrm{2}}}}  \mid
          \mathsf{H}\, \DualLNLLogicnt{e}  \mid  \mathsf{let}\,\mathsf{J}\, \DualLNLLogicmv{x}  =  \DualLNLLogicnt{e} \,\mathsf{in}\, \DualLNLLogicnt{t}  \mid  \mathsf{postp}\,( \DualLNLLogicmv{x}  \mapsto  \DualLNLLogicnt{t_{{\mathrm{1}}}} ,  \DualLNLLogicnt{t_{{\mathrm{2}}}} )  \mid \\
         & & &  \mathsf{let}\,\mathsf{H}\, \DualLNLLogicmv{x}  =  \DualLNLLogicnt{t_{{\mathrm{1}}}} \,\mathsf{in}\, \DualLNLLogicnt{t_{{\mathrm{2}}}} \\
         \\
         \text{(linear terms)} & \DualLNLLogicnt{e},\DualLNLLogicnt{u} & ::= & \DualLNLLogicmv{x} \mid  \mathsf{connect}_\perp\,\mathsf{to}\, \DualLNLLogicnt{e}  \mid  \mathsf{postp}_\perp\, \DualLNLLogicnt{e}  \mid  \mathsf{postp}\,( \DualLNLLogicmv{x}  \mapsto  \DualLNLLogicnt{e_{{\mathrm{1}}}} ,  \DualLNLLogicnt{e_{{\mathrm{2}}}} )  \mid  \mathsf{mkc}( \DualLNLLogicnt{e} , \DualLNLLogicmv{x} )  \mid \DualLNLLogicmv{x}  \DualLNLLogicsym{(}  \DualLNLLogicnt{e}  \DualLNLLogicsym{)} \mid \\
         & & &  \DualLNLLogicnt{e_{{\mathrm{1}}}}  \oplus  \DualLNLLogicnt{e_{{\mathrm{2}}}}  \mid  \mathsf{casel}\, \DualLNLLogicnt{e}  \mid  \mathsf{caser}\, \DualLNLLogicnt{e}  \mid  \mathsf{J}\, \DualLNLLogicnt{t} \\         
  \end{array}
  \\\\
  \begin{array}{cll}
    \text{(non-linear judgment)} &  \DualLNLLogicmv{x}  :  \DualLNLLogicnt{R}  \vdash_{\mathsf{C} }  \Psi \\
    \\
    \text{(linear judgment)} &  \DualLNLLogicmv{x}  :  \DualLNLLogicnt{A}  \vdash_{\mathsf{L} }  \Delta ; \Psi \\
  \end{array}
  \end{array}
  \]
  Contexts, $\Delta$ and $\Psi$, are the straightforward extension
  where each type is annotated with a term from the respective
  fragment.
\end{definition}

To aid the reader in understanding the variable structure, which
variable annotations are bound, deployed throughout the TND term
syntax we give the definitions of the free variable functions in the
following definition.
\begin{definition}
The free variable functions, $FV(t)$ and $FV(e)$, for linear and
non-linear terms $t$ and $e$ are defined by mutual recursion as
follows:
\[ \small
\begin{array}{lll}  
  \begin{array}{lll}
    \textbf{linear terms:}\\
    \,\,FV(\DualLNLLogicmv{x})  = \{ \DualLNLLogicmv{x}\}\\
    \,\,FV( \mathsf{connect}_\perp\,\mathsf{to}\, \DualLNLLogicnt{e} ) = FV(e) \\ 
    \,\,FV(\DualLNLLogicmv{x}  \DualLNLLogicsym{(}  \DualLNLLogicnt{e}  \DualLNLLogicsym{)}) = FV(e)\\
    \,\,FV( \mathsf{mkc}( \DualLNLLogicnt{e} , \DualLNLLogicmv{y} ) ) = FV(e)\\
    \,\,FV( \DualLNLLogicnt{e_{{\mathrm{1}}}}  \oplus  \DualLNLLogicnt{e_{{\mathrm{2}}}} ) = FV(e_1) \cup FV(e_2)\\
    \,\,FV( \mathsf{casel}\, \DualLNLLogicnt{e} ) = FV(e)\\
    \,\,FV( \mathsf{caser}\, \DualLNLLogicnt{e} ) = FV(e)\\    
    \,\,FV( \mathsf{J}\, \DualLNLLogicnt{t} ) = FV(t)\\
    \\
    \\
    \\[13px]
  \end{array}
  & \quad &  
  \begin{array}{lll}
    \textbf{non-linear terms:}\\
    \,\,FV(x) = \{ x\}\\
    \,\,FV( \varepsilon ) = \emptyset \\ 
    \,\,FV( \DualLNLLogicnt{t_{{\mathrm{1}}}}  \cdot  \DualLNLLogicnt{t_{{\mathrm{2}}}} )  = FV(t_1)\cup FV(t_2) \\
    \,\,FV( \mathsf{false}\, \DualLNLLogicnt{t} )  = FV(t) \\
    \,\,FV(\DualLNLLogicmv{x}  \DualLNLLogicsym{(}  \DualLNLLogicnt{t}  \DualLNLLogicsym{)}) = FV(t)\\
    \,\,FV( \mathsf{mkc}( \DualLNLLogicnt{t} , \DualLNLLogicmv{y} ) ) = FV(t)\\
    \,\,FV( \mathsf{inl}\, \DualLNLLogicnt{t} ) = FV( \mathsf{inr}\, \DualLNLLogicnt{t} ) = FV(t)\\
    \,\,FV( \mathsf{case}\, \DualLNLLogicnt{t_{{\mathrm{1}}}} \,\mathsf{of}\, \DualLNLLogicmv{x} . \DualLNLLogicnt{t_{{\mathrm{2}}}} , \DualLNLLogicmv{y} . \DualLNLLogicnt{t_{{\mathrm{3}}}} ) = \\
    \,\,\,\,\,\,\,\,\,\,FV(t_1)\cup FV(t_2)\smallsetminus\{x\}\cup FV(t_3)\smallsetminus \{y\}\\
    \,\,FV( \mathsf{let}\,\mathsf{J}\, \DualLNLLogicmv{y}  =  \DualLNLLogicnt{e} \,\mathsf{in}\, \DualLNLLogicnt{t} ) = FV(e)\cup FV(t)\smallsetminus \{y\} \\
    \,\,FV( \mathsf{let}\,\mathsf{H}\, \DualLNLLogicmv{y}  =  \DualLNLLogicnt{t_{{\mathrm{1}}}} \,\mathsf{in}\, \DualLNLLogicnt{t_{{\mathrm{2}}}} ) = FV(\DualLNLLogicnt{t_{{\mathrm{1}}}})\cup FV(\DualLNLLogicnt{t_{{\mathrm{2}}}})\smallsetminus \{y\} \\
    \,\,FV( \mathsf{H}\, \DualLNLLogicnt{e} ) = FV(e)\\
  \end{array}
\end{array}
\]
The free variables of a $p$-term are defined s follows: 
\[
\begin{array}{rll}
  FV(\mathtt{postp}_{\bot}\ e) & = & FV(e)\\
  FV(\mathtt{postp} (x 􏰀\mapsto e_1, e_2)) & = & FV(e_1) \smallsetminus \{x\} \cup FV(e_2)
\end{array}
\] 
and similarly for terms $\mathtt{postp} (x 􏰀\mapsto t_1, t_2)$.
\end{definition}

Terms are then typed by annotating the previous term structure over
DND derivations, and this is accomplished by annotating the DND
inference rules.  The typing rules for the non-linear fragment of TND
can be found in Figure~\ref{fig:non-linear-ta}, and the typing rules
for the linear fragment of TND can be found in
Figure~\ref{fig:linear-ta}.
\begin{figure}
  \begin{mdframed}
    \begin{mathpar}
      \DualLNLLogicdruleTCXXid{} \and
      \DualLNLLogicdruleTCXXweak{} \and
      \DualLNLLogicdruleTCXXcontr{} \and
      \DualLNLLogicdruleTCXXzI{} \and   
      \DualLNLLogicdruleTCXXdIOne{} \and
      \DualLNLLogicdruleTCXXdITwo{} \and
      \DualLNLLogicdruleTCXXdE{} \and
      \DualLNLLogicdruleTCXXsubI{} \and
      \DualLNLLogicdruleTCXXsubE{} \and
      \DualLNLLogicdruleTCXXHE{} \and      
    \end{mathpar}
  \end{mdframed}
  \caption{Non-linear fragment of the term assignment for TND}
  \label{fig:non-linear-ta}
\end{figure}
\begin{figure}
  \begin{mdframed}
    \begin{mathpar}
      \DualLNLLogicdruleTLXXid{} \and
      \DualLNLLogicdruleTLXXweak{} \and 
      \DualLNLLogicdruleTLXXcontr{} \and
      \DualLNLLogicdruleTLXXpI{} \and
      \DualLNLLogicdruleTLXXpE{} \and
      \DualLNLLogicdruleTLXXparI{} \and
      \DualLNLLogicdruleTLXXparE{} \and
      \DualLNLLogicdruleTLXXsubI{} \and
      \DualLNLLogicdruleTLXXsubE{} \and
      \DualLNLLogicdruleTLXXJI{} \and
      \DualLNLLogicdruleTLXXJE{} \and
      \DualLNLLogicdruleTLXXHI{} \and
      \DualLNLLogicdruleTLXXHE{} \and      
    \end{mathpar}
  \end{mdframed}
  \caption{Linear fragment of the term assignment for TND}
  \label{fig:linear-ta}
\end{figure}
\begin{remark}
  Let us call terms of the form $ \mathsf{postp}\,( \DualLNLLogicmv{x}  \mapsto  \DualLNLLogicnt{t_{{\mathrm{1}}}} ,  \DualLNLLogicnt{t_{{\mathrm{2}}}} ) $, $ \mathsf{postp}\,( \DualLNLLogicmv{x}  \mapsto  \DualLNLLogicnt{e_{{\mathrm{1}}}} ,  \DualLNLLogicnt{e_{{\mathrm{2}}}} ) $, and $ \mathsf{postp}_\perp\, \DualLNLLogicnt{e} $ $p$-\emph{terms}. Then say that
  a term $t$ is \emph{$p$-normal} if $t$ does not contain any $p$-term
  as a proper subterm.  In a typed calculus, linear $p$-terms can be typed
  with $\bot$.  Non-linear $p$ terms can be typed with $0$:  in
  presence of the $\DualLNLLogicdruleTCXXzIName$ rule this yieds
instances of the \emph{ex falso} rule. This is what happens in Crolard's 
calculus, where the analogue of the $ \mathsf{postp}\,( \DualLNLLogicmv{x}  \mapsto  \DualLNLLogicnt{t_{{\mathrm{1}}}} ,  \DualLNLLogicnt{t_{{\mathrm{2}}}} ) $, namely, 
$\mathtt{resume}\ t_2\ \mathtt{with}\ x\mapsto t_1$, always goes with a 
\emph{weakening} operation. The term $ \varepsilon $ is the identity of the
contraction binary operator $ \DualLNLLogicnt{t_{{\mathrm{1}}}}  \cdot  \DualLNLLogicnt{t_{{\mathrm{2}}}} $.

\noindent
However when within a non-$p$-normal term an expression of the form 
$ \mathsf{postp}\,( \DualLNLLogicmv{x}  \mapsto  \DualLNLLogicnt{t_{{\mathrm{1}}}} ,  \DualLNLLogicnt{t_{{\mathrm{2}}}} ) $ is eliminated as a $\beta$-redex, there is a choice
of the syntax for the contextual reduction. In absence of a more detailed 
analysis of the matter, we prefer to leave the typing of $p$ terms implicit 
in the syntax, to enforce the requirement of $p$-normality and to use the 
rule of weakening in place of the  $\DualLNLLogicdruleTCXXzIName$ 
rule in this context. 
\end{remark}

The typing rules depend on the extension of let and case expressions
to typing contexts.  We use the following notation for \emph{parallel
  composition} of typing contexts:
\[
\Delta = e_1: A_1\,\|\, \cdots \,\|\, e_n: A_n\qquad 
\]
This operation should be regarded as associative, commutative and
having the empty context as its identity. 
The extension of let expressions to contexts is given as follows:
\begin{center}
  \begin{tabular}{rcl}
    $\mathtt{let}\; p\; = t\; \mathtt{in}\ \cdot$ &\ =\ & $\cdot$\\
    $\mathtt{let}\; p\; = \DualLNLLogicnt{t_{{\mathrm{1}}}}\; \mathtt{in}\ (\DualLNLLogicnt{t_{{\mathrm{2}}}}: A)$ &\ =\ & $\mathtt{let}\; p\; = \DualLNLLogicnt{t_{{\mathrm{1}}}}\,\mathsf{in}\,\DualLNLLogicnt{t_{{\mathrm{2}}}} : A$\\
    $\mathtt{let}\; p\; = t\; \mathtt{in}\ (\Psi_{{\mathrm{1}}}\,\|\,\Psi_{{\mathrm{2}}})$ &\ =\ & 
    $(\mathtt{let}\; p\; = t\; \mathtt{in}\ \Psi_{{\mathrm{1}}})\,\|\, (\mathtt{let}\; p\; = t\ \mathtt{in}\ \Psi_{{\mathrm{2}}})$
  \end{tabular}
\end{center}
where $p =  \mathsf{H}\, \DualLNLLogicmv{y} $ or $p =  \mathsf{J}\, \DualLNLLogicmv{y} $.
Case expressions are handled similarly.

Similarly to DND logic we have the following admissible rules.
\begin{lemma}[Admissible Typing Rules]
  The term assignment for the admissible rules of the calculus is as follows:
  \begin{mathpar}
    \DualLNLLogicdruleTCXXcut{} \and     
    \DualLNLLogicdruleTLXXCcut{} \and
    \DualLNLLogicdruleTLXXcut{}      
  \end{mathpar}
\end{lemma}
\noindent
We generalize the rule of contraction on the non-linear
side to contexts. Let $m_1$ and $m_2$ be multisets of terms, then we
denote by $m_1 \cdot m_2$ the sum of multisets; if multisets are
represented as lists, then the sum is representable as the appending
of the lists. We denote singleton multisets, $\{\DualLNLLogicnt{t}\}$, by the term that inhabits
it, e.g. $\DualLNLLogicnt{t}$.  We extend this to contexts, $\Psi_1 \cdot \Psi_2$, recursively as
follows:
\begin{center}
\begin{tabular}{rcl}
$(\cdot) \cdot (\cdot)$ &\ =\ & $(\cdot)$\\
$(t_1: S)\cdot (t_2:S)$ &\ =\ & $t_1\cdot t_2 : S$\\
  $(\Psi_1\,\|\,\Psi_3)\cdot(\Psi_2\,\|\,\Psi_4)$ &\ =\ & $(\Psi_1\cdot\Psi_2)\,\|\,(\Psi_3\cdot\Psi_4)$  
\end{tabular}
\end{center}
\noindent
where $|\Psi_1|=|\Psi_3|$ and $|\Psi_2|=|\Psi_4|$. Whenever we write $\Psi_1 \cdot \Psi_2$ 
we assume that  $|\Psi_1|=|\Psi_2|$.

At this point we are now ready to turn to computing in TND by
specifying the reduction relation.  This definition is perhaps the
most interesting aspect of the theory, because reducing one term may
affect others.

\textbf{$\beta$-Reduction in TND.} As we discussed above
cointuitionistic logic can be interpreted as a theory of coroutines that
manipulate local context.  Thus, reducing one term in a typing context
could affect other terms in the context.  This implies that the
definition of the reduction relation for TND must account for more
than a single term. We accomplish this by defining the reduction
relation of terms in context, $ \DualLNLLogicmv{x}  :  \DualLNLLogicnt{S}  \vdash_{\mathsf{C} }  \Psi_{{\mathrm{1}}}  \DualLNLLogicsym{,}  \DualLNLLogicnt{t}  \DualLNLLogicsym{:}  \DualLNLLogicnt{T}  \DualLNLLogicsym{,}  \Psi_{{\mathrm{2}}} $ and $ \DualLNLLogicmv{x}  :  \DualLNLLogicnt{A}  \vdash_{\mathsf{L} }  \Delta_{{\mathrm{1}}}  \DualLNLLogicsym{,}  \DualLNLLogicnt{e}  \DualLNLLogicsym{:}  \DualLNLLogicnt{B}  \DualLNLLogicsym{,}  \Delta_{{\mathrm{2}}} ; \Psi $, so that the manipulation of the context
is made explicit.
\begin{figure}
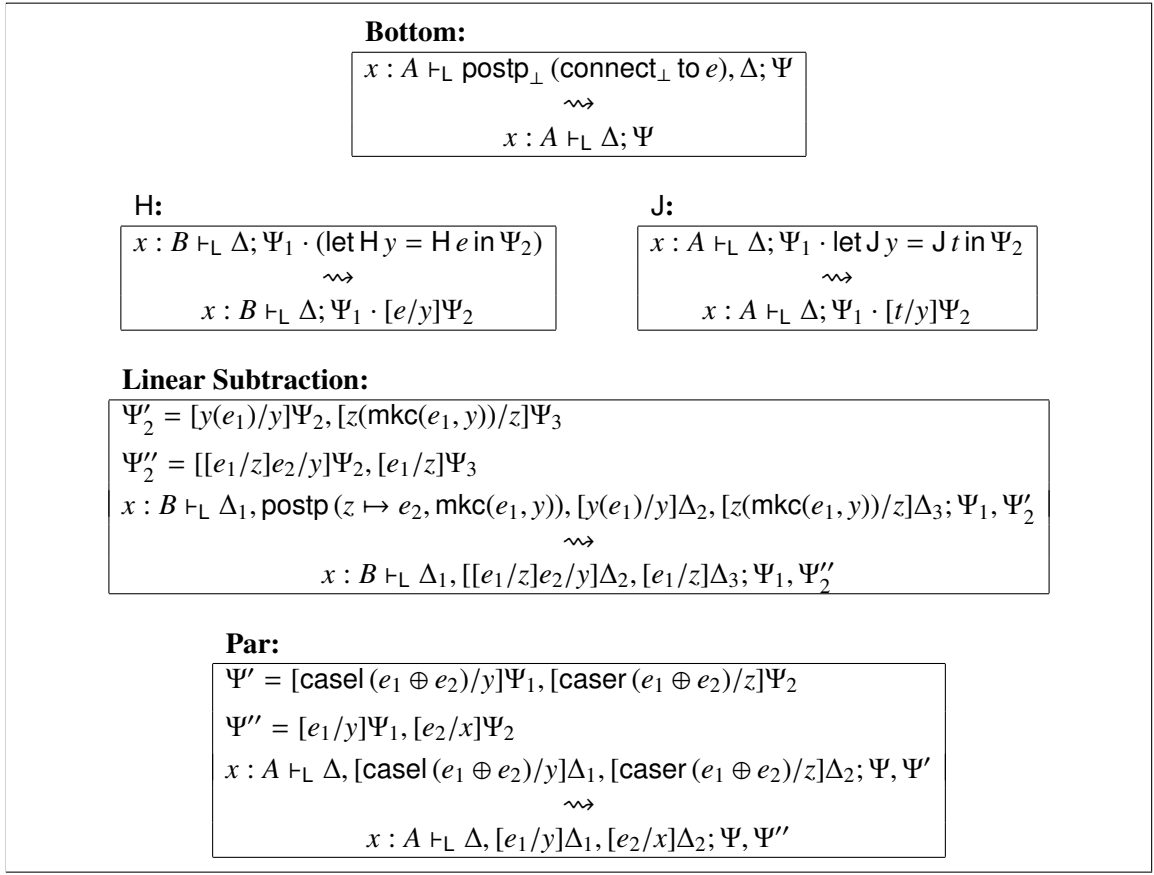

  \begin{mdframed}     
  \begin{center}
    \begin{math}
      \begin{array}{c}
        \begin{array}{|c|}
          \multicolumn{1}{l}{\textbf{Bottom:}}\\
          \hline
           \DualLNLLogicmv{x}  :  \DualLNLLogicnt{A}  \vdash_{\mathsf{L} }   \mathsf{postp}_\perp\, \DualLNLLogicsym{(}   \mathsf{connect}_\perp\,\mathsf{to}\, \DualLNLLogicnt{e}   \DualLNLLogicsym{)}   \DualLNLLogicsym{,}  \Delta ; \Psi \\
          \rightsquigarrow\\
           \DualLNLLogicmv{x}  :  \DualLNLLogicnt{A}  \vdash_{\mathsf{L} }  \Delta ; \Psi \\
          \hline
        \end{array}\\
        \\
        \begin{array}{lll}
          \begin{array}{|c|}
            \multicolumn{1}{l}{\textbf{$\mathsf{H}$:}}\\
            \hline
             \DualLNLLogicmv{x}  :  \DualLNLLogicnt{B}  \vdash_{\mathsf{L} }  \Delta ; \Psi_{{\mathrm{1}}}  \cdot  \DualLNLLogicsym{(}   \mathsf{let}\,\mathsf{H}\, \DualLNLLogicmv{y}  =   \mathsf{H}\, \DualLNLLogicnt{e}  \,\mathsf{in}\, \Psi_{{\mathrm{2}}}   \DualLNLLogicsym{)} \\
            \rightsquigarrow\\
             \DualLNLLogicmv{x}  :  \DualLNLLogicnt{B}  \vdash_{\mathsf{L} }  \Delta ; \Psi_{{\mathrm{1}}}  \cdot  \DualLNLLogicsym{[}  \DualLNLLogicnt{e}  \DualLNLLogicsym{/}  \DualLNLLogicmv{y}  \DualLNLLogicsym{]}  \Psi_{{\mathrm{2}}} \\
            \hline
          \end{array}
          & \quad &
          \begin{array}{|c|}
            \multicolumn{1}{l}{\textbf{$\mathsf{J}$:}}\\
            \hline
             \DualLNLLogicmv{x}  :  \DualLNLLogicnt{A}  \vdash_{\mathsf{L} }  \Delta ; \Psi_{{\mathrm{1}}}  \cdot   \mathsf{let}\,\mathsf{J}\, \DualLNLLogicmv{y}  =   \mathsf{J}\, \DualLNLLogicnt{t}  \,\mathsf{in}\, \Psi_{{\mathrm{2}}}  \\
            \rightsquigarrow\\
             \DualLNLLogicmv{x}  :  \DualLNLLogicnt{A}  \vdash_{\mathsf{L} }  \Delta ; \Psi_{{\mathrm{1}}}  \cdot  \DualLNLLogicsym{[}  \DualLNLLogicnt{t}  \DualLNLLogicsym{/}  \DualLNLLogicmv{y}  \DualLNLLogicsym{]}  \Psi_{{\mathrm{2}}} \\
            \hline
          \end{array}
        \end{array}\\
        \\
        \begin{array}{|c|}
          \multicolumn{1}{l}{\textbf{Linear Subtraction:}}\\
          \hline
          \multicolumn{1}{|l|}{\Psi'_{{\mathrm{2}}} = \DualLNLLogicsym{[}  \DualLNLLogicmv{y}  \DualLNLLogicsym{(}  \DualLNLLogicnt{e_{{\mathrm{1}}}}  \DualLNLLogicsym{)}  \DualLNLLogicsym{/}  \DualLNLLogicmv{y}  \DualLNLLogicsym{]}  \Psi_{{\mathrm{2}}}  \DualLNLLogicsym{,}  \DualLNLLogicsym{[}  \DualLNLLogicmv{z}  \DualLNLLogicsym{(}   \mathsf{mkc}( \DualLNLLogicnt{e_{{\mathrm{1}}}} , \DualLNLLogicmv{y} )   \DualLNLLogicsym{)}  \DualLNLLogicsym{/}  \DualLNLLogicmv{z}  \DualLNLLogicsym{]}  \Psi_{{\mathrm{3}}}}\\[5px]
          \multicolumn{1}{|l|}{\Psi''_{{\mathrm{2}}} = \DualLNLLogicsym{[}  \DualLNLLogicsym{[}  \DualLNLLogicnt{e_{{\mathrm{1}}}}  \DualLNLLogicsym{/}  \DualLNLLogicmv{z}  \DualLNLLogicsym{]}  \DualLNLLogicnt{e_{{\mathrm{2}}}}  \DualLNLLogicsym{/}  \DualLNLLogicmv{y}  \DualLNLLogicsym{]}  \Psi_{{\mathrm{2}}}  \DualLNLLogicsym{,}  \DualLNLLogicsym{[}  \DualLNLLogicnt{e_{{\mathrm{1}}}}  \DualLNLLogicsym{/}  \DualLNLLogicmv{z}  \DualLNLLogicsym{]}  \Psi_{{\mathrm{3}}}}\\
          \\[-10px]
           \DualLNLLogicmv{x}  :  \DualLNLLogicnt{B}  \vdash_{\mathsf{L} }  \Delta_{{\mathrm{1}}}  \DualLNLLogicsym{,}   \mathsf{postp}\,( \DualLNLLogicmv{z}  \mapsto  \DualLNLLogicnt{e_{{\mathrm{2}}}} ,   \mathsf{mkc}( \DualLNLLogicnt{e_{{\mathrm{1}}}} , \DualLNLLogicmv{y} )  )   \DualLNLLogicsym{,}  \DualLNLLogicsym{[}  \DualLNLLogicmv{y}  \DualLNLLogicsym{(}  \DualLNLLogicnt{e_{{\mathrm{1}}}}  \DualLNLLogicsym{)}  \DualLNLLogicsym{/}  \DualLNLLogicmv{y}  \DualLNLLogicsym{]}  \Delta_{{\mathrm{2}}}  \DualLNLLogicsym{,}  \DualLNLLogicsym{[}  \DualLNLLogicmv{z}  \DualLNLLogicsym{(}   \mathsf{mkc}( \DualLNLLogicnt{e_{{\mathrm{1}}}} , \DualLNLLogicmv{y} )   \DualLNLLogicsym{)}  \DualLNLLogicsym{/}  \DualLNLLogicmv{z}  \DualLNLLogicsym{]}  \Delta_{{\mathrm{3}}} ; \Psi_{{\mathrm{1}}}  \DualLNLLogicsym{,}  \Psi'_{{\mathrm{2}}} \\
          \rightsquigarrow\\
           \DualLNLLogicmv{x}  :  \DualLNLLogicnt{B}  \vdash_{\mathsf{L} }  \Delta_{{\mathrm{1}}}  \DualLNLLogicsym{,}  \DualLNLLogicsym{[}  \DualLNLLogicsym{[}  \DualLNLLogicnt{e_{{\mathrm{1}}}}  \DualLNLLogicsym{/}  \DualLNLLogicmv{z}  \DualLNLLogicsym{]}  \DualLNLLogicnt{e_{{\mathrm{2}}}}  \DualLNLLogicsym{/}  \DualLNLLogicmv{y}  \DualLNLLogicsym{]}  \Delta_{{\mathrm{2}}}  \DualLNLLogicsym{,}  \DualLNLLogicsym{[}  \DualLNLLogicnt{e_{{\mathrm{1}}}}  \DualLNLLogicsym{/}  \DualLNLLogicmv{z}  \DualLNLLogicsym{]}  \Delta_{{\mathrm{3}}} ; \Psi_{{\mathrm{1}}}  \DualLNLLogicsym{,}  \Psi''_{{\mathrm{2}}} \\
          \hline
        \end{array}\\
        \\
        \begin{array}{|c|}
          \multicolumn{1}{l}{\textbf{Par:}}\\
          \hline
          \multicolumn{1}{|l|}{\Psi' = \DualLNLLogicsym{[}   \mathsf{casel}\, \DualLNLLogicsym{(}   \DualLNLLogicnt{e_{{\mathrm{1}}}}  \oplus  \DualLNLLogicnt{e_{{\mathrm{2}}}}   \DualLNLLogicsym{)}   \DualLNLLogicsym{/}  \DualLNLLogicmv{y}  \DualLNLLogicsym{]}  \Psi_{{\mathrm{1}}}  \DualLNLLogicsym{,}  \DualLNLLogicsym{[}   \mathsf{caser}\, \DualLNLLogicsym{(}   \DualLNLLogicnt{e_{{\mathrm{1}}}}  \oplus  \DualLNLLogicnt{e_{{\mathrm{2}}}}   \DualLNLLogicsym{)}   \DualLNLLogicsym{/}  \DualLNLLogicmv{z}  \DualLNLLogicsym{]}  \Psi_{{\mathrm{2}}}}\\[5px]
          \multicolumn{1}{|l|}{\Psi'' = \DualLNLLogicsym{[}  \DualLNLLogicnt{e_{{\mathrm{1}}}}  \DualLNLLogicsym{/}  \DualLNLLogicmv{y}  \DualLNLLogicsym{]}  \Psi_{{\mathrm{1}}}  \DualLNLLogicsym{,}  \DualLNLLogicsym{[}  \DualLNLLogicnt{e_{{\mathrm{2}}}}  \DualLNLLogicsym{/}  \DualLNLLogicmv{x}  \DualLNLLogicsym{]}  \Psi_{{\mathrm{2}}}}\\
          \\[-10px]
           \DualLNLLogicmv{x}  :  \DualLNLLogicnt{A}  \vdash_{\mathsf{L} }  \Delta  \DualLNLLogicsym{,}  \DualLNLLogicsym{[}   \mathsf{casel}\, \DualLNLLogicsym{(}   \DualLNLLogicnt{e_{{\mathrm{1}}}}  \oplus  \DualLNLLogicnt{e_{{\mathrm{2}}}}   \DualLNLLogicsym{)}   \DualLNLLogicsym{/}  \DualLNLLogicmv{y}  \DualLNLLogicsym{]}  \Delta_{{\mathrm{1}}}  \DualLNLLogicsym{,}  \DualLNLLogicsym{[}   \mathsf{caser}\, \DualLNLLogicsym{(}   \DualLNLLogicnt{e_{{\mathrm{1}}}}  \oplus  \DualLNLLogicnt{e_{{\mathrm{2}}}}   \DualLNLLogicsym{)}   \DualLNLLogicsym{/}  \DualLNLLogicmv{z}  \DualLNLLogicsym{]}  \Delta_{{\mathrm{2}}} ; \Psi  \DualLNLLogicsym{,}  \Psi' \\
          \rightsquigarrow\\
           \DualLNLLogicmv{x}  :  \DualLNLLogicnt{A}  \vdash_{\mathsf{L} }  \Delta  \DualLNLLogicsym{,}  \DualLNLLogicsym{[}  \DualLNLLogicnt{e_{{\mathrm{1}}}}  \DualLNLLogicsym{/}  \DualLNLLogicmv{y}  \DualLNLLogicsym{]}  \Delta_{{\mathrm{1}}}  \DualLNLLogicsym{,}  \DualLNLLogicsym{[}  \DualLNLLogicnt{e_{{\mathrm{2}}}}  \DualLNLLogicsym{/}  \DualLNLLogicmv{x}  \DualLNLLogicsym{]}  \Delta_{{\mathrm{2}}} ; \Psi  \DualLNLLogicsym{,}  \Psi'' \\
          \hline
        \end{array}
      \end{array}
    \end{math}
  \end{center}
  \end{mdframed}
  \caption{Reductions for Linear Terms}
  \label{fig:red-linear-terms}
\end{figure}
\begin{figure}
  \begin{mdframed}    
  \begin{center}
  \begin{math}
    \begin{array}{c}
      \begin{array}{|c|}
        \multicolumn{1}{l}{\textbf{Subtraction:}}\\
        \hline
             \DualLNLLogicmv{x}  :  \DualLNLLogicnt{S}  \vdash_{\mathsf{C} }  \Psi_{{\mathrm{1}}}  \DualLNLLogicsym{,}    \mathsf{postp}\,( \DualLNLLogicmv{z}  \mapsto  \DualLNLLogicnt{t_{{\mathrm{2}}}} ,   \mathsf{mkc}( \DualLNLLogicnt{t_{{\mathrm{1}}}} , \DualLNLLogicmv{y} )  )    \DualLNLLogicsym{,}  \DualLNLLogicsym{[}  \DualLNLLogicmv{y}  \DualLNLLogicsym{(}  \DualLNLLogicnt{t_{{\mathrm{1}}}}  \DualLNLLogicsym{)}  \DualLNLLogicsym{/}  \DualLNLLogicmv{y}  \DualLNLLogicsym{]}  \Psi_{{\mathrm{2}}}  \DualLNLLogicsym{,}  \DualLNLLogicsym{[}   \mathsf{mkc}( \DualLNLLogicnt{t_{{\mathrm{1}}}} , \DualLNLLogicmv{y} )   \DualLNLLogicsym{/}  \DualLNLLogicmv{z}  \DualLNLLogicsym{]}  \Psi_{{\mathrm{3}}} \\
            \rightsquigarrow\\
                 \DualLNLLogicmv{x}  :  \DualLNLLogicnt{S}  \vdash_{\mathsf{C} }  \Psi_{{\mathrm{1}}}  \DualLNLLogicsym{,}  \DualLNLLogicsym{[}  \DualLNLLogicsym{[}  \DualLNLLogicnt{t_{{\mathrm{1}}}}  \DualLNLLogicsym{/}  \DualLNLLogicmv{z}  \DualLNLLogicsym{]}  \DualLNLLogicnt{t_{{\mathrm{2}}}}  \DualLNLLogicsym{/}  \DualLNLLogicmv{y}  \DualLNLLogicsym{]}  \Psi_{{\mathrm{2}}}  \DualLNLLogicsym{,}  \DualLNLLogicsym{[}  \DualLNLLogicnt{t_{{\mathrm{1}}}}  \DualLNLLogicsym{/}  \DualLNLLogicmv{x}  \DualLNLLogicsym{]}  \Psi_{{\mathrm{3}}} \\
                \hline
      \end{array}\\
      \\
      \begin{array}{lll}
        \begin{array}{|c|}
          \multicolumn{1}{l}{\textbf{Coproduct Left:}}\\
          \hline
               \DualLNLLogicmv{x}  :  \DualLNLLogicnt{S}  \vdash_{\mathsf{C} }  \Psi_{{\mathrm{1}}}  \DualLNLLogicsym{,}   \mathsf{case}\, \DualLNLLogicsym{(}   \mathsf{inl}\, \DualLNLLogicnt{t_{{\mathrm{1}}}}   \DualLNLLogicsym{)} \,\mathsf{of}\, \DualLNLLogicmv{y} . \Psi_{{\mathrm{2}}} ,  \DualLNLLogicmv{z} . \Psi_{{\mathrm{3}}}  \\
              \rightsquigarrow\\
                   \DualLNLLogicmv{x}  :  \DualLNLLogicnt{S}  \vdash_{\mathsf{C} }  \Psi_{{\mathrm{1}}}  \DualLNLLogicsym{,}  \DualLNLLogicsym{[}  \DualLNLLogicnt{t_{{\mathrm{1}}}}  \DualLNLLogicsym{/}  \DualLNLLogicmv{y}  \DualLNLLogicsym{]}  \Psi_{{\mathrm{2}}} \\
                  \hline
        \end{array}
        & \quad &
        \begin{array}{|c|}
          \multicolumn{1}{l}{\textbf{Coproduct Right:}}\\
          \hline
               \DualLNLLogicmv{x}  :  \DualLNLLogicnt{S}  \vdash_{\mathsf{C} }  \Psi_{{\mathrm{1}}}  \DualLNLLogicsym{,}   \mathsf{case}\, \DualLNLLogicsym{(}   \mathsf{inr}\, \DualLNLLogicnt{t_{{\mathrm{1}}}}   \DualLNLLogicsym{)} \,\mathsf{of}\, \DualLNLLogicmv{y} . \Psi_{{\mathrm{2}}} ,  \DualLNLLogicmv{z} . \Psi_{{\mathrm{3}}}  \\
              \rightsquigarrow\\
                   \DualLNLLogicmv{x}  :  \DualLNLLogicnt{S}  \vdash_{\mathsf{C} }  \Psi_{{\mathrm{1}}}  \DualLNLLogicsym{,}  \DualLNLLogicsym{[}  \DualLNLLogicnt{t_{{\mathrm{1}}}}  \DualLNLLogicsym{/}  \DualLNLLogicmv{z}  \DualLNLLogicsym{]}  \Psi_{{\mathrm{3}}} \\
                  \hline
        \end{array}
      \end{array}\\
      \\
      \begin{array}{|c|}
        \multicolumn{1}{l}{\textbf{$\mathsf{H}$:}}\\
        \hline        
         \DualLNLLogicmv{x}  :   \mathsf{H}\, \DualLNLLogicnt{B}   \vdash_{\mathsf{C} }  \DualLNLLogicsym{(}   \mathsf{let}\,\mathsf{H}\, \DualLNLLogicmv{x}  =  \DualLNLLogicmv{y} \,\mathsf{in}\, \Psi_{{\mathrm{1}}}   \DualLNLLogicsym{)}  \cdot   \mathsf{let}\,\mathsf{H}\, \DualLNLLogicmv{z}  =  \DualLNLLogicsym{(}   \mathsf{let}\,\mathsf{H}\, \DualLNLLogicmv{x}  =  \DualLNLLogicmv{y} \,\mathsf{in}\,  \mathsf{H}\, \DualLNLLogicnt{e}    \DualLNLLogicsym{)} \,\mathsf{in}\, \Psi_{{\mathrm{2}}}  \\
        \rightsquigarrow\\
         \DualLNLLogicmv{x}  :   \mathsf{H}\, \DualLNLLogicnt{B}   \vdash_{\mathsf{C} }  \DualLNLLogicsym{(}   \mathsf{let}\,\mathsf{H}\, \DualLNLLogicmv{x}  =  \DualLNLLogicmv{y} \,\mathsf{in}\, \Psi_{{\mathrm{1}}}   \DualLNLLogicsym{)}  \cdot  \DualLNLLogicsym{(}   \mathsf{let}\,\mathsf{H}\, \DualLNLLogicmv{x}  =  \DualLNLLogicmv{y} \,\mathsf{in}\, \DualLNLLogicsym{[}  \DualLNLLogicnt{e}  \DualLNLLogicsym{/}  \DualLNLLogicmv{z}  \DualLNLLogicsym{]}  \Psi_{{\mathrm{2}}}   \DualLNLLogicsym{)} \\
        \hline
      \end{array}\\
      \\
      \begin{array}{|c|}
        \multicolumn{1}{l}{\textbf{\textbf{Contraction with $\DualLNLLogicdruleTCXXdEName{}$}:}}\\
        \hline
         \DualLNLLogicmv{x}  :  \DualLNLLogicnt{S}  \vdash_{\mathsf{C} }  \Psi_{{\mathrm{1}}}  \DualLNLLogicsym{,}   \mathsf{case}\, \DualLNLLogicsym{(}   \DualLNLLogicnt{t_{{\mathrm{1}}}}  \cdot  \DualLNLLogicnt{t_{{\mathrm{2}}}}   \DualLNLLogicsym{)} \,\mathsf{of}\, \DualLNLLogicmv{y} . \Psi_{{\mathrm{2}}} ,  \DualLNLLogicmv{z} . \Psi_{{\mathrm{3}}}  \\
        \rightsquigarrow\\
         \DualLNLLogicmv{x}  :  \DualLNLLogicnt{S}  \vdash_{\mathsf{C} }  \Psi_{{\mathrm{1}}}  \DualLNLLogicsym{,}  \DualLNLLogicsym{(}   \mathsf{case}\, \DualLNLLogicnt{t_{{\mathrm{1}}}} \,\mathsf{of}\, \DualLNLLogicmv{y} . \Psi_{{\mathrm{2}}} ,  \DualLNLLogicmv{z} . \Psi_{{\mathrm{3}}}   \DualLNLLogicsym{)}  \cdot  \DualLNLLogicsym{(}   \mathsf{case}\, \DualLNLLogicnt{t_{{\mathrm{2}}}} \,\mathsf{of}\, \DualLNLLogicmv{y} . \Psi_{{\mathrm{2}}} ,  \DualLNLLogicmv{z} . \Psi_{{\mathrm{3}}}   \DualLNLLogicsym{)} \\
        \hline
      \end{array}\\
      \\
      \begin{array}{cc}
        \begin{array}{|c|}
        \multicolumn{1}{l}{\textbf{\textbf{Contraction with $\DualLNLLogicdruleTCXXdIOneName{}$}:}}\\
        \hline
         \DualLNLLogicmv{x}  :  \DualLNLLogicnt{S}  \vdash_{\mathsf{C} }   \mathsf{inl}\, \DualLNLLogicsym{(}   \DualLNLLogicnt{t_{{\mathrm{1}}}}  \cdot  \DualLNLLogicnt{t_{{\mathrm{2}}}}   \DualLNLLogicsym{)}   \DualLNLLogicsym{:}   \DualLNLLogicnt{S_{{\mathrm{1}}}}  +  \DualLNLLogicnt{S_{{\mathrm{2}}}}   \DualLNLLogicsym{,}  \Psi \\
        \rightsquigarrow\\
         \DualLNLLogicmv{x}  :  \DualLNLLogicnt{S}  \vdash_{\mathsf{C} }   \DualLNLLogicsym{(}   \mathsf{inl}\, \DualLNLLogicnt{t_{{\mathrm{1}}}}   \DualLNLLogicsym{)}  \cdot  \DualLNLLogicsym{(}   \mathsf{inl}\, \DualLNLLogicnt{t_{{\mathrm{2}}}}   \DualLNLLogicsym{)}   \DualLNLLogicsym{:}   \DualLNLLogicnt{S_{{\mathrm{1}}}}  +  \DualLNLLogicnt{S_{{\mathrm{2}}}}   \DualLNLLogicsym{,}  \Psi \\
        \hline
        \end{array}
        &
        \begin{array}{|c|}
        \multicolumn{1}{l}{\textbf{\textbf{Contraction with $\DualLNLLogicdruleTCXXdITwoName{}$}:}}\\
        \hline
         \DualLNLLogicmv{x}  :  \DualLNLLogicnt{S}  \vdash_{\mathsf{C} }   \mathsf{inr}\, \DualLNLLogicsym{(}   \DualLNLLogicnt{t_{{\mathrm{1}}}}  \cdot  \DualLNLLogicnt{t_{{\mathrm{2}}}}   \DualLNLLogicsym{)}   \DualLNLLogicsym{:}   \DualLNLLogicnt{S_{{\mathrm{1}}}}  +  \DualLNLLogicnt{S_{{\mathrm{2}}}}   \DualLNLLogicsym{,}  \Psi \\
        \rightsquigarrow\\
         \DualLNLLogicmv{x}  :  \DualLNLLogicnt{S}  \vdash_{\mathsf{C} }   \DualLNLLogicsym{(}   \mathsf{inr}\, \DualLNLLogicnt{t_{{\mathrm{1}}}}   \DualLNLLogicsym{)}  \cdot  \DualLNLLogicsym{(}   \mathsf{inr}\, \DualLNLLogicnt{t_{{\mathrm{2}}}}   \DualLNLLogicsym{)}   \DualLNLLogicsym{:}   \DualLNLLogicnt{S_{{\mathrm{1}}}}  +  \DualLNLLogicnt{S_{{\mathrm{2}}}}   \DualLNLLogicsym{,}  \Psi \\
        \hline
      \end{array}
      \end{array}\\
      \\
      \begin{array}{|c|}
        \multicolumn{1}{l}{\textbf{\textbf{Contraction with $\DualLNLLogicdruleTCXXsubIName{}$}:}}\\
        \hline
         \DualLNLLogicmv{x}  :  \DualLNLLogicnt{S}  \vdash_{\mathsf{C} }  \Psi_{{\mathrm{1}}}  \DualLNLLogicsym{,}   \mathsf{mkc}(  \DualLNLLogicnt{t_{{\mathrm{1}}}}  \cdot  \DualLNLLogicnt{t_{{\mathrm{2}}}}  , \DualLNLLogicmv{y} )   \DualLNLLogicsym{:}   \DualLNLLogicnt{T_{{\mathrm{1}}}}  -  \DualLNLLogicnt{T_{{\mathrm{2}}}}   \DualLNLLogicsym{,}  \DualLNLLogicsym{[}  \DualLNLLogicmv{y}  \DualLNLLogicsym{(}   \DualLNLLogicnt{t_{{\mathrm{1}}}}  \cdot  \DualLNLLogicnt{t_{{\mathrm{2}}}}   \DualLNLLogicsym{)}  \DualLNLLogicsym{/}  \DualLNLLogicmv{y}  \DualLNLLogicsym{]}  \Psi_{{\mathrm{2}}} \\
        \rightsquigarrow\\
         \DualLNLLogicmv{x}  :  \DualLNLLogicnt{S}  \vdash_{\mathsf{C} }  \Psi_{{\mathrm{1}}}  \DualLNLLogicsym{,}  \DualLNLLogicsym{(}    \mathsf{mkc}( \DualLNLLogicnt{t_{{\mathrm{1}}}} , \DualLNLLogicmv{y} )   \cdot   \mathsf{mkc}( \DualLNLLogicnt{t_{{\mathrm{2}}}} , \DualLNLLogicmv{y} )    \DualLNLLogicsym{)}  \DualLNLLogicsym{:}   \DualLNLLogicnt{T_{{\mathrm{1}}}}  -  \DualLNLLogicnt{T_{{\mathrm{2}}}}   \DualLNLLogicsym{,}  \DualLNLLogicsym{(}  \DualLNLLogicsym{[}  \DualLNLLogicmv{y}  \DualLNLLogicsym{(}  \DualLNLLogicnt{t_{{\mathrm{1}}}}  \DualLNLLogicsym{)}  \DualLNLLogicsym{/}  \DualLNLLogicmv{y}  \DualLNLLogicsym{]}  \Psi_{{\mathrm{2}}}  \cdot  \DualLNLLogicsym{[}  \DualLNLLogicmv{y}  \DualLNLLogicsym{(}  \DualLNLLogicnt{t_{{\mathrm{2}}}}  \DualLNLLogicsym{)}  \DualLNLLogicsym{/}  \DualLNLLogicmv{y}  \DualLNLLogicsym{]}  \Psi_{{\mathrm{2}}}  \DualLNLLogicsym{)} \\
        \hline
      \end{array}\\
      \\
      \begin{array}{|c|}
        \multicolumn{1}{l}{\textbf{\textbf{Contraction with $\DualLNLLogicdruleTCXXsubEName{}$}:}}\\
        \hline
         \DualLNLLogicmv{x}  :  \DualLNLLogicnt{S}  \vdash_{\mathsf{C} }  \Psi_{{\mathrm{1}}}  \DualLNLLogicsym{,}    \mathsf{postp}\,( \DualLNLLogicmv{z}  \mapsto  \DualLNLLogicnt{s} ,   \DualLNLLogicnt{t_{{\mathrm{1}}}}  \cdot  \DualLNLLogicnt{t_{{\mathrm{2}}}}  )    \DualLNLLogicsym{,}  \DualLNLLogicsym{[}  \DualLNLLogicmv{y}  \DualLNLLogicsym{(}   \DualLNLLogicnt{t_{{\mathrm{1}}}}  \cdot  \DualLNLLogicnt{t_{{\mathrm{2}}}}   \DualLNLLogicsym{)}  \DualLNLLogicsym{/}  \DualLNLLogicmv{y}  \DualLNLLogicsym{]}  \Psi_{{\mathrm{2}}} \\
        \rightsquigarrow\\
         \DualLNLLogicmv{x}  :  \DualLNLLogicnt{S}  \vdash_{\mathsf{C} }  \DualLNLLogicsym{(}  \Psi_{{\mathrm{1}}}  \DualLNLLogicsym{,}    \mathsf{postp}\,( \DualLNLLogicmv{z}  \mapsto  \DualLNLLogicnt{s} ,  \DualLNLLogicnt{t_{{\mathrm{1}}}} )    \DualLNLLogicsym{,}  \DualLNLLogicsym{[}  \DualLNLLogicmv{y}  \DualLNLLogicsym{(}  \DualLNLLogicnt{t_{{\mathrm{1}}}}  \DualLNLLogicsym{)}  \DualLNLLogicsym{/}  \DualLNLLogicmv{y}  \DualLNLLogicsym{]}  \Psi_{{\mathrm{2}}}  \DualLNLLogicsym{)}  \cdot  \DualLNLLogicsym{(}  \Psi_{{\mathrm{1}}}  \DualLNLLogicsym{,}    \mathsf{postp}\,( \DualLNLLogicmv{z}  \mapsto  \DualLNLLogicnt{s} ,  \DualLNLLogicnt{t_{{\mathrm{2}}}} )    \DualLNLLogicsym{,}  \DualLNLLogicsym{[}  \DualLNLLogicmv{y}  \DualLNLLogicsym{(}  \DualLNLLogicnt{t_{{\mathrm{2}}}}  \DualLNLLogicsym{)}  \DualLNLLogicsym{/}  \DualLNLLogicmv{y}  \DualLNLLogicsym{]}  \Psi_{{\mathrm{2}}}  \DualLNLLogicsym{)} \\
        \hline
      \end{array}\\
      \\
      \begin{array}{|c|}
        \multicolumn{1}{l}{\textbf{\textbf{Contraction with $\DualLNLLogicdruleTCXXHEName{}$}:}}\\
        \hline
         \DualLNLLogicmv{x}  :  \DualLNLLogicnt{S}  \vdash_{\mathsf{C} }  \Psi_{{\mathrm{1}}}  \DualLNLLogicsym{,}   \mathsf{let}\,\mathsf{H}\, \DualLNLLogicmv{y}  =   \DualLNLLogicnt{t_{{\mathrm{1}}}}  \cdot  \DualLNLLogicnt{t_{{\mathrm{2}}}}  \,\mathsf{in}\, \Psi_{{\mathrm{2}}}  \\
        \rightsquigarrow\\
         \DualLNLLogicmv{x}  :  \DualLNLLogicnt{S}  \vdash_{\mathsf{C} }  \Psi_{{\mathrm{1}}}  \DualLNLLogicsym{,}  \DualLNLLogicsym{(}   \mathsf{let}\,\mathsf{H}\, \DualLNLLogicmv{y}  =  \DualLNLLogicnt{t_{{\mathrm{1}}}} \,\mathsf{in}\, \Psi_{{\mathrm{2}}}   \DualLNLLogicsym{)}  \cdot  \DualLNLLogicsym{(}   \mathsf{let}\,\mathsf{H}\, \DualLNLLogicmv{y}  =  \DualLNLLogicnt{t_{{\mathrm{2}}}} \,\mathsf{in}\, \Psi_{{\mathrm{2}}}   \DualLNLLogicsym{)} \\
        \hline
      \end{array}\\      
    \end{array}
  \end{math}
  \end{center}
  \end{mdframed}
  \caption{Reductions for Non-linear Terms}
  \label{fig:red-non-linear}
\end{figure}
\begin{figure}
  \begin{mdframed}
    \begin{center}
    \begin{math}
      \begin{array}{c}    
        \begin{array}{|c|}
          \multicolumn{1}{l}{\textbf{\textbf{Weakening with $\DualLNLLogicdruleTCXXdEName{}$}:}}\\
          \hline
            \DualLNLLogicmv{x}  :  \DualLNLLogicnt{S}  \vdash_{\mathsf{C} }  \Psi_{{\mathrm{1}}}  \DualLNLLogicsym{,}   \mathsf{case}\, \DualLNLLogicsym{(}   \varepsilon   \DualLNLLogicsym{)} \,\mathsf{of}\, \DualLNLLogicmv{y} . \Psi_{{\mathrm{2}}} ,  \DualLNLLogicmv{z} . \Psi_{{\mathrm{3}}}  \\
           \rightsquigarrow\\
            \DualLNLLogicmv{x}  :  \DualLNLLogicnt{S}  \vdash_{\mathsf{C} }  \Psi_{{\mathrm{1}}}  \DualLNLLogicsym{,}   \varepsilon   \DualLNLLogicsym{:}  \DualLNLLogicnt{S_{{\mathrm{1}}}}  \DualLNLLogicsym{,} \, ... \, \DualLNLLogicsym{,}   \varepsilon   \DualLNLLogicsym{:}  \DualLNLLogicnt{S_{\DualLNLLogicmv{i}}} \\
           \text{where } \DualLNLLogicsym{\mbox{$\mid$}}  \Psi_{{\mathrm{2}}}  \DualLNLLogicsym{\mbox{$\mid$}}  \DualLNLLogicsym{=}  \DualLNLLogicsym{\mbox{$\mid$}}  \Psi_{{\mathrm{3}}}  \DualLNLLogicsym{\mbox{$\mid$}} \text{ and } |\Psi_{{\mathrm{2}}}| = \DualLNLLogicnt{S_{{\mathrm{1}}}}  \DualLNLLogicsym{,} \, ... \, \DualLNLLogicsym{,}  \DualLNLLogicnt{S_{\DualLNLLogicmv{i}}}\\    
          \hline
        \end{array}\\
        \\
        \begin{array}{cc}
          \begin{array}{|c|}
          \multicolumn{1}{l}{\textbf{\textbf{Weakening with $\DualLNLLogicdruleTCXXdIOneName{}$}:}}\\
          \hline
            \DualLNLLogicmv{x}  :  \DualLNLLogicnt{S}  \vdash_{\mathsf{C} }  \Psi  \DualLNLLogicsym{,}   \mathsf{inl}\,  \varepsilon    \DualLNLLogicsym{:}   \DualLNLLogicnt{S_{{\mathrm{1}}}}  +  \DualLNLLogicnt{S_{{\mathrm{2}}}}  \\
           \rightsquigarrow\\
            \DualLNLLogicmv{x}  :  \DualLNLLogicnt{S}  \vdash_{\mathsf{C} }  \Psi  \DualLNLLogicsym{,}   \varepsilon   \DualLNLLogicsym{:}   \DualLNLLogicnt{S_{{\mathrm{1}}}}  +  \DualLNLLogicnt{S_{{\mathrm{2}}}}  \\               
          \hline
        \end{array}
        &
        \begin{array}{|c|}
          \multicolumn{1}{l}{\textbf{\textbf{Weakening with $\DualLNLLogicdruleTCXXdITwoName{}$}:}}\\
          \hline
            \DualLNLLogicmv{x}  :  \DualLNLLogicnt{S}  \vdash_{\mathsf{C} }  \Psi  \DualLNLLogicsym{,}   \mathsf{inr}\,  \varepsilon    \DualLNLLogicsym{:}   \DualLNLLogicnt{S_{{\mathrm{1}}}}  +  \DualLNLLogicnt{S_{{\mathrm{2}}}}  \\
           \rightsquigarrow\\
            \DualLNLLogicmv{x}  :  \DualLNLLogicnt{S}  \vdash_{\mathsf{C} }  \Psi  \DualLNLLogicsym{,}   \varepsilon   \DualLNLLogicsym{:}   \DualLNLLogicnt{S_{{\mathrm{1}}}}  +  \DualLNLLogicnt{S_{{\mathrm{2}}}}  \\               
          \hline
        \end{array}
        \end{array}\\
        \\
        \begin{array}{cc}
          \begin{array}{|c|}
          \multicolumn{1}{l}{\textbf{\textbf{Weakening with $\DualLNLLogicdruleTCXXsubEName{}$}:}}\\
          \hline
            \DualLNLLogicmv{x}  :  \DualLNLLogicnt{S}  \vdash_{\mathsf{C} }  \Psi_{{\mathrm{1}}}  \DualLNLLogicsym{,}    \mathsf{postp}\,( \DualLNLLogicmv{z}  \mapsto  \DualLNLLogicnt{s} ,   \varepsilon  )    \DualLNLLogicsym{,}  \DualLNLLogicsym{[}  \DualLNLLogicmv{y}  \DualLNLLogicsym{(}   \varepsilon   \DualLNLLogicsym{)}  \DualLNLLogicsym{/}  \DualLNLLogicmv{y}  \DualLNLLogicsym{]}  \Psi_{{\mathrm{2}}} \\
           \rightsquigarrow\\
            \DualLNLLogicmv{x}  :  \DualLNLLogicnt{S}  \vdash_{\mathsf{C} }  \Psi_{{\mathrm{1}}}  \DualLNLLogicsym{,}  \DualLNLLogicsym{[}   \varepsilon   \DualLNLLogicsym{/}  \DualLNLLogicmv{y}  \DualLNLLogicsym{]}  \Psi_{{\mathrm{2}}} \\               
          \hline
        \end{array}
        &
        \begin{array}{|c|}
          \multicolumn{1}{l}{\textbf{\textbf{Weakening with $\DualLNLLogicdruleTCXXsubIName{}$}:}}\\
          \hline
            \DualLNLLogicmv{x}  :  \DualLNLLogicnt{S}  \vdash_{\mathsf{C} }  \Psi_{{\mathrm{1}}}  \DualLNLLogicsym{,}   \mathsf{mkc}(  \varepsilon  , \DualLNLLogicmv{y} )   \DualLNLLogicsym{:}   \DualLNLLogicnt{T_{{\mathrm{1}}}}  -  \DualLNLLogicnt{T_{{\mathrm{2}}}}   \DualLNLLogicsym{,}  \DualLNLLogicsym{[}  \DualLNLLogicmv{y}  \DualLNLLogicsym{(}   \varepsilon   \DualLNLLogicsym{)}  \DualLNLLogicsym{/}  \DualLNLLogicmv{y}  \DualLNLLogicsym{]}  \Psi_{{\mathrm{2}}} \\
           \rightsquigarrow\\
            \DualLNLLogicmv{x}  :  \DualLNLLogicnt{S}  \vdash_{\mathsf{C} }  \Psi_{{\mathrm{1}}}  \DualLNLLogicsym{,}   \varepsilon   \DualLNLLogicsym{:}   \DualLNLLogicnt{T_{{\mathrm{1}}}}  -  \DualLNLLogicnt{T_{{\mathrm{2}}}}   \DualLNLLogicsym{,}  \DualLNLLogicsym{[}   \varepsilon   \DualLNLLogicsym{/}  \DualLNLLogicmv{y}  \DualLNLLogicsym{]}  \Psi_{{\mathrm{2}}} \\               
          \hline
        \end{array}
        \end{array}\\
        \\
        \begin{array}{|c|}
          \multicolumn{1}{l}{\textbf{\textbf{Weakening with $\DualLNLLogicdruleTCXXHEName{}$}:}}\\
          \hline
            \DualLNLLogicmv{x}  :  \DualLNLLogicnt{S}  \vdash_{\mathsf{C} }  \Psi_{{\mathrm{1}}}  \DualLNLLogicsym{,}   \mathsf{let}\,\mathsf{H}\, \DualLNLLogicmv{y}  =   \varepsilon  \,\mathsf{in}\, \Psi_{{\mathrm{2}}}  \\
           \rightsquigarrow\\
            \DualLNLLogicmv{x}  :  \DualLNLLogicnt{S}  \vdash_{\mathsf{C} }  \Psi_{{\mathrm{1}}}  \DualLNLLogicsym{,}  \DualLNLLogicsym{[}   \varepsilon   \DualLNLLogicsym{/}  \DualLNLLogicmv{y}  \DualLNLLogicsym{]}  \Psi_{{\mathrm{2}}} \\               
          \hline
        \end{array}
      \end{array}
    \end{math}
  \end{center}
  \end{mdframed}
  \caption{Reductions for Non-linear Terms Continued}
  \label{fig:red-non-linear-cont}
\end{figure}

The reduction rules for the linear and non-linear fragments can be
found in Figure~\ref{fig:red-linear-terms} and
Figure~\ref{fig:red-non-linear} respectively.  We denote the judgments
for reduction by $ \DualLNLLogicmv{x}  :  \DualLNLLogicnt{S}  \vdash_{\mathsf{C} }  \Psi_{{\mathrm{1}}}  \rightsquigarrow  \DualLNLLogicmv{x}  :  \DualLNLLogicnt{S}  \vdash_{\mathsf{C} }  \Psi_{{\mathrm{2}}} $ and $ \DualLNLLogicmv{x}  :  \DualLNLLogicnt{A}  \vdash_{\mathsf{L} }  \Delta_{{\mathrm{1}}} ; \Psi_{{\mathrm{1}}}  \rightsquigarrow  \DualLNLLogicmv{x}  :  \DualLNLLogicnt{A}  \vdash_{\mathsf{L} }  \Delta_{{\mathrm{2}}} ; \Psi_{{\mathrm{2}}} $.  In the interest of readability we do not show full
derivations, but it should be noted that it is assumed that every term
mentioned in a reduction rule is typable with the expected type given
where it occurs in the judgment.  Furthermore, the reduction relation
depends on a few standard definitions and non-standard binding operations.

The non-standard binding operations concern the variable $y$ in $\mathtt{mkc}(t,y)$ and in
$\mathtt{postp}(y\mapsto t, s)$ and the related expressions $y(t)$ and $y(s)$, respectively, 
occurring in the non-linear context; similar operations occur in the linear case.
Consider term assignment to the rule subtraction introduction  $\mathrm{TC}\_-_I$
in Figure \ref{fig:non-linear-ta}. The variable $y$ is the unique free variable occurring in the 
sequent $y:T_2 \vdash_C \Psi_2$, the minor premise of the inference. In the conclusion 
$x:S \vdash_{\Psi_1}, \mathtt{mkc}(t,y): T_1 - T_2, [y(t)/y]\Psi_2$ the variable $y$ is bound 
in $\mathtt{mkc}(t,y)$; moreover, the occurrences of the free variable $y$ have been
substituted simultaneously in the context $\Psi$ by the expression $y(t)$ which denotes 
a bound varianble, indexed with $t$. Similar explanations apply to the term assignment 
for subtraction elimination, and to the corresponding linear rules in Figure \ref{fig:linear-ta}. 

An analogue of the capture of a free variable by a binder in the $\lambda$-calculus, 
is an occurrence of a bound variable $y(t)$ whose binder is ambiguous, for instance 
in a context where there were two occurrences of $\mathtt{mkc}(t,y)$, as a result of 
a contraction/cut reduction in a derivation. Such a context may be the conclusion of 
the following derivation, if $x_1 = x_2$, $y_1 = y_2$; here $t_1 = \mathtt{false}\; x_1$, 
$t_2 = \mathtt{false}\; x_2$:
{\small
 \[
 \AxiomC{$z:0 \vdash_C z:0$}
 \AxiomC{$x_1:S \vdash_C x_1:S$}
 \AxiomC{$y_1:T \vdash_C y_1:T$}
 \BinaryInfC{$x_1:S \vdash_C \mathtt{mkc}(x_1,y_1): S - T, y_1(x_1): T$}
\AxiomC{$x_2:A \vdash_C x_2:A$}
\AxiomC{$y_2:B \vdash_C y_2:B$}
 \BinaryInfC{$x_2:S \vdash_C \mathtt{mkc}(x_2,y_2): S - T, y_2(x_2): T$}
\TrinaryInfC{$z:0 \vdash_C, \mathtt{mkc}(t_1,y_1): S - T, \mathtt{mkc}(t_2,y_2): S - T, y_1(t_1): T, y_2(t_2): T$}
\DisplayProof
\]
}

A formal notion of $\alpha$ conversion has been proposed 
 for this notion of binding in untyped linear contexts in  \cite{Bellin:2012}. 
Here (capture-avoiding)  substitution,  denoted by 
$\DualLNLLogicsym{[}  \DualLNLLogicnt{t_{{\mathrm{1}}}}  \DualLNLLogicsym{/}  \DualLNLLogicmv{x}  \DualLNLLogicsym{]}  \DualLNLLogicnt{t_{{\mathrm{2}}}}$, $\DualLNLLogicsym{[}  \DualLNLLogicnt{e}  \DualLNLLogicsym{/x}  \DualLNLLogicsym{]}  \DualLNLLogicnt{t}$, $\DualLNLLogicsym{[}  \DualLNLLogicnt{t}  \DualLNLLogicsym{/}  \DualLNLLogicmv{x}  \DualLNLLogicsym{]}  \DualLNLLogicnt{e}$, and $\DualLNLLogicsym{[}  \DualLNLLogicnt{e_{{\mathrm{1}}}}  \DualLNLLogicsym{/}  \DualLNLLogicmv{x}  \DualLNLLogicsym{]}  \DualLNLLogicnt{e_{{\mathrm{2}}}}$, 
is defined in the usual way.  We extend capture-avoiding substitution to multisets in
the following way:
\begin{itemize}
\item $[t_1\cdot\ldots\cdot t_n/ z]s = [t_1/z]s\cdot\ldots\cdot [t_n/ z]s$
\item $[t_1\cdot\ldots\cdot t_n/ z]p = [t_1/z]p\,\|\, \ldots\,\|\, [t_n/ z]p,$ where $p$ is a $p$-term
\end{itemize}
The extension of the other flavors of substitution to multisets are
similar.  Standard extension of substitution to contexts was also
necessary.

Finally, there are several commuting conversions that are required
for reduction, for example, the following is one:
\[
\AxiomC{$\hskip1.6in  \DualLNLLogicmv{y}  :  \DualLNLLogicnt{T_{{\mathrm{2}}}}  \vdash_{\mathsf{C} }  \Psi_{{\mathrm{2}}}  \DualLNLLogicsym{,}  \DualLNLLogicnt{t_{{\mathrm{1}}}}  \DualLNLLogicsym{:}   \DualLNLLogicnt{T_{{\mathrm{4}}}}  +  \DualLNLLogicnt{T_{{\mathrm{5}}}}  $}
\noLine
\UnaryInfC{$ \DualLNLLogicmv{x}  :  \DualLNLLogicnt{S}  \vdash_{\mathsf{C} }  \Psi_{{\mathrm{1}}}  \DualLNLLogicsym{,}  \DualLNLLogicnt{t}  \DualLNLLogicsym{:}   \DualLNLLogicnt{T_{{\mathrm{2}}}}  +  \DualLNLLogicnt{T_{{\mathrm{3}}}}  \quad  \DualLNLLogicmv{z}  :  \DualLNLLogicnt{T_{{\mathrm{3}}}}  \vdash_{\mathsf{C} }  \Psi_{{\mathrm{3}}}  \DualLNLLogicsym{,}  \DualLNLLogicnt{t_{{\mathrm{2}}}}  \DualLNLLogicsym{:}   \DualLNLLogicnt{T_{{\mathrm{4}}}}  +  \DualLNLLogicnt{T_{{\mathrm{5}}}}  $}
\UnaryInfC{$ \DualLNLLogicmv{x}  :  \DualLNLLogicnt{S}  \vdash_{\mathsf{C} }  \Psi_{{\mathrm{1}}}  \DualLNLLogicsym{,}   \mathsf{case}\, \DualLNLLogicnt{t} \,\mathsf{of}\, \DualLNLLogicmv{y} . \DualLNLLogicnt{t_{{\mathrm{1}}}} , \DualLNLLogicmv{z} . \DualLNLLogicnt{t_{{\mathrm{2}}}}   \DualLNLLogicsym{:}   \DualLNLLogicnt{T_{{\mathrm{4}}}}  +  \DualLNLLogicnt{T_{{\mathrm{5}}}}  $}
\AxiomC{$ \DualLNLLogicmv{v_{{\mathrm{1}}}}  :  \DualLNLLogicnt{T_{{\mathrm{4}}}}  \vdash_{\mathsf{C} }  \Psi_{{\mathrm{4}}}  \quad  \DualLNLLogicmv{v_{{\mathrm{2}}}}  :  \DualLNLLogicnt{T_{{\mathrm{5}}}}  \vdash_{\mathsf{C} }  \Psi_{{\mathrm{5}}} $}
\BinaryInfC{$ \DualLNLLogicmv{x}  :  \DualLNLLogicnt{S}  \vdash_{\mathsf{C} }  \Psi_{{\mathrm{1}}}  \DualLNLLogicsym{,}    \mathsf{case}\, \DualLNLLogicsym{(}   \mathsf{case}\, \DualLNLLogicnt{t} \,\mathsf{of}\, \DualLNLLogicmv{y} . \DualLNLLogicnt{t_{{\mathrm{1}}}} , \DualLNLLogicmv{z} . \DualLNLLogicnt{t_{{\mathrm{2}}}}   \DualLNLLogicsym{)} \,\mathsf{of}\, \DualLNLLogicmv{v_{{\mathrm{1}}}} . \Psi_{{\mathrm{4}}} ,  \DualLNLLogicmv{v_{{\mathrm{2}}}} . \Psi_{{\mathrm{5}}}   $}
\DisplayProof
\]
commutes to
\[
\inferrule* [right=] {
   \DualLNLLogicmv{x}  :  \DualLNLLogicnt{S}  \vdash_{\mathsf{C} }  \Psi_{{\mathrm{1}}}  \DualLNLLogicsym{,}  \DualLNLLogicnt{t}  \DualLNLLogicsym{:}   \DualLNLLogicnt{T_{{\mathrm{2}}}}  +  \DualLNLLogicnt{T_{{\mathrm{3}}}}   \\ \Pi_1 \\ \Pi_2
}{ \DualLNLLogicmv{x}  :  \DualLNLLogicnt{S}  \vdash_{\mathsf{C} }  \Psi_{{\mathrm{1}}}  \DualLNLLogicsym{,}   \mathsf{case}\, \DualLNLLogicnt{t} \,\mathsf{of}\, \DualLNLLogicmv{y_{{\mathrm{2}}}} . \DualLNLLogicsym{(}  \Psi_{{\mathrm{2}}}  \DualLNLLogicsym{,}   \mathsf{case}\, \DualLNLLogicnt{t_{{\mathrm{1}}}} \,\mathsf{of}\, \DualLNLLogicmv{v_{{\mathrm{1}}}} . \Psi_{{\mathrm{4}}} ,  \DualLNLLogicmv{v_{{\mathrm{2}}}} . \Psi_{{\mathrm{5}}}   \DualLNLLogicsym{)} ,  \DualLNLLogicmv{y_{{\mathrm{3}}}} . \DualLNLLogicsym{(}  \Psi_{{\mathrm{3}}}  \DualLNLLogicsym{,}   \mathsf{case}\, \DualLNLLogicnt{t_{{\mathrm{2}}}} \,\mathsf{of}\, \DualLNLLogicmv{v_{{\mathrm{1}}}} . \Psi_{{\mathrm{4}}} ,  \DualLNLLogicmv{v_{{\mathrm{2}}}} . \Psi_{{\mathrm{5}}}   \DualLNLLogicsym{)}  }
\]
where
\[
\begin{array}{lll}
  \Pi_1:\\
  & \inferrule* [right=] {
     \DualLNLLogicmv{y_{{\mathrm{2}}}}  :  \DualLNLLogicnt{T_{{\mathrm{2}}}}  \vdash_{\mathsf{C} }  \Psi_{{\mathrm{2}}}  \DualLNLLogicsym{,}  \DualLNLLogicnt{t_{{\mathrm{1}}}}  \DualLNLLogicsym{:}   \DualLNLLogicnt{T_{{\mathrm{4}}}}  +  \DualLNLLogicnt{T_{{\mathrm{5}}}}   \\  \DualLNLLogicmv{v_{{\mathrm{1}}}}  :  \DualLNLLogicnt{T_{{\mathrm{4}}}}  \vdash_{\mathsf{C} }  \Psi_{{\mathrm{4}}}  \\  \DualLNLLogicmv{v_{{\mathrm{2}}}}  :  \DualLNLLogicnt{T_{{\mathrm{5}}}}  \vdash_{\mathsf{C} }  \Psi_{{\mathrm{5}}} 
  }{ \DualLNLLogicmv{y_{{\mathrm{2}}}}  :  \DualLNLLogicnt{T_{{\mathrm{2}}}}  \vdash_{\mathsf{C} }  \Psi_{{\mathrm{2}}}  \DualLNLLogicsym{,}   \mathsf{case}\, \DualLNLLogicnt{t_{{\mathrm{1}}}} \,\mathsf{of}\, \DualLNLLogicmv{v_{{\mathrm{1}}}} . \Psi_{{\mathrm{4}}} ,  \DualLNLLogicmv{v_{{\mathrm{2}}}} . \Psi_{{\mathrm{5}}}  }\\
  \\
  \Pi_2:\\
  & \inferrule* [right=] {
     \DualLNLLogicmv{y_{{\mathrm{3}}}}  :  \DualLNLLogicnt{T_{{\mathrm{3}}}}  \vdash_{\mathsf{C} }  \Psi_{{\mathrm{3}}}  \DualLNLLogicsym{,}  \DualLNLLogicnt{t_{{\mathrm{2}}}}  \DualLNLLogicsym{:}   \DualLNLLogicnt{T_{{\mathrm{4}}}}  +  \DualLNLLogicnt{T_{{\mathrm{5}}}}   \\  \DualLNLLogicmv{v_{{\mathrm{1}}}}  :  \DualLNLLogicnt{T_{{\mathrm{4}}}}  \vdash_{\mathsf{C} }  \Psi_{{\mathrm{4}}}  \\  \DualLNLLogicmv{v_{{\mathrm{2}}}}  :  \DualLNLLogicnt{T_{{\mathrm{5}}}}  \vdash_{\mathsf{C} }  \Psi_{{\mathrm{5}}} 
  }{ \DualLNLLogicmv{y_{{\mathrm{3}}}}  :  \DualLNLLogicnt{T_{{\mathrm{3}}}}  \vdash_{\mathsf{C} }  \Psi_{{\mathrm{3}}}  \DualLNLLogicsym{,}   \mathsf{case}\, \DualLNLLogicnt{t_{{\mathrm{2}}}} \,\mathsf{of}\, \DualLNLLogicmv{v_{{\mathrm{1}}}} . \Psi_{{\mathrm{4}}} ,  \DualLNLLogicmv{v_{{\mathrm{2}}}} . \Psi_{{\mathrm{5}}}  }
\end{array}
\]
If $\DualLNLLogicnt{t_{{\mathrm{1}}}} =  \mathsf{inl}\, \DualLNLLogicnt{s_{{\mathrm{1}}}} $ and $\DualLNLLogicnt{t_{{\mathrm{2}}}} =  \mathsf{inr}\, \DualLNLLogicnt{s_{{\mathrm{2}}}} $ then after commutation 
\[
 \DualLNLLogicmv{y_{{\mathrm{2}}}}  :  \DualLNLLogicnt{T_{{\mathrm{2}}}}  \vdash_{\mathsf{C} }  \Psi_{{\mathrm{2}}}  \DualLNLLogicsym{,}   \mathsf{case}\, \DualLNLLogicsym{(}   \mathsf{inl}\, \DualLNLLogicnt{s_{{\mathrm{1}}}}   \DualLNLLogicsym{)} \,\mathsf{of}\, \DualLNLLogicmv{v_{{\mathrm{1}}}} . \Psi_{{\mathrm{4}}} ,  \DualLNLLogicmv{v_{{\mathrm{2}}}} . \Psi_{{\mathrm{5}}}   \rightsquigarrow_{\beta}
 \DualLNLLogicmv{y_{{\mathrm{2}}}}  :  \DualLNLLogicnt{T_{{\mathrm{2}}}}  \vdash_{\mathsf{C} }  \Psi_{{\mathrm{2}}}  \DualLNLLogicsym{,}  \DualLNLLogicsym{[}  \DualLNLLogicnt{s_{{\mathrm{1}}}}  \DualLNLLogicsym{/}  \DualLNLLogicmv{v_{{\mathrm{1}}}}  \DualLNLLogicsym{]}  \Psi_{{\mathrm{4}}}  
\]
and
\[
 \DualLNLLogicmv{y_{{\mathrm{3}}}}  :  \DualLNLLogicnt{T_{{\mathrm{3}}}}  \vdash_{\mathsf{C} }  \Psi_{{\mathrm{3}}}  \DualLNLLogicsym{,}   \mathsf{case}\, \DualLNLLogicsym{(}   \mathsf{inr}\, \DualLNLLogicnt{s_{{\mathrm{2}}}}   \DualLNLLogicsym{)} \,\mathsf{of}\, \DualLNLLogicmv{v_{{\mathrm{1}}}} . \Psi_{{\mathrm{4}}} ,  \DualLNLLogicmv{v_{{\mathrm{2}}}} . \Psi_{{\mathrm{5}}}   \rightsquigarrow_{\beta}
 \DualLNLLogicmv{y_{{\mathrm{2}}}}  :  \DualLNLLogicnt{T_{{\mathrm{2}}}}  \vdash_{\mathsf{C} }  \Psi_{{\mathrm{3}}}  \DualLNLLogicsym{,}  \DualLNLLogicsym{[}  \DualLNLLogicnt{s_{{\mathrm{2}}}}  \DualLNLLogicsym{/}  \DualLNLLogicmv{v_{{\mathrm{2}}}}  \DualLNLLogicsym{]}  \Psi_{{\mathrm{5}}} 
\]
There are other commuting conversions as well, but as one can see, due
to the complexities introduced in reduction arising from the fact that
multiple terms in the context are affected during reduction results in
the commuting conversions from being very compact.  The remainder of
the commuting conversions can be found in
Appendix~\ref{sec:commuting_conversions}.  In the next section we give
the interpretation of TND into the categorical model.

\subsection{Categorical interpretation of rules}
\label{sec:categorical_interpretation_of_rules}

We now turn to the interpretation of Dual LNL Logic into our
categorical model given in Section~\ref{sec:adjoint_model}.  We
structure the proof similarly to Bierman~\cite{Bierman:1994}, but the
proof itself follows similarly to Benton's~\cite{Benton:1994} proof
for LNL Logic.

Given a {\em signature} $\mathsf{Sg}$, consisting of a collection of types $\sigma_i$, where $\sigma_i = A\ \hbox{or}\ S$, 
and a collection of {\em sorted function symbols} $f_j : \sigma_1, \ldots, \sigma_n \rightarrow \tau$ and given 
a Symmetric Monoidal Category (SMC) $(\mathbb{C}, \bullet, 1, \alpha, \lambda, \rho, \gamma)$,  a {\em structure} 
$\mathcal{M}$ for $\mathsf{Sg}$ is an assignment of an object $[\![\sigma]\!]$ of $\mathcal{L}$ for each type $\sigma$ and of 
a morphism $[\![f]\!] : [\![\sigma_1]\!]\bullet\ldots\bullet[\![\sigma_n]\!]\rightarrow [\![\tau]\!]$ for each function
$f : \sigma_1, \ldots, \sigma_n \rightarrow \tau$ of $\mathsf{Sg}$.\footnote{ In this subsection only we use the symbol $\bullet$
  and 1 for the monoidal binary operation and its unit in the categorical structure, distinguished from the $\oplus$ and
  $\bot$ symbols in the formal language. We shall show that the interpretation of $\oplus$ is isomorphic to the operation
  $\bullet$, so we shall be able to identify them ( and similarly for $\bot$ and 1).}  The types of terms in context $\Delta = [e_1: A_1, \ldots, e_n: A_n]$ or $\Delta = [t_1: T_1, \ldots, t_n: T_n]$ are 
interpreted into the SMC as $[\![\sigma_1, \sigma_2, \ldots, \sigma_n]\!]$ =   
$(\ldots ([\![\sigma_1]\!]\bullet[\![\sigma_2]\!])\ldots ) \bullet [\![\sigma_n]\!]$; left associativity is also 
intended for concatenations of type sequences $\Gamma, \Delta$. Thus, we need the ``book-keeping'' functions 
$\mathtt{Split}(\Gamma, \Delta): [\![\Gamma, \Delta]\!] \rightarrow [\![\Gamma]\!]\bullet [\![\Delta]\!]$ and 
$\mathtt{Join}(\Gamma, \Delta): [\![\Gamma]\!] \bullet  [\![\Delta]\!] \rightarrow [\![\Gamma, \Delta]\!]$ inductively defined 
using the associativity laws $\alpha$ and its inverse $\alpha^{-1}$ (cfr Bierman 1994, given also in Bellin 2015).  

The semantics of terms in context is then specified by induction on terms: 
\begin{center}
  \begin{tabular}{c}
    $[\![x: A\vdash_{\mathsf{L}} x: A]\!] =_{df} id_{[\![A\!]}$\\
    \\
    $[\![x: A\vdash_{\mathsf{L}}  f(e_1,\ldots, e_n): B]\!] =_{df} [\![x: A\vdash_{\mathsf{L}} e_1: A_1]\!]\bullet
    \ldots\bullet [\![x:\sigma\vdash_{\mathsf{L}} e_n: A_n]\!]; [\![ f ]\!]$
  \end{tabular}
\end{center}   
and similarly with non-linear types. 
Following this one then proves by induction on the type derivation that substitution in the term calculus 
corresponds to composition in the category (\cite{Bierman:1994}, Lemma 13). 

In the mixed sequents $x:A \vdash_{\mathsf{L}} \Delta ; \Psi, t:T$ of TND non-linear terms are interpreted through the functor $J :
\mathcal{C} \rightarrow \mathcal{L}$. Thus, we have the following:
\[
   \DualLNLLogicmv{x}  :  \DualLNLLogicnt{A}  \vdash_{\mathsf{L} }  \Delta ; \Psi  \DualLNLLogicsym{,}  \DualLNLLogicnt{t}  \DualLNLLogicsym{:}  \DualLNLLogicnt{T}   =  \DualLNLLogicmv{x}  :  \DualLNLLogicnt{A}  \vdash_{\mathsf{L} }  \Delta  \DualLNLLogicsym{,}    \mathsf{J}\, \Psi    \DualLNLLogicsym{,}   \mathsf{J}\, \DualLNLLogicnt{t}   \DualLNLLogicsym{:}   \mathsf{J}\, \DualLNLLogicnt{T}  ;  \cdot  
\]  
Let $\mathcal{M}$ be a structure for a signature $\mathsf{Sg}$ in a
SMC $\mathcal{L}$. Equations in context will be denoted by $ \DualLNLLogicmv{x}  :  \DualLNLLogicnt{A}  \vdash_{\mathsf{L} }  \Gamma  ,  \DualLNLLogicnt{e_{{\mathrm{1}}}}  =  \DualLNLLogicnt{e_{{\mathrm{2}}}}  :  \DualLNLLogicnt{B}  ;  \Psi $ and $ \DualLNLLogicmv{x}  :  \DualLNLLogicnt{S}  \vdash_{\mathsf{C} }  \Psi ,  \DualLNLLogicnt{t_{{\mathrm{1}}}}  =  \DualLNLLogicnt{t_{{\mathrm{2}}}}  :  \DualLNLLogicnt{T} $, and
are both defined to be the reflexive, symmetric, and transitive
closure of the reduction relations defined by the rules in
Figure~\ref{fig:red-linear-terms} and Figure~\ref{fig:red-non-linear}
respectively.  Given such an equation:
$$
 \DualLNLLogicmv{x}  :  \DualLNLLogicnt{A}  \vdash_{\mathsf{L} }  \Gamma  ,  \DualLNLLogicnt{e_{{\mathrm{1}}}}  =  \DualLNLLogicnt{e_{{\mathrm{2}}}}  :  \DualLNLLogicnt{B}  ;  \Psi 
$$
we say that the structure \emph{satisfies} the equation if it assigns the same morphisms to  
$ \DualLNLLogicmv{x}  :  \DualLNLLogicnt{A}  \vdash_{\mathsf{L} }  \Gamma  \DualLNLLogicsym{,}  \DualLNLLogicnt{e_{{\mathrm{1}}}}  \DualLNLLogicsym{:}  \DualLNLLogicnt{B} ; \Psi $.  and to $ \DualLNLLogicmv{x}  :  \DualLNLLogicnt{A}  \vdash_{\mathsf{L} }  \Gamma  \DualLNLLogicsym{,}  \DualLNLLogicnt{e_{{\mathrm{2}}}}  \DualLNLLogicsym{:}  \DualLNLLogicnt{B} ; \Psi $.
Similarly, $\mathcal{M}$ satisfies $ \DualLNLLogicmv{x}  :  \DualLNLLogicnt{S}  \vdash_{\mathsf{C} }  \Psi ,  \DualLNLLogicnt{t_{{\mathrm{1}}}}  =  \DualLNLLogicnt{t_{{\mathrm{2}}}}  :  \DualLNLLogicnt{T} $ if it assigns the same morphism
to $ \DualLNLLogicmv{x}  :  \DualLNLLogicnt{S}  \vdash_{\mathsf{C} }  \Psi  \DualLNLLogicsym{,}  \DualLNLLogicnt{t_{{\mathrm{1}}}}  \DualLNLLogicsym{:}  \DualLNLLogicnt{T} $ and to $ \DualLNLLogicmv{x}  :  \DualLNLLogicnt{S}  \vdash_{\mathsf{C} }  \Psi  \DualLNLLogicsym{,}  \DualLNLLogicnt{t_{{\mathrm{2}}}}  \DualLNLLogicsym{:}  \DualLNLLogicnt{T} $.
Then given an algebraic theory 
$\mathsf{Th} = (\mathsf{Sg}, \mathsf{Ax})$, a structure $\mathcal{M}$ for $\mathsf{Sg}$ is a {\em model} for $\mathsf{Th}$ if it satisfies all the 
axioms in $\mathsf{Ax}$. 

We now go through some cases of the rules in TND to specify their
categorical interpretation so as to satisfy the equations in context
and to prove consistency of TND, and hence, DLNL logic in the model.
We do not give every case, but the ones we do not give are similar to
the ones given here. We analyze the linear connectives, giving an
argument for co-ILL that is analogue to Bierman's for ILL. We conclude
that as expected:
\begin{itemize}
\item the cotensor \emph{par} can be identified with the bifunctor
  $\bullet$ of the structure;
\item linear subtraction $\lsub$ is the left adjoint to the bifunctor
  $\bullet$;
\item the unit $\bot$ can be identified with $1$.
\end{itemize}

\subsubsection{Linear Disjunction}\label{lindisj} 

The introduction rule for Par is of the form 
\[
\inferrule* [right=$\DualLNLLogicdruleTLXXparIName{}$] {
   \DualLNLLogicmv{x}  :  \DualLNLLogicnt{A}  \vdash_{\mathsf{L} }  \Delta  \DualLNLLogicsym{,}  \DualLNLLogicnt{e_{{\mathrm{1}}}}  \DualLNLLogicsym{:}  \DualLNLLogicnt{B}  \DualLNLLogicsym{,}  \DualLNLLogicnt{e_{{\mathrm{2}}}}  \DualLNLLogicsym{:}  \DualLNLLogicnt{C} ; \Psi 
}{ \DualLNLLogicmv{x}  :  \DualLNLLogicnt{A}  \vdash_{\mathsf{L} }  \Delta  \DualLNLLogicsym{,}   \DualLNLLogicnt{e_{{\mathrm{1}}}}  \oplus  \DualLNLLogicnt{e_{{\mathrm{2}}}}   \DualLNLLogicsym{:}   \DualLNLLogicnt{B}  \oplus  \DualLNLLogicnt{C}  ; \Psi }
\]
This suggests an operation on Hom-sets of the form:%
\footnote{Notice that given a sequent $x: A \vdash_{\mathsf{L}} \Delta; \Psi$ 
where $\Delta = e_1: A_1, \ldots,e_n:A_n$ and $\Psi = t_1: T_1, \ldots, t_m: T_m$ 
we write $\mathcal{L}(A, \Delta\bullet J\Psi)$ for the Hom-set 
\[
\mathcal{L}([\![A]\!], [\![A_1]\!]\bullet\ldots\bullet[\![A_n]\!]\bullet J[\![T_1]\!]\bullet\ldots\bullet J[\![T_m]\!]).
\]}
$$
\Phi_{A, \Delta J\Psi}: \mathcal{L}(A, \Delta \bullet (B \bullet C) \bullet J\Psi) \rightarrow 
\mathcal{L}(A, \Delta\bullet B\oplus C \bullet J\Psi)
$$
{\em natural in} $\Delta$, $A$ and $J\Psi$. Given 
$e: A\rightarrow \Delta \bullet (B\bullet C)\bullet J\Psi,$ $a: A' \rightarrow A$
 $h: \Delta\rightarrow\Delta'$, and  $p: J\Psi\rightarrow J\Psi'$, naturality yields:
$$
\Phi_{A', \Delta',J\Psi'}(a; e; h\bullet (id_B\bullet id_{\mathsf{C}}) \bullet p) = a; \Phi_{A, \Delta, J\Psi}(e); h\bullet id_{B\oplus C}\bullet p
$$
In particular, suppose we have $d: A \rightarrow \Delta\bullet (B\bullet C) \bullet  \mathsf{J}\, \Psi $,
and let $e = id_{\Delta} \bullet (id_B \bullet id_{\mathsf{C}}) \bullet \id_{ \mathsf{J}\, \Psi }$,
$h = id_{\Delta}$, and $p = id_{J\Psi}$. Then we have
$\Phi_{A, \Delta,  \mathsf{J}\, \Psi }(d) = d; \Phi_{(\Delta \bullet (B \bullet C) \bullet  \mathsf{J}\, \Psi ),\Delta, \mathsf{J}\, \Psi } (id_{\Delta}\bullet (id_B \bullet id_{\mathsf{C}}) \bullet id_{J\Psi})$.
By functorality of $\bullet$ we have $id_{B}\bullet id_{C} = id_{B\bullet C}$. Hence, writing  
$\bigoplus$ for $\Phi_{(\Delta \bullet (B \bullet C) \bullet  \mathsf{J}\, \Psi ),\Delta, \mathsf{J}\, \Psi } (id_{\Delta}\bullet id_{B \bullet C} \bullet id_{J\Psi})$ we have 
$\Phi_{A, \Delta, \mathsf{J}\, \Psi }(d) = d; \bigoplus$. Finally, given the morphism $\psi_{\Delta,B,C,P} : ((\Delta \bullet B) \bullet C) \bullet  \mathsf{J}\, \Psi  \rightarrow \Delta \bullet (B \bullet C) \bullet  \mathsf{J}\, \Psi $, which is natural in all arguments and is definable using $\mathsf{Split}$ and $\mathsf{Join}$, we define:
$$
\interp{ \DualLNLLogicmv{x}  :  \DualLNLLogicnt{A}  \vdash_{\mathsf{L} }  \Delta  \DualLNLLogicsym{,}   \DualLNLLogicnt{e_{{\mathrm{1}}}}  \oplus  \DualLNLLogicnt{e_{{\mathrm{2}}}}   \DualLNLLogicsym{:}   \DualLNLLogicnt{B}  \oplus  \DualLNLLogicnt{C}   \DualLNLLogicsym{,}    \mathsf{J}\, \Psi   ;  \cdot  } =_{df} \interp{ \DualLNLLogicmv{x}  :  \DualLNLLogicnt{A}  \vdash_{\mathsf{L} }  \Delta  \DualLNLLogicsym{,}  \DualLNLLogicnt{e_{{\mathrm{1}}}}  \DualLNLLogicsym{:}  \DualLNLLogicnt{B}  \DualLNLLogicsym{,}  \DualLNLLogicnt{e_{{\mathrm{2}}}}  \DualLNLLogicsym{:}  \DualLNLLogicnt{C}  \DualLNLLogicsym{,}    \mathsf{J}\, \Psi   ;  \cdot  };\psi;\bigoplus.
$$

\ \\
\noindent
The Par elimination rule has the form 
\begin{center} 
  
  \begin{math}
    $$\mprset{flushleft}
    \inferrule* [right=$\DualLNLLogicdruleTLXXparEName{}$] {
       \DualLNLLogicmv{z}  :  \DualLNLLogicnt{A}  \vdash_{\mathsf{L} }  \Delta_{{\mathrm{1}}}  \DualLNLLogicsym{,}  \DualLNLLogicnt{e}  \DualLNLLogicsym{:}   \DualLNLLogicnt{B}  \oplus  \DualLNLLogicnt{C}  ; \Psi_{{\mathrm{1}}} 
      \\
       \DualLNLLogicmv{x}  :  \DualLNLLogicnt{B}  \vdash_{\mathsf{L} }  \Delta_{{\mathrm{2}}} ; \Psi_{{\mathrm{2}}} 
      \\
       \DualLNLLogicmv{y}  :  \DualLNLLogicnt{C}  \vdash_{\mathsf{L} }  \Delta_{{\mathrm{3}}} ; \Psi_{{\mathrm{3}}} 
    }{ \DualLNLLogicmv{z}  :  \DualLNLLogicnt{A}  \vdash_{\mathsf{L} }  \Delta_{{\mathrm{1}}}  \DualLNLLogicsym{,}  \DualLNLLogicsym{[}   \mathsf{casel}\, \DualLNLLogicsym{(}  \DualLNLLogicnt{e}  \DualLNLLogicsym{)}   \DualLNLLogicsym{/}  \DualLNLLogicmv{x}  \DualLNLLogicsym{]}  \Delta_{{\mathrm{2}}}  \DualLNLLogicsym{,}  \DualLNLLogicsym{[}   \mathsf{caser}\, \DualLNLLogicsym{(}  \DualLNLLogicnt{e}  \DualLNLLogicsym{)}   \DualLNLLogicsym{/}  \DualLNLLogicmv{y}  \DualLNLLogicsym{]}  \Delta_{{\mathrm{3}}} ; \Psi_{{\mathrm{1}}}  \DualLNLLogicsym{,}  \DualLNLLogicsym{[}   \mathsf{casel}\, \DualLNLLogicsym{(}  \DualLNLLogicnt{e}  \DualLNLLogicsym{)}   \DualLNLLogicsym{/}  \DualLNLLogicmv{x}  \DualLNLLogicsym{]}  \Psi_{{\mathrm{2}}}  \DualLNLLogicsym{,}  \DualLNLLogicsym{[}   \mathsf{caser}\, \DualLNLLogicsym{(}  \DualLNLLogicnt{e}  \DualLNLLogicsym{)}   \DualLNLLogicsym{/}  \DualLNLLogicmv{y}  \DualLNLLogicsym{]}  \Psi_{{\mathrm{3}}} }  
  \end{math}
\end{center}
This suggests an operation on Hom-sets of the form 
$$
\Psi_{A,\Delta,J\Psi}: \mathcal{L}(A, B\oplus C \bullet \Delta_1\bullet J\Psi_1)\times
\mathcal{L}(B, \Delta_2\bullet J\Psi_2)\times\mathcal{L}(C, \Delta_3\bullet J\Psi_3)\rightarrow 
\mathcal{L}(A, \Delta\,\bullet\, J\Psi)
$$
{\em natural in} $A,\Delta, J\Psi$ where we write
$\Delta = \Delta_1 \bullet \Delta_2 \bullet \Delta_3$ and $J\Psi = J\Psi_1 \bullet J\Psi_2 \bullet J\Psi_3$.
Given the following morphisms:
\begin{center}
  \begin{math}
    \begin{array}{lll}
      \begin{array}{lll}
        g:A\rightarrow B\oplus C \bullet \Delta_1\bullet J\Psi_1\\
        e: B\rightarrow\Delta_2\bullet J\Psi_2\\
        f: C\rightarrow \Delta_3\bullet J\Psi_3\\
        a: A'\rightarrow A\\
      \end{array}
      &
      \begin{array}{lll}
        d_1:\Delta_1\rightarrow \Delta_1'\\
        d_2: \Delta_2 \rightarrow \Delta'_2\\
        d_3: \Delta_3 \rightarrow \Delta_3'\\
        \\
      \end{array}
      &
      \begin{array}{lll}
        p_1: J\Psi_1\rightarrow J\Psi_1'\\
        p_2: J\Psi_2 \rightarrow J\Psi_2'\\
        p_3: J\Psi_3 \rightarrow J\Psi_3'\\
        \\
      \end{array}
    \end{array}
  \end{math}
\end{center}
naturality yields:
\begin{center}
\begin{tabular}{rl}
$\Psi_{A',\Delta',\Gamma',J\Psi'}\bigl((a;g; id_{B\oplus C} \bullet d_1 \bullet p_1 ), (e;d_2\bullet p_2),$ & 
$(f;d_3\bullet p_3)\bigr)$ =\\ 
$a; \Psi_{A,\Delta,J\Psi}(g,e,f);$ & 
$d_1 \bullet d_2\bullet d_3\bullet p_1\bullet p_2\bullet p_3; \mathtt{Join}(\Delta',J\Psi').$
\end{tabular}
\end{center}
In particular, set $e = id_B, f = id_{\mathsf{C}}, a = id_A, d_i = id_{\Delta_i}$, and $p_i = id_{J\Psi_i}$, and we get 
$$
\Psi_{A,\Delta J\Psi}(g, e, f) = \Psi_{A,\Delta, J\Psi}(g, id_B, id_{\mathsf{C}}); 
id_{\Delta}\bullet c\bullet d; \mathtt{Join}(\Delta, J\Psi) 
$$
where the operation $\mathtt{Join}$ implements the required associativity.
Writing $(x)^{\ast}$ for $\Psi_{D,\Delta}(x, id_B, id_{\mathsf{C}})$ we define
\begin{center}
\begin{tabular}{l}
$[\![z:A\vdash_{\mathsf{L}} \Delta_1, [\mathtt{casel}\ e/x]\Delta_2, [\mathtt{caser}\ e/y] \Delta_3; \Psi_1, [\mathtt{casel}\ e/x] \Psi_2, [\mathtt{caser}\ e/y] \Psi_3]\!]=_{df}$\\
$[\![  \DualLNLLogicmv{z}  :  \DualLNLLogicnt{A}  \vdash_{\mathsf{L} }  \Delta_{{\mathrm{1}}}  \DualLNLLogicsym{,}  \DualLNLLogicnt{e}  \DualLNLLogicsym{:}   \DualLNLLogicnt{B}  \oplus  \DualLNLLogicnt{C}  ; \Psi_{{\mathrm{1}}}  ]\!]^*; 
(id_{\Delta_1}\bullet [\![  \DualLNLLogicmv{x}  :  \DualLNLLogicnt{B}  \vdash_{\mathsf{L} }  \Delta_{{\mathrm{2}}} ; \Psi_{{\mathrm{2}}}  ]\!]\bullet[\![  \DualLNLLogicmv{y}  :  \DualLNLLogicnt{C}  \vdash_{\mathsf{L} }  \Delta_{{\mathrm{3}}} ; \Psi_{{\mathrm{3}}}   ]\!]);
\mathtt{Join}(\Delta,J\Psi)$.\\
\end{tabular}
\end{center}

We now turn to the equations in context. Consider the following case:
\begin{center} \footnotesize
  \begin{math} 
    $$\mprset{flushleft}
    \inferrule* [right=$\mathbf \oplus\text{-}\beta$] {
      { \setlength{\tabcolsep}{15px}
        \begin{tabular}{lll}          
          & & $e\equiv\mathtt{casel}(e_1 \oplus e_2)$\\          
          & & $e' \equiv \mathtt{caser}(e_1 \oplus e_2)$\\
          $|\Delta_1| = |\Delta'_1|$ & $| \Psi_1| = |\Psi'_1|$ & $y:A_2 \vdash_{\mathsf{L}} \Delta_3; J\Psi_3 =\Delta'_3; J\Psi'_3$\\
          $|\Delta_2|=|\Delta'_2|$   & $|\Psi_2|=|\Psi'_2|$    & $x : A_1 \vdash_{\mathsf{L}} \Delta_2; J\Psi_2 = \Delta'_2; J\Psi'_2$\\
          $|\Delta_3|=|\Delta_3'|$   & $|\Psi_3|=|\Psi'_3|$    & $z:B \vdash_{\mathsf{L}} e_1 :A_1,e_2 :A_2, \Delta_1; J\Psi_1 = e'_1 :A_1,e'_2 :A_2,\Delta'_1;J\Psi'_1$
      \end{tabular}}
    }{z:B\vdash_{\mathsf{L}} \Delta_1,[e/x]\Delta_2,[e'/x]\Delta_3;\Psi_1,[e/x]\Psi_2,[e'/x]\Psi_3 = \Delta'_1, [e'_1/x]\Delta_2,[e'_2/x]\Delta′_3;\Psi′_1,[e'_1/x]\Psi'_2,[e'_2/x]\Psi'_3}
  \end{math}
\end{center}
Let 
$$
q: B\rightarrow A_1\bullet A_2 \bullet \Delta_1\bullet J\Psi,\qquad m:A_1\rightarrow\Delta_2\bullet J\Psi_2 \quad
\hbox{ and }\quad  
n: A_2\rightarrow\Delta_3\bullet J\Psi_3.
$$
Then to satisfy the above equations in context we need that the following diagram commutes:
\begin{center}
$\xymatrix@R=.4in{B\ar[r]^(.25){q} & \Delta_1\bullet J\Psi_1\bullet (A_1\bullet A_2) \ar[d]_{\oplus} \ar[rr]^(.4){id_{\Delta_1}\bullet m\bullet n} &
&\Delta_1\bullet J\Psi_1\bullet\Delta_2\bullet J\Psi_2\bullet \Delta_3\bullet J\Psi_3\\
&\Delta_1\bullet  J\Psi_1\bullet A\oplus B \ar[rr]_{\ast} & &\Delta_1\bullet  J\Psi_1\bullet A\bullet B \ar[u]_{id_{\Delta_1}\bullet m\bullet n}\\}$
\end{center}   
We make the assumption that the above decomposition is unique. 
Moreover, supposing $\Delta_1$ to be empty and $m = id_A$, $n = id_B$, $q = id_A\bullet id_B = id_{A\bullet B}$ 
we obtain $(id_A\bullet id_B; \mathbf{\bigoplus})^{\ast} = id_A\bullet id_B$ and similarly 
$(id_{A \oplus B})^*; \mathbf{\bigoplus} = id_{A \oplus B}$; hence we may conclude that there is a natural isomorphism 
\begin{center}
\AxiomC{$D \rightarrow \Gamma\bullet A\bullet B$}
\doubleLine
\UnaryInfC{$D \rightarrow \Gamma\bullet A\oplus B$}
\DisplayProof
\end{center}
so we can identify $\bullet$ and $\oplus$. Finally we see that the following $\eta$ equation in context 
is also satisfied:  
\begin{equation}\label{eta:par}
\quad\framebox{
\AxiomC{$\oplus - \eta$ rule}
\noLine
\UnaryInfC{$|\Delta| = |\Delta'| \quad|\Psi| = |\Psi'| \quad z:B \vdash_{\mathsf{L}} \Delta; \Psi = \Delta'; \Psi'$}
\UnaryInfC{$z:B \vdash_{\mathsf{L}} (\mathtt{casel\, e}\oplus\mathtt{caser}\, e): A_1\oplus A_2, \Delta; \Psi = e:A_1\oplus A_2, \Delta'; \Psi'$}
\noLine
\UnaryInfC{\strut}
\DisplayProof
}
\end{equation}

\subsubsection{Linear subtraction}\label{linsubtr}
\noindent
\ref{linsubtr}.1. {\em Subtraction introduction.} The introduction rule for subtraction has the form:
\begin{center}
  \begin{math}
    $$\mprset{flushleft}
    \inferrule* [right=$\DualLNLLogicdruleTLXXsubIName{}$] {
       \DualLNLLogicmv{x}  :  \DualLNLLogicnt{A}  \vdash_{\mathsf{L} }  \Delta_{{\mathrm{1}}}  \DualLNLLogicsym{,}  \DualLNLLogicnt{e}  \DualLNLLogicsym{:}  \DualLNLLogicnt{B} ; \Psi_{{\mathrm{1}}} 
      \\
       \DualLNLLogicmv{y}  :  \DualLNLLogicnt{C}  \vdash_{\mathsf{L} }  \Delta_{{\mathrm{2}}} ; \Psi_{{\mathrm{2}}} 
      \\
      \DualLNLLogicsym{\mbox{$\mid$}}  \Psi_{{\mathrm{1}}}  \DualLNLLogicsym{\mbox{$\mid$}}  \DualLNLLogicsym{=}  \DualLNLLogicsym{\mbox{$\mid$}}  \Psi_{{\mathrm{2}}}  \DualLNLLogicsym{\mbox{$\mid$}}
    }{ \DualLNLLogicmv{x}  :  \DualLNLLogicnt{A}  \vdash_{\mathsf{L} }  \Delta_{{\mathrm{1}}}  \DualLNLLogicsym{,}   \mathsf{mkc}( \DualLNLLogicnt{e} , \DualLNLLogicmv{y} )   \DualLNLLogicsym{:}   \DualLNLLogicnt{B}  \colimp  \DualLNLLogicnt{C}   \DualLNLLogicsym{,}  \DualLNLLogicsym{[}  \DualLNLLogicmv{y}  \DualLNLLogicsym{(}  \DualLNLLogicnt{e}  \DualLNLLogicsym{)}  \DualLNLLogicsym{/}  \DualLNLLogicmv{y}  \DualLNLLogicsym{]}  \Delta_{{\mathrm{2}}} ; \Psi_{{\mathrm{1}}}  \DualLNLLogicsym{,}  \DualLNLLogicsym{[}  \DualLNLLogicmv{y}  \DualLNLLogicsym{(}  \DualLNLLogicnt{e}  \DualLNLLogicsym{)}  \DualLNLLogicsym{/}  \DualLNLLogicmv{y}  \DualLNLLogicsym{]}  \Psi_{{\mathrm{2}}} }
  \end{math}
\end{center}
This suggests a natural transformation with components:
$$
\Phi_{A, \Delta,J\Psi}: \mathcal{L}(A, \Delta_1\bullet B\bullet J\Psi_1)\times\mathcal{L}(C, \Delta_2\bullet J\Psi_2) \rightarrow 
\mathcal{L}(A, \Delta_1\bullet (B\lsub C)\bullet \Delta_2\bullet J\Psi_1\bullet J\Psi_2)
$$
natural in $A, \Delta_1, \Delta_2, J \Psi_1, J\Psi_2$.
Taking morphisms 
$$e: A\rightarrow \Delta_1\bullet B \bullet J\Psi_1, \quad  f:C\rightarrow \Delta_2 \bullet J\Psi_2$$ 
and also $a: A'\rightarrow A$, $d_1:\Delta_1\rightarrow \Delta_1'$, 
$d_2: \Delta_2\rightarrow \Delta_2'$, $p_1: J\Psi_1\rightarrow J\Psi_1'$, $p_2: J\Psi_2\rightarrow J\Psi_2'$, 
by naturality we have
\begin{center}
\begin{tabular}{c}
$\Phi_{A', \Delta_1',\Delta'_2, J\Psi'_1J\Psi'_2}\left((a; e; d_1\bullet id_B\bullet p_1), (f;d_2;p_2)\right) =\qquad\qquad $\\
$\qquad\qquad = a; \Phi_{A,\Delta, J\Psi}(e, f); d_1\bullet d_2\bullet id_{B\lsub C}; 
\mathtt{Join}(\Delta_1',\Delta_2', B\lsub C, J\Psi_1', J\Psi_2')$
\end{tabular}
\end{center}
In particular, taking $a=id_A$, $d_1=id_{\Delta_1}$, $p_1 = id_{J\Psi_1}, p_2 = id_{J\Psi_2}$ but 
$d_2: C\rightarrow\Delta_2\cdot J\Psi_2$ and $f= id_{\mathsf{C}}$ we have:
\begin{center}
\begin{tabular}{c}
$\Phi_{A, \Delta_1, \Delta_2, J\Psi_1, J\Psi_2}(e, d_2) = \Phi_{A,\Delta_1} (e, id_{\mathsf{C}}); 
id_{\Delta_1}\bullet d_2 \bullet id_{A\lsub B}\bullet id_{J\Psi_1}\bullet id{J\Psi_2}; $\\
$\hskip2in\mathtt{Join}(\Delta_1, \Delta_2, A\lsub B, J\Psi_1, J\Psi_2)$\\
\end{tabular}
\end{center}
Writing $\mathbf{MKC}^C_{A,\Delta_1,J\Psi_1}(e)$ for $ \Phi_{A,\Delta_1 J\Psi_1} (e, id_{\mathsf{C}})$, 
$\Phi_{A, \Delta, J\Psi}(e, d_2)$ can be expressed as the composition 
$$
\mathbf{MKC}^C_{A,\Delta_1, J\Psi_1}(e); id_{\Delta_1}\bullet d_2\bullet id_{B\lsub C}
$$ 
where $\mathbf{MKC}^C_{A,\Delta_1, J\Psi_1}$ is a natural transformation with components 
$$
\mathcal{L}(A, \Delta_1\bullet B\bullet J\Psi)\times\mathcal{L}(C, C) \rightarrow 
\mathcal{L}(A, \Delta_1\bullet C \bullet C\lsub C)
$$
so we make the definition 
\begin{center}
\begin{tabular}{l}
$[\![x:A \vdash_{\mathsf{L}} \Delta_1, \mathtt{mkc}(e,y):\, B\lsub C, [y(e)/y+\Delta_2; \Psi_1\cdot [y(e)/y] \Psi_2]\!] =_{df}$\\
\quad $\mathbf{MKC}^C_{A,\Delta_1,J\Psi_1}[\![x:A\vdash_{\mathsf{L}}\Delta_1, e_1:B]\!];
id_{\Delta}\bullet [\![y:C\vdash_{\mathsf{L}}\Delta_2; \Psi_2]\!]\bullet id_{B\lsub C}; $\\
$\hskip2in \mathtt{Join}(\Delta_1,\Delta_2, B\lsub C, J\Psi_1, J\Psi_2)$\\
\end{tabular}
\end{center}
Notice that $\mathbf{MKC}^C_{A,\Delta_1, J\Psi_1}$ corresponds to the one-premise form of the 
subtraction introduction rule
\begin{center} 
\AxiomC{$x:A\vdash_{\mathsf{L}} \Delta_1, e: B; \Psi_1$}
\RightLabel{$\DualLNLLogicdruleTLXXsubIName{}$}
\UnaryInfC{$x:A\vdash_{\mathsf{L}} \Delta_1, \mathtt{mkc}(e,y) : B \lsub C, y(e): C; \Psi_1$}
\DisplayProof
\vspace{3ex}
\end{center}
which is equivalent in terms of provability to the more general form
considered here \cite{Crolard:2004}.

The subtraction elimination rule has the form:
\begin{center}
  \begin{math}
    $$\mprset{flushleft}
    \inferrule* [right=$\DualLNLLogicdruleTLXXsubEName{}$] {
       \DualLNLLogicmv{x}  :  \DualLNLLogicnt{A}  \vdash_{\mathsf{L} }  \Delta_{{\mathrm{1}}}  \DualLNLLogicsym{,}  \DualLNLLogicnt{e_{{\mathrm{1}}}}  \DualLNLLogicsym{:}   \DualLNLLogicnt{B}  \colimp  \DualLNLLogicnt{C}  ; \Psi_{{\mathrm{1}}} 
      \\
       \DualLNLLogicmv{y}  :  \DualLNLLogicnt{B}  \vdash_{\mathsf{L} }  \DualLNLLogicnt{e_{{\mathrm{2}}}}  \DualLNLLogicsym{:}  \DualLNLLogicnt{C}  \DualLNLLogicsym{,}  \Delta_{{\mathrm{2}}} ; \Psi_{{\mathrm{2}}} 
      \\
      \DualLNLLogicsym{\mbox{$\mid$}}  \Psi_{{\mathrm{1}}}  \DualLNLLogicsym{\mbox{$\mid$}}  \DualLNLLogicsym{=}  \DualLNLLogicsym{\mbox{$\mid$}}  \Psi_{{\mathrm{2}}}  \DualLNLLogicsym{\mbox{$\mid$}}            
    }{ \DualLNLLogicmv{x}  :  \DualLNLLogicnt{A}  \vdash_{\mathsf{L} }   \mathsf{postp}\,( \DualLNLLogicmv{y}  \mapsto  \DualLNLLogicnt{e_{{\mathrm{2}}}} ,  \DualLNLLogicnt{e_{{\mathrm{1}}}} )   \DualLNLLogicsym{,}  \Delta_{{\mathrm{1}}}  \DualLNLLogicsym{,}  \DualLNLLogicsym{[}  \DualLNLLogicmv{y}  \DualLNLLogicsym{(}  \DualLNLLogicnt{e_{{\mathrm{1}}}}  \DualLNLLogicsym{)}  \DualLNLLogicsym{/}  \DualLNLLogicmv{y}  \DualLNLLogicsym{]}  \Delta_{{\mathrm{2}}} ; \Psi_{{\mathrm{1}}}  \DualLNLLogicsym{,}  \DualLNLLogicsym{[}  \DualLNLLogicmv{y}  \DualLNLLogicsym{(}  \DualLNLLogicnt{e_{{\mathrm{1}}}}  \DualLNLLogicsym{)}  \DualLNLLogicsym{/}  \DualLNLLogicmv{y}  \DualLNLLogicsym{]}  \Psi_{{\mathrm{2}}} }
  \end{math}
\end{center}
This suggests a natural transformation with components
$$
\Psi_{A,\Delta_1, \Delta_2, J\Psi_1, J\Psi_2}: \mathcal{L}(A, \Delta_1\bullet (B\lsub C)\bullet J\Psi_1)\times
\mathcal{L}(B, C\bullet\Delta_2\bullet J\Psi_2) \rightarrow 
\mathcal{L}(A,  \Delta_1\bullet \Delta_2\bullet J\Psi_1\bullet J\Psi_2)
$$
natural in $A, \Delta_1, \Delta_2, J\Psi_1, J\Psi_2$. Here $\mathtt{postp}(y\mapsto e_2, e_1)$ is given 
type $1$ and an application of left identity $\lambda_{1,\Delta_2}$ is assumed implicitly. 

Given 
$$e: A\rightarrow \Delta_1\bullet (B\lsub C) \bullet J\Psi_1,\qquad  
f: B\rightarrow C\bullet \Delta_2\bullet J\Psi_2
$$ and also $a: A'\rightarrow A$, $d_1: \Delta_1\rightarrow\Delta_1'$, 
$d_2:\Delta_2\rightarrow\Delta_2'$, $p_1: J\Psi_1 \rightarrow J\Psi'_1$ $p_2: J\Psi_2 \rightarrow J\Psi'_2$
naturality yields 
\begin{center}
\begin{tabular}{l}
$\Psi_{A', \Delta_1',\Delta_2', J\Psi_1', J\Psi_2'}\left((a; e; d_1\bullet id_{B\lsub C}\bullet p_1), (f; id_{\mathsf{C}}\bullet d_2\bullet p_2)\right) = $\\
$a; \Psi_{A, \Delta_1,\Delta_2,J\Psi_1}(e, f);\lambda_{1,\Delta_1} \bullet d_1\bullet d_2\bullet p_1\bullet p_2; \mathtt{Join}(\Delta_1',\Delta_2', J\Psi_1,J\Psi_2)$
\end{tabular}
\end{center}
In particular, taking $a: A\rightarrow \Delta_1\bullet(B\lsub C)$, 
$e = id_{\Delta_1\bullet(B\lsub C)}$, $d_1= id_{\Delta_1}$, $d_2: id_{\Delta_2}$, $p_1 = id_{J\Psi_1}$, $p_2 = id{J\Psi_2}$ we obtain
$$
\Psi_{A, \Delta_1,\Delta_2}(a, f) = a; \Psi_{A, \Delta_1,\Delta_2}(id_{\Delta_1\bullet(C\lsub D)\bullet id_{J\Psi_1}}, f); \mathtt{Join}(\Delta_1,\Delta_2, J\Psi_1, J\Psi_2)
$$
Writing $\mathbf{POSTP}(f)$ for $\Psi_{A, \Delta_1,\Delta_2,J\Psi_1, J\Psi_2}(id_{\Delta_1\bullet(B\lsub C)\bullet J\Psi_1}, f)$
we define 
\begin{center}
\begin{tabular}{l}
  $[\![x:A \vdash_{\mathsf{L}} \Delta_1, \mathtt{postp}(y\mapsto e_2, e_1), [y(e_1)/y]\Delta_2; \Psi_1, [y(e_1)/y], \Psi_2]\!] =_{df}$\\
\quad $[\![x:A\vdash_{\mathsf{L}} \Delta_1, e_1: B\lsub C]\!]; id_{\Delta_1}\bullet
\mathbf{POSTP}[\![y:B \vdash_{\mathsf{L}} e_2: C,\Delta_2; \Psi_2]\!]; \mathtt{Join}(\Delta_1,\Delta_2,J\Psi_1,J\Psi_2)$ 
\end{tabular}
\end{center}

\ \\
\noindent
\ref{linsubtr}.3. {\em Equations in context.}  We have equations in context of the form 
\begin{center}
  \framebox{ \footnotesize
      \begin{math}
    $$\mprset{flushleft}
    \inferrule* [right={\scriptsize $\mathbf \lsub - \beta$}] {
      { \setlength{\tabcolsep}{15px}
        \begin{tabular}{lll}
                                     & $e_z\equiv z(\mathtt{mkc}(e_1,y))$\\
          $|\Delta_1| = |\Delta'_1|$ & $e_p\equiv \mathtt{postp}(z\mapsto e_2, \mathtt{mkc}(e_1,y))$\\
          $|\Delta_2|=|\Delta'_2|$   & $x : B \vdash_{\mathsf{L}} e_1:A_1, \Delta_1; \Psi_1 = e'_1: A_1,\Delta'_1; \Psi'_1$\\
          $| \Psi_1| = |\Psi'_1|$    & $y:A_2 \vdash_{\mathsf{L}} \Delta_2; \Psi_2 =\Delta'_2; \Psi'_2$\\
          $|\Psi_2|=|\Psi'_2|$       & $z:A_1 \vdash_{\mathsf{L}} e_2 :A_2, \Delta_3; \Psi_3 =e'_2 :A_2,\Delta'_3;\Psi'_3$\\
        \end{tabular}
      }
    }{x: B\vdash_{\mathsf{L}} \Delta_1, e_p, [y(e_1)/y]\Delta_2, [e_z/z]\Delta_3; \Psi_1, [y(e_1)/y]\Psi_2, [e_z/z]\Psi_3 = \\
      \,\,\,\,\,\,\,\,\,\,\,\,\,\,\,\,\,\,\,\,\,\,\,\,\,\,\,
      \Delta'_1, [ [e'_1/z]e'_2/y]\Delta_2, [e'_1/z]\Delta'_3; \Psi'_1, [ [e'_1/z]e'_2/y]\Psi'_2, [e'_1/z]\Psi'_3}
      \end{math}
      }
\end{center}

We repeat the derivations of the redex and of the reductum.

\centerline{\bf Redex:}
{\small
\begin{center}
\AxiomC{$x:B \vdash_{\mathsf{L}} e_1:A_1, \Delta_1; \Psi_1\qquad y:A_2\vdash_{\mathsf{L}}\Delta_2; \Psi_2$}
\UnaryInfC{$x:B \vdash_{\mathsf{L}} \mathtt{mkc}(e_1, y): A_1\lsub A_2, \Delta_1, [y(e_1)/y]\Delta_2; \Psi_1,  [y(e_1)/y]\Psi_2$}
\AxiomC{$z:A_1\vdash_{\mathsf{L}} e_2:A_2, \Delta_3; \Psi_3$}
\BinaryInfC{$x:B\vdash_{\mathsf{L}} \Delta_1, 
\overbrace{\mathtt{postp}(z\mapsto e_2, \mathtt{mkc}(e_1,y))}^{e_p}, 
 [y(e_1)/y]\Delta_2,  [\overbrace{z(\mathtt{mkc}(e_1,y))}^{e_z}/z]\Delta_3 ;$}
\noLine
\UnaryInfC{$\hskip2in \Psi_1,  [y(e_1)/y]\Psi_2, [z(\mathtt{mkc}(e_1,y)/z)]:\Psi_3 $}
\DisplayProof
\end{center}}
\centerline{\bf Reductum:}
{\small
\begin{center}
\AxiomC{$x:B \vdash_{\mathsf{L}} e'_1:A_1, \Delta'_1; \Psi'_1\qquad  z:A_1\vdash_{\mathsf{L}} e'_2:A_2, \Delta'_3; \Psi'_3$}
\UnaryInfC{$x:B \vdash_{\mathsf{L}} \Delta'_1, [e'_1/z]\Delta'_3,       [e'_1/z]e'_2 : A_2; \Psi'_1,  [e'_1/z]\Psi'_3$}
\AxiomC{$y:A_2\vdash_{\mathsf{L}} \Delta_2; \Psi_2$}
\BinaryInfC{$x:B \vdash_{\mathsf{L}}\Delta'_1, [ [e'_1/z]e'_2/y]\Delta'_2,\ [e'_1/z]\Delta'_3; \Psi'_1,[ [e'_1/z]e'_2]\Psi'_2,\ [e'_1/z]\Psi'_3$}
\DisplayProof
\end{center}}

Given morphisms $n: B\rightarrow \Delta_1\bullet A_1$ and $m: A_1\rightarrow \Delta_3\bullet A_2$, for these equations to be satisfied we need the following diagram to commute (omitting non-linear terms):
\begin{center}
$\xymatrix@R=.4in{B\ar[d]_{\mathbf{MKC}^{A_2}(n)} \ar[rr]^n &\qquad & 
\Delta_1\bullet A_1\ar[d]^{id_{\Delta_1}\bullet m}\\
\Delta_1\bullet(A_1\lsub A_2)\bullet A_2\ar[rr]_{\mathbf{POSTP}(m)\bullet id_{A_2}} &&\Delta_1\bullet\Delta_3\bullet A_2\\
}$
\end{center}
in particular, taking $n = id_{A_1}$ we have 
\begin{center}
$\xymatrix@R=.4in{A_1 \ar[d]_{\mathbf{MKC}^A_2(id_{A_1})}\ar[r]^m&\Delta_3\bullet A_2\\
(A_1\lsub A_2)\bullet A_2 \ar[ur]|{\mathbf{POSTP}(m)\bullet id_{A_2}} &\\
}$
\end{center}
Assuming the above decomposition to be unique, we can show that the $\eta$ equation in context
is also satisfied:     
\begin{equation}\label{eta:subtr}
\framebox{\small 
\AxiomC{$|\Delta| = |\Delta'| \quad |\Psi| = |\Psi'|\quad z;B \vdash_{\mathsf{L}} \Delta, \Psi = \Delta', \Psi'$}
\UnaryInfC{$z:B \vdash_{\mathsf{L}} \mathtt{postp}(x\mapsto y,e), \mathtt{mkc}(x(e),y): A_1\lsub A_2, \Delta, \Psi = 
e: A_1\lsub A_2, \Delta'; \Psi'$}
\DisplayProof
}
\end{equation}
and conclude that there is a natural isomorphism between the maps 
\begin{center}
\AxiomC{$A \rightarrow \Delta\bullet B$}
\doubleLine
\UnaryInfC{$A\lsub B\rightarrow \Delta$}
\DisplayProof
\end{center}
i.e., that $\lsub$ is the left adjoint to the bifunctor $\bullet$.

\subsubsection{Functors}\label{functors}

Recall that a model of Linear-Non Linear co-intuitionistic logic consists of  a symmetric comonoidal adjunction 
$\mathcal{L} : H \dashv  J : \mathcal{C}$ where $\mathcal{L} = (\mathcal{L}, \bot, \oplus, \limp)$ 
is a symmetric monoidal coclosed category and $\mathcal{C} = (\mathcal{C}, 0, +, -)$
 is a cocartesian coclosed category.

We use the same symbols for the functors
$J: \mathcal{C} \rightarrow \mathcal{L}$ and $\mathcal{H}: \mathcal{L} \rightarrow  \mathcal{C}$ 
in the models and for the operators that represent them in the language.

\vspace{1ex}

\ref{functors}.1 {\em rules for $J: \mathcal{C} \rightarrow \mathcal{L}$.} 
\begin{equation}
\framebox{
\AxiomC{\it $TL\_\func{J}_I$}
\noLine
\UnaryInfC{$x:A \vdash_{\mathsf{L}} ,\Delta; t:T, \Psi$}
\UnaryInfC{$x:A\vdash_{\mathsf{L}} \Delta, Jt: JT; \Psi$}
\DisplayProof
}
\end{equation}
If $\Delta = \overline{R}: |\Delta|$ and $\Psi = \overline{S}: |\Psi|$, then the categorical interpetation of the rule is 
an application of $\alpha^{-1}$: 
\begin{center}
\AxiomC{$\xymatrix@R=.4in{A\ar[rr]^{\overline{R} \bullet Jt \bullet J\overline{S}} & &\Delta\bullet JT \bullet J\Psi}$}
\UnaryInfC{$\xymatrix@R=.4in{A\ar[rr]^{(\overline{R} \bullet Jt )\bullet J\overline{S}} & &(\Delta\bullet JT) \bullet J\Psi}$}
\DisplayProof
\end{center}

\begin{equation}
\framebox{
\AxiomC{\it $TL\_\func{J}_E$ elimination}
\noLine
\UnaryInfC{$x:A \vdash_{\mathsf{L}} \Delta, e: JT;\Psi_1\qquad  y:T \vdash_{\mathsf{C}} \Psi_2\ \hbox{where}\ |\Psi_1|=|\Psi_2|$}
\UnaryInfC{$x:A\vdash_{\mathsf{L}} \Delta; \Psi_1 \cdot \mathtt{let}\ Jy = e\  \mathtt{in}\ \Psi_2$}
\DisplayProof
}
\end{equation}
If $\Delta = \overline{R}: |\Delta|$,  $\Psi_1 = \overline{R'}: |\Psi_1|$, $\Psi_2 = \overline{S}:|\Psi_2|$, 
then the categorical interpretation of the rule is given by an operation of the form 
\[
\mathcal{L}(A, \Delta\bullet JT\bullet J\Psi_1) \times \mathcal{C}(T, \Psi_2)\rightarrow 
\mathcal{L}(A, \Delta\bullet J\Psi_1\bullet J\Psi_2) 
\]
given by the following compositions
\begin{center}
\AxiomC{$\xymatrix@R=.4in{A \ar[rr]^{\overline{R}\bullet e \bullet J\overline{R}'} && \Delta\bullet JT\bullet J\Psi_1}$}
\AxiomC{$\xymatrix@R=.4in{T \ar[r]^{\overline{S}} & \Psi_2} \quad \hbox{in} \  \mathcal{C}$}
\UnaryInfC{$\xymatrix@R=.4in{JT \ar[r]^{\overline{JS}}& J\Psi_2} \quad \hbox{in} \  \mathcal{L}$}
\BinaryInfC{$\xymatrix@R=.1in@C=.6cm{A \ar[rr]^{\overline{R}\bullet e \bullet J(\overline{R}')\quad} &
& \Delta\bullet J(T) \bullet J(\Psi_1) 
\ar[rrr]^{id_{\Delta}\bullet J(\overline{S}) \bullet id_{J(\Psi_1)}}&&&\Delta\bullet J(\Psi_1)\bullet J(\Psi_2)
\ar[rr]^{id_{\Delta} \bullet j^{-1}_{\Psi_1\Psi_2}} &&\Delta\bullet J(\Psi_1+\Psi_2)}$}
\noLine
\UnaryInfC{\hskip3in since $|\Psi_1| = |\Psi_2|, \quad \xymatrix@R=.1in@C=.6cm{ 
\ar[rr]^{id_{\Delta}\bullet\nabla_{\Psi_1}} &&\Delta \bullet J(\Psi_1)}$}
\DisplayProof
\end{center}

\vspace{3ex}

\ref{functors}.2 {\em rules for $H: \mathcal{L} \rightarrow \mathcal{C}$.} 
\begin{equation}
\framebox{
\AxiomC{\it $TC\_\func{H}_I$}
\noLine
\UnaryInfC{$x:A\vdash_{\mathsf{L}} \overline{R}: \Delta, e: B; \Psi$}
\UnaryInfC{$x:A\vdash_{\mathsf{L}} \overline{R}: \Delta;He:HB, \Psi$}
\DisplayProof
}
\end{equation}
Let $\Delta = \overline{R}: |\Delta|$ and  $\Psi = \overline{S}:|\Psi|$: then 
\begin{center}
\AxiomC{$\xymatrix@R=.1in@C.5cm{A \ar[rr]^{\overline{R}\oplus e\oplus J(\overline{S})\quad}&&
\Delta \bullet B \bullet  J(\Psi)}$}
\RightLabel{using $\eta_B: B \rightarrow JHB$}
\UnaryInfC{$\xymatrix@R=.1in@C.5cm{A \ar[rrr]^{\overline{R}\oplus JH(e) \bullet J(\overline{S})\quad}&&
& \Delta\bullet JH(B)\bullet J(\Psi)}$}
\DisplayProof
\end{center}

\vspace{3ex}
 
\begin{equation}
\framebox{
\AxiomC{\it $TC\_\func{H}_{E}$}
\noLine
\UnaryInfC{$x:B \vdash_{\mathsf{L}} \Delta ;  t:HA, \Psi_1\qquad y:A\vdash_{\mathsf{L}} ·; \Psi_2\qquad 
\mathrm{where}\ |\Psi_1| = |\Psi_2|$}
\UnaryInfC{$x:B \vdash_{\mathsf{L}} \Delta; \Psi_1, \mathtt{let}\ Hy=t\ \mathtt{in}\ \Psi_2$}
\DisplayProof
}
\end{equation}

The categorical interpretation of $H$ elim$_1$ is as follows:
Let $\Psi_1 = \overline{R}: |\Psi_1|$ and $\Psi_2 = \overline{S}: |\Psi_2|$. Then we have the following compositions: 
\begin{center}
\AxiomC{$\strut\quad$}
\noLine
\UnaryInfC{$\xymatrix@R=.1in@C.5cm{S \ar[rr]^{t + \overline{R}\quad}&& H(A) + \Psi}$}
\AxiomC{$\xymatrix@R=.1in@C.5cm{A \ar[rr]^{J(\overline{S})} && J(\Psi)}\quad \hbox{in}\ \mathcal{L}$}
\UnaryInfC{$\xymatrix@R=.1in@C.5cm{HA \ar[rr]^{HJ(\overline{S})} && HJ(\Psi) 
\ar[r]^{\quad\epsilon_{\Psi}} &\Psi}\quad \hbox{in}\ \mathcal{C}$}
\BinaryInfC{$\xymatrix@R=.1in@C.5cm{S \ar[rr]^{t + \overline{R}\quad}&& H(A) + \Psi
\ar[rrr]^{HJ(\overline{S}) + id_{\Psi}} &&& \Psi + \Psi \ar[r]^{\quad\nabla_{\Psi}}& \Psi }$}
\DisplayProof\\
\end{center}



\section{Related and Future Work}
\label{sec:related_work}
The most comprehensive treatment of ILL is in Gavin Bierman's thesis
\cite{Bierman:1994}.  There one finds the Proof Theory (Chapter 2),
i.e, the sequent calculus with cut-eliminaton, natural deduction and
axiomatic versions of ILL. Then (Chapter 3) a term assignment to the
natural deduction and to the sequent calculus versions are presented
with $\beta$-reductions and commutative conversions, and strong
normalization and confluence are proved for the resulting calculus. A
painstaking analysis of the rules of the labeled calculus leads to the
construction of a categorical model of ILL, a \emph{linear category},
in particular of the exponential part, a main contribution of Bierman
and of the Cambridge school of the 1990s with respect to previous
models by Seely and Lafont.  Bellin \cite{Bellin:2012} presents a
categorical model of co-intuitionistic linear logic based on a
dualization of Bierman \cite{Bierman:1994} construction for ILL.

Benton's work \cite{Benton:1994} on LNL logic presents the categorical
model for Linear-Non-Linear Intuitionistic logic LNL.  Chapter 2 shows
how to obtain a LNL model from a Linear Category and viceversa.
Versions of the sequent calculus for LNL are considered and
cut-elimination is proved for one such version. Then Natural Deduction
is given with term assignment and the categorical interpretation of a
fragment of the natural deduction system. Then $\beta$-reductions and
commuting conversions are presented.  The present work follows
Benton's paper aiming at a (non-trivial) dualization of it.

Bi-intuitionistic logic was introduced by C.Rauszer
\cite{Rauszer:1974} with an algebraic and Kripke semantics
\cite{Rauszer:1980} and a Gentzen style sequent calculus
\cite{Rauszer:1974a}.  Co-intuitionistic logic requires a multiple
conclusion system, because of the cotensor in the linear case and of
contraction right in the non-linear one.  This raises the problem of
the relations between intuitionistic implication and disjunction, and,
dually, between subtraction and conjunction.  In the case of the logic
FILL that extends ILL with the cotensor (\emph{par}) applying Maheara
and Dragalin's restriction that only one formula occurs in the
succedent of the premise of an implication right, yields a calculus
that does not satisfies cut-elimination, as noticed by Schellinx
\cite{Schellinx:1991}. Similarly, in the logic BILL
(\emph{Bi-Intuitionistic Linear Logic}) requiring that only one
formula occurs in the antecedent of the premise of a subtraction left
yields a system that does not satisfy cut-elimination.

\begin{center}
\begin{tabular}{ccc}
\AxiomC{$\Gamma, A \vdash B$}
\RightLabel{$\limp$ R}
\UnaryInfC{$\Gamma \vdash A \limp B$}
\DisplayProof & \hskip1in\strut& 
\AxiomC{$ A \vdash B, \Delta$}
\RightLabel{$\lsub$ E}
\UnaryInfC{$A \lsub B \vdash \Delta$}
\DisplayProof 
\end{tabular}
\end{center}
As a simple counterexample, consider the sequent $p \Rightarrow q, r
\rightarrow ((p - q) \wedge r)$ given by Pinto and Uustalu around 2003
\cite{Pinto-Uustalu:2010}, which is provable with cut but not cut-free
with Dragalin's restrictions.

Hyland and de Paiva introduced a sequent calculus for FILL labeled
with terms
\begin{center}
\begin{tabular}{c}
\AxiomC{$\overline{y}:\Gamma, x:A \vdash t:B, \overline{u}:\Delta$}
\RightLabel{$\limp$ R}
\UnaryInfC{$\overline{y}: \Gamma \vdash \lambda x:T A \limp B, \overline{u}:\Delta$}
\DisplayProof
\end{tabular} 
\end{center}
where $x: A$ occurs in $t:B$ if and only if there is an ``essential
dependency'' of $B$ from $A$.  The restriction on the $\limp$ I is
that $x$ does not occur in the terms $\overline{u}:\Delta$.  The
original term assignment did not guarantee cut-elimination, as noticed
by Bierman \cite{Bierman:1996}; the assignment to \emph{par left}
($\oplus$ L) had to be fine tuned, as indicated by Bellin
\cite{Bellin:1997}.


A detailed presentation of the term calculus for FILL with a full
proof of cut elimination by Eades and de Paiva is in
\cite{EadesP:2016}, where the correctness for a categorical semantics
for FILL is also proved.  Another correct formalization of FILL, a
sequent calculus with a relational annotation, was given by Bra\"uner
and de Paiva \cite{BraunerDePaiva:1997}, with a proof of
cut-elimination.  The second author \cite{Bellin:1997} gave a system
of proof nets for FILL which sequentialize in the sequent calculus
with term assignment; the essential fact here is that \emph{$x:A$
  occurs in $t:B$ if and only if there is a ``directed chain'' between
  $A$ and $B$ in the proof structure.}  Here cut elimination is proved
by reduction to cut-elimination for proof nets.

A system of two-sided proof nets (in the style of natural deduction)
was given by Cockett and Seely \cite{Cockett:1997}.  For
Bi-Intuitionistic Linear Logic, they gave also a system of proof nets,
corresponding to a sequent calculus without annotations and
restrictions that therefore collapses into classical MLL.  Recently,
Clouston, Dawson, Goré and Tiu \cite{CloustonDGT:2013} gave an 
annotation-free formalization for BILL, alternative to sequent calculi, 
in the form of deep-inference and display calculi for BILL. This calculus 
enjoys cut-elimination and is relevant to the categorical semantics
bi-intuitionistic linear logic. Annotation-free formalizations of 
Bi-Intuitionistic Logic use the display calculus \cite{Gore:2000}, nested sequents
\cite{GorePT:2008} and deep inference \cite{Postniece:2009}.

Tristan Crolard \cite{Crolard:2001,Crolard:2004} made an in-depth
study of Rauszer's logic. In \cite{Crolard:2001} he showed that models
of Rauszer logic (called ``subtractive logic'') based on bi-cartesian
closed categories (with co-exponents) collapse to preorders.  He also
studied models of subtractive logic and showed that its first order
theory is constant-domain logic, thus it is not a conservative
extension of intuitionistic logic.

Crolard \cite{Crolard:2004} develops the type theory for subtractive
logic, extending a system of multiple conclusion classical natural
deduction with a connective of subtraction and then decorating proofs
with a system of annotations of dependencies that allows us to
identify ``constructive proofs'': these are derivations where only the
premise of an implication introduction depends on the discharged
assumption and only the premise of a subtraction elimination depends
on the discharged conclusion. Therefore Crolard's sequent calculus
with annotations is not affected by the counterexamples to
cut-eliminations.

The type theory is Parigot $\lambda\mu$-calculus extended with
operators for sums, products and subtraction, where the operators for
subtraction introduction and elimination are understood as a calculus
of co-routines.  A constructive system of co-routines is then obtained
by imposing restrictions on terms corresponding to the restrictions on
constructive proofs.  

In a series of papers the second author gave a ``pragmatic''
interpretation of bi-intuitionism, where intuitionistic and
co-intuitionistic logic are interpreted as logics of the acts of
assertion and making a hypothesis, respectively, the interactions
between the two sides depending on negations, see \cite{Bellin:2014}.
Here the separation between intuitionistic and co-intuitionistic logic
and their models is given a linguistic motivation. Writing $\vdash p$
for the type of assertions that $p$ is true and using intuitionistic
connectives with the BHK interpretation, one gives a ``pragmatic
interpretatiion'' of ILL, where an expression $A$ is \emph{justified}
or \emph{unjustified} \cite{dalla1995pragmatic}. Similarly, writing
$\mathcal{H}\; p$ for the type of hypotheses that $p$ is true, and
using co-intuitionistic connectives, one builds a co-intuitionistic
language, for which an analogue ``pragmatic interpretation'' has been
attempted.  Both languages may be given a modal interpretattion in S4,
with $(\vdash p)^M = \Box p$ and $(\mathcal{H}\; p)^M = \diamondsuit
p$. Notice that here there is a semantic duality between an assertion
$\vdash p$ and a hypothesis $\mathcal{H}\; \neg p$, as $\Box p$ and
$\diamondsuit \neg p$ are contradictory. Similarly there is a semantic
duality between $\mathcal{H}\; p$ and $\vdash \neg p$, but not between
$\vdash p$ and the hypothesis $\mathcal{H}\, p$. A useful direction of
research in the proof theory of bi-intuitionism may be the
investigation the relations between co-intuitionistic proofs and
intuitionistic refutations.

It is in this context that a term assignment for co-intuitionistic  logic 
has been developed, starting from Crolard's definition but independently
of the $\lambda\mu$-framework. This calculus was used
here as a term assignment of Dual LNL logic.

Trafford \cite{trafford2016structuring} defines an interpretation of
co-intuitionistic logic into a topos-theoretic model to represent both
proofs, in an elementary topoi, and refutations, in a complement
topoi.  He then shows that classical logic can be simulated in his
model.  Earlier Estrada-Gonz\'alez \cite{estrada2010complement} gave a
sequent calculus for BINT based on complement topoi.

Finally, to achieve the project outlined in the introduction of putting
together intuitionistic and co-intuitionistic adjoint logic in the
environment of BILL the definition of a suitable syntax for BILL will
play a key role.


\bibliographystyle{plainurl}



\appendix
\section{Commuting Conversions}
\label{sec:commuting_conversions}
\begin{enumerate}
\item[] {\bf Non linear rules. }
\item \emph{disjunction intro} TC$_{+_{I_1} }$ and TC$_{+_{I_2} }$ commute upwards with every inference and the terms obtained are the same. 
\item \emph{disjunction  elim} TC$_{+_E }$ commutes upwrds with inferences in the derivation of the major premise, 
  the terms assigned to the resulting subderivations are equated. For instance
  \begin{center} \scriptsize
    \begin{math} 
      $$\mprset{flushleft}
      \inferrule* [right=] {
        $$\mprset{flushleft}
        \inferrule* [right=] {
           \DualLNLLogicmv{y}  :  \DualLNLLogicnt{T_{{\mathrm{2}}}}  \vdash_{\mathsf{C} }  \Psi_{{\mathrm{2}}}  \DualLNLLogicsym{,}  \DualLNLLogicnt{t_{{\mathrm{1}}}}  \DualLNLLogicsym{:}   \DualLNLLogicnt{T_{{\mathrm{4}}}}  +  \DualLNLLogicnt{T_{{\mathrm{5}}}}  
          \\\\                        
               \DualLNLLogicmv{z}  :  \DualLNLLogicnt{T_{{\mathrm{3}}}}  \vdash_{\mathsf{C} }  \Psi_{{\mathrm{3}}}  \DualLNLLogicsym{,}  \DualLNLLogicnt{t_{{\mathrm{2}}}}  \DualLNLLogicsym{:}   \DualLNLLogicnt{T_{{\mathrm{4}}}}  +  \DualLNLLogicnt{T_{{\mathrm{5}}}}   \\  \DualLNLLogicmv{x}  :  \DualLNLLogicnt{S}  \vdash_{\mathsf{C} }  \Psi_{{\mathrm{1}}}  \DualLNLLogicsym{,}  \DualLNLLogicnt{t}  \DualLNLLogicsym{:}   \DualLNLLogicnt{T_{{\mathrm{2}}}}  +  \DualLNLLogicnt{T_{{\mathrm{3}}}}   \\ \DualLNLLogicsym{\mbox{$\mid$}}  \Psi_{{\mathrm{2}}}  \DualLNLLogicsym{\mbox{$\mid$}}  \DualLNLLogicsym{=}  \DualLNLLogicsym{\mbox{$\mid$}}  \Psi_{{\mathrm{3}}}  \DualLNLLogicsym{\mbox{$\mid$}}
        }{ \DualLNLLogicmv{x}  :  \DualLNLLogicnt{S}  \vdash_{\mathsf{C} }  \Psi_{{\mathrm{1}}}  \DualLNLLogicsym{,}    \mathsf{case}\, \DualLNLLogicnt{t} \,\mathsf{of}\, \DualLNLLogicmv{y} . \Psi_{{\mathrm{2}}} ,  \DualLNLLogicmv{z} . \Psi_{{\mathrm{3}}}    \DualLNLLogicsym{,}   \mathsf{case}\, \DualLNLLogicnt{t} \,\mathsf{of}\, \DualLNLLogicmv{y} . \DualLNLLogicnt{t_{{\mathrm{1}}}} , \DualLNLLogicmv{z} . \DualLNLLogicnt{t_{{\mathrm{2}}}}   \DualLNLLogicsym{:}   \DualLNLLogicnt{T_{{\mathrm{4}}}}  +  \DualLNLLogicnt{T_{{\mathrm{5}}}}  }
        \\
           \DualLNLLogicmv{v_{{\mathrm{1}}}}  :  \DualLNLLogicnt{T_{{\mathrm{4}}}}  \vdash_{\mathsf{C} }  \Psi_{{\mathrm{4}}}  \\
           \DualLNLLogicmv{v_{{\mathrm{2}}}}  :  \DualLNLLogicnt{T_{{\mathrm{5}}}}  \vdash_{\mathsf{C} }  \Psi_{{\mathrm{5}}}  \\
          \\
            \DualLNLLogicsym{\mbox{$\mid$}}  \Psi_{{\mathrm{4}}}  \DualLNLLogicsym{\mbox{$\mid$}}  \DualLNLLogicsym{=}  \DualLNLLogicsym{\mbox{$\mid$}}  \Psi_{{\mathrm{5}}}  \DualLNLLogicsym{\mbox{$\mid$}}
      }{ \DualLNLLogicmv{x}  :  \DualLNLLogicnt{S}  \vdash_{\mathsf{C} }  \Psi_{{\mathrm{1}}}  \DualLNLLogicsym{,}    \mathsf{case}\, \DualLNLLogicnt{t} \,\mathsf{of}\, \DualLNLLogicmv{y} . \Psi_{{\mathrm{2}}} ,  \DualLNLLogicmv{z} . \Psi_{{\mathrm{3}}}    \DualLNLLogicsym{,}   \mathsf{case}\, \DualLNLLogicsym{(}   \mathsf{case}\, \DualLNLLogicnt{t} \,\mathsf{of}\, \DualLNLLogicmv{y} . \DualLNLLogicnt{t_{{\mathrm{1}}}} , \DualLNLLogicmv{z} . \DualLNLLogicnt{t_{{\mathrm{2}}}}   \DualLNLLogicsym{)} \,\mathsf{of}\, \DualLNLLogicmv{v_{{\mathrm{1}}}} . \Psi_{{\mathrm{4}}} ,  \DualLNLLogicmv{v_{{\mathrm{2}}}} . \Psi_{{\mathrm{5}}}  }
    \end{math}
  \end{center}
  commutes to
  \begin{center} \scriptsize
    \begin{math} 
      $$\mprset{flushleft}
      \inferrule* [right=] {
         \DualLNLLogicmv{x}  :  \DualLNLLogicnt{S}  \vdash_{\mathsf{C} }  \Psi_{{\mathrm{1}}}  \DualLNLLogicsym{,}  \DualLNLLogicnt{t}  \DualLNLLogicsym{:}   \DualLNLLogicnt{T_{{\mathrm{2}}}}  +  \DualLNLLogicnt{T_{{\mathrm{3}}}}  
        \\
        $$\mprset{flushleft}
        \inferrule* [right=] {
          \DualLNLLogicsym{\mbox{$\mid$}}  \Psi_{{\mathrm{4}}}  \DualLNLLogicsym{\mbox{$\mid$}}  \DualLNLLogicsym{=}  \DualLNLLogicsym{\mbox{$\mid$}}  \Psi_{{\mathrm{5}}}  \DualLNLLogicsym{\mbox{$\mid$}}  \\\\  \DualLNLLogicmv{v_{{\mathrm{1}}}}  :  \DualLNLLogicnt{T_{{\mathrm{4}}}}  \vdash_{\mathsf{C} }  \Psi_{{\mathrm{4}}}  \\\\  \DualLNLLogicmv{v_{{\mathrm{2}}}}  :  \DualLNLLogicnt{T_{{\mathrm{5}}}}  \vdash_{\mathsf{C} }  \Psi_{{\mathrm{5}}}  \\  \DualLNLLogicmv{y}  :  \DualLNLLogicnt{T_{{\mathrm{2}}}}  \vdash_{\mathsf{C} }  \Psi_{{\mathrm{2}}}  \DualLNLLogicsym{,}  \DualLNLLogicnt{t_{{\mathrm{1}}}}  \DualLNLLogicsym{:}   \DualLNLLogicnt{T_{{\mathrm{4}}}}  +  \DualLNLLogicnt{T_{{\mathrm{5}}}}  
        }{ \DualLNLLogicmv{y}  :  \DualLNLLogicnt{T_{{\mathrm{2}}}}  \vdash_{\mathsf{C} }  \Psi_{{\mathrm{2}}}  \DualLNLLogicsym{,}   \mathsf{case}\, \DualLNLLogicnt{t_{{\mathrm{1}}}} \,\mathsf{of}\, \DualLNLLogicmv{v_{{\mathrm{1}}}} . \Psi_{{\mathrm{4}}} ,  \DualLNLLogicmv{v_{{\mathrm{2}}}} . \Psi_{{\mathrm{5}}}  }
        \\
        $$\mprset{flushleft}
        \inferrule* [right=] {
          \DualLNLLogicsym{\mbox{$\mid$}}  \Psi_{{\mathrm{4}}}  \DualLNLLogicsym{\mbox{$\mid$}}  \DualLNLLogicsym{=}  \DualLNLLogicsym{\mbox{$\mid$}}  \Psi_{{\mathrm{5}}}  \DualLNLLogicsym{\mbox{$\mid$}}  \\\\  \DualLNLLogicmv{v_{{\mathrm{1}}}}  :  \DualLNLLogicnt{T_{{\mathrm{4}}}}  \vdash_{\mathsf{C} }  \Psi_{{\mathrm{4}}}  \\\\  \DualLNLLogicmv{v_{{\mathrm{2}}}}  :  \DualLNLLogicnt{T_{{\mathrm{5}}}}  \vdash_{\mathsf{C} }  \Psi_{{\mathrm{5}}}  \\  \DualLNLLogicmv{z}  :  \DualLNLLogicnt{T_{{\mathrm{3}}}}  \vdash_{\mathsf{C} }  \Psi_{{\mathrm{3}}}  \DualLNLLogicsym{,}  \DualLNLLogicnt{t_{{\mathrm{2}}}}  \DualLNLLogicsym{:}   \DualLNLLogicnt{T_{{\mathrm{4}}}}  +  \DualLNLLogicnt{T_{{\mathrm{5}}}}  
        }{ \DualLNLLogicmv{z}  :  \DualLNLLogicnt{T_{{\mathrm{3}}}}  \vdash_{\mathsf{C} }  \Psi_{{\mathrm{3}}}  \DualLNLLogicsym{,}   \mathsf{case}\, \DualLNLLogicnt{t_{{\mathrm{2}}}} \,\mathsf{of}\, \DualLNLLogicmv{v_{{\mathrm{1}}}} . \Psi_{{\mathrm{4}}} ,  \DualLNLLogicmv{v_{{\mathrm{2}}}} . \Psi_{{\mathrm{5}}}  }
      }{ \DualLNLLogicmv{x}  :  \DualLNLLogicnt{S}  \vdash_{\mathsf{C} }  \Psi_{{\mathrm{1}}}  \DualLNLLogicsym{,}   \mathsf{case}\, \DualLNLLogicnt{t} \,\mathsf{of}\, \DualLNLLogicmv{y} . \DualLNLLogicsym{(}  \Psi_{{\mathrm{2}}}  \DualLNLLogicsym{,}   \mathsf{case}\, \DualLNLLogicnt{t_{{\mathrm{1}}}} \,\mathsf{of}\, \DualLNLLogicmv{v_{{\mathrm{1}}}} . \Psi_{{\mathrm{4}}} ,  \DualLNLLogicmv{v_{{\mathrm{2}}}} . \Psi_{{\mathrm{5}}}   \DualLNLLogicsym{)} ,  \DualLNLLogicmv{z} . \DualLNLLogicsym{(}  \Psi_{{\mathrm{3}}}  \DualLNLLogicsym{,}   \mathsf{case}\, \DualLNLLogicnt{t_{{\mathrm{2}}}} \,\mathsf{of}\, \DualLNLLogicmv{v_{{\mathrm{1}}}} . \Psi_{{\mathrm{4}}} ,  \DualLNLLogicmv{v_{{\mathrm{2}}}} . \Psi_{{\mathrm{5}}}   \DualLNLLogicsym{)}  }
    \end{math}
  \end{center}   
  \begin{remark}If $\DualLNLLogicnt{t_{{\mathrm{1}}}} =  \mathsf{inl}\, \DualLNLLogicnt{s_{{\mathrm{1}}}} $ and $\DualLNLLogicnt{u_{{\mathrm{2}}}} =  \mathsf{inr}\, \DualLNLLogicnt{s_{{\mathrm{2}}}} $
    then after commutation 
    \[
       \DualLNLLogicmv{y}  :  \DualLNLLogicnt{T_{{\mathrm{2}}}}  \vdash_{\mathsf{C} }  \Psi_{{\mathrm{2}}}  \DualLNLLogicsym{,}   \mathsf{case}\,  \mathsf{inl}\, \DualLNLLogicnt{s_{{\mathrm{1}}}}  \,\mathsf{of}\, \DualLNLLogicmv{v_{{\mathrm{1}}}} . \Psi_{{\mathrm{4}}} ,  \DualLNLLogicmv{v_{{\mathrm{2}}}} . \Psi_{{\mathrm{5}}}   \rightsquigarrow_{\beta}  \DualLNLLogicmv{y}  :  \DualLNLLogicnt{T_{{\mathrm{2}}}}  \vdash_{\mathsf{C} }  \Psi_{{\mathrm{2}}}  \DualLNLLogicsym{,}  \DualLNLLogicsym{[}  \DualLNLLogicnt{s_{{\mathrm{1}}}}  \DualLNLLogicsym{/}  \DualLNLLogicmv{v_{{\mathrm{1}}}}  \DualLNLLogicsym{]}  \Psi_{{\mathrm{4}}}  
      \]
      \[
         \DualLNLLogicmv{y}  :  \DualLNLLogicnt{T_{{\mathrm{3}}}}  \vdash_{\mathsf{C} }  \Psi_{{\mathrm{3}}}  \DualLNLLogicsym{,}   \mathsf{case}\,  \mathsf{inr}\, \DualLNLLogicnt{s_{{\mathrm{2}}}}  \,\mathsf{of}\, \DualLNLLogicmv{v_{{\mathrm{1}}}} . \Psi_{{\mathrm{4}}} ,  \DualLNLLogicmv{v_{{\mathrm{2}}}} . \Psi_{{\mathrm{5}}}   \rightsquigarrow_{\beta}  \DualLNLLogicmv{y}  :  \DualLNLLogicnt{T_{{\mathrm{2}}}}  \vdash_{\mathsf{C} }  \Psi_{{\mathrm{3}}}  \DualLNLLogicsym{,}  \DualLNLLogicsym{[}  \DualLNLLogicnt{s_{{\mathrm{1}}}}  \DualLNLLogicsym{/}  \DualLNLLogicmv{v_{{\mathrm{2}}}}  \DualLNLLogicsym{]}  \Psi_{{\mathrm{5}}} 
        \]
  \end{remark}
\item Subtraction introduction TC$_{-_I}$ commutes upwards  with inferences in both branches with any inference $\mathcal{I}$:
  \begin{center} \footnotesize
    \begin{tabular}{ccc}
      \AxiomC{$ \DualLNLLogicmv{x}  :  \DualLNLLogicnt{S}  \vdash_{\mathsf{C} }  \DualLNLLogicnt{t'_{{\mathrm{1}}}}  \DualLNLLogicsym{:}  \DualLNLLogicnt{T_{{\mathrm{1}}}}  \DualLNLLogicsym{,}  \Psi'_{{\mathrm{1}}} $}
      \RightLabel{$\mathcal{I}$}
      \UnaryInfC{$ \DualLNLLogicmv{x}  :  \DualLNLLogicnt{S}  \vdash_{\mathsf{C} }  \DualLNLLogicnt{t_{{\mathrm{1}}}}  \DualLNLLogicsym{:}  \DualLNLLogicnt{T_{{\mathrm{1}}}}  \DualLNLLogicsym{,}  \Psi_{{\mathrm{1}}} $}
      \AxiomC{$ \DualLNLLogicmv{y}  :  \DualLNLLogicnt{T_{{\mathrm{2}}}}  \vdash_{\mathsf{C} }  \Psi_{{\mathrm{2}}} $}
      \RightLabel{TC$_{-_I}$}
      \BinaryInfC{$ \DualLNLLogicmv{x}  :  \DualLNLLogicnt{S}  \vdash_{\mathsf{C} }  \Psi_{{\mathrm{1}}}  \DualLNLLogicsym{,}   \mathsf{mkc}( \DualLNLLogicnt{t} , \DualLNLLogicmv{y} )   \DualLNLLogicsym{:}   \DualLNLLogicnt{T_{{\mathrm{1}}}}  -  \DualLNLLogicnt{T_{{\mathrm{2}}}}   \DualLNLLogicsym{,}  \DualLNLLogicsym{[}  \DualLNLLogicmv{y}  \DualLNLLogicsym{(}  \DualLNLLogicnt{t}  \DualLNLLogicsym{)}  \DualLNLLogicsym{/}  \DualLNLLogicmv{y}  \DualLNLLogicsym{]}  \Psi_{{\mathrm{2}}} $}
      \DisplayProof 
      &\quad 
      &
      \AxiomC{$ \DualLNLLogicmv{x}  :  \DualLNLLogicnt{S}  \vdash_{\mathsf{C} }  \DualLNLLogicnt{t_{{\mathrm{1}}}}  \DualLNLLogicsym{:}  \DualLNLLogicnt{T_{{\mathrm{1}}}}  \DualLNLLogicsym{,}  \Psi_{{\mathrm{1}}} $}
      \AxiomC{$ \DualLNLLogicmv{y}  :  \DualLNLLogicnt{T_{{\mathrm{2}}}}  \vdash_{\mathsf{C} }  \Psi'_{{\mathrm{2}}} $}
      \RightLabel{$\mathcal{I}$}
      \UnaryInfC{$ \DualLNLLogicmv{y}  :  \DualLNLLogicnt{T_{{\mathrm{2}}}}  \vdash_{\mathsf{C} }  \Psi_{{\mathrm{2}}} $}
      \RightLabel{TC$_{-_I}$}
      \BinaryInfC{$ \DualLNLLogicmv{x}  :  \DualLNLLogicnt{S}  \vdash_{\mathsf{C} }  \Psi_{{\mathrm{1}}}  \DualLNLLogicsym{,}   \mathsf{mkc}( \DualLNLLogicnt{t_{{\mathrm{1}}}} , \DualLNLLogicmv{y} )   \DualLNLLogicsym{:}   \DualLNLLogicnt{T_{{\mathrm{1}}}}  -  \DualLNLLogicnt{T_{{\mathrm{2}}}}   \DualLNLLogicsym{,}  \DualLNLLogicsym{[}  \DualLNLLogicmv{y}  \DualLNLLogicsym{(}  \DualLNLLogicnt{t}  \DualLNLLogicsym{)}  \DualLNLLogicsym{/}  \DualLNLLogicmv{y}  \DualLNLLogicsym{]}  \Psi_{{\mathrm{2}}} $}
      \DisplayProof \\
      \\
      {\normalsize commutes to} & & {\normalsize commutes to} \\
      \\
      \AxiomC{$ \DualLNLLogicmv{x}  :  \DualLNLLogicnt{S}  \vdash_{\mathsf{C} }  \DualLNLLogicnt{t_{{\mathrm{1}}}}  \DualLNLLogicsym{:}  \DualLNLLogicnt{T_{{\mathrm{1}}}}  \DualLNLLogicsym{,}  \Psi'_{{\mathrm{1}}} $}
      \AxiomC{$ \DualLNLLogicmv{y}  :  \DualLNLLogicnt{T_{{\mathrm{2}}}}  \vdash_{\mathsf{C} }  \Psi_{{\mathrm{2}}} $}
      \RightLabel{TC$_{-_I}$}
      \BinaryInfC{$ \DualLNLLogicmv{x}  :  \DualLNLLogicnt{S}  \vdash_{\mathsf{C} }  \Psi'_{{\mathrm{1}}}  \DualLNLLogicsym{,}   \mathsf{mkc}( \DualLNLLogicnt{t_{{\mathrm{1}}}} , \DualLNLLogicmv{y} )   \DualLNLLogicsym{:}   \DualLNLLogicnt{T_{{\mathrm{1}}}}  -  \DualLNLLogicnt{T_{{\mathrm{2}}}}   \DualLNLLogicsym{,}  \DualLNLLogicsym{[}  \DualLNLLogicmv{y}  \DualLNLLogicsym{(}  \DualLNLLogicnt{t'}  \DualLNLLogicsym{)}  \DualLNLLogicsym{/}  \DualLNLLogicmv{y}  \DualLNLLogicsym{]}  \Psi_{{\mathrm{2}}} $}
      \RightLabel{$\mathcal{I}$}
      \UnaryInfC{$ \DualLNLLogicmv{x}  :  \DualLNLLogicnt{S}  \vdash_{\mathsf{C} }  \Psi_{{\mathrm{1}}}  \DualLNLLogicsym{,}   \mathsf{mkc}( \DualLNLLogicnt{t_{{\mathrm{1}}}} , \DualLNLLogicmv{y} )   \DualLNLLogicsym{:}   \DualLNLLogicnt{T_{{\mathrm{1}}}}  -  \DualLNLLogicnt{T_{{\mathrm{2}}}}   \DualLNLLogicsym{,}  \DualLNLLogicsym{[}  \DualLNLLogicmv{y}  \DualLNLLogicsym{(}  \DualLNLLogicnt{t}  \DualLNLLogicsym{)}  \DualLNLLogicsym{/}  \DualLNLLogicmv{y}  \DualLNLLogicsym{]}  \Psi_{{\mathrm{2}}} $}
      \DisplayProof& &
      \AxiomC{$ \DualLNLLogicmv{x}  :  \DualLNLLogicnt{S}  \vdash_{\mathsf{C} }  \DualLNLLogicnt{t_{{\mathrm{1}}}}  \DualLNLLogicsym{:}  \DualLNLLogicnt{T_{{\mathrm{1}}}}  \DualLNLLogicsym{,}  \Psi_{{\mathrm{1}}} $}
      \AxiomC{$ \DualLNLLogicmv{y}  :  \DualLNLLogicnt{T_{{\mathrm{2}}}}  \vdash_{\mathsf{C} }  \Psi'_{{\mathrm{2}}} $}
      \RightLabel{TC$_{-_I}$}
      \BinaryInfC{$ \DualLNLLogicmv{x}  :  \DualLNLLogicnt{S}  \vdash_{\mathsf{C} }  \Psi_{{\mathrm{1}}}  \DualLNLLogicsym{,}   \mathsf{mkc}( \DualLNLLogicnt{t_{{\mathrm{1}}}} , \DualLNLLogicmv{y} )   \DualLNLLogicsym{:}   \DualLNLLogicnt{T_{{\mathrm{1}}}}  -  \DualLNLLogicnt{T_{{\mathrm{2}}}}   \DualLNLLogicsym{,}  \Psi'_{{\mathrm{2}}} $}
      \RightLabel{TC$_{-_I}$}
      \RightLabel{$\mathcal{I}$}
      \UnaryInfC{$ \DualLNLLogicmv{x}  :  \DualLNLLogicnt{S}  \vdash_{\mathsf{C} }  \Psi_{{\mathrm{1}}}  \DualLNLLogicsym{,}   \mathsf{mkc}( \DualLNLLogicnt{t_{{\mathrm{1}}}} , \DualLNLLogicmv{y} )   \DualLNLLogicsym{:}   \DualLNLLogicnt{T_{{\mathrm{1}}}}  -  \DualLNLLogicnt{T_{{\mathrm{2}}}}   \DualLNLLogicsym{,}  \Psi_{{\mathrm{2}}} $}
      \DisplayProof
    \end{tabular}
  \end{center}
\item Subtraction elimination TC$_{-_E}$ commutes upwards. For instance, 
  \begin{center}
    \begin{tabular}{c}
      \AxiomC{$ \DualLNLLogicmv{x}  :  \DualLNLLogicnt{S}  \vdash_{\mathsf{C} }  \DualLNLLogicmv{w}  \DualLNLLogicsym{:}  \DualLNLLogicnt{S_{{\mathrm{1}}}}  \DualLNLLogicsym{,}  \Psi_{{\mathrm{1}}}  \qquad  \DualLNLLogicmv{z}  :  \DualLNLLogicnt{S_{{\mathrm{2}}}}  \vdash_{\mathsf{C} }  \Psi_{{\mathrm{2}}}  \DualLNLLogicsym{,}  \DualLNLLogicnt{t_{{\mathrm{1}}}}  \DualLNLLogicsym{:}   \DualLNLLogicnt{T_{{\mathrm{1}}}}  -  \DualLNLLogicnt{T_{{\mathrm{2}}}}  $}
      \UnaryInfC{$ \DualLNLLogicmv{x}  :  \DualLNLLogicnt{S}  \vdash_{\mathsf{C} }  \Psi_{{\mathrm{1}}}  \DualLNLLogicsym{,}   \mathsf{mkc}( \DualLNLLogicmv{w} , \DualLNLLogicmv{z} )   \DualLNLLogicsym{:}   \DualLNLLogicnt{S_{{\mathrm{1}}}}  -  \DualLNLLogicnt{S_{{\mathrm{2}}}}   \DualLNLLogicsym{,}  \DualLNLLogicsym{[}  \DualLNLLogicmv{z}  \DualLNLLogicsym{(}  \DualLNLLogicmv{w}  \DualLNLLogicsym{)}  \DualLNLLogicsym{/}  \DualLNLLogicmv{z}  \DualLNLLogicsym{]}  \Psi_{{\mathrm{2}}}  \DualLNLLogicsym{,}  \DualLNLLogicsym{[}  \DualLNLLogicmv{z}  \DualLNLLogicsym{(}  \DualLNLLogicmv{w}  \DualLNLLogicsym{)}  \DualLNLLogicsym{/}  \DualLNLLogicmv{z}  \DualLNLLogicsym{]}  \DualLNLLogicnt{t_{{\mathrm{1}}}}  \DualLNLLogicsym{:}   \DualLNLLogicnt{T_{{\mathrm{1}}}}  -  \DualLNLLogicnt{T_{{\mathrm{2}}}}  $}
      \AxiomC{$ \DualLNLLogicmv{y}  :  \DualLNLLogicnt{T_{{\mathrm{1}}}}  \vdash_{\mathsf{C} }  \DualLNLLogicnt{t}  \DualLNLLogicsym{:}  \DualLNLLogicnt{T_{{\mathrm{2}}}}  \DualLNLLogicsym{,}  \Psi_{{\mathrm{3}}} $}
      \BinaryInfC{$ \DualLNLLogicmv{x}  :  \DualLNLLogicnt{S}  \vdash_{\mathsf{C} }  \Psi_{{\mathrm{1}}}  \DualLNLLogicsym{,}   \mathsf{mkc}( \DualLNLLogicmv{w} , \DualLNLLogicmv{z} )   \DualLNLLogicsym{:}   \DualLNLLogicnt{S_{{\mathrm{1}}}}  -  \DualLNLLogicnt{S_{{\mathrm{2}}}}   \DualLNLLogicsym{,}  \DualLNLLogicsym{[}  \DualLNLLogicmv{z}  \DualLNLLogicsym{(}  \DualLNLLogicmv{w}  \DualLNLLogicsym{)}  \DualLNLLogicsym{/}  \DualLNLLogicmv{z}  \DualLNLLogicsym{]}  \Psi_{{\mathrm{2}}}  \DualLNLLogicsym{,}    \mathsf{postp}\,( \DualLNLLogicmv{y}  \mapsto  \DualLNLLogicnt{t} ,  \DualLNLLogicsym{[}  \DualLNLLogicmv{z}  \DualLNLLogicsym{(}  \DualLNLLogicmv{w}  \DualLNLLogicsym{)}  \DualLNLLogicsym{/}  \DualLNLLogicmv{z}  \DualLNLLogicsym{]}  \DualLNLLogicnt{t_{{\mathrm{1}}}} )    \DualLNLLogicsym{,}  \DualLNLLogicsym{[}  \DualLNLLogicmv{y}  \DualLNLLogicsym{(}  \DualLNLLogicsym{[}  \DualLNLLogicmv{z}  \DualLNLLogicsym{(}  \DualLNLLogicmv{w}  \DualLNLLogicsym{)}  \DualLNLLogicsym{/}  \DualLNLLogicmv{z}  \DualLNLLogicsym{]}  \DualLNLLogicnt{t_{{\mathrm{1}}}}  \DualLNLLogicsym{)}  \DualLNLLogicsym{/}  \DualLNLLogicmv{y}  \DualLNLLogicsym{]}  \Psi_{{\mathrm{3}}} $}
      \DisplayProof\\
      \\
      commutes to \\
      \\
      \AxiomC{$ \DualLNLLogicmv{x}  :  \DualLNLLogicnt{S}  \vdash_{\mathsf{C} }  \DualLNLLogicmv{w}  \DualLNLLogicsym{:}  \DualLNLLogicnt{S_{{\mathrm{1}}}}  \DualLNLLogicsym{,}  \Psi_{{\mathrm{1}}} $}
      \AxiomC{$ \DualLNLLogicmv{z}  :  \DualLNLLogicnt{S_{{\mathrm{2}}}}  \vdash_{\mathsf{C} }  \Psi_{{\mathrm{2}}}  \DualLNLLogicsym{,}  \DualLNLLogicnt{t_{{\mathrm{1}}}}  \DualLNLLogicsym{:}   \DualLNLLogicnt{T_{{\mathrm{1}}}}  -  \DualLNLLogicnt{T_{{\mathrm{2}}}}   \qquad  \DualLNLLogicmv{y}  :  \DualLNLLogicnt{T_{{\mathrm{1}}}}  \vdash_{\mathsf{C} }  \DualLNLLogicnt{t}  \DualLNLLogicsym{:}  \DualLNLLogicnt{T_{{\mathrm{2}}}}  \DualLNLLogicsym{,}  \Psi_{{\mathrm{3}}} $}
      \UnaryInfC{$ \DualLNLLogicmv{z}  :  \DualLNLLogicnt{S_{{\mathrm{2}}}}  \vdash_{\mathsf{C} }  \Psi_{{\mathrm{2}}}  \DualLNLLogicsym{,}    \mathsf{postp}\,( \DualLNLLogicmv{y}  \mapsto  \DualLNLLogicnt{t} ,  \DualLNLLogicnt{t_{{\mathrm{1}}}} )    \DualLNLLogicsym{,}  \DualLNLLogicsym{[}  \DualLNLLogicmv{y}  \DualLNLLogicsym{(}  \DualLNLLogicnt{t_{{\mathrm{1}}}}  \DualLNLLogicsym{)}  \DualLNLLogicsym{/}  \DualLNLLogicmv{y}  \DualLNLLogicsym{]}  \Psi_{{\mathrm{3}}} $}
      \BinaryInfC{$ \DualLNLLogicmv{x}  :  \DualLNLLogicnt{S}  \vdash_{\mathsf{C} }  \Psi_{{\mathrm{1}}}  \DualLNLLogicsym{,}   \mathsf{mkc}( \DualLNLLogicmv{w} , \DualLNLLogicmv{z} )   \DualLNLLogicsym{:}   \DualLNLLogicnt{S_{{\mathrm{1}}}}  -  \DualLNLLogicnt{S_{{\mathrm{2}}}}   \DualLNLLogicsym{,}  \DualLNLLogicsym{[}  \DualLNLLogicmv{z}  \DualLNLLogicsym{(}  \DualLNLLogicmv{w}  \DualLNLLogicsym{)}  \DualLNLLogicsym{/}  \DualLNLLogicmv{z}  \DualLNLLogicsym{]}  \Psi_{{\mathrm{2}}}  \DualLNLLogicsym{,}    \mathsf{postp}\,( \DualLNLLogicmv{y}  \mapsto  \DualLNLLogicnt{t} ,  \DualLNLLogicsym{[}  \DualLNLLogicmv{z}  \DualLNLLogicsym{(}  \DualLNLLogicmv{w}  \DualLNLLogicsym{)}  \DualLNLLogicsym{/}  \DualLNLLogicmv{z}  \DualLNLLogicsym{]}  \DualLNLLogicnt{t_{{\mathrm{1}}}} )    \DualLNLLogicsym{,}  \DualLNLLogicsym{[}  \DualLNLLogicmv{z}  \DualLNLLogicsym{(}  \DualLNLLogicmv{w}  \DualLNLLogicsym{)}  \DualLNLLogicsym{/}  \DualLNLLogicmv{z}  \DualLNLLogicsym{]}  \DualLNLLogicsym{[}  \DualLNLLogicmv{y}  \DualLNLLogicsym{(}  \DualLNLLogicnt{t_{{\mathrm{1}}}}  \DualLNLLogicsym{)}  \DualLNLLogicsym{/}  \DualLNLLogicmv{y}  \DualLNLLogicsym{]}  \Psi_{{\mathrm{3}}} $}
      \DisplayProof\\
    \end{tabular}
  \end{center}
  where 
  \begin{equation}
    \DualLNLLogicsym{[}  \DualLNLLogicmv{z}  \DualLNLLogicsym{(}  \DualLNLLogicmv{w}  \DualLNLLogicsym{)}  \DualLNLLogicsym{/}  \DualLNLLogicmv{z}  \DualLNLLogicsym{]}  \DualLNLLogicsym{[}  \DualLNLLogicmv{y}  \DualLNLLogicsym{(}  \DualLNLLogicnt{t_{{\mathrm{1}}}}  \DualLNLLogicsym{)}  \DualLNLLogicsym{/}  \DualLNLLogicmv{y}  \DualLNLLogicsym{]}  \Psi_{{\mathrm{3}}} =  \DualLNLLogicsym{[}  \DualLNLLogicmv{y}  \DualLNLLogicsym{(}  \DualLNLLogicsym{[}  \DualLNLLogicmv{z}  \DualLNLLogicsym{(}  \DualLNLLogicmv{w}  \DualLNLLogicsym{)}  \DualLNLLogicsym{/}  \DualLNLLogicmv{z}  \DualLNLLogicsym{]}  \DualLNLLogicnt{t_{{\mathrm{1}}}}  \DualLNLLogicsym{)}  \DualLNLLogicsym{/}  \DualLNLLogicmv{y}  \DualLNLLogicsym{]}  \Psi_{{\mathrm{3}}}
  \end{equation}
  since $\DualLNLLogicmv{z} \notin \Psi_{{\mathrm{2}}}$.
\item[] {\bf Linear rules.}
\item The $\bot$ introduction rule TILL$_{\bot I}$ rule commutes with any inference, as $\mathtt{connect}_{\bot}
  \mathtt{to} e$ can be ``rewired'' to any term in the context. 
\item The commutations of the rules for linear subtraction TILL$_{\lsub I}$ and TILL$_{\lsub E}$ are similar to those for non-linear subtraction. 
\item Linear disjunction (\emph{par}) introduction (TLL$_{\oplus I}$) commutes with any inference. 
  Linear disjunction elimination  (TLL$_{\oplus E}$) also commutes upwards. For example (writing a proof without 
  non-linear parts for simplicity) we have the following:
  {\small
    \begin{center}
      \begin{tabular}{c}
        \AxiomC{$ \DualLNLLogicmv{x}  :  \DualLNLLogicnt{A}  \vdash_{\mathsf{L} }  \Delta  \DualLNLLogicsym{,}  \DualLNLLogicnt{e}  \DualLNLLogicsym{:}  \DualLNLLogicsym{(}   \DualLNLLogicnt{B_{{\mathrm{1}}}}  \oplus  \DualLNLLogicnt{B_{{\mathrm{2}}}}   \DualLNLLogicsym{)} $}
        \AxiomC{$ \DualLNLLogicmv{y}  :  \DualLNLLogicnt{B_{{\mathrm{1}}}}  \vdash_{\mathsf{L} }  \Delta_{{\mathrm{1}}}  \DualLNLLogicsym{,}  \DualLNLLogicnt{e_{{\mathrm{1}}}}  \DualLNLLogicsym{:}   \DualLNLLogicnt{C_{{\mathrm{1}}}}  \oplus  \DualLNLLogicnt{C_{{\mathrm{2}}}}   \qquad  \DualLNLLogicmv{v}  :  \DualLNLLogicnt{C_{{\mathrm{1}}}}  \vdash_{\mathsf{L} }  \Delta_{{\mathrm{3}}} \qquad  \DualLNLLogicmv{w}  :  \DualLNLLogicnt{C_{{\mathrm{2}}}}  \vdash_{\mathsf{L} }  \Delta_{{\mathrm{4}}} $}
        \RightLabel{$\oplus$ E}
        \UnaryInfC{$ \DualLNLLogicmv{y}  :  \DualLNLLogicnt{B_{{\mathrm{1}}}}  \vdash_{\mathsf{L} }  \Delta_{{\mathrm{1}}}  \DualLNLLogicsym{,}  \DualLNLLogicsym{[}   \mathsf{casel}\, \DualLNLLogicnt{e_{{\mathrm{1}}}}   \DualLNLLogicsym{/}  \DualLNLLogicmv{v}  \DualLNLLogicsym{]}  \Delta_{{\mathrm{3}}}  \DualLNLLogicsym{,}  \DualLNLLogicsym{[}   \mathsf{caser}\, \DualLNLLogicnt{e_{{\mathrm{1}}}}   \DualLNLLogicsym{/}  \DualLNLLogicmv{w}  \DualLNLLogicsym{]}  \Delta_{{\mathrm{4}}} $}
        \noLine
        \UnaryInfC{$ \DualLNLLogicmv{z}  :  \DualLNLLogicnt{B_{{\mathrm{2}}}}  \vdash_{\mathsf{L} }  \Delta_{{\mathrm{2}}}  \hskip1.7in$}
        \RightLabel{$\oplus$ E}
        \BinaryInfC{$ \DualLNLLogicmv{x}  :  \DualLNLLogicnt{A}  \vdash_{\mathsf{L} }  \Delta  \DualLNLLogicsym{,}  \DualLNLLogicsym{[}   \mathsf{casel}\, \DualLNLLogicnt{e}   \DualLNLLogicsym{/}  \DualLNLLogicmv{y}  \DualLNLLogicsym{]}  \Delta_{{\mathrm{1}}}  \DualLNLLogicsym{,}  \DualLNLLogicsym{[}   \mathsf{casel}\, \DualLNLLogicnt{e}   \DualLNLLogicsym{/}  \DualLNLLogicmv{y}  \DualLNLLogicsym{]}  \DualLNLLogicsym{[}   \mathsf{casel}\, \DualLNLLogicnt{e_{{\mathrm{1}}}}   \DualLNLLogicsym{/}  \DualLNLLogicmv{v}  \DualLNLLogicsym{]}  \Delta_{{\mathrm{3}}}  \DualLNLLogicsym{,}  \DualLNLLogicsym{[}   \mathsf{casel}\, \DualLNLLogicnt{e}   \DualLNLLogicsym{/}  \DualLNLLogicmv{y}  \DualLNLLogicsym{]}  \DualLNLLogicsym{[}   \mathsf{caser}\, \DualLNLLogicnt{e_{{\mathrm{1}}}}   \DualLNLLogicsym{/}  \DualLNLLogicmv{w}  \DualLNLLogicsym{]}  \Delta_{{\mathrm{4}}}  \DualLNLLogicsym{,}  \DualLNLLogicsym{[}   \mathsf{caser}\, \DualLNLLogicnt{e}   \DualLNLLogicsym{/}  \DualLNLLogicmv{z}  \DualLNLLogicsym{]}  \Delta_{{\mathrm{2}}}  $}
        \DisplayProof\\
        \\
        commutes to\\
        \\
        \AxiomC{$ \DualLNLLogicmv{x}  :  \DualLNLLogicnt{A}  \vdash_{\mathsf{L} }  \Delta  \DualLNLLogicsym{,}  \DualLNLLogicnt{e}  \DualLNLLogicsym{:}  \DualLNLLogicsym{(}   \DualLNLLogicnt{B_{{\mathrm{1}}}}  \oplus  \DualLNLLogicnt{B_{{\mathrm{2}}}}   \DualLNLLogicsym{)}  \qquad  \DualLNLLogicmv{y}  :  \DualLNLLogicnt{B_{{\mathrm{1}}}}  \vdash_{\mathsf{L} }  \Delta_{{\mathrm{1}}}  \DualLNLLogicsym{,}  \DualLNLLogicnt{e_{{\mathrm{1}}}}  \DualLNLLogicsym{:}   \DualLNLLogicnt{C_{{\mathrm{1}}}}  \oplus  \DualLNLLogicnt{C_{{\mathrm{2}}}}   \quad  \DualLNLLogicmv{z}  :  \DualLNLLogicnt{B_{{\mathrm{2}}}}  \vdash_{\mathsf{L} }  \Delta_{{\mathrm{2}}} $}
        \RightLabel{$\oplus$ E}
        \UnaryInfC{$ \DualLNLLogicmv{x}  :  \DualLNLLogicnt{A}  \vdash_{\mathsf{L} }  \Delta  \DualLNLLogicsym{,}  \DualLNLLogicsym{[}   \mathsf{casel}\, \DualLNLLogicnt{e}   \DualLNLLogicsym{/}  \DualLNLLogicmv{y}  \DualLNLLogicsym{]}  \Delta_{{\mathrm{1}}}  \DualLNLLogicsym{,}  \DualLNLLogicsym{[}   \mathsf{casel}\, \DualLNLLogicnt{e}   \DualLNLLogicsym{/}  \DualLNLLogicmv{y}  \DualLNLLogicsym{]}  \DualLNLLogicnt{e_{{\mathrm{1}}}}  \DualLNLLogicsym{:}   \DualLNLLogicnt{C_{{\mathrm{1}}}}  \oplus  \DualLNLLogicnt{C_{{\mathrm{2}}}}   \DualLNLLogicsym{,}  \DualLNLLogicsym{[}   \mathsf{caser}\, \DualLNLLogicnt{e}   \DualLNLLogicsym{/}  \DualLNLLogicmv{z}  \DualLNLLogicsym{]}  \Delta_{{\mathrm{2}}} $}
        \noLine
        \UnaryInfC{$\hskip3in   \DualLNLLogicmv{v}  :  \DualLNLLogicnt{C_{{\mathrm{1}}}}  \vdash_{\mathsf{L} }  \Delta_{{\mathrm{3}}}  \qquad  \DualLNLLogicmv{w}  :  \DualLNLLogicnt{C_{{\mathrm{2}}}}  \vdash_{\mathsf{L} }  \Delta_{{\mathrm{4}}} $}
        \RightLabel{$\oplus$ E}
        \UnaryInfC{$ \DualLNLLogicmv{x}  :  \DualLNLLogicnt{A}  \vdash_{\mathsf{L} }  \Delta  \DualLNLLogicsym{,}  \DualLNLLogicsym{[}   \mathsf{casel}\, \DualLNLLogicnt{e}   \DualLNLLogicsym{/}  \DualLNLLogicmv{y}  \DualLNLLogicsym{]}  \Delta_{{\mathrm{1}}}  \DualLNLLogicsym{,}  \DualLNLLogicsym{[}   \mathsf{casel}\, \DualLNLLogicsym{[}   \mathsf{casel}\, \DualLNLLogicnt{e}   \DualLNLLogicsym{/}  \DualLNLLogicmv{y}  \DualLNLLogicsym{]}  \DualLNLLogicnt{e_{{\mathrm{1}}}}   \DualLNLLogicsym{/}  \DualLNLLogicmv{v}  \DualLNLLogicsym{]}  \Delta_{{\mathrm{3}}}  \DualLNLLogicsym{,}  \DualLNLLogicsym{[}   \mathsf{caser}\, \DualLNLLogicsym{[}   \mathsf{casel}\, \DualLNLLogicnt{e}   \DualLNLLogicsym{/}  \DualLNLLogicmv{y}  \DualLNLLogicsym{]}  \DualLNLLogicnt{e_{{\mathrm{1}}}}   \DualLNLLogicsym{/}  \DualLNLLogicmv{w}  \DualLNLLogicsym{]}  \Delta_{{\mathrm{4}}}  \DualLNLLogicsym{,}  \DualLNLLogicsym{[}   \mathsf{caser}\, \DualLNLLogicnt{e}   \DualLNLLogicsym{/}  \DualLNLLogicmv{z}  \DualLNLLogicsym{]}  \Delta_{{\mathrm{2}}} $}
        \DisplayProof
      \end{tabular}
  \end{center} }
  Now 
  \begin{equation}
    \DualLNLLogicsym{[}   \mathsf{casel}\, \DualLNLLogicnt{e}   \DualLNLLogicsym{/}  \DualLNLLogicmv{y}  \DualLNLLogicsym{]}  \DualLNLLogicsym{[}   \mathsf{casel}\, \DualLNLLogicnt{e_{{\mathrm{1}}}}   \DualLNLLogicsym{/}  \DualLNLLogicmv{v}  \DualLNLLogicsym{]}  \Delta_{{\mathrm{3}}} \quad = \quad  \DualLNLLogicsym{[}   \mathsf{casel}\, \DualLNLLogicsym{[}   \mathsf{casel}\, \DualLNLLogicnt{e}   \DualLNLLogicsym{/}  \DualLNLLogicmv{y}  \DualLNLLogicsym{]}  \DualLNLLogicnt{e_{{\mathrm{1}}}}   \DualLNLLogicsym{/}  \DualLNLLogicmv{v}  \DualLNLLogicsym{]}  \Delta_{{\mathrm{3}}}
  \end{equation}
  because $y$ does not occur in $\Delta_{{\mathrm{3}}}$, only in $\DualLNLLogicnt{e_{{\mathrm{1}}}}$ and 
  \[
    \DualLNLLogicsym{[}   \mathsf{casel}\, \DualLNLLogicnt{e}   \DualLNLLogicsym{/}  \DualLNLLogicmv{y}  \DualLNLLogicsym{]}  \DualLNLLogicsym{[}   \mathsf{casel}\, \DualLNLLogicnt{e_{{\mathrm{1}}}}   \DualLNLLogicsym{/}  \DualLNLLogicmv{w}  \DualLNLLogicsym{]}  \Delta_{{\mathrm{4}}} \quad = \quad \DualLNLLogicsym{[}   \mathsf{caser}\, \DualLNLLogicsym{[}   \mathsf{casel}\, \DualLNLLogicnt{e}   \DualLNLLogicsym{/}  \DualLNLLogicmv{y}  \DualLNLLogicsym{]}  \DualLNLLogicnt{e_{{\mathrm{1}}}}   \DualLNLLogicsym{/}  \DualLNLLogicmv{w}  \DualLNLLogicsym{]}  \Delta_{{\mathrm{4}}}
    \]
    because $y$ does not occur in $\Delta_{{\mathrm{4}}}$, only in $\DualLNLLogicnt{e_{{\mathrm{1}}}}$.

\end{enumerate}


\section{Proofs}
\label{sec:proofs}

\subsection{Proof of Lemma~\ref{lemma:symmetric_comonoidal_isomorphisms}}
\label{subsec:proof_of_lemma:symmetric_comonoidal_isomorphisms}
We show that both of the  maps:
\[
\begin{array}{lll}
  \jinv{R,S} := \func{J}R \oplus \func{J}S \mto^{\eta} \func{JH}(\func{J}R \oplus \func{J}S) \mto^{\func{J}\h{A,B}} \func{J}(\func{HJ}R + \func{HJ}S) \mto^{\J(\varepsilon_R + \varepsilon_S)} \func{J}(R + S)\\
  \\
  \jinv{0} := \perp \mto^{\eta} \func{JH}\perp \mto^{\func{J}\h{\perp}} \func{J}0
\end{array}
\]
are mutual inverses with $\j{R,S} : \func{J}(R + S) \mto \func{J}R
\oplus \func{J}S$ and $\j{0} : \perp \mto \func{J}0$ respectively.

\begin{itemize}
\item[Case.] The following diagram implies that $\jinv{R,S};\j{R,S} = \id$:
  \begin{diagram}
    \square|ammm|/->`->`->`<-/<950,500>[
      \func{J}R \oplus \func{J}S`
      \func{JH}(\func{J}R \oplus \func{J}S)`
      \func{JHJ}R \oplus \func{JHJ}S`
      \func{J}(\func{HJ}R + \func{HJ}S);
      \eta`
      \eta \oplus \eta`
      \func{J}\h{}`
      \j{}
    ]
    \dtriangle(-950,0)|amm|/=``<-/<950,500>[
      \func{J}R \oplus \func{J}S`
      \func{J}R \oplus \func{J}S`
      \func{JHJ}R \oplus \func{JHJ}S;``
      \func{J}\varepsilon \oplus \func{J}\varepsilon]

    \qtriangle(-950,-500)/`<-`->/<1900,500>[
      \func{J}R \oplus \func{J}S`
      \func{J}(\func{HJ}R + \func{HJ}S)`
      \func{J}(R + S);`
      \j{}`
      \func{J}(\varepsilon + \varepsilon)]        
  \end{diagram}
  The two top diagrams both commute because $\eta$ and $\varepsilon$
  are the unit and counit of the adjunction respectively, and the
  bottom diagram commutes by naturality of $\j{}$.
  
\item[Case.] The following diagram implies that $\j{R,S};\jinv{R,S} = \id$:
  \begin{diagram}
    \square|ammm|/->`->`->`->/<950,500>[
      \func{J}(R + S)`
      \func{J}R \oplus \func{J}S`
      \func{JHJ}(R + S)`
      \func{JH}(\func{J}R \oplus \func{J}S);
      \j{}`
      \eta`
      \eta`
      \func{JH}\j{}
    ]
    \dtriangle(-950,0)|amm|/=``<-/<950,500>[
      \func{J}(R + S)`
      \func{J}(R + S)`
      \func{JHJ}(R + S);``
      \func{J}\varepsilon]

    \qtriangle(-950,-500)/`<-`->/<1900,500>[
      \func{J}(R + S)`
      \func{JH}(\func{J}R \oplus \func{J}S)`
      \func{J}(\func{HJ}R + \func{HJ}S);`
      \func{J}(\varepsilon + \varepsilon)`
      \func{J}\h{}]
  \end{diagram}
  The top left and bottom diagrams both commute because $\eta$ and $\varepsilon$
  are the unit and counit of the adjunction respectively, and the
  top right diagram commutes by naturality of $\eta$.
  
\item[Case.] The following diagram implies that $\jinv{0};\j{0} = \id$:
  \begin{diagram}
    \square|amma|/->`=`->`<-/<950,500>[
      \perp`
      \func{JH}\perp`
      \perp`
      \func{J}0;
      \eta``
      \func{J}\h{\perp}`
      \j{0}]
  \end{diagram}
  This diagram holds because $\eta$ is the unit of the adjunction.

\item[Case.] The following diagram implies that $\j{0};\jinv{0} = \id$:        
  \begin{diagram}
    \Atriangle|aaa|/->`->`<-/<950,500>[
      \func{JHJ}0`
      \func{J}0`
      \func{JH}\perp;
      \func{J}\varepsilon`
      \func{JH}\j{0}`
      \func{J}\h{\perp}]

    \Dtriangle|aaa|/=`->`/<950,500>[
      \func{J}0`
      \func{JHJ}0`
      \func{J}0;`
      \eta`]

    \square/->``->`/<1900,1000>[
      \func{J}0`
      \perp`
      \func{J}0`
      \func{JH}\perp;
      \j{0}``
      \eta`]
  \end{diagram}
  The top-left and bottom diagrams commute because $\eta$ and
  $\varepsilon$ are the unit and counit of the adjunction
  respectively, and the top-right digram commutes by naturality of
  $\eta$.
\end{itemize}

\subsection{Proof of Lemma~\ref{lemma:symmetric_comonoidal_monad}}
\label{subsec:proof_of_lemma:symmetric_comonoidal_monad}
Since $\wn$ is the composition of two symmetric comonoidal functors we know it is also symmetric comonoidal, and hence, the following diagrams all hold:
\begin{mathpar}
  \bfig
  \vSquares|ammmmma|/->`->`->``->`->`->/[
    \wn ((A \oplus B) \oplus C)`
    \wn (A \oplus B) \oplus \wn C`
    \wn (A \oplus (B \oplus C))`
    (\wn A \oplus \wn B) \oplus \wn C`
    \wn A \oplus \wn (B \oplus C))`
    \wn A \oplus (\wn B \oplus \wn C);
    \r{A \oplus B,C}`
    \wn \alpha_{A,B,C}`
    \r{A,B} \oplus \id_{\wn C}``
    \r{A,B \oplus C}`
    \alpha_{\wn A,\wn B,\wn C}`
    \id_{\wn A} \oplus \r{B,C}]    
  \efig
\end{mathpar}
\begin{mathpar}
  \bfig
  \square|amma|/->`->`->`->/<1000,500>[
    \wn (\perp \oplus A)`
    \wn \perp \oplus \wn A`
    \wn A`
    \perp \oplus \wn A;
    \r{\perp,A}`
    \wn {\lambda}_{A}`
    \r{\perp} \oplus \id_{\wn A}`
      {\lambda^{-1}}_{\wn A}]
  \efig
  \and
  \bfig
  \square|amma|/->`->`->`->/<1000,500>[
    \wn (A \oplus \perp)`
    \wn A \oplus \wn \perp`
    \wn A`
    \wn A \oplus \perp;
    \r{A,\perp}`
    \wn {\rho}_{A}`
    \id_{\wn A} \oplus \r{\perp}`
       {\rho^{-1}}_{\wn A}]
  \efig
\end{mathpar}

\begin{diagram}
  \square|amma|/->`->`->`->/<1000,500>[
    \wn (A \oplus B)`
    \wn A \oplus \wn B`
    \wn (B \oplus A)`
    \wn B \oplus \wn A;
    \r{A,B}`
    \wn {\beta}_{A,B}`
        {\beta}_{\wn A,\wn B}`
        \r{B,A}]
\end{diagram}
Next we show that $(\wn,\eta,\mu)$ defines a monad where
$\eta_A : A \mto ?A$ is the unit of the adjunction, and
$\mu_A = \func{J}\varepsilon_{\func{H}\,A} : \wn\wn A \mto \wn A$.  It
suffices to show that every diagram of
Definition~\ref{def:symm-comonoidal-monad} holds.
\begin{itemize}
\item[Case.]
  $$\bfig
  \square|ammb|<600,600>[
    \wn^3 A`
    \wn^2 A`
    \wn^2 A`
    \wn A;
    \mu_{\wn A}`
    \wn\mu_A`
    \mu_A`
    \mu_A]
  \efig$$
  It suffices to show that the following diagram commutes:
  $$\bfig
  \square|ammb|<600,600>[
    \func{J}(\func{H}(\wn^2 A))`
    \func{J}(\func{H}\,\wn A)`
    \func{J}(\func{H}\,\wn A)`
    \func{J}(\func{H}\,A);
    \func{J}\varepsilon_{\func{H}\,\wn A}`
    \func{J}(\func{H}\,\mu_A)`
    \func{J}\varepsilon_{\func{H}\,A}`
    \func{J}\varepsilon_{\func{H}\,A}]
  \efig$$
  But this diagram is equivalent to the following:
  $$\bfig
  \square|ammb|<600,600>[
    \func{H}\func{JHJH} A`
    \func{H}\,\func{JH} A`
    \func{H}\,\func{JH} A`
    \func{H}\,A;
    \varepsilon_{\func{H}\,\func{JH} A}`
    \func{H}\,\func{J}\varepsilon_{\func{H}\,A}`
    \varepsilon_{\func{H}\,A}`
    \varepsilon_{\func{H}\,A}]
  \efig$$
  The previous diagram commutes by naturality of $\varepsilon$.

\item[Case.]
  $$\bfig
  \Atrianglepair/=`<-`=`->`<-/<600,600>[
    \wn A`
    \wn A`
    \wn^2 A`
    \wn A;`
    \mu_A``
    \eta_{\wn A}`
    \wn \eta_A]
  \efig$$
  It suffices to show that the following diagrams commutes:
  $$\bfig
  \Atrianglepair/=`<-`=`->`<-/<600,600>[
    JH A`
    JH A`
    JHJH A`
    JH A;`
    J\varepsilon_{HA}``
    \eta_{JH A}`
    JH \eta_A]
  \efig$$
  Both of these diagrams commute because $\eta$ and $\varepsilon$
  are the unit and counit of an adjunction.
\end{itemize}

It remains to be shown that $\eta$ and $\mu$ are both
symmetric comonoidal natural transformations, but this easily follows
from the fact that we know $\eta$ is by assumption, and that $\mu$
is because it is defined in terms of $\varepsilon$ which is a
symmetric comonoidal natural transformation.  Thus, all of the
following diagrams commute:
\begin{mathpar}
  \bfig
  \ptriangle|amm|/->`->`<-/<1000,600>[
    A \oplus B`
    \wn A \oplus \wn B`
    \wn (A \oplus B);
    \eta_A \oplus \eta_B`
    \eta_A`
    \r{A,B}]    
  \efig
  \and
  \bfig
  \Vtriangle/->`=`->/<600,600>[
    \perp`
    \wn\perp`
    \perp;
    \eta_\perp``
    \r{\perp}]
  \efig
\end{mathpar}
\begin{mathpar}
  \bfig
  \square|ammm|/->`->``/<1050,600>[
    \wn^2(A \oplus B)`
    \wn (\wn A \oplus \wn B)`
    \wn (A \oplus B)`;
    \wn\r{A,B}`
    \mu_{A \oplus B}``]

  \square(850,0)|ammm|/->``->`/<1050,600>[
    \wn (\wn A \oplus \wn B)`
    \wn^2 A \oplus \wn^2 B``
    \wn A \oplus \wn B;
    \r{\wn A,\wn B}``
    \mu_A \oplus \mu_B`]
  \morphism(-200,0)<2100,0>[\wn(A \oplus B)`\wn A \oplus \wn B;\r{A,B}]
  \efig
  \and
  \bfig
  \square|ammb|/->`->`->`->/<600,600>[
    \wn^2\perp`
    \wn\perp`
    \wn\perp`
    \perp;
    \wn\r{\perp}`
    \mu_\perp`
    \r{\perp}`
    \r{\perp}]
  \efig
\end{mathpar}

\subsection{Proof of Lemma~\ref{lemma:right_weakening_and_contraction}}
\label{subsec:proof_of_lemma:right_weakening_and_contraction}
Suppose $(\func{H},\h{})$ and $(\func{J},\j{})$ are two symmetric
comonoidal functors, such that, $\cat{L} : \func{H} \dashv \func{J}
: \cat{C}$ is a dual LNL model.  Again, we know $\wn A = H;J : \cat{L}
\mto \cat{L}$ is a symmetric comonoidal monad by
Lemma~\ref{lemma:symmetric_comonoidal_monad}.  

We define the following morphisms:
\[
\begin{array}{lll}
  \w{A} := \perp \mto^{\jinv{0}} \func{J} 0 \mto^{\func{J}\diamond_{\func{H} A}} \func{JH}A \mto/=/ \wn A\\
  \c{A} := \wn A \oplus \wn A \mto/=/ \func{JH}A \oplus \func{JH}A \mto^{\jinv{\func{H}A,\func{H}A}} \func{J}(\func{H}A + \func{H}A) \mto^{\func{J}\codiag{\func{H}A}} \func{JHA} \mto/=/ \wn A
\end{array}
\]

Next we show that both of these are symmetric comonoidal natural
transformations, but for which functors?  Define $\func{W}(A) =
\perp$ and $\func{C}(A) = \wn A \oplus \wn A$ on objects of
$\cat{L}$, and $\func{W}(f : A \mto B) = \id_\perp$ and $\func{C}(f
: A \mto B) = \wn f \oplus \wn f$ on morphisms.  So we must show
that $\w{} : \func{W} \mto \wn$ and $\c{} : \func{C} \mto \wn$ are
symmetric comonoidal natural transformations.  We first show that
$\w{}$ is and then we show that $\c{}$ is.  Throughout the proof we
drop subscripts on natural transformations for readability.
\begin{itemize}
\item[Case.] To show $\w{}$ is a natural transformation we must show
  the following diagram commutes for any morphism $f : A \mto B$:
  \[
  \bfig
  \square[W(A)`\wn A`W(B)`\wn B;\w{A}`W(f)`\wn f`\w{B}]
  \efig
  \]
  This diagram is equivalent to the following:
  \[
  \bfig
  \square[\perp`\wn A`\perp`\wn B;\w{A}`\id_{\perp}`\wn f`\w{B}]
  \efig
  \]
  It further expands to the following:
  \[
  \bfig
  \hSquares/->`->`->``->`->`->/[\perp`\func{J}0`\func{JH}A`\perp`\func{J}0`\func{JH}B;\jinv{0}`\func{J}(\diamond_{\func{H}A})`\id_\perp``\func{JH}f`\jinv{0}`\func{J}(\diamond_{\func{H}B})]
  \efig
  \]
  This diagram commutes, because
  $\func{J}(\diamond_{\func{H}A});\func{J}f =
  \func{J}(\diamond_{\func{H}A};f) =
  \func{J}(\diamond_{\func{H}B})$, by the uniqueness of the initial
  map.

\item[Case.] The functor $\func{W}$ is comonoidal itself.  To see this we
  must exhibit a map
  \[\s{\perp} := \id_\perp : \func{W}\perp \mto \perp\]
  and a natural transformation
  \[\s{A,B} := \rho^{-1}_\perp : \func{W}(A \oplus B) \mto \func{W}A \oplus \func{W}B\]
  subject to the coherence conditions in
  Definition~\ref{def:coSMCFUN}.  Clearly, the second map is a natural
  transformation, but we leave showing they respect the coherence
  conditions to the reader.  Now we can show that $\w{}$ is indeed
  symmetric comonoidal.
  \begin{itemize}
  \item[Case.] \ \\
    \begin{diagram}
      \square|amma|<1000,500>[
        \func{W}(A \oplus B)`
        \func{W}A \oplus \func{W}B`
        \wn (A \oplus B)`
        \wn A \oplus \wn B;
        \s{A,B}`
        \w{A \oplus B}`
        \w{A} \oplus \w{B}`
        \r{A,B}]
    \end{diagram}
    Expanding the objects of the previous diagram results in the
    following:
    \begin{diagram}
      \square|amma|<1000,500>[
        \perp`
        \perp \oplus \perp`
        \wn (A \oplus B)`
        \wn A \oplus \wn B;
        \s{A,B}`
        \w{A \oplus B}`
        \w{A} \oplus \w{B}`
        \r{A,B}]
    \end{diagram}
    This diagram commutes, because the following fully expanded
    diagram commutes:
    \begin{diagram}
      \square|amma|/<-`->`->`->/<950,500>[
        \J 0`
        \J (0 + 0)`
        \J\H (A \oplus B)`
        \J (\H A + \H B);
        \J\rho`
        \J\diamond`
        \J (\diamond + \diamond)`
        \J\h{}]

      \square(950, 0)|amma|/->``->`->/<950,500>[
        \J (0 + 0)`
        \J 0 + \J 0`
        \J (\H A + \H B)`
        \J\H A \oplus \J\H B;
        \j{}``
        \J\diamond \oplus \J\diamond`
        \j{}]

      \square(0,500)/->`->``/<1900,1500>[
        \perp`
        \perp \oplus \perp`
        \J 0`;
        \rho^{-1}`
        \jinv{0}``]

      \dtriangle(950,1300)|mma|<950,700>[
        \perp \oplus \perp`
        \J 0 \oplus \perp`
        \J 0 \oplus \J 0;
        \jinv{0} \oplus \id`
        \jinv{0} \oplus \jinv{0}`
        \id \oplus \jinv{0}]

      \ptriangle(950,800)|amm|/`=`->/<950,500>[
        \J 0 \oplus \perp`
        \J 0 \oplus \J 0`
        \J 0 \oplus \perp;``
        \id \oplus \j{0}]

      \morphism(1900,1300)|m|/=/<0,-800>[
        \J 0 \oplus \J 0`
        \J 0 \oplus \J 0;]

      \morphism(0,500)|m|<950,800>[
        \J 0`
        \J0 \oplus \perp;
        \rho^{-1}]

      \place(475,250)[(1)]
      \place(1425,250)[(2)]
      \place(950,650)[(3)]
      \place(1180,1100)[(4)]
      \place(1620,1550)[(5)]
      \place(475,1550)[(6)]
    \end{diagram}
    Diagram 1 commutes because $0$ is the initial object, diagram 2
    commutes by naturality of $\j{}$, diagram 3 commutes because
    $\J$ is a symmetric comonoidal functor, diagram 4 commutes
    because $\j{0}$ is an isomorphism
    (Lemma~\ref{lemma:symmetric_comonoidal_isomorphisms}), diagram 5
    commutes by functorality of $\J$, and diagram 6 commutes by
    naturality of $\rho$.
    
  \item[Case.]\ \\
    \begin{diagram}
      \Vtriangle/<-`<-`<-/[
        \perp`
        \wn \perp`
        \func{W}\perp;
        \r{\perp}`
        \s{\perp}`
        \w{\perp}]
    \end{diagram}
    Expanding the objects in the previous diagram results in the
    following:
    \begin{diagram}
      \Vtriangle/<-`=`<-/[
        \perp`
        \wn \perp`
        \perp;
        \r{\perp}``
        \w{\perp}]
    \end{diagram}
    This diagram commutes because the following one does:
    \begin{diagram}
      \dtriangle|ama|/=`<-`->/<950,500>[
        \J 0`
        \J 0`
        \J\H \perp;`
        \J\h{\perp}`
        \J\diamond]
      \square(-950,0)|aaaa|/`=``->/<950,500>[
        \perp``
        \perp`
        \J 0;```
        \jinv{0}]
      \morphism(-950,500)/<-/<1900,0>[
        \perp`
        \J 0;
        \j{0}]        
    \end{diagram}
    The diagram on the left commutes because $\j{0}$ is an
    isomorphism
    (Lemma~\ref{lemma:symmetric_comonoidal_isomorphisms}), and the
    diagram on the right commutes because $0$ is the initial object.

  \end{itemize}

\item[Case.] Now we show that $\c{A} : \wn A \oplus \wn A \mto \wn
  A$ is a natural transformation.  This requires the following
  diagram to commute (for any $f : A \mto B$):
  \[
  \bfig
  \square[
    \func{C}A`
    \wn A`
    \func{C}B`
    \wn B;
    \c{A}`
    \func{C}f`
    \wn f`
    \c{B}]
  \efig
  \]
  This expands to the following diagram:
  \[
  \bfig
  \square[
    \wn A \oplus \wn A`
    \wn A`
    \wn B \oplus \wn B`
    \wn B;
    \c{A}`
    \wn f \oplus \wn f`
    \wn f`
    \c{B}]
  \efig
  \]
  This diagram commutes because the following diagram does:
  \begin{diagram}
    \hSquares[
      \func{JH}A \oplus \func{JH}A`
      \func{J}(\func{H}A + \func{H}A)`
      \func{JH}A`
      \func{JH}B \oplus \func{JH}B`
      \func{J}(\func{H}B + \func{H}B)`
      \func{JH}B;
      \jinv{\func{H}A, \func{H}A}`
      \func{J}\bigtriangledown_{\func{H}A}`
      \func{JH}f \oplus \func{JH}f`
      \func{J}(\func{H}f + \func{H}f)`
      \func{JH}f`
      \jinv{\func{H}B, \func{H}B}`
      \func{J}\bigtriangledown_{\func{H}B}]
  \end{diagram}
  The left square commutes by naturality of $\jinv{}$, and the right square commutes by naturality of the codiagonal
  $\bigtriangledown_{A} : A + A \mto A$.

\item[Case.] The functor $\func{C} : \cat{L} \mto \cat{L}$ is indeed
  symmetric comonoidal where the required maps are defined as follows:
  \[
  \small
  \begin{array}{lll}
    \quad\quad\quad\t_\perp := \wn \perp \oplus \wn \perp \mto/=/ \J\H\perp \oplus \J\H\perp \mto^{\jinv{}} \J(\H\perp + \H\perp) \mto^{\J\codiag{}} \J\H\perp \mto^{\J\h{\perp}} \J 0 \mto^{\j{0}} \perp\\
    \\
    \quad\quad\quad\t_{A,B} := \wn (A \oplus B) \oplus \wn (A \oplus B) \mto^{\r{A,B} \oplus \r{A,B}} (\wn A \oplus \wn B) \oplus (\wn A \oplus \wn B) \mto^{\iso} (\wn A \oplus \wn A) \oplus (\wn B \oplus \wn B)
  \end{array}
  \]
  where $\mathsf{iso}$ is a natural isomorphism that can easily be
  defined using the symmetric monoidal structure of
  $\cat{L}$. Clearly, $\t$ is indeed a natural transformation, but
  we leave checking that the required diagrams in
  Definition~\ref{def:coSMCFUN} commute to the reader.  We can now
  show that $\c{A} : \wn A \oplus \wn A \mto \wn A$ is symmetric
  comonoidal.  The following diagrams from
  Definition~\ref{def:coSMCNAT} must commute:
  \begin{itemize}
  \item[Case.] \ \\
    \begin{diagram}
      \square<1000,500>[
        \func{C}(A \oplus B)`
        \func{C}A \oplus \func{C}B`
        \wn (A \oplus B)`
        \wn A \oplus \wn B;
        \t_{A,B}`
        \c{A \oplus B}`
        \c{A} \oplus \c{B}`
        \r{A,B}]
    \end{diagram}
    Expanding the objects in the previous diagram results in the following:
    \begin{diagram}
      \square<2000,500>[
        \wn (A \oplus B) \oplus \wn (A \oplus B)`
        (\wn A \oplus \wn A) \oplus (\wn B \oplus \wn B)`
        \wn (A \oplus B)`
        \wn A \oplus \wn B;
        \t_{A,B}`
        \c{A \oplus B}`
        \c{A} \oplus \c{B}`
        \r{A,B}]
    \end{diagram}
    This diagram commutes, because the following fully expanded one
    does:
    \begin{center}
      \rotatebox{90}{$
        \bfig
        \square|amma|<1500,500>[
          \J(\H (A \oplus B) + \H (A \oplus B))`
          \J((\H A + \H B) + (\H A + \H B))`
          \J\H (A \oplus B)`
          \J (\H A + \H B);
          \J(\h{} + \h{})`
          \J\codiag{}`
          \J\codiag{}`
          \J\h{}]

        \square(0,500)|amma|<1500,500>[
          \J\H(A \oplus B) \oplus \J\H(A \oplus B)`
          \J(\H A + \H B) \oplus \J(\H A + \H B)`
          \J(\H(A \oplus B) + \H(A \oplus B))`
          \J((\H A + \H B) + (\H A + \H B));
          \J\h{} \oplus \J\h{}`
          \jinv{}`
          \jinv{}`
          \J(\h{} + \h{})]

        \square(1500,0)|amma|/->`->`->`=/<1500,500>[
          \J((\H A + \H B) + (\H A + \H B))`
          \J((\H A + \H A) + (\H B + \H B))`
          \J (\H A + \H B)`
          \J (\H A + \H B);
          \J\iso`
          \J\codiag{}`
          \J(\codiag{} + \codiag{})`]

        \square(3000,0)|amma|<1500,500>[
          \J((\H A + \H A) + (\H B + \H B))`
          \J(\H A + \H A) \oplus \J(\H B + \H B)`
          \J (\H A + \H B)`
          \J\H A \oplus \J\H B;
          \j{}`
          \J(\codiag{} + \codiag{})`
          \J\codiag{} \oplus \J\codiag{}`
          \j{}]

        \dtriangle(3000,500)|ama|/<-`->`->/<1500,500>[
          (\J\H A \oplus \J\H A) \oplus (\J\H B \oplus \J\H B)`
          \J((\H A + \H A) + (\H B + \H B))`
          \J(\H A + \H A) \oplus \J(\H B + \H B);
          \j{};(\j{} \oplus \j{})`
          \jinv{} \oplus \jinv{}`
          \j{}]          

        \morphism(1500,1000)<1500,0>[
          \J(\H A + \H B) \oplus \J(\H A + \H B)`
          (\J\H A \oplus \J\H B) \oplus (\J\H A \oplus \J\H B);
          \j{} \oplus \j{}]

        \morphism(3000,1000)<1500,0>[
          (\J\H A \oplus \J\H B) \oplus (\J\H A \oplus \J\H B)`
          (\J\H A \oplus \J\H A) \oplus (\J\H B \oplus \J\H B);
          \iso]

        \place(750,250)[(1)]
        \place(750,750)[(2)]
        \place(2250,250)[(3)]
        \place(2550,750)[(4)]
        \place(3750,250)[(5)]
        \place(4125,700)[(6)]
        \efig
        $}
    \end{center}
    Diagram 1 commutes by naturality of $\codiag{}$, diagram 2
    commutes by naturality of $\jinv{}$, diagram 3 commutes by
    straightforward reasoning on coproducts, diagram 4 commutes by
    straightforward reasoning on the symmetric monoidal structure of
    $\J$ after expanding the definition of the two isomorphisms --
    here $\J\iso$ is the corresponding isomorphisms on coproducts --
    diagram 5 commutes by naturality of $\j{}$, and diagram 6
    commutes because $\j{}$ is an isomorphism
    (Lemma~\ref{lemma:symmetric_comonoidal_isomorphisms}).
    
  \item[Case.] \ \\
    \begin{diagram}
      \Vtriangle/<-`<-`<-/[
        \perp`
        \wn \perp`
        \func{C} \perp;
        \r{\perp}`
        \t{\perp}`
        \c{\perp}]
    \end{diagram}
    Expanding the objects of this diagram results in the following:
    \begin{diagram}
      \square/<-`<-`<-`=/<950,500>[
        \perp`
        \wn \perp`
        \wn \perp \oplus \wn \perp`
        \wn \perp \oplus \wn \perp;
        \r{\perp}`
        \t{\perp}`
        \c{\perp}`]          
    \end{diagram}
    Simply unfolding the morphisms in the previous diagram reveals the following:
    \begin{diagram}
      \square/`<-`<-`=/<950,500>[
        \J(\H\perp + \H\perp)`
        \J(\H\perp + \H\perp)`
        \J\H\perp \oplus \J\H\perp`
        \J\H\perp \oplus \J\H\perp;`
        \jinv{}`
        \jinv{}`]

      \square(0,500)/`<-`<-`/<950,500>[
        \J\H\perp`
        \J\H\perp`
        \J(\H\perp + \H\perp)`
        \J(\H\perp + \H\perp);
        `
        \J\codiag{}`
        \J\codiag{}`]

      \square(0,1000)/`<-`<-`/<950,500>[
        \J0`
        \J0`
        \J\H\perp`
        \J\H\perp;
        `
        \J\h{\perp}`
        \J\h{\perp}`]

      \square(0,1500)/=`<-`<-`/<950,500>[
        \perp`
        \perp`
        \J0`
        \J0;
        `
        \j{}`
        \j{}`]
    \end{diagram}
    Clearly, this diagram commutes.
  \end{itemize}
\end{itemize}
At this point we have shown that $\w{A} : \perp \mto \wn A$ and
$\c{A} : \wn A \oplus \wn A \mto \wn A$ are symmetric comonoidal
naturality transformations.  Now we show that for any $\wn A$ the
triple $(\wn A,\w{A},\c{A})$ forms a commutative monoid.  This means
that the following diagrams must commute:
\begin{itemize}
\item[Case.]\ \\
  \begin{diagram}
    \hSquares|aammmmm|/->`->```->``/[
      (\wn A \oplus \wn A) \oplus \wn A`
      \wn A \oplus (\wn A \oplus \wn A)`
      \wn A \oplus \wn A```
      \wn A;
      \alpha_{\wn A,\wn A,\wn A}`
      \id_{\wn A} \oplus \c{A}```
      \c{A}``]
    \btriangle|maa|/->``->/<2407,500>[(\wn A \oplus \wn A) \oplus \wn A`
      \wn A \oplus \wn A`
      \wn A;
      \c{A} \oplus \id_{A}``
      \c{A}]
  \end{diagram}
  The previous diagram commutes, because the following one does (we
  omit subscripts for readability):
  \begin{diagram}
    \scriptsize
    \square|amma|/->`->``->/<950,500>[
      (\func{JH}A \oplus \func{JH}A) \oplus \func{JH}A`
      \func{JH}A \oplus (\func{JH}A \oplus \func{JH}A)`
      \func{J}(\func{H}A + \func{H}A) \oplus \func{JH}A`
      \func{J}((\func{H}A + \func{H}A) + \func{H}A);
      \alpha`
      \jinv{} \oplus \id``
      \jinv{}]

    \square(950,0)|amma|/->``->`->/<950,500>[
      \func{JH}A \oplus (\func{JH}A \oplus \func{JH}A)`
      \func{JH}A \oplus \func{J}(\func{H}A + \func{H}A)`
      \func{J}((\func{H}A + \func{H}A) + \func{H}A)`
      \func{J}(\func{H}A + (\func{H}A + \func{H}A));
      \id \oplus \jinv{}``
      \jinv{}`
      \func{J}\alpha]

    \square(1900,0)|amma|/->``->`->/<950,500>[
      \func{JH}A \oplus \func{J}(\func{H}A + \func{H}A)`
      \func{JH}A \oplus \func{JH}A`
      \func{J}(\func{H}A + (\func{H}A + \func{H}A))`
      \func{J}(\func{H}A + \func{H}A);
      \id \oplus \func{J}\codiag{}``
      \jinv{}`
      \func{J}(\id + \codiag{})]

    \square(0,-500)|amma|/`->`->`->/<950,500>[
      \func{J}(\func{H}A + \func{H}A) \oplus \func{JH}A`
      \func{J}((\func{H}A + \func{H}A) + \func{H}A)`
      \func{JH}A \oplus \func{JH}A`
      \func{J}(\func{H}A + \func{H}A);`
      \func{J}\codiag{} \oplus \id`
      \func{J}(\codiag{} + \id)`
      \jinv{}]

    \dtriangle(950,-500)|ama|/`->`->/<1900,500>[
      \func{J}(\func{H}A + \func{H}A)`
      \func{J}(\func{H}A + \func{H}A)`
      \func{JH}A;`
      \func{J}\codiag{}`
      \func{J}\codiag{}]

    \place(950,250)[(1)]
    \place(2375,250)[(2)]
    \place(475,-250)[(3)]
    \place(1900,-250)[(4)]
  \end{diagram}
  Diagram 1 commutes because $\func{J}$ is a symmetric monoidal
  functor (Corollary~\ref{corollary:J-SMMF}), diagrams 2 and 3
  commute by naturality of $\jinv{}$, and diagram 4 commutes because
  $(\func{H}A, \diamond, \codiag{})$ is a commutative monoid in
  $\cat{C}$, but we leave the proof of this to the reader.

\item[Case.]\ \\
  \begin{diagram}
    \btriangle|maa|/->`->`->/<1000,600>[
      \wn A \oplus \perp`
      \wn A \oplus \wn A`
      \wn A;
      \id_{\wn A} \oplus \w{A}`
      \rho_{\wn A}`
      \c{A}]
  \end{diagram}
  The previous diagram commutes, because the following one does:
  \begin{diagram}
    \square|amma|/->`->`->`->/<950,500>[
      \func{JH}A \oplus \func{J}0`
      \func{J}(\func{H}A + 0)`
      \func{JH}A \oplus \func{JH}A`
      \func{J}(\func{H}A + \func{H}A);
      \jinv{}`
      \id \oplus \func{J}\diamond`
      \func{J}(\id \oplus \diamond)`
      \jinv{}]

    \square(950,0)|amma|/->``=`->/<950,500>[
      \func{J}(\func{H}A + 0)`
      \func{JH}A`
      \func{J}(\func{H}A + \func{H}A)`
      \func{JH}A;
      \func{J}\rho```
      \func{J}\codiag{}]

    \square(0,500)|amma|/->`->`=`/<1900,500>[
      \func{JH}A \oplus \perp`
      \func{JH}A`
      \func{JH}A \oplus \func{J}0`
      \func{JH}A;
      \rho`
      \id \oplus \jinv{0}``]

    \place(950,750)[(1)]
    \place(475,250)[(2)]
    \place(1425,250)[(3)]      
  \end{diagram}
  Diagram 1 commutes because $\func{J}$ is a symmetric monoidal
  functor (Corollary~\ref{corollary:J-SMMF}), diagram 2 commutes by
  naturality of $\jinv{}$, and diagram 3 commutes because
  $(\func{H}A, \diamond, \codiag{})$ is a commutative monoid in
  $\cat{C}$, but we leave the proof of this to the reader.
  
\item[Case.]\ \\
  \begin{diagram}
    \btriangle|maa|/->`->`->/<1000,600>[
      \wn A \oplus \wn A`
      \wn A \oplus \wn A`
      \wn A;
      \beta_{\wn A,\wn A}`
      \c{A}`
      \c{A}]
  \end{diagram}
  This diagram commutes, because the following one does:
  \begin{diagram}
    \hSquares/->`->`->`->`=`->`->/[
      \func{JH}A \oplus \func{JH}A`
      \func{J}(\func{H}A + \func{H}A)`
      \func{JH}A`
      \func{JH}A \oplus \func{JH}A`
      \func{J}(\func{H}A + \func{H}A)`
      \func{JH}A;
      \jinv{}`
      \func{J}\codiag{}`
      \beta`
      \func{J}\beta``
      \jinv{}`
      \func{J}\codiag{}]
  \end{diagram}
  The left diagram commutes by naturality of $\jinv{}$, and the right
  diagram commutes because $(\func{H}A, \diamond, \codiag{})$ is a
  commutative monoid in $\cat{C}$, but we leave the proof of this to
  the reader.
\end{itemize}

Finally, we must show that $\w{A} : \perp \mto \wn A$ and $\c{A} :
\wn A \oplus \wn A \mto \wn A$ are $\wn\text{-algebra}$ morphisms.
The algebras in play here are $(\wn A,\mu : \wn\wn A \mto \wn A)$,
$(\perp, \r{\perp} : \wn \perp \mto \perp)$, and $(\wn A \oplus \wn
A, u_A : \wn (\wn A \oplus \wn A) \mto \wn A \oplus \wn A)$, where
$u_A := \wn (\wn A \oplus \wn A) \mto^{\r{\wn A,\wn A}} ?^2 A \oplus
?^2 A \mto^{\mu_A \oplus \mu_A} \wn A \oplus \wn A$.  It suffices to
show that the following diagrams commute:
\begin{itemize}
\item[Case.]\ \\
  \begin{diagram}
    \square<950,500>[
      \wn \perp`
      \perp`
      \wn\wn A`
      \wn A;
      \r{\perp}`
      \wn\w{}`
      \w{}`
      \mu]
  \end{diagram}
  This diagram commutes, because the following fully expanded one does:
  \begin{diagram}
    \square|mmma|<2500,500>[
      \J\H\J 0`
      \J 0`
      \J\H\J\H A`
      \J\H A;
      \J\varepsilon_0`
      \J\H\J\diamond`
      \J\diamond`
      \J\varepsilon]

    \square(0,500)|amma|/->`->``/<1250,1000>[
      \J\H\perp`
      \J 0`
      \J\H\J 0`;
      \J\h{\perp}`
      \J\H\jinv{0}``]
    \square(1250,500)|amma|/->``->`/<1250,1000>[
      \J 0`
      \perp``
      \J 0;
      \j{0}``
      \jinv{0}`]

    \Dtriangle(0,500)|mmm|/`->`=/<1250,500>[
      \J\H \perp`
      \J\H\J 0`
      \J\H\J 0;`
      \J\H\jinv{0}`]

    \morphism(1250,1000)|m|<700,0>[
      \J\H\J 0`
      \J\H \perp;
      \J\H\j{0}]

    \morphism(1950,1000)|m|<550,-500>[
      \J\H \perp`
      \J 0;
      \J\h{\perp}]

    \place(1250,250)[(1)]
    \place(1500,750)[(2)]
    \place(380,1000)[(3)]
    \place(2000,1250)[(4)]
  \end{diagram}
  Diagram 1 commutes by naturality of $\varepsilon$, diagram 2
  commutes because $\varepsilon$ is the counit of the symmetric
  comonoidal adjunction, diagram 3 clearly commutes, and diagram 4
  commutes because $\j{0}$ is an isomorphism
  (Lemma~\ref{lemma:symmetric_comonoidal_isomorphisms}).
  
\item[Case.]\ \\
  \begin{diagram}
    \square|amma|<950,500>[
      \wn (\wn A \oplus \wn A)`
      \wn A \oplus \wn A`
      \wn\wn A`
      \wn A;
      u`
      \wn\c{}`
      \c{}`
      \mu]
  \end{diagram}
  This diagram commutes because the following fully expanded one does:
  \begin{center}
    \tiny
    \rotatebox{90}{$\bfig
      \square|amma|/->`->``/<1250,1500>[
        \J\H\J(\H A + \H A)`
        \J\H(\J\H A \oplus \J\H A)`
        \J\H\J\H A`;
        \J\H\j{}`
        \J\H\J\codiag{}``]
      \square(1250,0)|amma|/->``->`/<1250,1500>[
        \J\H(\J\H A \oplus \J\H A)`
        \J(\H\J\H A + \H\J\H A)``
        \J\H\J\H A;
        \J\h{}``
        \J\codiag{}`]
      \morphism/=/<2500,0>[\J\H\J\H A`\J\H\J\H A;]

      \Atriangle/<-`<-`=/<1250,500>[
        \J\H A`
        \J\H\J\H A`
        \J\H\J\H A;
        \J\varepsilon`
        \J\varepsilon`]

      \Vtriangle(0,1000)/`->`->/<1250,500>[
        \J\H\J(\H A + \H A)`
        \J(\H\J\H A + \H\J\H A)`
        \J(\H A + \H A);`
        \J\varepsilon`
        \J(\varepsilon + \varepsilon)]

      \morphism(1250,1000)|m|<0,-500>[
        \J(\H A + \H A)`
        \J\H A;
        \J\codiag{}]

      \square(0,1500)|amma|/->`->``/<1250,500>[
        \J\H(\J\H A \oplus \J\H A)`
        \J(\H\J\H A + \H\J\H A)`
        \J\H\J(\H A + \H A)`;
        \J\h{}`
        \J\H\jinv{}``]        

      \square(1250,1500)|amma|/->``->`/<1250,500>[
        \J(\H\J\H A + \H\J\H A)`
        \J\H\J\H A \oplus \J\H\J\H A`
        \J\H(\J\H A \oplus \J\H A)`
        \J(\H\J\H A + \H\J\H A);
        \j{}`
        `
        \jinv{}`]

      \square(2500,0)|amma|<1250,1500>[
        \J(\H\J\H A + \H\J\H A)`
        \J(\H A + \H A)`
        \J\H\J\H A`
        \J\H A;
        \J(\varepsilon + \varepsilon)`
        \J\codiag{}`
        \J\codiag{}`
        \J\varepsilon]

      \square(2500,1500)|amma|<1250,500>[
        \J\H\J\H A \oplus \J\H\J\H A`
        \J\H A \oplus \J\H A`
        \J(\H\J\H A + \H\J\H A)`
        \J(\H A + \H A);
        \J\varepsilon \oplus \J\varepsilon`
        \jinv{}`
        \jinv{}`
        \J(\varepsilon + \varepsilon)]

      \place(1250,250)[(1)]
      \place(625,750)[(2)]
      \place(1875,750)[(3)]
      \place(1250,1250)[(4)]
      \place(1250,1750)[(5)]
      \place(3125,1750)[(6)]
      \place(3125,750)[(7)]
      \efig$}
  \end{center}
\end{itemize}
Diagram 1 clearly commutes, diagram 2 commutes by naturality of
$\varepsilon$, diagram 3 commutes by naturality of $\codiag{}$,
diagram 4 commutes because $\varepsilon$ is the counit of the
symmetric comonoidal adjunction, diagram 5 commutes because $\j{}$
is an isomorphism
(Lemma~\ref{lemma:symmetric_comonoidal_isomorphisms}), diagram 6
commutes by naturality of $\jinv{}$, and diagram 7 is the same
diagram as 3, but this diagram is redundant for readability.

\subsection{Proof of Lemma~\ref{lemma:monoid-morphism}}
\label{sec:proof_of_lemma:monoid-morphism}
Suppose $\cat{L} : \func{H} \dashv \func{J} : \cat{C}$ is a dual LNL
model.  Then we know $\wn A = \J\H A$ is a symmetric comonoidal
monad by Lemma~\ref{lemma:symmetric_comonoidal_monad}.  Bellin
\cite{Bellin:2012} remarks that by Maietti, Maneggia de Paiva and
Ritter's Proposition~25 \cite{Maietti2005}, it suffices to show that
$\mu_A : \wn\wn A \mto \wn A$ is a monoid morphism.  Thus, the
following diagrams must commute:
\begin{itemize}
\item[Case.]\ \\
  \begin{diagram}
    \square|amma|<950,500>[
      \wn\wn A \oplus \wn\wn A`
      \wn\wn A`
      \wn A \oplus \wn A`
      \wn A;
      \c{\wn A}`
      \mu_A \oplus \mu_A`
      \mu_A`
      \c{A}]
  \end{diagram}
  This diagram commutes because the following fully expanded one
  does:
  \begin{diagram}
    \square|amma|<1000,500>[
      \J\H\J\H A \oplus \J\H\J\H A`
      \J(\H\J\H A + \H\J\H A)`
      \J\H A \oplus \J\H A`
      \J(\H A + \H A);
      \jinv{}`
      \J\varepsilon \oplus \J\varepsilon`
      \J(\varepsilon + \varepsilon)`
      \jinv{}]

    \square(1000,0)|amma|<1000,500>[
      \J(\H\J\H A + \H\J\H A)`
      \J\H\J\H A`
      \J(\H A + \H A)`
      \J\H A;
      \J\codiag{}`
      \J(\varepsilon + \varepsilon)`
      \J\varepsilon`
      \J\codiag{}]      
  \end{diagram}
  The left square commutes by naturality of $\jinv{}$ and the right
  square commutes by naturality of the codiagonal.
  
\item[Case.]\ \\
  \begin{diagram}
    \Atriangle|aaa|[
      \perp`
      \wn\wn A`
      \wn A;
      \w{\wn A}`
      \w{A}`
      \mu_A]
  \end{diagram}
  This diagram commutes because the following fully expanded one
  does:
  \begin{diagram}
    \square|amma|/=`->`->`=/<1000,500>[
      \perp`
      \perp`
      \J 0`
      \J 0;`
      \jinv{0}`
      \jinv{0}`]

    \square(0,-500)|amma|/=`->`->`->/<1000,500>[
      \J 0`
      \J 0`
      \J\H\J\H A`
      \J\H A;`
      \J\diamond`
      \J\diamond`
      \J\varepsilon]
  \end{diagram}
  The top square trivially commutes, and the bottom square commutes
  by uniqueness of the initial map.
\end{itemize}

\subsection{Proof of Cut Reduction (Lemma~\ref{lemma:cut_reduction})}
\label{subsec:proof_of_cut_reduction_lemma:cut-reduction}
By induction on $d(\Pi_1) + d(\Pi_2)$.  We consider only the case
where the last inferences of $\Pi_1$ and $\Pi_2$ are logical
inferences. The other cases are handled mainly by permutation of
inferences and use of the inductive hypothesis; we refer to Benton's
text for them.  Throughout the proof we will add an asterisk to the
name of an inference rule to indicate that the rule may be applied
zero or more times.

\vspace{1ex}

 \noindent
\emph{J right / J left}. We have 
\begin{center}
\AxiomC{$\pi_1$}
\noLine
\UnaryInfC{$A \vdash_{\mathsf{L}} \Delta, J T^n ; T, \Psi$}
\LeftLabel{$\Pi_1 =$} 
\RightLabel{$\DualLNLLogicdruleLXXjRName$}
\UnaryInfC{$A \vdash_{\mathsf{L}} \Delta, J T^{n+1}; \Psi$}
\DisplayProof
\qquad 
\AxiomC{$\pi_2$}
\noLine
\UnaryInfC{$T \vdash_{\mathsf{C}} \Psi'$}
\LeftLabel{$\Pi_2 =$} 
\RightLabel{$\DualLNLLogicdruleLXXjLName$}
\UnaryInfC{$JT \vdash_{\mathsf{L}} ; \Delta, J T^{n+1}, \Psi'$}
\DisplayProof
\end{center}
By the inductive hypothesis appled to $\Pi_2$ and $\pi_1$ there exists a proof $\Pi'$ of $A \vdash_{\mathsf{L}} \Delta; T, \Psi, \Psi'$
with $c(\Pi') \leq |J T| = |T | + 1$.  Then the following derivation 
\begin{center}
\AxiomC{$\Pi'$}
\noLine
\UnaryInfC{$A \vdash_{\mathsf{L}} \Delta ; T, \Psi, \Psi'$}
\AxiomC{$\pi_2$}
\noLine
\UnaryInfC{$T \vdash_{\mathsf{C}} \Psi'$}
\LeftLabel{$\Pi_1 =$} 
\RightLabel{$\DualLNLLogicdruleLXXCcutName$}
\BinaryInfC{$A \vdash_{\mathsf{L}} \Delta, \Psi, \Psi', \Psi'$}
\doubleLine
\RightLabel{$\DualLNLLogicdruleCXXcrName^*$}
\UnaryInfC{$A \vdash_{\mathsf{L}} \Delta, \Psi, \Psi'$}
\DisplayProof
\end{center}
has cut rank $\mathit{max}( |T|+1, c(\Pi'), c(\pi_2)) = |T|+1 = |J T|$. 

\vspace{1ex}

 \noindent
\emph{H right / H left}. We have 
\begin{center}
\AxiomC{$\pi_1$}
\noLine
\UnaryInfC{$B \vdash_{\mathsf{L}} \Delta, A; H A^n ; \Psi$}
\LeftLabel{$\Pi_1 =$} 
\RightLabel{$\DualLNLLogicdruleLXXhRName$}
\UnaryInfC{$B \vdash_{\mathsf{L}} \Delta; H A^{n+1}, \Psi$}
\DisplayProof
\qquad 
\AxiomC{$\pi_2$}
\noLine
\UnaryInfC{$A \vdash_{\mathsf{L}} ; \Psi'$}
\LeftLabel{$\Pi_2 =$} 
\RightLabel{$\DualLNLLogicdruleCXXhLName$}
\UnaryInfC{$H A \vdash_{\mathsf{C}} ; \Psi'$}
\DisplayProof
\end{center}
By the inductive hypothesis applied to $\Pi_2$ and $\pi_1$ there exists a proof $\Pi'$ of $B \vdash_{\mathsf{L}} \Delta; A, \Psi, \Psi'$
with $c(\Pi') \leq |H A| = |A | + 1$.  Then the following derivation 
\begin{center}
\AxiomC{$\Pi'$}
\noLine
\UnaryInfC{$B \vdash_{\mathsf{L}} \Delta ; A, \Psi, \Psi'$}
\AxiomC{$\pi_2$}
\noLine
\UnaryInfC{$A \vdash_{\mathsf{L}} ; \Psi'$}
\LeftLabel{$\Pi_1 =$} 
\RightLabel{$\DualLNLLogicdruleLXXcutName$}
\BinaryInfC{$B \vdash_{\mathsf{L}} \Delta, \Psi, \Psi', \Psi'$}
\doubleLine
\RightLabel{$\DualLNLLogicdruleCXXcrName^*$}
\UnaryInfC{$B \vdash_{\mathsf{L}} \Delta, \Psi, \Psi'$}
\DisplayProof
\end{center}
has cut rank $\mathit{max}( |A|+1, c(\Pi'), c(\pi_2)) = |A|+1 = |H A|$. 

\vspace{1ex}

 \noindent
\emph{+ right$_1$ / + left}.  We have 
\begin{center}
\AxiomC{$\pi_1$}
\noLine
\UnaryInfC{$S \vdash_{\mathsf{C}} T_1, (T_1 + T_2)^{n}, \Psi$}
\LeftLabel{$\Pi_1 =$} 
\RightLabel{$\DualLNLLogicdruleCXXdROneName$}
\UnaryInfC{$S \vdash_{\mathsf{C}} (T_1 + T_2)^{n+1}, \Psi$}
\DisplayProof
\qquad 
\AxiomC{$\pi_2$}
\noLine
\UnaryInfC{$T_1 \vdash_{\mathsf{C}} \Psi_1$}
\AxiomC{$\pi_3$}
\noLine
\UnaryInfC{$T_2 \vdash_{\mathsf{C}} \Psi_2$}
\LeftLabel{$\Pi_2 =$} 
\RightLabel{$\DualLNLLogicdruleCXXdLName$}
\BinaryInfC{$T_1+ T_2 \vdash_{\mathsf{C}} ; \Psi_1, \Psi_2$}
\DisplayProof
\end{center}

\noindent
If $n = 0$, then the reduction is as follows:
\begin{center}
\begin{tabular}{c}
\AxiomC{$\pi_1$}
\noLine
\UnaryInfC{$S \vdash_{\mathsf{C}} T_1, \Psi$}
\LeftLabel{$\Pi_1 =$} 
\RightLabel{$\DualLNLLogicdruleCXXdROneName$}
\UnaryInfC{$S \vdash_{\mathsf{C}} T_1 + T_2, \Psi$}
\AxiomC{$\pi_2$}
\noLine
\UnaryInfC{$T_1 \vdash_{\mathsf{C}} \Psi_1$}
\AxiomC{$\pi_3$}
\noLine
\UnaryInfC{$T_2 \vdash_{\mathsf{C}} \Psi_2$}
\LeftLabel{$\Pi_2 =$} 
\RightLabel{$\DualLNLLogicdruleCXXdLName$}
\BinaryInfC{$T_1+ T_2 \vdash_{\mathsf{C}}  \Psi_1, \Psi_2$}
\RightLabel{$\DualLNLLogicdruleCXXcutName$}
\BinaryInfC{$S\vdash_{\mathsf{C}} \Psi, \Psi_1, \Psi_2$}
\DisplayProof\\
\\
 reduces to \\ 
\\
\AxiomC{$\pi_1$}
\noLine
\UnaryInfC{$S \vdash_{\mathsf{C}} T_1, \Psi$}
\AxiomC{$\pi_2$}
\noLine
\UnaryInfC{$T_1 \vdash_{\mathsf{C}} \Psi_1$}
\RightLabel{$\DualLNLLogicdruleCXXcutName$}
\BinaryInfC{$S \vdash_{\mathsf{C}} \Psi, \Psi_1$}
\doubleLine
\LeftLabel{$\Pi$}
\RightLabel{$\DualLNLLogicdruleCXXwkName^*$}
\UnaryInfC{$S\vdash_{\mathsf{C}} \Psi, \Psi_1, \Psi_2$}
\DisplayProof
\end{tabular}
\end{center}

\noindent
Here $c(\Pi) = max(|T_1 + 1|, c(\pi_1), c(\pi_2)) \leq |T_1 + T_2|$.

\vspace{1ex}
 
\noindent
If $n > 0$, then by the inductive hypothesis applied to $\Pi_2$ and $\pi_1$ there exists a proof $\Pi'$ of 
$S \vdash_{\mathsf{C}} T_1, \Psi, \Psi_1, \Psi_2$ with $c(\Pi') \leq |T_1+ T_2| = |T_1|+|T_2 | + 1$.  Then the following derivation 
\begin{center}
\AxiomC{$\Pi' $}
\noLine
\UnaryInfC{$S \vdash_{\mathsf{C}} T_1, \Psi, \Psi_1, \Psi_2$}
\AxiomC{$\pi_2$}
\noLine
\UnaryInfC{$T_1 \vdash_{\mathsf{C}}  \Psi_1$}
\LeftLabel{$\Pi =$} 
\RightLabel{$\DualLNLLogicdruleCXXcutName$}
\BinaryInfC{$S \vdash_{\mathsf{C}} \Psi, \Psi_1, \Psi_1, \Psi_2$}
\doubleLine
\RightLabel{$\DualLNLLogicdruleCXXcrName^*$}
\UnaryInfC{$S \vdash_{\mathsf{C}} \Psi, \Psi_1, \Psi_2$}
\DisplayProof
\end{center}
has cut rank $\mathit{max}( |T_1|+1, c(\Pi'), c(\pi_2)) \leq |T_1 + T_2|$. 

\vspace{1ex}

\noindent
$\lsub$ right / $\lsub$ left. We have 
\begin{center}
\begin{tabular}{c}
\AxiomC{$\pi_1$ }
\noLine
\UnaryInfC{$A \vdash_{\mathsf{L}} \Delta_1; \Psi_1, B_1$}
 \AxiomC{$\pi_2$ } 
\noLine
\UnaryInfC{$B_2 \vdash_{\mathsf{L}} \Delta_2; \Psi_2$}
\LeftLabel{$\Pi_1 =$}
\RightLabel{$\DualLNLLogicdruleLXXsRName$}
\BinaryInfC{$A \vdash_{\mathsf{L}} B_1 \lsub B_2,  \Delta_1, \Delta_2; \Psi_1, \Psi_2$}
\AxiomC{$\pi_3$}
\noLine
\UnaryInfC{$B_1 \vdash_{\mathsf{L}} B_2, \Delta ;  \Psi$}
\LeftLabel{$\Pi_2 =$}
\RightLabel{$\DualLNLLogicdruleLXXsLName$}
\UnaryInfC{$B_1 \lsub B_2 \vdash_{\mathsf{L}} \Delta ; \Psi$}
\RightLabel{$\DualLNLLogicdruleLXXcutName$}
\BinaryInfC{$A \vdash_{\mathsf{L}} \Delta_1, \Delta_2, \Delta ; \Psi_1, \Psi_2, \Psi$}
\DisplayProof\\
\\
reduces to $\Pi$ 
\\
\\
\AxiomC{$\pi_1$}
\noLine
\UnaryInfC{$A\vdash_{\mathsf{L}} \Delta_1, B_1; \Psi_1$}
 \AxiomC{$\pi_3$}
\noLine
\UnaryInfC{$B_1 \vdash_{\mathsf{L}} B_2, \Delta ;  \Psi$}
\RightLabel{$\DualLNLLogicdruleLXXcutName$}
\BinaryInfC{$A\vdash_{\mathsf{L}} \Delta_1, \Delta, B_2; \Psi_1,\Psi$}
 \AxiomC{$\pi_2$ } 
\noLine
\UnaryInfC{$B_2 \vdash_{\mathsf{L}} \Delta_2; \Psi_2$}
\RightLabel{$\DualLNLLogicdruleLXXcutName$}
\BinaryInfC{$A \vdash_{\mathsf{L}} \Delta_1, \Delta_2, \Delta; \Psi_1, \Psi_2, \Psi$}
\DisplayProof
\end{tabular}
\end{center}
The resulting derivation $\Pi$ has cut rank $c(\Pi) = max(|B_1|+1, c(\pi_1), c(\pi_2), |B_2|+1, c(\pi_3)) \leq |B_1\lsub B_2|$.  





\end{document}